\documentclass[12pt,english]{article}
\usepackage[T1]{fontenc}
\usepackage[utf8]{inputenc}
\usepackage{color}
\usepackage{babel}
\usepackage{prettyref}
\usepackage{url}
\usepackage{enumitem}
\usepackage{bm}
\usepackage{amsmath}
\usepackage{amsthm}
\usepackage{amssymb}
\usepackage{graphicx}
\usepackage{geometry}
\usepackage{setspace}
\usepackage[authoryear]{natbib}
\usepackage{microtype}
\usepackage[pdfusetitle,
 bookmarks=true,bookmarksnumbered=false,bookmarksopen=false,
 breaklinks=false,pdfborder={0 0 0},pdfborderstyle={},backref=false,colorlinks=true]
 {hyperref}
\hypersetup{
 linkcolor=magenta, urlcolor=cyan, citecolor=blue}

\makeatletter
\theoremstyle{plain}
\newtheorem{assumption}{\protect\assumptionname}
\theoremstyle{remark}
\newtheorem{rem}{\protect\remarkname}
\theoremstyle{plain}
\newtheorem{lyxalgorithm}{\protect\algorithmname}
\theoremstyle{plain}
\newtheorem{prop}{\protect\propositionname}
\theoremstyle{plain}
\newtheorem{thm}{\protect\theoremname}
\theoremstyle{plain}
\newtheorem{cor}{\protect\corollaryname}
\theoremstyle{plain}
\newtheorem{assumptionprimepinner}{Assumption}
\newenvironment{assumptionprimep}[2][]{%
  \renewcommand\theassumptionprimepinner{#2\textsuperscript{$\mkern2mu\prime$}}%
  \def\temp{#1}\ifx\temp\empty%
  \begin{assumptionprimepinner}
  \else%
  \begin{assumptionprimepinner}[#1]
  \fi%
}{\end{assumptionprimepinner}}
\theoremstyle{plain}
\newtheorem{lem}{\protect\lemmaname}
\newcommand{\raisedtarget}[1]{%
  \raisebox{\baselineskip}[0pt][0pt]{\hypertarget{#1}{}}%
}

\newenvironment{delayedproof}[1]
  {\begin{proof}[\raisedtarget{#1}Proof of \prettyref{#1}]}
  {\end{proof}}

\date{}

\usepackage{babel}

\usepackage{microtype}

\usepackage{eucal}
\usepackage{relsize}
\usepackage{bm}

\usepackage[bottom]{footmisc}
\usepackage{mleftright}
\mleftright

\usepackage{enumitem}
\setlist[enumerate]{label=(\roman*),itemsep=0pt}

\usepackage{multirow}
\usepackage{longtable}
\usepackage{rotating}
\usepackage{makecell}
\usepackage{threeparttable}
\usepackage[table]{xcolor}

\usepackage{changepage}  % adjustwidth

\usepackage{siunitx}
\usepackage{booktabs}
\usepackage{caption}
\def\thm@space@setup{%
  \thm@preskip=12pt plus 4pt minus 3pt
  \thm@postskip=1.05\thm@preskip % or whatever, if you don't want them to be equal
}

\def\th@remark{%
  \thm@headfont{\bfseries}
  \itshape
  \thm@preskip=12pt plus 4pt minus 3pt
  \thm@postskip\thm@preskip
}

\renewenvironment{proof}[1][\proofname]{\par
    \pushQED{\qed}%
    \normalfont \topsep12\p@\@plus4\p@\@minus3\p@\relax
    \trivlist
    \item\relax
          {\itshape
      #1\@addpunct{.}}\hspace\labelsep\ignorespaces
  }{%
    \popQED\endtrivlist\@endpefalse
    \addvspace{13pt plus 4pt minus 3pt}
  }

\DefineFNsymbols*{noasterisk}{
   {\TextOrMath \textdagger \dagger}%
   {\TextOrMath \textdaggerdbl \ddagger}%
   {\TextOrMath \textsection  \mathsection}%
   {\TextOrMath \textparagraph \mathparagraph}%
   {\TextOrMath \textbardbl \|}%
   {\TextOrMath {\textdagger\textdagger}{\dagger\dagger}}%
   {\TextOrMath {\textdaggerdbl\textdaggerdbl}{\ddagger\ddagger}}%
}
\setfnsymbol{noasterisk}

\DeclareFontFamily{U}{mathx}{}
\DeclareFontShape{U}{mathx}{m}{n}{<-> mathx10}{}
\DeclareSymbolFont{mathx}{U}{mathx}{m}{n}
\DeclareMathAccent{\widehat}{0}{mathx}{"70}
\DeclareMathAccent{\widecheck}{0}{mathx}{"71}

\let\smallcdot\cdot
\DeclareRobustCommand*{\bigcdot}{%
  \mathbin{\mathpalette\bigcdot@{}}%
}
\newcommand*{\bigcdot@scalefactor}{.5}
\newcommand*{\bigcdot@widthfactor}{1.15}
\newcommand*{\bigcdot@}[2]{%
  \sbox0{$#1\vcenter{}$}% math axis
  \sbox2{$#1\smallcdot\m@th$}%
  \hbox to \bigcdot@widthfactor\wd2{%
    \hfil
    \raise\ht0\hbox{%
      \scalebox{\bigcdot@scalefactor}{%
        \lower\ht0\hbox{$#1\bullet\m@th$}%
      }%
    }%
    \hfil
  }%
}

\let\cdot\bigcdot

\let\Gamma\varGamma
\let\Delta\varDelta
\let\Theta\varTheta
\let\Lambda\varLambda
\let\Xi\varXi
\let\Pi\varPi
\let\Sigma\varSigma
\let\Upsilon\varUpsilon
\let\Phi\varPhi
\let\Psi\varPsi
\let\Omega\varOmega

\let\epsilon\varepsilon

\let\hat\widehat
\let\tilde\widetilde
\allowdisplaybreaks

\g@addto@macro\normalsize{%
 \abovedisplayskip=7pt plus 1pt minus 2pt
 \abovedisplayshortskip=5.5pt plus 1pt minus 2pt
 \belowdisplayskip=7pt plus 1pt minus 2pt
 \belowdisplayshortskip=6.5pt plus 1pt minus 2pt
}{}{}

\providecommand{\corollaryname}{Corollary}

\providecommand{\lemmaname}{Lemma}
\providecommand{\theoremname}{Theorem}
\providecommand{\propositionname}{Proposition}
\providecommand{\assumptionname}{Assumption}

\newrefformat{assump}{Assumption~\ref{#1}}
\newrefformat{lem}{Lemma~\ref{#1}}
\newrefformat{cor}{Corollary~\ref{#1}}
\newrefformat{thm}{Theorem~\ref{#1}}
\newrefformat{prop}{Proposition~\ref{#1}}
\newrefformat{rem}{Remark~\ref{#1}}

\newcommand{\proofref}[2][]{%
  \hyperlink{#2}{%
    \if\relax\detokenize{#1}\relax
      proof of \prettyref{#2}%
    \else
      #1%
    \fi
  }%
}

\usepackage{apptools}
\AtAppendix{
  \counterwithin{fact}{section}
  \counterwithin{rem}{section}
  \counterwithin{table}{section}
}
\AtAppendix{\counterwithin{lem}{section}}
\AtAppendix{\counterwithin{prop}{section}}
\AtAppendix{\counterwithin{cor}{section}}
\AtAppendix{\counterwithin{footnote}{section}}
\AtAppendix{
  \counterwithin{equation}{section}
  \renewcommand{\theequation}{\thesection-\arabic{equation}}
}

\def\td#1{\tilde{#1}\sbcorr{0mu}{1.5mu}}

\def\sbcorr#1#2{%
  \def\tmpa{#1}%
  \def\tmpb{#2}%
  \futurelet\next\sbcorrA%
}
\def\sbcorrA{%
  \ifx\next_%
    \expandafter\sbcorrB%
  \else%
    \expandafter\sbcorrC%
  \fi%
}
\def\sbcorrB_#1{%
  _{\mkern\tmpa#1}%
  \futurelet\next\sbcorrC%
}
\def\sbcorrC{%
  \ifx\next^%
    \expandafter\sbcorrD%
  \else%
    \fi%
}
\def\sbcorrD^#1{%
  ^{\mkern\tmpb#1}%
}

\newcommand{\ncal}{\mathcal{N}}

\newcommand{\xtd}{\tilde{x}}

\newcommand{\fe}{{\mathrm{WG}}}

\newcommand{\ivx}{{\mathrm{IVX}}}

\newcommand{\epct}{\mathbb{E}}

\newcommand{\bhatfe}{\widehat{\beta}^{\mkern2mu\fe}}

\newcommand{\bhativx}{\widehat{\beta}^{\mkern2mu\ivx}}

\newcommand{\rohatfe}{\widehat{\rho}^{\mkern2mu\fe}}

\newcommand{\joto}{(n,T)\to\infty}
\newcommand{\pto}{\to_p}
\newcommand{\dto}{\to_d}

\let\math@org=$
\def\itinlinemath#1{%
  \math@org%
  \mkern+1mu\relax%
  #1%
  \mkern-1.2mu\relax%
  \math@org%
}

\begingroup
  \catcode`\$=13
  \gdef\activateitalicmath{%
    \catcode`\$=13%
    \def${\math@org}%
    \def$##1${\itinlinemath{##1}}%
  }
\endgroup

\AtBeginEnvironment{assumption}{\activateitalicmath}
\AtBeginEnvironment{assumptionprime}{\activateitalicmath}
\AtBeginEnvironment{assumptionprimep}{\activateitalicmath}
\AtBeginEnvironment{assumptionppp}{\activateitalicmath}
\AtBeginEnvironment{lem}{\activateitalicmath}
\AtBeginEnvironment{thm}{\activateitalicmath}
\AtBeginEnvironment{cor}{\activateitalicmath}
\AtBeginEnvironment{prop}{\activateitalicmath}
\AtBeginEnvironment{lyxalgorithm}{\activateitalicmath}
\AtBeginEnvironment{rem}{\activateitalicmath}

\usepackage{upref}
\usepackage{xparse}
\usepackage{letltxmacro}

\AtBeginDocument{%
  \LetLtxMacro{\origcitep}{\citep}
  \RenewDocumentCommand{\citep}{o o m}{%
    \textup{%
      \IfNoValueTF{#1}{%
        \IfNoValueTF{#2}{\origcitep{#3}}{\origcitep[][#2]{#3}}%
      }{%
        \IfNoValueTF{#2}{\origcitep[#1]{#3}}{\origcitep[#1][#2]{#3}}%
      }%
    }%
  }%

  \LetLtxMacro{\origcitet}{\citet}
  \RenewDocumentCommand{\citet}{o o m}{%
    \textup{%
      \IfNoValueTF{#1}{%
        \IfNoValueTF{#2}{\origcitet{#3}}{\origcitet[][#2]{#3}}%
      }{%
        \IfNoValueTF{#2}{\origcitet[#1]{#3}}{\origcitet[#1][#2]{#3}}%
      }%
    }%
  }%

  \LetLtxMacro{\origcite}{\cite}
  \RenewDocumentCommand{\cite}{o o m}{%
    \textup{%
      \IfNoValueTF{#1}{%
        \IfNoValueTF{#2}{\origcite{#3}}{\origcite[][#2]{#3}}%
      }{%
        \IfNoValueTF{#2}{\origcite[#1]{#3}}{\origcite[#1][#2]{#3}}%
      }%
    }%
  }%
}

\makeatother

\providecommand{\algorithmname}{Algorithm}
\providecommand{\assumptionname}{Assumption}
\providecommand{\corollaryname}{Corollary}
\providecommand{\lemmaname}{Lemma}
\providecommand{\propositionname}{Proposition}
\providecommand{\remarkname}{Remark}
\providecommand{\theoremname}{Theorem}

\begin{document}
\title{Nickell Meets Stambaugh: A Tale of \\
Two Biases in Panel Predictive Regressions\thanks{\protect\setstretch{1.14}Chengwang Liao (\protect\url{lchw@link.cuhk.edu.hk});
Ziwei Mei (\protect\url{ziweimei@um.edu.mo}); Zhentao Shi (\protect\url{zhentao.shi@cuhk.edu.hk},
corresponding author). Address: 9/F Esther Lee Building, Department
of Economics, The Chinese University of Hong Kong, Shatin, New Territories,
Hong Kong SAR, China; Tel: (+852) 3943\,1432. Shi acknowledges the
partial financial support from the National Natural Science Foundation
of China (No.72425007) and the Direct Fund of the Social Science Faculty
of the Chinese University of Hong Kong (No.4052328). We thank James
Duffy, Erik Hjalmarsson, Anastasios Magdalinos, Peter C.B.~Phillips,
and Martin Weidner for helpful comments. We thank the Managing Editor,
Associate Editor, and four anonymous referees for their constructive
reviews that lead to substantial improvement and enrichment of this
paper. We thank Ji Pan and Menjie Shi for excellent research assistance.}}
\author{Chengwang Liao$^{a}$, Ziwei Mei$^{b}$, Zhentao Shi$^{a}$ \\
 $^{a}$The Chinese University of Hong Kong\\
$^{b}$University of Macau}

\maketitle
\bigskip{}

\begin{abstract}
In panel predictive regressions with persistent covariates, coexistence
of the Nickell bias and the Stambaugh bias imposes challenges for
estimation and hypothesis testing. This paper introduces an innovative
estimator, the Double IVX (DIVX), inspired by the IVX technique in
time series. DIVX effectively removes this composite Nickell-Stambaugh
bias and reinstates standard inferential procedures based on the $t$-statistic.
This new procedure achieves unified inference across a wide range
of modes of persistence in panel predictive regressions when the cross-sectional
dimension and the time dimension are comparably large. Such desirable
properties were unattainable by existing methods, including the popular
within-group estimator. Extensive Monte Carlo simulations demonstrate
the robustness of DIVX under a variety of settings. We apply DIVX
to panel data of financial markets in developed economies to examine
the predictability of stock returns.

\vspace{0.8cm}


\end{abstract}
\clearpage{}

\section{Introduction\label{sec:Introduction}}

Prediction has been one of the fundamental tasks of econometrics.
In time series, the least squares (LS) for a linear predictive regression
context is known to incur the \emph{Stambaugh bias} \citep{stambaugh1999predictive}
in finite samples. The Stambaugh bias is particularly severe when
the regressor is persistent. For example, in the autoregression of
order one (AR(1)) form $x_{t+1}=\mu_{x}+\rho^{*}x_{t}+v_{t+1}$ with
the AR parameter $\rho^{*}$ close to 1, the bias becomes a first-order
issue in that it will substantially distort the size of the conventional
testing procedure based on the asymptotic standard normal distribution
$\mathcal{N}(0,1)$ of the $t$-statistic.

With the advent of rich economic datasets covering cross sections
of countries and states, the time series predictive regression has
been introduced into panel data. The empirical literature has applied
panel predictive regressions to infer the predictability of various
economic and financial indicators using cross-country panel data,
including global stock returns \citep{hjalmarsson2010predicting,westerlund2017testing,davis2022leverage},
financial crises \citep{greenwood2022predictable,krishnamurthy2025credit},
and country cash flows \citep{gala2023global}. Panel regression not
only improves efficiency of estimation and inference by increasing
the sample size, but it also allows for fixed effects (FE) to capture
unobservable individual-specific heterogeneity.

This paper highlights the bias of commonly used estimators in panel
predictive regressions with potentially persistent regressors and
proposes a novel solution. Despite widespread applications, standard
approaches for panel predictive regressions, including the within-group
(WG) estimator, are prone to a non-vanishing bias known as the \emph{Nickell
bias} \citep{nickell1981biases}. Under a stationary time dimension,
rigorous theoretical development is already challenging when the number
of cross-sectional units $n$ and the number of time periods $T$
are comparably large, not to mention the difficulties brought about
by highly persistent regressors.

To fix ideas, we consider a target variable $y_{i,t+1}$ generated
by the following linear model
\begin{align}
y_{i,t+1} & =\mu_{y,i}+\beta^{*}x_{i,t}+e_{i,t+1},\quad\text{for }i=1,\dots,n\text{ and }t=1,\dots,T-1,\label{eq:predictive}
\end{align}
where $x_{i,t}$ follows an AR(1) process
\begin{equation}
x_{i,t+1}=\mu_{x,i}+\rho^{*}x_{i,t}+v_{i,t+1}.\label{eq:AR1}
\end{equation}
The slope coefficient $\beta^{*}$ is the parameter of key interest,
for it measures the predictability of $y_{i,t+1}$ using $x_{i,t}$.
The AR(1) coefficient $\rho^{*}$ signifies the persistence of $x_{i,t}$.
When $\rho^{*}\in(-1,1)$ is bounded away from unity, the vector $\left(y_{i,t},x_{i,t}\right)$
is stationary over time if the innovation vector $\left(e_{i,t},v_{i,t}\right)$
is stationary. This is a two-equation panel VAR system studied by
\citet{holtz1988estimating}, and the presence of the Nickell bias
in the WG estimator and the analytical bias correction have been explored
by \citet{hahn2002asymptotically} under the ``large-$n$-large-$T$''
asymptotics.

A widely used approach to characterizing the effect of a close-to-unity
$\rho^{*}$ is modeling it as a deterministic function of the sample
size $T$ in the form
\begin{equation}
\rho^{*}=\rho_{T}^{*}=1+c^{*}/T^{\gamma},\label{eq:rate rho}
\end{equation}
where $c^{*}\in\mathbb{R}$ and $\gamma\in[0,1]$ are fixed constants,
and the subscript $T$ is suppressed in $\rho^{*}$ when there is
no ambiguity. Note that this setup also includes the familiar \emph{stationary}
case (Case I) when $\gamma=0$ and $c^{*}\in(-2,0)$, under which
$\rho^{*}\in(-1,1)$ is a constant. As an asymptotic device, the representation
in \eqref{eq:rate rho} accommodates a persistent $x_{i,t}$, where
``persistent'' means $\rho^{*}\to1$ as $T\to\infty$, which is
the source of non-trivial Stambaugh bias. The following modes of persistence
are covered by \eqref{eq:rate rho}: Case II---\emph{mildly integrated}
(MI, $c^{*}<0$ and $\gamma\in\left(0,1\right)$); Case III---\emph{locally
integrated} (LI, $c^{*}<0$ and $\gamma=1$); Case IV---\emph{unit
root} (UR, $c^{*}=0$ and $\gamma=1$); and Case V---\emph{locally
explosive }(LE, $c^{*}>0$ and $\gamma=1$). For convenience, we wrap
Cases III-V into the category of \emph{local unit root} (LUR, $c^{*}\in\mathbb{R}$
and $\gamma=1$).\footnote{The persistence in panel data is indexed by $T$, without loss of
generality. If we write the AR(1) coefficient by a double-index sequence
$\rho_{n,T}^{*}=1+c^{*}/R_{n,T}$, where $R_{n,T}\to\infty$ as $(n,T)\to\infty$,
only the divergence speed of $R_{n,T}$ affects the asymptotics. For
example, if $\sqrt{T}/R_{n,T}\to R\in(0,\infty)$, then $x_{i,t}$
behaves asymptotically the same as a mildly integrated regressor with
$\rho^{*}=1+c^{*}/T^{\gamma}$ where $\gamma=0.5$; if $T/R_{n,T}\to R\in[0,\infty)$,
$x_{i,t}$ behaves asymptotically the same as a local unit root with
$\gamma=1$. } We remain agnostic about the degrees of persistence, and adopt \citet{phillips1999linear}'s
\emph{joint asymptotics} to allow $n$ and $T$ to simultaneously
tend to infinity.

We aim for a unified inferential procedure that simultaneously achieves
the following three goals. First, it accommodates all modes of persistence
stated above and misspecification of the AR(1) model featured by weakly
dependent AR(1) errors $\{v_{i,t}\}$. Second, it covers a wide range
of asymptotic schemes, characterized as $n/T\to c\in[0,\infty)$.
In many applications, the number of cross-sectional units $n$ is
at least proportional to the time span $T$, and it is therefore essential
to accommodate $n/T\to c>0.$ We refer to this scheme as the \emph{leading
asymptotic case.} Third, it possesses a nearly optimal rate of convergence,
maintaining the super-consistency for persistent regressors and achieving
high power of hypothesis testing for predictability of the outcome.

Our analysis starts with the popular WG estimator. In panel data,
the Stambaugh bias will be carried over and fused with the Nickell
bias in WG, resulting in a composite \emph{Nickell-Stambaugh bias
}in the $t$-statistic with an order substantially enlarged from $1/\sqrt{T^{1-\gamma}}$
to $\sqrt{n/T^{1-\gamma}}$. The inflated bias shifts the center of
the $t$-statistic much further away from zero, invalidating the standard
statistical inference which refers to the critical values from $\mathcal{N}(0,1)$.
To reinstate the standard inference based on the $t$-statistic, one
strategy is to find the analytical formula of the bias and subtract
it from the estimator. The bias formula for WG depends primarily on
$\rho^{*}$. Were there an ``oracle'' to reveal the true $\rho^{*}$,
bias correction would be straightforward. Unfortunately, bias correction
for WG by plugging in a consistent estimator $\widehat{\rho}$ is
infeasible under $\gamma=1$, because correcting this excessively
large bias demands an impossibly fast rate of convergence of $\hat{\rho}$.

Faced with the intrinsic difficulty of WG, we turn to \emph{IVX} proposed
by \citet{phillips2009econometric}. IVX constructs a mildly integrated
instrumental variable (IV) using the original predictor and runs a
two-stage least squares estimation. In time series, IVX enjoys standard
asymptotic distributions with a nearly optimal convergence rate, allowing
use of the critical values based on $\mathcal{N}(0,1)$ or the $\chi^{2}$
distribution for hypothesis testing. Our baseline estimator $\bhativx$
is a panel analog of the time series counterpart. Panel IVX also incurs
the Nickell-Stambaugh bias, with the expression of the bias again
relying on $\rho^{*}$. The silver lining is that the mildly integrated
IV reduces the order of the bias to $o_{p}\bigl(\sqrt{n/T^{1-\gamma}}\bigr)$,
thereby providing a small niche for bias correction by plugging an
estimator $\widehat{\rho}$ with a sufficiently fast convergence rate
into IVX's analytical bias formula.

Which $\widehat{\rho}$ is qualified? We find that the WG estimator
for $\rho^{*}$ in the panel AR fails the task due to its own bias.
Alternatively, \citet{han2014x}'s X-differencing (XDiff) estimator
removes the bias in the stationary case (Case I) and the exact unit
root case (Case IV), but it does not cope with Cases II, III, and
V, where $\rho^{*}\to1$ but $\rho^{*}\neq1$. In addition, XDiff
requires that the AR(1) error $\{v_{i,t}\}$ must be martingale difference
sequences, meaning that the AR(1) model for $x_{i,t}$ must be correctly
specified. Other works covering nonstationary panels, including \citet{westerlund2017testing},
do not accommodate the asymptotic scheme $n/T\to c>0$.

To this end, we propose a new solution that delivers valid statistical
inference for all the types of regressor persistence under consideration.
With a slightly tailored parameter, the IVX estimator $\hat{\rho}^{\mkern3mu {\rm IVX}}$
eradicates the Nickell bias in the panel AR(1) model \eqref{eq:AR1}
while retaining signal strength; hence it enjoys the desirable rate
of convergence. Plugging $\hat{\rho}^{\mkern3mu {\rm IVX}}$ into
the bias formula of $\bhativx$, we produce a bias-corrected estimator
named\emph{ Double IVX} (DIVX). DIVX restores the asymptotic normality
centered at zero, and thus the standard inferential procedure follows.
Furthermore, it achieves all three goals stated above. The key theoretical
insight lies in the delicate interplay of distinctive biases, in particular
the bias of the plug-in estimator in the analytical formula. The lessons
we learned from the drawbacks of other potential estimators culminate
in the DIVX estimator as a proper solution. To the best of our knowledge,
DIVX is the first and only available estimator that unifies the estimation
and inferential procedure in Cases I-V from the stationary regime
to the near-unity regime.

To broaden its usability, in the Online Appendices we further generalize
DIVX from the simple predictive model (\ref{eq:predictive}) to two-way
fixed effects (Appendix \ref{sec:appdx two-way}), multivariate models
(Appendix \ref{sec:appdxMultivariate}), multiple-period-ahead predictive
regression (Appendix \ref{sec:Local-Projection}), and heterogeneous
panel models (Appendix \ref{sec:Grouped-Heterogeneity}).

\medskip{}

\textbf{Literature review}. This paper stands on strands of vast literature.
How to conduct valid hypothesis testing in predictive regressions
has spanned into a large literature; see \citet{phillips2015halbert}
for a survey. IVX has witnessed many applications and extensions,
for instance \citet{phillips2013predictive}, \citet{Kostakis2015},
\citet{Xu2020}, \citet{hjalmarsson2022long}, and \citet{Demetrescu2023},
to name a few. In time series, the finite sample bias in the estimation
of AR(1) is found by \citet{Hurwicz1950} and \citet{Kendall1954}.
Numerous methods have been developed to mitigate the bias in dynamic
models of stationary time series, for instance, the jackknife methods
\citep{quenouille1956notes,tukey1958bias,phillips2005jackknifing,chambers2013jackknife}.
For persistent time series, \citet{chan1987asymptotic} study the
asymptotics of LS when $\rho^{*}$ is local to unity, and \citet{phillips1987towards}
provides a comprehensive treatment. The bias in the LS estimator of
$\rho^{*}$ carries over into that of the predictive coefficient $\beta^{*}$
\citep{stambaugh1999predictive}. The impact of the persistent regressor
on statistical inference has been explored by \citet{cavanagh1995inference},
\citet{campbell2006efficient}, and \citet{jansson2006optimal}, and
the effects on shrinkage estimation have been investigated by \citet{lee2022lasso},
\citet{mei2022lasso}, and \citet{gao2026lasso}. To avoid nonstandard
inference, \citet{phillips2009econometric} propose IVX, which serves
as the baseline estimator of the procedure recommended by this paper.
\citet{magdalinosUniformDistributionfreeInference2024} develop a
uniform inference procedure for time series autoregressive processes
using IVX instrumentation.

The Nickell bias in dynamic panel data distorts the inference by the
WG estimator. One of the classical solutions is the GMM-based estimators
\citep{anderson1981estimation,arellano1991some,arellano1995another,blundell1998initial},
which is well known to suffer from the weak instrument problem when
data are highly persistent \citep{kruinigerGMMEstimationInference2009,phillipsDynamicPanelAndersonHsiao2018}
and will not serve our purpose. Anderson and Hsiao \citeyearpar{anderson1981estimation,anderson1982formulation}
consider the maximum likelihood under various orders of $n$ and $T$,
and the likelihood-based methods are extensively discussed by recent
literature \citep{hsiaoMaximumLikelihoodEstimation2002,kruinigerQuasiMLEstimation2013,dhaeneLikelihoodInferenceAutoregression2016}.
However, they are mostly developed under fixed $T$ and their behaviors
under persistent regressors are not fully explored. In the large-$n$-large-$T$
asymptotics, \citet{hahn2002asymptotically} and \citet{okui2010asymptotically}
investigate the analytical bias correction, mainly focused on the
stationary case. Alternative proposals include split-panel jackknife
\citep{dhaene2015split,chudik2018half}, forwards and backwards recursive
detrending \citep{westerlund2017testing}, indirect inference \citep{gourieroux2010indirect,bao2023indirect},
and X-differencing \citep{han2014x}. While studies of the Nickell
bias mainly focus on the panel AR, the bias is inherited by panel
predictive regressions \citeyearpar[Hjalmarsson][]{hjalmarsson2008stambaugh,hjalmarsson2010predicting}.
Besides predictive models, biases are ubiquitous in large panel data;
see \citet{hahn2004jackknife} and \citet{fernandez2016individual}
for nonlinear models.\footnote{Among these works, \citet{han2014x} and \citet{westerlund2017testing}
rigorously cover persistent regressors in their theoretical justifications
with joint asymptotics. As mentioned before, \citet{han2014x} only
cover exact unit roots with $\rho^{*}=1$ without misspecifying the
AR(1) model. \citet{westerlund2017testing} proposes an estimator
forwards and backwards recursive detrending, while it has substantial
power loss in finite samples when the regressor becomes highly persistent;
see \citet[Section 4]{westerlund2017testing} for details. Section
\ref{subsec:comparison} of the Online Appendices conducts additional
simulations to compare DIVX to the split-panel jackknife estimator
\citep{dhaene2015split,chudik2018half}, the X-differencing estimator
\citep{han2014x}, and the forwards and backwards recursive detrending
\citep{westerlund2017testing} to illustrate the superiority of our
proposed method. } None of the aforementioned works simultaneously achieves the three
goals stated above: allowing for all five types of persistence with
model misspecification characterized by a weakly dependent AR(1) innovation
$v_{i,t}$, accommodating general asymptotic schemes $n/T\to c\in[0,\infty)$,
and maintaining nearly optimal convergence rate to achieve high power
of hypothesis testing.

The extension to multiple predictive horizons in Section \ref{sec:Local-Projection}
of our Online Appendices connects our study with the recent advancements
of local projection \citep{jorda2005estimation}, whose convenience
in estimating the impulse response has drawn considerable research
interest \citep{barnichon2019impulse,montiel2021local,plagborg2021local,herbst2024bias}.
While most empirical applications use the WG estimator for panel local
projection, \citet{mei2023implicit} show that the Nickell bias sustains
asymptotically and invalidates the standard inference for panel local
projection when $n$ and $T$ are comparably large. We extend DIVX
to further handle panel local projection with persistent regressors;
see Appendix \ref{sec:Local-Projection} for details. Furthermore,
the extension for heterogeneous panels in our Appendix \ref{sec:Grouped-Heterogeneity}
utilizes the techniques of identifying latent structures in panel
data \citep{su2016identifying,su2018identifying,huang2021nonstationary,wang2021identifying,liu2023panel}.

This paper is positioned as a stepping stone toward a comprehensive
theory for predictive regressions in practical and realistic empirical
settings, in particular when the regressors are persistent. The current
paper has not yet covered all characteristics of panel data. For example,
beyond individual effects, the literature has applied the interactive
effects to characterize cross-sectional correlation \citep{greenaway2012asymptotic,moon2015linear,moon2017dynamic,westerlund2017testing}.
In view of the lengthy proofs of the current paper, we leave these
important extensions for future studies. Furthermore, as a starting
point for developing debiased inference for highly persistent panel
data that meets the three aforementioned goals, our analysis focuses
on analytical bias corrections with WG and IVX that are relatively
tractable. Inference for nonstationary panels with alternative methodologies
merit separate papers.

\medskip{}

\textbf{Layout}. The rest of the paper is organized as follows. Section
\ref{sec:framework} sets up the model, heuristically explains the
inconvenience of WG estimators, and elaborates the DIVX estimator.
Section \ref{sec:theory} formally develops the asymptotic properties
of DIVX with a univariate regressor. Section \ref{sec:simulation}
carries out Monte Carlo simulations to demonstrate the validity and
necessity of bias correction by DIVX. An empirical example that examines
the predictability of stock returns is carried out in Section \ref{sec:Empirical-Application}.
The Online Appendices collect extensions of the DIVX estimator to
two-way fixed effects, multivariate models, multiple-period-ahead
predictive regression, and heterogeneous panel models, together with
additional numerical studies. The Supplementary Materials collect
all technical proofs.

\medskip{}

\textbf{Notations}. The symbols ``$\to_{p}$'' and ``$\to_{d}$''
signify \emph{convergence in probability} and \emph{convergence in
distribution}, respectively. For a time series $\{x_{t}\}$, let $\Delta x_{t}=x_{t}-x_{t-1}$
be its first-order difference. For a generic panel data random variable
$x_{i,t}$ with $i=1,\dots,n$ and $t=1,\dots,T-1$, we use $\bar{x}_{i}=\frac{1}{T-1}\sum_{t=1}^{T-1}x_{i,t}$
to denote the within-group average for the individual $i$, and $\xtd_{i,t}=x_{i,t}-\bar{x}_{i}$
is the within-group demeaned data. Let $a\wedge b=\min\left\{ a,b\right\} $,
$a\vee b=\max\left\{ a,b\right\} $, and $a_{+}=a\vee0$. For a vector
$\bm{a}$, we use $\|\bm{a}\|$ to denote its Euclidean norm. Let
$\bm{1}(A)$ be the indicator function of an event $A$, i.e., $\bm{1}(A)$
takes value 1 if $A$ occurs and 0 otherwise.

\section{Panel Predictive Regression\label{sec:framework}}

\subsection{Setup\label{subsec:unimodel}}

We start with a simple predictive regression of the two-equation system
\eqref{eq:predictive} and \eqref{eq:AR1}. Following \citet{phillips2019uniform},
let the regressor $x_{i,t}$ follow a state space representation
\begin{equation}
\begin{split}x_{i,t} & =\alpha_{i}+\delta_{i,t},\\
\delta_{i,t+1} & =\rho^{*}\delta_{i,t}+v_{i,t+1}.
\end{split}
\label{eq:state space}
\end{equation}
It implies that $x_{i,t}$ admits the AR(1) form \eqref{eq:AR1} with
the FE
\begin{equation}
\mu_{x,i}=(1-\rho^{*})\alpha_{i},\label{eq:FE_AR}
\end{equation}
under which \eqref{eq:AR1} can be rewritten as $x_{i,t+1}-\alpha_{i}=\rho^{*}\left(x_{i,t}-\alpha_{i}\right)+v_{i,t+1}.$
Such a specification of FE is standard in the literature \citep{han2014x}.\footnote{This specification avoids an unrestricted nonzero intercept in a local-to-unity
process. Otherwise, the intercept would accumulate and become a drift
that dominates the stochastic trend, thus drastically complicating
the asymptotic orders.}

The following regularity conditions characterize the data generating
processes.
\begin{assumption}[Initial values and drifts]
\label{assump:initval}Uniformly across all $i\in[n]$, the conditional
mean $\mathbb{E}(\delta_{i,0}|\alpha_{i})=0$, the unconditional variance
$\mathbb{E}(\delta_{i,0}^{2})=O(|1-\rho^{*}|^{-1}\wedge T^{1-\xi})$
where $\xi>0$ is an arbitrarily small absolute constant, and $\mathbb{E}(\delta_{i,0}^{4}+\alpha_{i}^{4})=O(|1-\rho^{*}|^{-2}\wedge T^{2})$.
(The convention $1/0=\infty$ is invoked if $\rho^{*}=1$).\footnote{There is a parametric alternative to this assumption following \citet{phillipsUnitRootCointegrating2009}.
We can assume $\delta_{i,0}=\sum_{s=0}^{\kappa_{T}}\rho^{*s}v_{i,-s}$
where if $\gamma=1$, then $\kappa_{T}/T\to0$ as $T\to\infty$; if
$\gamma\in[0,1)$, then $\kappa_{T}$ is unrestricted. This assumption
mandates that the initialization should not include infinitely distant
past innovations if $x_{i,t}$ are LUR; otherwise the infinite series
almost surely diverges. Under this assumption (along with \prettyref{assump:innov}),
we have $\delta_{i,0}=o_{p}(\sqrt{T})$ if $\gamma=1$ and $\delta_{i,0}=O_{p}(\sqrt{T^{\gamma}})$
if $\gamma\in[0,1)$.}
\end{assumption}
In Assumption \ref{assump:initval}, the zero conditional mean of
the initial values will facilitate the derivation of the analytical
bias formulas. The restriction on the second moment of the initial
values implies that the unconditional variance of $\delta_{i,t}$
is $O(T^{\gamma})$ when $\gamma<1$, and the small constant $\xi$
takes effect only on the LUR case when $\gamma=1,$ so that $\delta_{i,0}=o_{p}(\sqrt{T})$.
These conditions ensure that the initial values would not impact the
asymptotics. The order of the fourth moments of the initial values
and drifts bounds the estimator errors of $\rho^{*}$ and the standard
error of the $t$-statistic for $\beta^{*}$.

Suppose the AR(1) errors $v_{i,t}$ in \eqref{eq:AR1} follow the
linear process
\begin{equation}
v_{i,t}=\sum_{s=0}^{\infty}g_{s}\varepsilon_{i,t-s},\label{eq:linearProcess}
\end{equation}
where the innovations $\{\varepsilon_{i,t}\}$ are defined in the
following Assumption \ref{assump:innov}.
\begin{assumption}[Innovations]
\label{assump:innov} \mbox{}
\begin{enumerate}
\item \label{enu:w_cumu}For each $i$, let $\bm{w}_{i,t}=(e_{i,t},\varepsilon_{i,t})^{\prime}$,
with $e_{i,t}$ as in \eqref{eq:predictive}, denote a two-dimensional
strictly stationary and ergodic martingale difference sequence (m.d.s.)
adaptive to the filtration $\{\mathcal{F}_{i,t}=\sigma(\delta_{i,0},\alpha_{i},\bm{w}_{i,t},\bm{w}_{i,t-1},\dots)\}$.
$\{\bm{w}_{i,t}\}$ are i.i.d.~across $i$. In addition, we assume
absolutely summable fourth order cumulants: $\sup_{a,b,c,d\in\{1,2\}}\sum_{t,s,r=-\infty}^{\infty}|\kappa_{abcd}(0,t,s,r)|<\infty$,
where \newpage
\begin{align*}
\kappa_{abcd}(t_{1},t_{2},t_{3},t_{4}) & =\mathbb{E}(w_{a,i,t_{1}}w_{b,i,t_{2}}w_{c,i,t_{3}}w_{d,i,t_{4}})-\mathbb{E}(w_{a,i,t_{1}}w_{b,i,t_{2}})\mathbb{E}(w_{c,i,t_{3}}w_{d,i,t_{4}})\\
 & \qquad-\mathbb{E}(w_{a,i,t_{1}}w_{c,i,t_{3}})\mathbb{E}(w_{b,i,t_{2}}w_{d,i,t_{4}})-\mathbb{E}(w_{a,i,t_{1}}w_{d,i,t_{4}})\mathbb{E}(w_{b,i,t_{2}}w_{c,i,t_{3}}),
\end{align*}
 with $w_{a,i,t}$ being the $a$-th element of $\bm{w}_{i,t}$.
\item \label{enu:hetero_e} The sequence $\{e_{i,t}\}$ admits a $\mathrm{GARCH}(q,r)$
representation:
\[
e_{i,t}=h_{i,t}^{1/2}u_{i,t},\qquad h_{i,t}=\phi+\sum_{k=1}^{q}a_{k}e_{i,t-k}^{2}+\sum_{\ell=1}^{r}b_{\ell}h_{i,t-\ell},
\]
where $u_{i,t}$ are i.i.d.\ random variables with $\mathbb{E}(u_{i,t})=0$,
$\mathbb{E}(u_{i,t}^{2})=1$ and $\mathbb{E}(u_{i,t}^{4})<\infty$,
and the constant coefficients satisfy $\phi>0,$ $a_{k},b_{\ell}\geq0$
and $0\leq\sum_{k=1}^{q}a_{k}+\sum_{\ell=1}^{r}b_{\ell}<1$.
\item \label{enu:LP}The coefficients in the linear process \eqref{eq:linearProcess}
satisfy $|g_{s}|\leq C_{0}\exp(-C_{g}s)$ for any $s$ with positive
constants $C_{0}$ and $C_{g}$. Moreover, $\sup_{t\leq0}|\mathbb{E}(\delta_{i,0}\varepsilon_{i,t})|<\infty$.
\end{enumerate}
\end{assumption}
Overall, our theoretical framework covers a wide range of data generating
processes. In Assumption \ref{assump:innov}, Condition \prettyref{enu:w_cumu}
bounds the fourth order cumulants, which excludes overly strong high-order
temporal dependence in the innovations. This is a standard theoretical
assumption to address heteroskedasticity in the literature \citep{andrewsHeteroskedasticityAutocorrelationConsistent1991,hahn2002asymptotically,stockHeteroskedasticityrobustStandardErrors2008,montiel2021local}.
Condition \prettyref{enu:hetero_e} assumes a GARCH structure for
the error term of the main predictive regression, allowing for conditional
heteroskedasticity over time for the martingale difference sequences
$\{e_{i,t}\}$. This condition follows \citet{Kostakis2015} and \citet{magdalinosLeastSquaresIvx2022}
in time series predictive regressions with conditional heteroskedasticity.
Condition \prettyref{enu:LP} assumes that the AR(1) errors $\{v_{i,t}\}$
follow linear processes with exponentially decaying coefficients,
thereby accommodating misspecification of the AR(1) model. In particular,
the exponentially decaying coefficients are widely used in the literature
of bias correction for dynamic panels like the split-panel jackknife
estimators \citep{dhaene2015split,chudik2018half}, and have been
sufficiently general to cover the commonly used stationary ARMA processes.
Relaxing the exponentially decaying rate is possible at the cost of
substantial complications in theoretical expositions and proofs. Condition
\prettyref{enu:LP} also requires that $\sup_{t\leq0}|\mathbb{E}(\delta_{i,0}\varepsilon_{i,t})|<\infty$
to avoid overly strong endogeneity.

Our assumptions for the innovations share similarities to the literature
of panel predictive regressions. For example, the m.d.s.~condition,
independence across $i$, and linear processes, are also assumed in
\citet{hjalmarsson2010predicting} and \citet{westerlund2017testing}.
In Assumption \ref{assump:innov}, the bounded cumulants in Condition
\prettyref{enu:w_cumu} and the GARCH representation in Condition
\prettyref{enu:hetero_e} are new for panel predictive regressions.
As mentioned in the previous paragraph, they are widely used in the
literature of time series \citep{andrewsHeteroskedasticityAutocorrelationConsistent1991,hahn2002asymptotically,stockHeteroskedasticityrobustStandardErrors2008,Kostakis2015,montiel2021local,magdalinosLeastSquaresIvx2022},
and indispensable for rigorously deriving asymptotic distributions
of the feasible test statistics under conditional heteroskedasticity
from low-level assumptions. These common and mild assumptions will
earn rigorous theoretical analysis and attractive asymptotic properties
for our estimator.

In what follows, Section \ref{subsec:Inconvenience-of-WG=0000DF}
demonstrates that the WG estimator is unsuitable for reliable inference
on $\beta^{*}$. Section \ref{subsec:IVX-as-the} explains the feasibility
of bias correction for the IVX estimator of $\beta^{*}$ using a tailored
IVX estimator of $\rho^{*}$. Section \ref{subsec:IVX-for-rho} introduces
this IVX estimator of $\rho^{*}$. Section \ref{subsec:Double IVX}
completes the bias-correction strategy by defining the DIVX estimator,
which enables valid inference for $\beta^{*}$. Theoretical justifications
are provided in Section \ref{sec:theory}.

\subsection{Inconvenience of WG\label{subsec:Inconvenience-of-WG=0000DF}}

We first focus on the WG estimator
\[
\bhatfe=\sum_{i}\sum_{t}\tilde{x}_{i,t}y_{i,t+1}\bigg/\sum_{i}\sum_{t}\td{x}_{i,t}^{2},
\]
 the default option of empirical studies. To simplify the illustrations,
in the discussion of the WG estimator $\bhatfe$ we assume the innovations
$\{e_{i,t}\}$ and $\{v_{i,t}\}$ are i.i.d.~across both $i$ and
$t$. If WG fails in this simplified scenario, it will not work in
more general cases.

Note that we have the following decomposition:
\begin{align*}
\bhatfe-\beta^{*} & =\sum_{i}\sum_{t}\tilde{x}_{i,t}e_{i,t+1}\bigg/\sum_{i}\sum_{t}\td{x}_{i,t}^{2}.
\end{align*}
Since we allow for correlation between the error term $\{e_{i,t}\}$
in the main regression and the AR(1) error $\{v_{i,t}\}$, the error
term $e_{i,t+1}$ and the demeaned regressor $\tilde{x}_{i,t}$ are
correlated even under the i.i.d.~condition. Therefore, WG is biased.
Specifically, when the regressor along the time dimension is stationary
($\gamma=0$), in the leading asymptotic case where $n/T\to c\in\left(0,\infty\right)$
we have
\begin{equation}
\sqrt{nT}\left[(\bhatfe-\beta^{*})+\omega_{ev}^{*}\cdot b_{n,T}^{\mathrm{WG}}\left(\rho^{*}\right)\right]\dto\ncal(0,\Sigma^{\text{WG}}),\label{eq:beta_sq_NT}
\end{equation}
where $\omega_{ev}^{*}=\mathbb{E}(e_{i,t}v_{i,t})$, and the asymptotic
variance $\Sigma^{\text{WG}}$ is a positive constant with a complicated
formula (\ref{eq:fe_var}). The bias formula has an analytical form
\begin{align}
b_{n,T}^{{\rm WG}}(\rho) & =\frac{n\sum_{t=2}^{T-1}\sum_{s=2}^{t}\rho{}^{t-s}}{(T-1)\sum_{i=1}^{n}\sum_{t=1}^{T-1}\td{x}_{i,t}^{2}}.\label{eq:fe_bias}
\end{align}
On the left-hand side of \eqref{eq:beta_sq_NT} the estimator is inflated
by $\sqrt{nT}$, which we refer to as the \emph{(standard) panel factor}.
Equation~\eqref{eq:beta_sq_NT} highlights the fact that $\sqrt{nT}(\bhatfe-\beta^{*})$
involves an asymptotic bias that shifts the center of the asymptotic
normal distribution away from 0.

When the regressor $x_{i,t}$ is persistent, that is, $\gamma\in(0,1]$,
the order of the bias depends on $\gamma$. This is the \emph{Nickell-Stambaugh
bias} to which the title of the paper alludes. We will show that the
panel factor $\sqrt{nT}$ will be further multiplied by the \emph{persistence
factor} $\sqrt{T^{\gamma}}$ to deliver asymptotic normality:
\[
\sqrt{nT^{1+\gamma}}\left[(\bhatfe-\beta^{*})+\omega_{ev}^{*}\cdot b_{n,T}^{\mathrm{WG}}\left(\rho^{*}\right)\right]\dto\ncal(0,\Sigma^{\text{WG}}).
\]
What is worse, even if $\omega_{ev}^{*}$ is known, it is difficult
to obtain a feasible estimator of $\rho^{*}$ to be plugged into the
bias expression. Again, we use $\hat{\rho}$ to denote a generic estimator
of $\rho^{*}$. When $\hat{\rho}-\rho^{*}\pto0$, a simple Taylor
expansion gives
\begin{equation}
\sqrt{nT^{1+\gamma}}\left[b_{n,T}^{\mathrm{WG}}\left(\hat{\rho}\right)-b_{n,T}^{\mathrm{WG}}\left(\rho^{*}\right)\right]=\frac{\mathrm{d}}{\mathrm{d}\rho}b_{n,T}^{\mathrm{WG}}\left(\rho^{*}\right)\cdot\sqrt{nT^{1+\gamma}}(\hat{\rho}-\rho^{*})+h.o.t,\label{eq:Stambaugh deriv-1}
\end{equation}
where ``$h.o.t$'' collects the higher order terms that we omit
here in heuristic discussions. We can deduce that $\frac{\mathrm{d}}{\mathrm{d}\rho}b_{n,T}^{\mathrm{WG}}\left(\rho^{*}\right)=O_{p}(1)$
as $(n,T)\to\infty$. Therefore, to make the feasible estimator of
the bias $b_{n,T}^{\mathrm{WG}}\left(\hat{\rho}\right)$ asymptotically
equivalent to $b_{n,T}^{\mathrm{WG}}\left(\rho^{*}\right)$, it is
required that
\[
\sqrt{nT^{1+\gamma}}(\hat{\rho}-\rho^{*})=o_{p}(1)
\]
 as $(n,T)\to\infty$; see Section~\ref{subsec:Drawback other} for
details. This is \emph{mission impossible} because a regular estimator
$\hat{\rho}$ can at most achieve $\sqrt{nT^{1+\gamma}}(\hat{\rho}-\rho^{*})=O_{p}(1)$
in the panel AR regression, but not $o_{p}(1)$; see \prettyref{rem:rho_opt_rate}
about the optimal rate. The simulations in Section~\ref{sec:simulation}
provide evidence of the conspicuous bias of the WG estimator and the
undesirable performance of bias-correction procedures when WG is used
as the baseline estimator.
\begin{rem}[Optimal Convergence Rate of $\hat{\rho}$]
\label{rem:rho_opt_rate} Under the extra assumption $\alpha_{i}=0$,
the pooled OLS estimator
\[
\hat{\rho}^{\,\mathrm{LS}}=\frac{\sum_{i=1}^{n}\sum_{t=1}^{T-1}x_{i,t}x_{i,t+1}}{\sum_{i=1}^{n}\sum_{t=1}^{T-1}x_{i,t}^{2}}
\]
is $\sqrt{nT^{1+\gamma}}$-consistent. In reality when $\alpha_{i}\neq0$,
this rate of convergence cannot be improved. Therefore, $\sqrt{nT^{1+\gamma}}(\hat{\rho}-\rho^{*})=o_{p}(1)$
is not achievable.
\end{rem}

\subsection{IVX as the main estimator \label{subsec:IVX-as-the}}

In this subsection we provide the estimation and inference procedure
for $\beta^{*}$ based on IVX. For the time series models many approaches
to inferring $\beta^{*}$ have been proposed, such as the Bonferroni
method \citep{cavanagh1995inference,campbell2006efficient} and the
conditional likelihood approach \citep{jansson2006optimal}. These
procedures are designed for a univariate regressor, and the respective
limit distributions are nonstandard. In contrast, the IVX method achieves
valid standard inference in time series multiple regressions and is
robust to degrees of persistence \citep{magdalinos2009limit,Kostakis2015,phillips2013predictive}.
In addition, it allows for weak temporal dependence of the AR(1) error
$v_{i,t}$, thereby accommodating misspecification of the AR(1) model
for the regressor $x_{i,t}$.

Our procedure is based on a panel version of IVX. First, we produce
a mildly integrated IV by filtering the regressor:
\begin{equation}
z_{i,t}=\sum_{s=1}^{t}\rho_{z}^{t-s}\Delta x_{i,s},\quad\rho_{z}=1+c_{z}/T^{\theta},\label{eq:def IV}
\end{equation}
where $c_{z}<0$ and $\theta\in(0,1)$ are constants. On the one hand,
$\rho_{z}\to1$ as $T\to\infty$, and therefore the instrument $z_{i,t}$
is still persistent and maintains a fast convergence rate (the third
goal listed in the Introduction). On the other hand, $\rho_{z}$'s
speed of convergence to $1$ is slower than $\rho^{*}$'s when $\gamma>\theta$,
which is always true when $x_{i,t}$ is LUR with $\gamma=1$. Therefore,
the IV mitigates the bias order and accommodates a slightly less accurate
estimator for $\rho^{*}$, thereby leaving small room for finding
an admissible $\hat{\rho}$. We suggest following the literature \citep{Kostakis2015,phillips2016robust}
to choose $c_{z}=-1$ and $\theta=0.95$, which performs well in all
our numerical and empirical examples. The relatively large choice
$\theta=0.95$ maintains high efficiency in estimating $\beta^{*}$
when $x_{i,t}$ is a highly persistent LUR.

We propose the following panel IVX estimator:
\begin{equation}
\bhativx=\frac{\sum_{i=1}^{n}\sum_{t=1}^{T-1}\tilde{z}_{i,t}y_{i,t+1}}{\sum_{i=1}^{n}\sum_{t=1}^{T-1}\tilde{z}_{i,t}x_{i,t}},\label{eq:beta_IVX}
\end{equation}
where $\tilde{z}_{i,t}$ is the within-group demeaned $z_{i,t}$.
Again, the bias comes from the correlation between $e_{i,t+1}$ and
the transformed instrument $\tilde{z}_{i,t}$. Define $\omega_{ev,h}^{*}=\mathbb{E}(e_{i,t}v_{i,t+h})$
as the $h$-period intertemporal covariance between the two error
terms, and denote $\omega_{ev}^{*}=\omega_{ev,0}^{*}$. In an oracle
setting where $\rho^{*}$ and $\{\omega_{ev,h}^{*}\}$ were known,
we would have
\begin{equation}
\sqrt{nT^{1+(\theta\wedge\gamma)}}\left[(\hat{\beta}^{\mathrm{IVX}}-\beta^{*})+b_{n,T}^{\mathrm{IVX}}(\{\omega_{ev,h}^{*}\},\rho^{*},\rho_{z})\right]\to_{d}\mathcal{N}(0,\Sigma^{\mathrm{IVX}}),\label{eq:beta_IVX_rate}
\end{equation}
where the asymptotic variance $\Sigma^{\mathrm{IVX}}$ is a positive
constant with a complicated formula (see \eqref{eq:SigmaIVX}), and
the bias formula is
\begin{equation}
b_{n,T}^{\mathrm{IVX}}(\{\omega_{ev,h}^{*}\},\rho^{*},\rho_{z})=\frac{n\sum_{h=0}^{T-3}\Psi_{h,T-1}(\rho^{*},\rho_{z})\omega_{ev,h}^{*}}{(T-1)\sum_{i=1}^{n}\sum_{t=1}^{T-1}\tilde{z}_{i,t}x_{i,t}},\label{eq:ivx_bias}
\end{equation}
where
\[
\Psi_{h,T-1}(\rho^{*},\rho_{z})=\sum_{k=h+2}^{T-1}\frac{\rho_{z}^{T-k}-\rho^{*T-k}}{\rho_{z}-\rho^{*}}.
\]
Shown in \prettyref{prop:bhativx} below, the order of the bias $b_{n,T}^{\mathrm{IVX}}(\{\omega_{ev,h}^{*}\},\rho^{*},\rho_{z})$
is $O_{p}(1/T^{2-(\theta\vee\gamma)})$, but the inflating factor
on the left-hand side of \eqref{eq:beta_IVX_rate} is $\sqrt{nT^{1+(\theta\wedge\gamma)}}$,
though slightly mitigated relative to WG's $\sqrt{nT^{1+\gamma}}$.
If $n/T\to c\in(0,\infty)$ and $\gamma+\theta+(\theta\vee\gamma)>2$
(e.g., $\gamma=1$ under which $x_{i,t}$ is LUR), then the inflating
factor $\sqrt{nT^{1+(\theta\wedge\gamma)}}$ dominates the bias order
$O_{p}(1/T^{2-(\theta\vee\gamma)})$, thereby shifting the center
of the asymptotic normal distribution of $\sqrt{nT^{1+(\theta\wedge\gamma)}}(\hat{\beta}^{\mathrm{IVX}}-\beta^{*})$
away from 0.

The bias formula (\ref{eq:ivx_bias}) incorporates unknown parameters
$\rho^{*}$ and $\{\omega_{ev,h}^{*}\}$. The discussion here focuses
on $\rho^{*}$, the essential trouble maker, by temporarily taking
$\{\omega_{ev,h}^{*}\}$ as given. The IVX keeps the standard panel
factor $\sqrt{nT}$ on the left-hand side of \eqref{eq:beta_IVX_rate}
and in the meantime its persistence factor is $\sqrt{T^{\theta\wedge\gamma}}$,
in contrast to WG's $\sqrt{T^{\gamma}}$. When $\hat{\rho}-\rho^{*}\to_{p}0$,
a Taylor expansion of the bias around a generic plug-in estimator
$\hat{\rho}$ becomes
\begin{align*}
 & \sqrt{nT^{1+(\theta\wedge\gamma)}}\left[b_{n,T}^{\mathrm{IVX}}(\{\omega_{ev,h}^{*}\},\hat{\rho},\rho_{z})-b_{n,T}^{\mathrm{IVX}}(\{\omega_{ev,h}^{*}\},\rho^{*},\rho_{z})\right]\\
 & \qquad=\frac{\mathrm{\partial}}{\mathrm{\partial}\rho}b_{n,T}^{\mathrm{IVX}}(\{\omega_{ev,h}^{*}\},\rho^{*},\rho_{z})\cdot\sqrt{nT^{1+(\theta\wedge\gamma)}}(\hat{\rho}-\rho^{*})+h.o.t.
\end{align*}
It can be verified that $\frac{\mathrm{\partial}}{\mathrm{\partial}\rho}b_{n,T}^{\mathrm{IVX}}(\{\omega_{ev,h}^{*}\},\rho^{*},\rho_{z})=O_{p}\bigl(T^{-[2-(\theta\vee\gamma)-\gamma]}\bigr)$,
and therefore a desirable order $\sqrt{nT^{1+(\theta\wedge\gamma)}}(\hat{\rho}-\rho^{*})=o_{p}\bigl(T^{2-(\theta\vee\gamma)-\gamma}\bigr)$,
or alternatively
\begin{equation}
\hat{\rho}-\rho^{*}=o_{p}\left(\left(nT^{\theta+3\gamma+(\theta\vee\gamma)-3}\right)^{-1/2}\right)\label{eq:cond IVX bias correct}
\end{equation}
would be sufficient to remove the bias in $\hat{\beta}^{\mathrm{IVX}}$.
In the LUR case where $\gamma=1$, this amounts to $\hat{\rho}-\rho^{*}=o_{p}\bigl((nT^{1+\theta})^{-1/2}\bigr)$.

Which estimator of $\rho^{*}$ satisfies \eqref{eq:cond IVX bias correct}
under any $\gamma\in[0,1]$? The WG estimator
\[
\hat{\rho}^{\mkern2mu {\rm WG}}=\frac{\sum_{i=1}^{n}\sum_{t=1}^{T-1}\tilde{x}_{i,t}x_{i,t+1}}{\sum_{i=1}^{n}\sum_{t=1}^{T-1}\tilde{x}_{i,t}x_{i,t}}
\]
is not qualified due to its inherent Nickell-Stambaugh bias. As mentioned
in the Introduction, there remain no asymptotic guarantees of the
widely used GMM or likelihood-based methods under highly persistent
panels with misspecification of the AR(1) model. In the literature,
the X-differencing estimator by \citet{han2014x} allows for nonstationary
panels, but only admits exact unit roots and rules out weak temporal
dependence of $v_{i,t}$. The forwards and backwards recursive detrending
by \citet{westerlund2017testing} has an excessively large variance
in finite samples when the regressor is highly persistent; see \citet[Section 4]{westerlund2017testing}
for details.

Is IVX again applicable? When $\gamma=1$, the inflation factor for
IVX is $\sqrt{nT^{1+\theta}}$ as in \eqref{eq:beta_IVX_rate}. Therefore,
the condition $\hat{\rho}-\rho^{*}=o_{p}\bigl((nT^{1+\theta})^{-1/2}\bigr)$
requires that $\hat{\rho}$ converges faster than the inflation factor
of IVX with the parameter $\theta$ as in (\ref{eq:def IV}). This
is achievable by enlarging the parameter $\theta$ for IVX instrumentation.
We formalize this idea in the following.

\subsection{IVX for $\rho^{*}$\label{subsec:IVX-for-rho}}

We generate the following instrumental variable:
\begin{equation}
z_{i,t}^{(1)}=\sum_{s=1}^{t}\left(1+\frac{c_{z}}{T^{\theta_{1}}}\right)^{t-s}\Delta x_{i,s},\label{eq:def IV (1)}
\end{equation}
where $c_{z}<0$ follows (\ref{eq:def IV}), and $\theta<\theta_{1}<1$.
The enlarged parameter $\theta_{1}$ accelerates the convergence of
IVX estimator of $\rho^{*}$ to meet the condition (\ref{eq:cond IVX bias correct}).
We suggest $\theta_{1}=\left(1+\theta\right)/2$ as a convenient choice,
which produces robust performance in all numerical exercises throughout
this paper.

With the tailored instrumental variable $z_{i,t}^{(1)}$, the IVX
estimator of $\rho^{*}$ is given by
\begin{equation}
\hat{\rho}^{\mkern3mu \mathrm{IVX}}=\frac{\sum_{i=1}^{n}\sum_{t=1}^{T-1}(z_{i,t}^{(1)}x_{i,t+1}-\hat{\varDelta}_{vv})}{\sum_{i=1}^{n}\sum_{t=1}^{T-1}z_{i,t}^{(1)}x_{i,t}},\label{eq:rhoIVX}
\end{equation}
where
\begin{equation}
\hat{\varDelta}_{vv}=\frac{1}{nT}\sum_{h=1}^{G}\sum_{i=1}^{n}\sum_{t=h+1}^{T}\hat{v}_{i,t}^{{\rm WG}}\hat{v}_{i,t-h}^{{\rm WG}},\label{eq:NW Deltavv}
\end{equation}
with $G=\lfloor T^{1/4}\rfloor$ and $\hat{v}_{i,t+1}^{{\rm WG}}=\tilde{x}_{i,t+1}-\hat{\rho}^{{\rm WG}}\tilde{x}_{i,t}$.
Here $\hat{\varDelta}_{vv}$ estimates the long-run covariance $\varDelta_{vv}^{*}=\mathbb{E}(\sum_{h=1}^{\infty}v_{i,t-h}v_{i,t})$.
Note that
\begin{equation}
\hat{\rho}^{\mkern3mu \mathrm{IVX}}-\rho^{*}=\frac{\sum_{i=1}^{n}\sum_{t=1}^{T-1}(z_{i,t}^{(1)}v_{i,t+1}-\hat{\varDelta}_{vv})}{\sum_{i=1}^{n}\sum_{t=1}^{T-1}z_{i,t}^{(1)}x_{i,t}},\label{eq:rho IVX error}
\end{equation}
and $\mathbb{E}(z_{i,t}^{(1)}v_{i,t+1})\neq0$ when $v_{i,t}$ is
weakly dependent over time. To accommodate misspecification of the
AR(1) process, we need to correct the bias from $\mathbb{E}(z_{i,t}^{(1)}v_{i,t+1})$
using the log-run covariance estimator $\hat{\varDelta}_{vv}$. Equation~\eqref{eq:NW Deltavv}
is a familiar kernel estimator of the long-run covariance, where the
bandwidth $G=\lfloor T^{1/4}\rfloor$ follows the common practice
\citep{driscoll1998consistent,greene} and proved asymptotically valid
in Section \ref{subsec:theoryDIVX}. In \prettyref{prop:rho_convergence},
we will show that $\hat{\rho}^{\mkern3mu \mathrm{IVX}}$ satisfies
the condition (\ref{eq:cond IVX bias correct}) and is thus applicable
for correcting the bias in $\bhativx.$
\begin{rem}
To the best of our knowledge, $\hat{\rho}^{\mkern3mu \mathrm{IVX}}$
is currently the only estimator that achieves the fast convergence
rate required by \eqref{eq:cond IVX bias correct} under the leading
asymptotic case $n/T\to c>0$ and weak temporal dependence of the
AR(1) error $v_{i,t}$. Detailed theoretical justifications of $\hat{\rho}^{\mkern3mu \mathrm{IVX}}$
refer to Proposition \ref{prop:rho_convergence} in Section \ref{subsec:theoryDIVX},
followed by Remarks \ref{rem:bcLUR} and \ref{rem:biasStationary}
discussing the special cases of LURs and stationary regressors. An
earlier version of this paper, \citet{liao2024nickellmeetsstambaughtale},
proposed an \emph{X-Jackknife }estimator for $\rho^{*}$. When there
is no serial correlation in $v_{i,t}$, the X-Jackknife estimator
is unbiased by leveraging a unique jackknife scheme that divides the
time dimension into the odd and even indices. The current setting
allows serial dependence in $v_{i,t}$, and thus DIVX supersedes X-Jackknife.
\end{rem}
With the aforementioned two IVX estimators $\hat{\beta}^{{\rm IVX}}$
and $\hat{\rho}^{\mkern3mu \mathrm{IVX}}$, we are ready to present
our core estimator DIVX.

\subsection{Double IVX \label{subsec:Double IVX}}

The oracle bias formula \eqref{eq:ivx_bias} involves $\{\omega_{ev,h}^{*}\}$
for $h=0,1,...,T-3$. The estimation of $\omega_{ev,h}^{*}$ is feasible
using the intertemporal sample covariances of the two-step residuals:
\begin{equation}
\hat{\omega}_{ev,h}=\frac{1}{nT}\sum_{i=1}^{n}\sum_{t=1}^{T-h}\hat{v}_{i,t+h}\hat{e}_{i,t},\label{eq:omega ev hat}
\end{equation}
where
\begin{equation}
\hat{v}_{i,t}=\tilde{x}_{i,t}-\hat{\rho}^{\mkern3mu \mathrm{IVX}}\tilde{x}_{i,t-1},\quad\hat{e}_{i,t}=\tilde{y}_{i,t}-\hat{\beta}^{{\rm IVX}}\tilde{x}_{i,t-1}.\label{eq:residuals}
\end{equation}
The plug-in estimator of the bias \eqref{eq:ivx_bias} with $\hat{\omega}_{ev,h}$
for all $h=0,1,...,T-3$ will induce overly large variances. We therefore
consider the truncated bias formula:
\begin{equation}
\hat{b}_{n,T}^{\mathrm{IVX}}(\boldsymbol{\omega}_{ev,G}^{*},\rho^{*})=\frac{n\sum_{h=0}^{G}\Psi_{h,T-1}(\rho^{*},\rho_{z})\omega_{ev,h}^{*}}{(T-1)\sum_{i=1}^{n}\sum_{t=1}^{T-1}\tilde{z}_{i,t}x_{i,t}},\quad G=\lfloor T^{1/4}\rfloor,\label{eq:truncated bias IVX}
\end{equation}
where $\boldsymbol{\omega}_{ev,G}^{*}=\{\omega_{ev,h}^{*}\}_{0\leq h\leq G}.$
Under weak dependence characterized by Assumption \ref{assump:innov}\ref{enu:LP},
the covariance sequence $\{\omega_{ev,h}^{*}\}$ exponentially decays
as $h$ increases, and therefore the truncation at a polynomial rate
of $T$ is sufficient to approximate the oracle bias formula. The
choice $G=\lfloor T^{1/4}\rfloor$ again follows the common practice
of long-run variance estimation.

Let $\hat{\boldsymbol{\omega}}_{ev,G}=\{\hat{\omega}_{ev,h}\}_{0\leq h\leq G}$
collect the estimated covariances in (\ref{eq:omega ev hat}). The
DIVX estimator is then given as
\begin{equation}
\hat{\beta}^{{\rm DIVX}}=\hat{\beta}^{\mathrm{IVX}}+\hat{b}_{n,T}^{\mathrm{IVX}}(\hat{\boldsymbol{\omega}}_{ev,G},\hat{\rho}^{\mkern3mu \mathrm{IVX}}).\label{eq:DIVX}
\end{equation}
To conduct statistical inference for $\beta^{*}$, we adopt the following
standard error that is robust to conditional heteroskedasticity:
\begin{equation}
\hat{\varsigma}^{\mkern3mu \mathrm{IVX}}=\left[\sum_{i=1}^{n}\left(\sum_{t=1}^{T-1}z_{i,t}^{2}\hat{e}_{i,t+1}^{\mkern3mu 2}-T\hat{\lambda}\bar{z}_{i}^{2}\hat{\omega}_{ee}\right)\right]^{1/2}\Bigg/\left|\sum_{i=1}^{n}\sum_{t=1}^{T-1}\tilde{z}_{i,t}x_{i,t}\right|,\label{eq:se IVX}
\end{equation}
where $\hat{\omega}_{ee}=(n(T-1))^{-1}\sum_{i=1}^{n}\sum_{t=1}^{T-1}\hat{e}_{i,t+1}^{\mkern3mu 2}$
estimates the variance of the error term $\omega_{ee}^{*}={\rm E}(e_{i,t}^{2})$,
and $\widehat{\lambda}=(1-n\hat{\varrho}_{ev}^{\mkern3mu 2}/T^{3/2})_{+}$
with $\hat{\varrho}_{ev}=\hat{\omega}_{ev}/\hat{\omega}_{ee}$ for
finite sample correction. The following remark elaborates the standard
error.

\begin{rem}[Standard Error]
\label{rem:finite sample correction} In time series, \citet[p.~1516]{Kostakis2015}
points out that IVX inference has finite-sample distortion caused
by the estimation of intercepts when the regressor is persistent.
In panel data, the individual-specific intercepts further worsen the
distortion. Equation~\eqref{eq:se IVX} mimics \citet{Kostakis2015}
for finite-sample correction. Thanks to IVX instrumentation, $\bar{z}_{i}$
appearing in the standard error is asymptotically negligible. Therefore,
it is asymptotically equivalent to calculating the standard error
either using the original IV $z_{i,t}$ or the within-group demeaned
$\tilde{z}_{i,t}$, and $\hat{\lambda}$ weights between these two
options. When the relative sample size $n/T^{3/2}$ or the ratio $\hat{\varrho}_{ev}^{\mkern3mu 2}$
is larger, the estimator $\hat{\beta}^{\mathrm{IVX}}$ suffers from
a larger first-order bias that is harder to remove, thereby requiring
a larger standard error with a smaller $\hat{\lambda}$ for robust
finite-sample performance. Compared to \citet{Kostakis2015}, our
finite-sample correction adopts an additional factor $n/T^{3/2}$
for panel data, motivated by the theoretical assumption $n/T^{3/2}\to0$
for persistent regressors; see Remark \ref{rem: n T1.5 LUR} below
for details. Most importantly, this finite-sample correction preserves
asymptotic guarantees in the leading case $n/T\to c>0$, under which
$\hat{\lambda}=O_{p}(1)$ ensures that the finite-sample correction
term is asymptotically negligible.
\end{rem}
Before diving into asymptotic theory, we summarize the procedures
of DIVX inference in the following Algorithm \ref{alg:ivx}.
\begin{lyxalgorithm}[Double IVX]
\label{alg:ivx}\mbox{}
\begin{description}[itemsep=0pt,leftmargin=*,parsep=0pt,topsep=0pt]
\item [{Step~1}]  (Slope coefficient estimation) Obtain $\bhativx$ in
\eqref{eq:beta_IVX}, and $\hat{\rho}^{\mkern3mu \mathrm{IVX}}$ in
\eqref{eq:rhoIVX}.
\item [{Step~2}] (Variance and covariance) Calculate $\hat{\omega}_{ee}=(n(T-1))^{-1}\sum_{i=1}^{n}\sum_{t=1}^{T-1}\hat{e}_{i,t+1}^{\mkern3mu 2}$
and $\hat{\omega}_{ev,h}=(nT)^{-1}\sum_{i=1}^{n}\sum_{t=1}^{T-h}\hat{v}_{i,t+h}\hat{e}_{i,t}$
for $h\geq0$, where the residuals $\hat{v}_{i,t}$ and $\hat{e}_{i,t}$
follow \eqref{eq:residuals}.
\item [{Step~3}] (Bias correction) Compute the Double IVX estimator $\hat{\beta}^{{\rm DIVX}}$
by \eqref{eq:DIVX} and its standard error by \eqref{eq:se IVX}.
\item [{Step~4}] (Confidence interval and hypothesis testing) Let $\Phi_{1-\alpha/2}$
be the $100(1-\alpha/2)$-th percentile of the standard normal distribution,
for example $\Phi_{0.975}=1.96$ for $\alpha=0.05$. The $100(1-\alpha)\%$
two-sided confidence interval is
\[
\Bigl(\hat{\beta}^{{\rm DIVX}}-\Phi_{1-\alpha/2}\cdot\hat{\varsigma}^{\mkern3mu \mathrm{IVX}},\ \hat{\beta}^{{\rm DIVX}}+\Phi_{1-\alpha/2}\cdot\hat{\varsigma}^{\mkern3mu \mathrm{IVX}}\Bigr).
\]
A null hypothesis $\mathbb{H}_{0}\colon\beta^{*}=\beta_{0}$ is rejected
under significance level $\alpha$ if $|t^{{\rm DIVX}}(\beta_{0})|>\Phi_{1-\alpha/2}$,
where the $t$-statistic is
\begin{equation}
t^{{\rm DIVX}}(\beta_{0})=\frac{\hat{\beta}^{{\rm DIVX}}-\beta_{0}}{\hat{\varsigma}^{\mkern3mu \mathrm{IVX}}}.\label{eq:t DIVX}
\end{equation}
\end{description}
\end{lyxalgorithm}

\section{Asymptotic Theory\label{sec:theory}}

We will compare the theoretical results of various estimators. In
Section \ref{sec:framework}, we have heuristically discussed the
WG and IVX estimators of the main regression coefficient $\beta^{*}$,
and the estimators of $\rho^{*}$. We refer to ``WG-IVX'' as the
process of the initial estimator $\hat{\beta}^{{\rm WG}}$ and subsequently
applying a bias correction based on $\hat{\rho}^{{\rm IVX}}$.\footnote{The bias formula of WG in (\ref{eq:fe_bias}) is deduced under the
simplified i.i.d.~assumption. The general bias formula with weakly
dependent AR(1) error $v_{i,t}$ in (\ref{eq:AR1}) is given as (\ref{eq:fe_bias_general})
in the simulation studies of Section \ref{sec:simulation}. } We also denote ``WG-WG'', ``IVX-WG'' in a parallel way, and ``IVX-IVX''
is our core method Double IVX. We will take \citet{phillips1999linear}'s
joint asymptotics that simultaneously sends both $n$ and $T$ to
infinity, with particular attention to the leading asymptotic case
of $n/T\to c>0$.

Figure \ref{fig:Diagram-of-applicability} uses traffic lights to
signify validity of the aforementioned four methods in the leading
asymptotic case. The amber lights under MI indicate that the validity
depends on the user's choice of $\theta$ relative to the degree of
persistence $\gamma$ in the DGP, but since $\gamma$ is unknown there
is no asymptotic guarantee. If either the main regression or the AR
regression is estimated by WG, red lights flash in all the LUR cases.
DIVX is the only procedure that secures green lights in all Cases
I--V.

\begin{figure}[t]
\begin{centering}
\includegraphics[width=0.7\columnwidth]{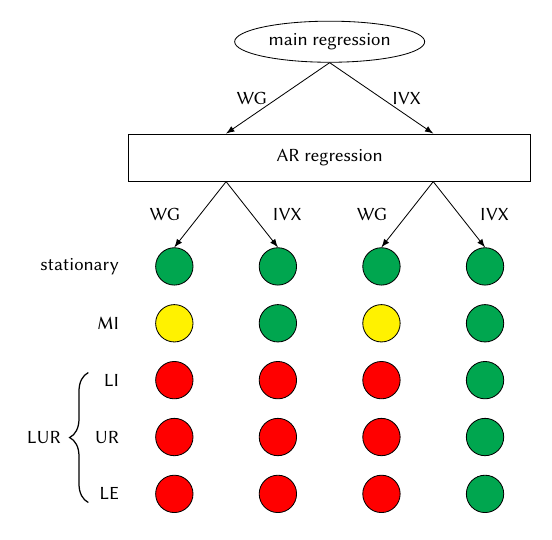}
\par\end{centering}
\centering{}\caption{\label{fig:Diagram-of-applicability} Diagram of asymptotic validity}
\end{figure}

In the following, Section \ref{subsec:theoryDIVX} establishes the
asymptotic guarantees of our proposed DIVX method, and Section \ref{subsec:Drawback other}
argues against the three alternatives.

\subsection{Asymptotic Guarantees of Double IVX\label{subsec:theoryDIVX}}

We first establish the asymptotic normality of panel IVX under infeasible
$\{\omega_{ev,h}^{*}\}$ and $\rho^{*}$.
\begin{prop}
\label{prop:bhativx} Under Assumptions~\ref{assump:initval} and
\ref{assump:innov}, as $(n,T)\to\infty$ we have
\begin{gather*}
\sqrt{nT^{1+(\theta\wedge\gamma)}}\left[\hat{\beta}^{\mathrm{IVX}}-\beta^{*}+b_{n,T}^{\mathrm{IVX}}(\{\omega_{ev,h}^{*}\},\rho^{*},\rho_{z})\right]\to_{d}\mathcal{N}\bigl(0,\Sigma^{\mathrm{IVX}}\bigr),\\
b_{n,T}^{\mathrm{IVX}}(\{\omega_{ev,h}^{*}\},\rho^{*},\rho_{z})=O_{p}\bigl(T^{-[2-(\theta\vee\gamma)]}\bigr),
\end{gather*}
where $b_{n,T}^{\mathrm{IVX}}(\{\omega_{ev,h}^{*}\},\rho^{*},\rho_{z})$
is defined in \eqref{eq:ivx_bias} and $\Sigma^{\mathrm{IVX}}$ in
\eqref{eq:SigmaIVX}.
\end{prop}
Proposition \ref{prop:bhativx} formulates the Nickell-Stambaugh bias
of $\hat{\beta}^{\mathrm{IVX}}$. Section \ref{subsec:IVX-as-the}
has highlighted the importance of an accurate estimator of $\rho^{*}$
with a sufficiently fast rate of convergence. The following Proposition
establishes the convergence rate of the IVX estimator $\hat{\rho}^{\mkern3mu \mathrm{IVX}}$.
Recall that $\theta_{1}$ is the essential parameter in the IVX estimator
$\hat{\rho}^{\mkern3mu \mathrm{IVX}}$ in (\ref{eq:rhoIVX}), and
we recommended $\theta_{1}=(1+\theta)/2$ as a convenient choice satisfying
$\theta<\theta_{1}<1$.
\begin{prop}
\label{prop:rho_convergence}Under Assumptions~\ref{assump:initval}
and \ref{assump:innov}, for any fixed $\gamma\in[0,1]$ we have
\[
\hat{\rho}^{\mkern3mu \mathrm{IVX}}-\rho^{*}=O_{p}\left(\frac{1}{\sqrt{nT^{1+(\theta_{1}\wedge\gamma)}}}+\frac{1}{T^{2(\theta_{1}\wedge\gamma)}}+\frac{G}{\sqrt{nT^{1+2(\theta_{1}\wedge\gamma)}}}+\frac{G}{T^{1+(\theta_{1}\wedge\gamma)}}\right)
\]
 as $(n,T)\to\infty$.
\end{prop}
\prettyref{prop:rho_convergence} conveys the key message: the IVX
estimator $\hat{\rho}^{\mkern3mu \mathrm{IVX}}$ enjoys a desirable
convergence rate, thereby enabling valid bias correction in the leading
asymptotic case. In Remark \ref{rem:verify_IVX_Cond} of the Supplementary
Materials, we derive that when $G=\lfloor T^{1/4}\rfloor$ and $\theta_{1}=(1+\theta)/2$,
the essential condition \eqref{eq:cond IVX bias correct} is satisfied.
Two remarks are in order to elaborate two special cases: LUR and stationary
panels.
\begin{rem}[Bias Correction for LUR]
\label{rem:bcLUR}In the LUR case where $\gamma=1$, the rate of
convergence amounts to $\hat{\rho}^{\mkern3mu \mathrm{IVX}}-\rho^{*}=O_{p}\bigl((nT^{1+\frac{1+\theta}{2}})^{-1/2}\bigr)$
and accommodates the requirement $\hat{\rho}-\rho^{*}=o_{p}\bigl((nT^{1+\theta})^{-1/2}\bigr)$
as discussed right after \eqref{eq:cond IVX bias correct}.
\end{rem}
\begin{rem}[Negligible Bias for Stationary Panels]
\label{rem:biasStationary} When $x_{i,t}$ is stationary with $\gamma=0$,
in the leading asymptotic case the condition \eqref{eq:cond IVX bias correct}
reduces to $\hat{\rho}-\rho^{*}=o_{p}(T^{(1-\theta)})$, and thus
$\hat{\rho}^{\mkern3mu \mathrm{IVX}}-\rho^{*}=O_{p}(1)$ is sufficient
for our purpose. Therefore, despite the fact that the AR(1) coefficient
$\rho^{*}$ cannot be consistently estimated due to endogeneity arising
from the misspecified AR(1) regression \eqref{eq:AR1} when $v_{i,t}$
is weakly dependent, the bias of $\hat{\beta}^{{\rm IVX}}$ is negligible
with stationary regressors and DIVX remains valid.
\end{rem}
Though the bias $b_{n,T}^{\mathrm{IVX}}(\{\omega_{ev,h}^{*}\},\rho^{*},\rho_{z})$
defined in (\ref{eq:ivx_bias}) involves $\{\omega_{ev,h}^{*}\}$,
its estimation in (\ref{eq:omega ev hat}) is straightforward using
the sample covariances of residuals, without affecting the asymptotic
guarantees of DIVX. Furthermore, the truncated bias formula (\ref{eq:truncated bias IVX})
well approximates the oracle bias formula due to the AR(1) error's
weak dependence depicted by the Assumption \ref{assump:innov}\ref{enu:LP}.
The following Theorem \ref{thm:DIVX} is our core theoretical result,
establishing the asymptotic normality of $\hat{\beta}^{{\rm DIVX}}$.
\begin{thm}
\label{thm:DIVX} Suppose that Assumptions~\ref{assump:initval}
and \ref{assump:innov} hold. Under $G=\lfloor T^{1/4}\rfloor$ and
$\theta_{1}=(1+\theta)/2$, if $(n,T)\to\infty$ and $n/T\to c\in[0,\infty)$,
we have
\[
\frac{\hat{\beta}^{{\rm DIVX}}-\beta^{*}}{\hat{\varsigma}^{\mkern3mu \mathrm{IVX}}}\to_{d}\mathcal{N}(0,1),
\]
where $\hat{\varsigma}^{\mkern3mu \mathrm{IVX}}$ is given by \eqref{eq:se IVX}.
In addition, the standard error $\hat{\varsigma}^{\mkern3mu \mathrm{IVX}}=O_{p}\bigl(1/\sqrt{nT^{1+(\theta\wedge\gamma)}}\bigr)$.
\end{thm}
\prettyref{thm:DIVX} unifies the inference procedure under the polynomial
rate $\rho^{*}=1+c^{*}/T^{\gamma}$. This result shows that DIVX delivers
valid inference with standard critical values, without requiring prior
knowledge of the regressor's persistence.
\begin{rem}[Local Power]
\label{rem:local power} The order of the standard error $O_{p}\bigl(1/\sqrt{nT^{1+(\theta\wedge\gamma)}}\bigr)$
suggests that DIVX has high power to detect local-to-zero violation
of the null hypothesis. In particular, the test with DIVX is consistent
whenever the true coefficient violates the null hypothesis $H_{0}\colon\beta^{*}=\beta_{0}$
in Algorithm \ref{alg:ivx} with $|\beta^{*}-\beta_{0}|\gg1/\sqrt{nT^{1+(\theta\wedge\gamma)}}$.
For LUR with $\gamma=1$, the rate $\sqrt{nT^{1+\theta}}$ is very
close to the \emph{optimal} $\sqrt{n}T$-consistency discussed in
Remark \ref{rem:rho_opt_rate} with the choice $\theta=0.95$. For
stationary panels with $\gamma=0$, DIVX achieves the standard $\sqrt{nT}$-consistency.
\end{rem}
In fact, we can further enhance it to achieve asymptotic normality
\emph{uniformly }in $\rho^{*}$ using the drifting parameter sequence
approach in \citet{andrews2020generic}.
\begin{cor}
\label{cor:DIVX_uniform} Fix three absolute constants $m_{1}^{*}\in(0,1)$,
$m_{2}^{*}\in(0,\infty)$, and $\alpha\in[0,1]$. Define $B_{T}=[-1+m_{1}^{*},1+m_{2}^{*}/T]$.
The conditions in Theorem \ref{thm:DIVX} yield
\[
\sup_{\rho^{*}\in B_{T}}\left|\Pr\left\{ t^{{\rm \mathrm{DIVX}}}(\beta^{*})<\Phi^{-1}\left(\alpha\right)\right\} -\alpha\right|\to0,
\]
where $\Phi\left(\cdot\right)$ is the cumulative distribution function
of \textup{$\mathcal{N}(0,1)$}\textup{\emph{, and the $t$-statistic
$t^{\mathrm{DIVX}}$ is defined in \eqref{eq:t DIVX}.}}
\end{cor}
\begin{rem}
The compact support $B_{T}=[-1+m_{1}^{*},1+m_{2}^{*}/T]$ for the
admissible $\rho^{*}$ is a sequence of closed sets. The left-end
is invariant and bounded away from $-1$, whereas the right-end exceeds
but converges to $1$ as in the LE case. Inside such sequence of closed
sets $\rho^{*}$ can be an arbitrary sequence. This uniform result
is more general and flexible than the convergent sequences specified
in \eqref{eq:rate rho}.
\end{rem}
\begin{rem}
\label{rem: n T1.5 LUR}When the regressor is LUR, DIVX admits a wider
range of asymptotic sequences satisfying $n/T^{3/2}\to0$, more general
than the relative rate specified in Theorem \ref{thm:DIVX}. Intuitively,
the large variance of the LUR regressor dominates the endogeneity
caused by the weak dependence of the AR(1) error $v_{i,t}$. It validates
the bias correction for the IVX estimator \eqref{eq:rhoIVX} by the
long-run covariance. The AR(1) coefficient $\rho^{*}$ becomes consistently
estimable by IVX with fast convergence rate displayed in Proposition
\ref{prop:rho_convergence}. This is in sharp contrast to the stationary
case, under which $\rho^{*}$ cannot be consistently estimated, as
discussed in Remark \ref{rem:biasStationary}. Interested readers
may refer to \eqref{eq:rIVX Case 2} in the proof of Theorem \ref{thm:DIVX}
in the supplementary materials for technical details.
\end{rem}
We have shown that, in the simple predictive regression system \eqref{eq:predictive}
and \eqref{eq:AR1} with a scalar $x_{i,t}$, DIVX achieves unified
and uniform inference for the parameter $\beta^{*}$ of interest.
Admittedly, this model is simplistic in order to illustrate the ideas.
Our Online Appendices extend the DIVX estimation and inference into
four empirically-oriented scenarios, including two-way fixed effects,
multiple panel predictive regressions, local projections with multiple
horizons, and cross-sectional heterogeneity with latent group structures.
Please refer to Sections \ref{sec:appdx two-way}- \ref{sec:Grouped-Heterogeneity}
in the Online Appendices for models and methodologies, and Section
\ref{sec:Additional-Simulations} for additional simulations of these
extensions.

\subsection{Failures of the Alternative Estimators\label{subsec:Drawback other}}

To convey the negative message about the other estimator, we use a
simplifying i.i.d.\ condition on the error terms $e_{i,t}$ and $v_{i,t}$,
stated in the following \prettyref{assump:iid} to replace \prettyref{assump:innov}.
Furthermore, we pretend that $\omega_{ev}^{*}$ is known, so we only
need to focus on the estimation of $\rho^{*}$. If an estimator does
not work under such a special case, it is expected to stay invalid
under general data generating processes (DGPs).
\begin{assumptionprimep}[Innovations]{\ref*{assump:innov}}
\label{assump:iid} The error terms $\{(e_{i,t},v_{i,t})^{\prime}\}$
are i.i.d.\ across both $i$ and $t$, and have finite fourth moment,
i.e., $\epct(e_{i,t}^{4}+v_{i,t}^{4})<C<\infty$ for some positive
constant $C$.
\end{assumptionprimep}
Assumption \ref{assump:iid} simplifies $\omega_{ev,h}^{*}=0$ for
all $h\geq1$, and the bias formula of $\hat{\beta}^{{\rm IVX}}$
reduces to $\omega_{ev}^{*}\cdot\tilde{b}_{n,T}^{\ivx}\bigl(\rho^{*}\bigr)$,
where
\[
\tilde{b}_{n,T}^{\ivx}\bigl(\rho\bigr)=\frac{n}{T-1}\cdot\frac{\sum_{t=2}^{T-1}\sum_{s=2}^{t}\rho_{z}^{t-s}\rho^{s-2}}{\sum_{i=1}^{n}\sum_{t=1}^{T-1}\tilde{z}_{i,t}x_{i,t}}.
\]
We then look at the scenario when WG is used in the panel AR regression
\eqref{eq:AR1} and the bias correction for $\hat{\beta}^{{\rm IVX}}$.
With a known $\omega_{ev}^{*}$, the IVX-WG estimator by plugging
$\rohatfe$ into the bias function is
\[
\hat{\beta}^{\,\textrm{IVX-WG}}=\bhativx+\omega_{ev}^{*}\cdot\tilde{b}_{n,T}^{\ivx}\bigl(\hat{\rho}^{{\rm WG}}\bigr).
\]
We immediately obtain the following corollary of Proposition \ref{prop:bhativx}.
\begin{cor}
\label{cor:ivx-wg}Under \prettyref{assump:initval} and \ref{assump:iid},
as $\joto$ we have
\begin{equation}
\rohatfe-\rho^{*}=O_{p}\biggl(\frac{1}{\sqrt{nT^{1+\gamma}}}+\frac{1}{T}\biggr).\label{eq:rho WG rate}
\end{equation}
Therefore, if $n/T^{5-(\theta\vee\gamma)-\theta-3\gamma}\to0$, then
\[
(\hat{\beta}^{\textup{IVX-WG}}-\beta^{*})/\widehat{\varsigma}^{{\rm IVX}}\dto\ncal(0,1).
\]
\end{cor}
\begin{rem}[IVX-WG excludes the leading case]
The $1/T$ term in the convergence rate of $\rohatfe$ arises from
the Nickell-Stambaugh bias in panel AR, and is slower than that of
$\hat{\rho}^{{\rm IVX}}$ in \prettyref{prop:rho_convergence}. It
leads to a much narrower range of $n$ and $T$ for asymptotic normality.
In particular, in the leading asymptotic case $n/T\to c>0$, when
$\gamma=1$ the expansion rate condition for $\hat{\beta}^{\,\textup{IVX-WG}}$
is violated since $n/T^{5-(\theta\vee\gamma)-\theta-3\gamma}=n/T^{1-\theta}\to\infty$.
\end{rem}
We then turn to the WG estimator $\hat{\beta}^{{\rm WG}}$ for the
main regression. As explained in Section \ref{subsec:Inconvenience-of-WG=0000DF},
in panel predictive regressions the Nickell-Stambaugh bias of WG is
severe.
\begin{prop}
\label{prop:bhatfe} Under \prettyref{assump:initval} and \ref{assump:iid},
as $(n,T)\to\infty$ we have
\begin{gather*}
\sqrt{nT^{1+\gamma}}\bigl[\bhatfe-\beta^{*}+\omega_{ev}^{*}\cdot b_{n,T}^{\fe}(\rho^{*})\bigr]\dto\ncal\bigl(0,\Sigma^{\fe}\bigr),\\
b_{n,T}^{\fe}(\rho^{*})=O_{p}\left(1/T\right),
\end{gather*}
 where the variance $\Sigma^{\fe}$ is laid out in \eqref{eq:fe_var}.
Furthermore,
\[
\left[\bhatfe-\beta^{*}+\omega_{ev}^{*}\cdot b_{n,T}^{\fe}(\rho^{*})\right]\Big/\varsigma^{\mathrm{WG}}\dto\ncal(0,1)
\]
 as $(n,T)\to\infty$, where
\begin{equation}
\varsigma^{{\rm WG}}=\frac{\sqrt{n\cdot\mathrm{var}\bigl(\sum_{t=1}^{T-1}\tilde{x}_{i,t}e_{i,t+1}\bigr)}}{\sum_{i=1}^{n}\sum_{t=1}^{T-1}\tilde{x}_{i,t}^{2}}.\label{eq:fe_se-1}
\end{equation}
\end{prop}
With i.i.d.~errors, \prettyref{prop:bhatfe} characterizes the stochastic
order of the bias, which is proportional to $b_{n,T}^{\fe}(\rho^{*})$.
From this proposition we have $\bhatfe-\beta^{*}=O_{p}\bigl(1/\sqrt{nT^{1+\gamma}}+1/T\bigr)$,
which means that the WG estimator is consistent when both $n$ and
$T$ pass to infinity. However, the main focus of predictive regressions
lies in the inference to determine whether the variable $x_{i,t}$
retains predictive power to the targeted dependent variable, and mere
consistency is insufficient for this purpose. The bias $\sqrt{nT^{1+\gamma}}b_{n,T}^{\fe}(\rho^{*})=O_{p}\bigl(\sqrt{n/T^{1-\gamma}}\bigr)$
can diverge to infinity and dominate the variance of $\bhatfe$ when
$\gamma=1$.

Eliminating the bias in WG is a challenging task, in particular when
the regressor is highly persistent. We try $\rohatfe$ and $\hat{\rho}^{\mkern3mu \mathrm{IVX}}$
for bias correction, and the two respective estimators are given as
\[
\hat{\beta}^{\textup{WG-WG}}=\bhatfe+\omega_{ev}^{*}\cdot b_{n,T}^{\fe}(\rohatfe)\quad\text{and}\quad\hat{\beta}^{\textup{WG-IVX}}=\bhatfe+\omega_{ev}^{*}\cdot b_{n,T}^{\fe}(\hat{\rho}^{\mkern3mu \mathrm{IVX}}).
\]
The following proposition summarizes the asymptotics of these two
estimators.
\begin{prop}
\label{prop:FE-fail} Suppose \prettyref{assump:initval} and \ref{assump:iid}
hold, and $\joto$.
\begin{enumerate}
\item If $n/T^{3(1-\gamma)}\to0$, then $(\hat{\beta}^{\,\textup{WG-WG}}-\beta^{*})/\varsigma^{{\rm WG}}\dto\ncal(0,1).$
\item Suppose that $\theta_{1}>3/4$. If $n/T\to c\in[0,\infty)$ and $1/T^{1-\gamma}\to0$,
then $(\hat{\beta}^{\,\textup{WG-IVX}}-\beta^{*})/\varsigma^{{\rm WG}}\dto\ncal(0,1).$
\end{enumerate}
\end{prop}
We have stated in \eqref{eq:rho WG rate} that $\rohatfe$'s rate
of convergence is $O_{p}\bigl(1/\sqrt{nT^{1+\gamma}}+1/T\bigr)$,
which reflects the Nickell-Stambaugh bias in the panel AR. The asymptotic
bias vanishes when $n/T^{3(1-\gamma)}\to0$ for the $t$-statistics
based on WG-WG. However, this is not helpful for unified inference,
as it obviously rules out the LUR cases with the persistence index
$\gamma=1$, under which $n/T^{3(1-\gamma)}=n\to\infty$.

In \prettyref{prop:FE-fail}(ii), the addition condition $\theta_{1}>3/4$
is merely to simplify the exposition of the complex asymptotic regimes
without impacting the asymptotic validity. This condition is consistent
with our recommendation $\theta_{1}=0.975$ that is close to one to
maintain fast convergence of IVX. WG-IVX based on $\hat{\rho}^{\mkern3mu {\rm IVX}}$
tightens the valid asymptotic regime of the $t$-statistic from $n/T^{3(1-\gamma)}\to0$
in \prettyref{prop:FE-fail}(i) to allow for $n/T\to c\in[0,\infty)$
but under the restrictive condition $1/T^{1-\gamma}\to0$ in \prettyref{prop:FE-fail}(ii).
Though this is a substantial enhancement, it still rules out $\gamma=1$
in the leading asymptotic case, under which $1/T^{1-\gamma}=1\not\to0$.
It is therefore also undesirable for statistical inference.

This section makes it clear that if either the main regression \eqref{eq:predictive}
or the AR regression \eqref{eq:AR1} is estimated by WG, unified inference
is not achievable unless $n$ is much smaller than $T$. Since the
theory of WG does not cover the leading asymptotic case, we do not
recommend using WG for panel predictive regressions. The simulations
in Section \ref{sec:simulation} will illustrate the validity of DIVX
and unsatisfactory performance of WG in finite samples, which are
in line with our theoretical results.

\section{Simulations \label{sec:simulation}}

\subsection{Baseline Setup and Results}

In this section we conduct Monte Carlo simulations for the six estimators
covered in the theoretical section, including the vanilla WG, WG-WG,
WG-IVX, the vanilla IVX, IVX-WG, and finally, our recommended DIVX.
We consider panels with $n\in\{50,100,200\}$ and $T\in\{100,200,500\}$.
The relatively large time span $T$ demonstrates that the biases of
the alternative methods are not merely finite-sample issues.

For the DGP of the predictive regression \eqref{eq:predictive}, we
generate the dependent variable by setting the true coefficient $\beta^{*}=0$;
that is, $x_{i,t}$ has no predictive power for $y_{i,t+1}$. We set
the AR(1) coefficient in (\ref{eq:state space}) as $\rho^{*}\in\{0.6,1-1/T^{0.75},1-1/T,1,1+1/T\}$
to reflect various degrees of the regressor's persistence. They are
the finite sample embodiment of the stationary (ST), MI, LI, UR, and
LE regressors, respectively.

The fixed effects $\mu_{y,i}=T^{-1}\sum_{t}x_{i,t}$, is specified
to be correlated with the regressor. The drift $\alpha_{i}$ and the
initial value $\delta_{i,0}$ in (\ref{eq:state space}) are both
independently drawn from $\mathcal{N}(0,1)$. To showcase the validity
of DIVX when the AR(1) model (\ref{eq:AR1}) is misspecified, we generate
the AR(1) error $v_{i,t}=0.5v_{i,t-1}+\varepsilon_{i,t}+0.4\varepsilon_{i,t-1}$
from a stationary ARMA(1,1) process. The i.i.d.~innovations $\varepsilon_{i,t}$
and the error term $e_{i,t}$ in the main regression \eqref{eq:predictive}
are generated from a bivariate normal distribution
\begin{equation}
\begin{pmatrix}e_{i,t}\\
\varepsilon_{i,t}
\end{pmatrix}\sim{\rm i.i.d.}\ \mathcal{N}\left(\begin{pmatrix}0\\
0
\end{pmatrix},\begin{pmatrix}1 & \omega_{12}^{*}\\
\omega_{12}^{*} & 1
\end{pmatrix}\right),\label{eq:omega12 def}
\end{equation}
where $\omega_{12}^{*}$ measures the strength of correlation. In
this section, we specify the contemporaneous correlation as $\omega_{12}^{*}=-0.95$
to produce strong negative correlation between the two error terms,
which characterizes the typical case in stock-return predictive regressions
\citep{Kostakis2015,phillips2016robust}. Additional simulations with
a variety of $\omega_{12}^{*}$ values are relegated to Section \ref{sec:Additional-Simulations}
of the Online Appendices.

For the IVX-based estimators, we adopt \citet{Kostakis2015}'s choices
of $c_{z}=-1$ and $\theta=0.95$ in the user-specified persistence
index $\rho_{z}=1+c_{z}/T^{\theta}$ as in (\ref{eq:def IV}), and
$\theta_{1}=(1+\theta)/2=0.975$ for the IVX instrumentation (\ref{eq:def IV (1)});
moreover, in (\ref{eq:NW Deltavv}) and (\ref{eq:truncated bias IVX})
$G=\lfloor T^{1/4}\rfloor$ for the long-run variance. These choices
are fixed throughout the simulation studies here and the empirical
application in Section \ref{sec:Empirical-Application}.

For the alternative estimators with WG in either stage, we conduct
bias corrections parallel to DIVX. Though we assume i.i.d.~innovations
and a known covariance $\omega_{ev}^{*}$ for simplicity of theoretical
discussions in Section \ref{subsec:Drawback other}, in the numerical
studies we keep agnostic about the DGP and the parameters in the bias
formula. Specifically, we re-define the IVX-WG estimator as $\hat{\beta}^{\text{IVX-WG}}=\hat{\beta}^{\mathrm{IVX}}+\hat{b}_{n,T}^{\mathrm{IVX}}(\hat{\boldsymbol{\omega}}_{ev,G},\hat{\rho}^{{\rm WG}})$,
where the truncated bias $\hat{b}_{n,T}^{\mathrm{IVX}}$ follows (\ref{eq:truncated bias IVX})
as in DIVX. The covariance estimators $\hat{\boldsymbol{\omega}}_{ev,G}=\{\hat{\omega}_{ev,h}\}_{h=0}^{G}$
in (\ref{eq:omega ev hat}) are tailored with the AR(1) residual $\hat{v}_{i,t}$
in (\ref{eq:residuals}) replaced by $\hat{v}_{i,t}^{{\rm WG}}=\tilde{x}_{i,t}-\hat{\rho}^{{\rm WG}}\tilde{x}_{i,t-1}$.
The standard errors of the vanilla IVX, IVX-WG, and DIVX follow (\ref{eq:se IVX}).
In addition, the bias formula of $\hat{\beta}^{{\rm WG}}$ is
\begin{equation}
\hat{b}_{n,T}^{{\rm WG}}(\boldsymbol{\omega}_{ev,G}^{*},\rho^{*})=\frac{n}{T}\cdot\frac{\sum_{h=0}^{G}\Phi_{h}\omega_{ev,h}^{*}}{\sum_{i=1}^{n}\sum_{t=1}^{T-1}\tilde{x}_{i,t}^{2}},\label{eq:fe_bias_general}
\end{equation}
where
\[
\Phi_{h}=\frac{1}{2}(T-h-1)(T-h)\cdot\bm{1}(\rho^{*}=1)+\frac{1}{1-\rho^{*}}\left(T-h-1-\frac{\rho^{*}-\rho^{*T-h}}{1-\rho^{*}}\right)\bm{1}(\rho^{*}\neq1).
\]
Similarly, we re-define the WG-WG and WG-IVX estimators as $\hat{\beta}^{\text{WG-WG}}=\hat{\beta}^{\mathrm{WG}}+\hat{b}_{n,T}^{\mathrm{WG}}(\hat{\boldsymbol{\omega}}_{ev,G},\hat{\rho}^{{\rm WG}})$
and $\hat{\beta}^{\text{WG-IVX}}=\hat{\beta}^{\mathrm{WG}}+\hat{b}_{n,T}^{\mathrm{WG}}(\hat{\boldsymbol{\omega}}_{ev,G},\hat{\rho}^{{\rm IVX}}).$
We tailor the covariance estimators $\hat{\boldsymbol{\omega}}_{ev,G}$
with the residual of the main regression model $\hat{e}_{i,t}$ in
(\ref{eq:residuals}) replaced by $\hat{e}_{i,t}^{{\rm WG}}=\tilde{y}_{i,t}-\hat{\beta}^{{\rm WG}}\tilde{x}_{i,t-1}$
for WG-IVX, and with both residuals replaced by $\hat{e}_{i,t}^{{\rm WG}}$
and $\hat{v}_{i,t}^{{\rm WG}}$ for WG-WG. The standard error for
$\hat{\beta}^{{\rm WG}}$ and its bias-corrected variants is $\hat{\varsigma}^{\mkern2mu {\rm WG}}=\sqrt{n\cdot\sum_{t=1}^{T-1}\tilde{x}_{i,t}^{2}(\hat{e}_{i,t+1}^{{\rm WG}})^{2}}/\sum_{i=1}^{n}\sum_{t=1}^{T-1}\tilde{x}_{i,t}^{2}$.

\begin{figure}[ph]
\begin{centering}
\includegraphics[width=0.95\columnwidth]{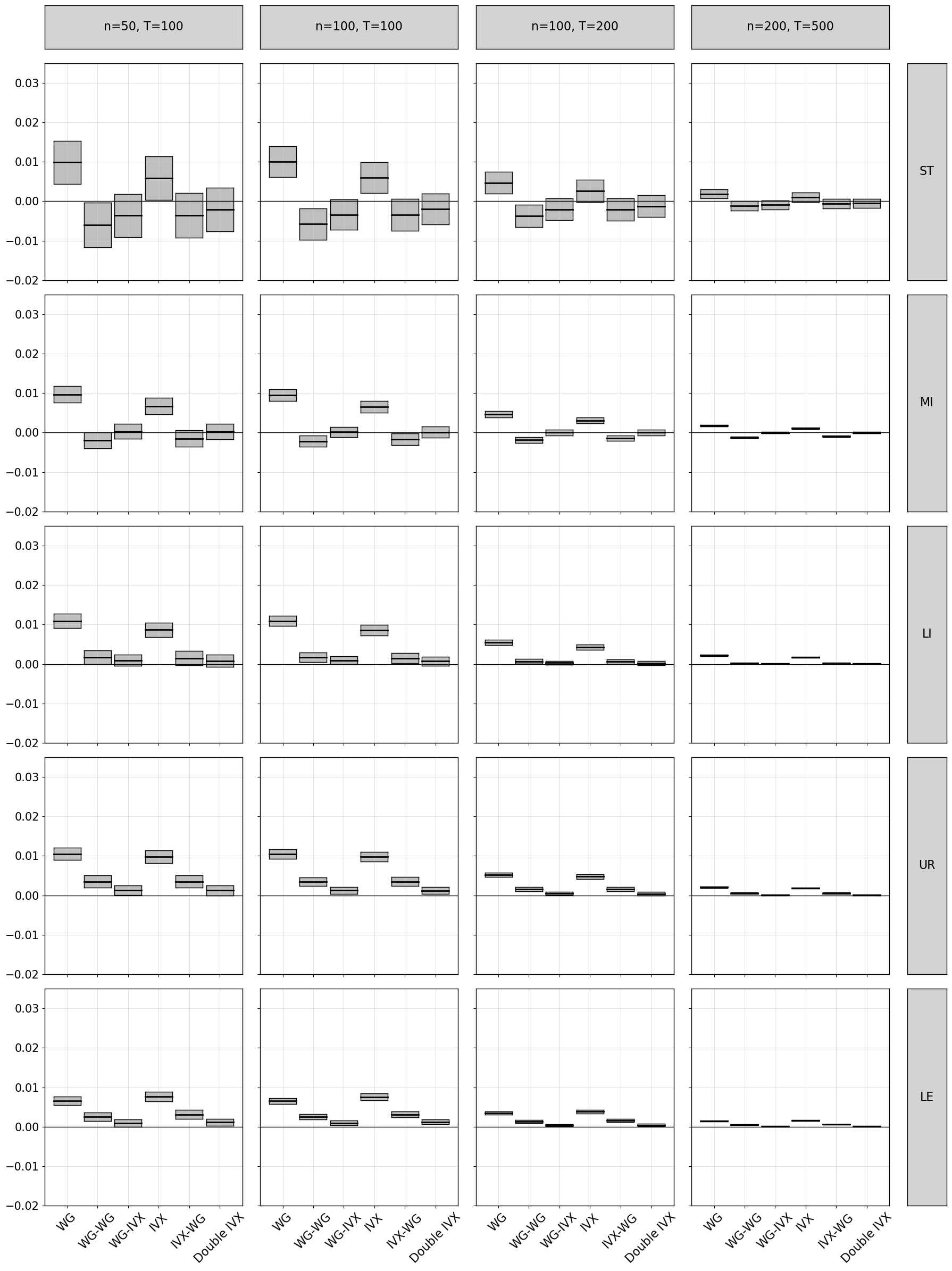}\smallskip{}
\noindent\begin{minipage}[t]{1\linewidth}%
{\footnotesize Notes: In each small box, the central line indicates
the empirical bias of $\hat{\beta}$, and the total height is twice
the empirical standard deviation, marking the lower and upper limits
$(\hat{\beta}-{\rm s.d.},\hat{\beta}+{\rm s.d.})$. To save space,
this figure only exhibits the results under $(n,T)\in\{(50,100),(100,100),(100,200),(200,500)\}$.}%
\end{minipage}
\par\end{centering}
\centering{}\caption{\label{fig:bias_sd} Bias and standard deviation}
\end{figure}

All simulations are repeated $1000$ times. Figure \ref{fig:bias_sd}
presents the relative point estimation performances. The center of
each bar is the empirical bias, with the height equal to twice the
empirical standard deviation. It is obvious that the vanilla WG and
IVX are severely biased, and the bias is more substantial when $n$
is large relative to $T$. All the bias correction methods are helpful
in mitigating the bias. DIVX is competitive in terms of bias correction,
and well centered around the true value under all five cases of $\rho^{*}$
as the sample size gets large.

\begin{figure}[tp]
\begin{centering}
\includegraphics[width=1\columnwidth]{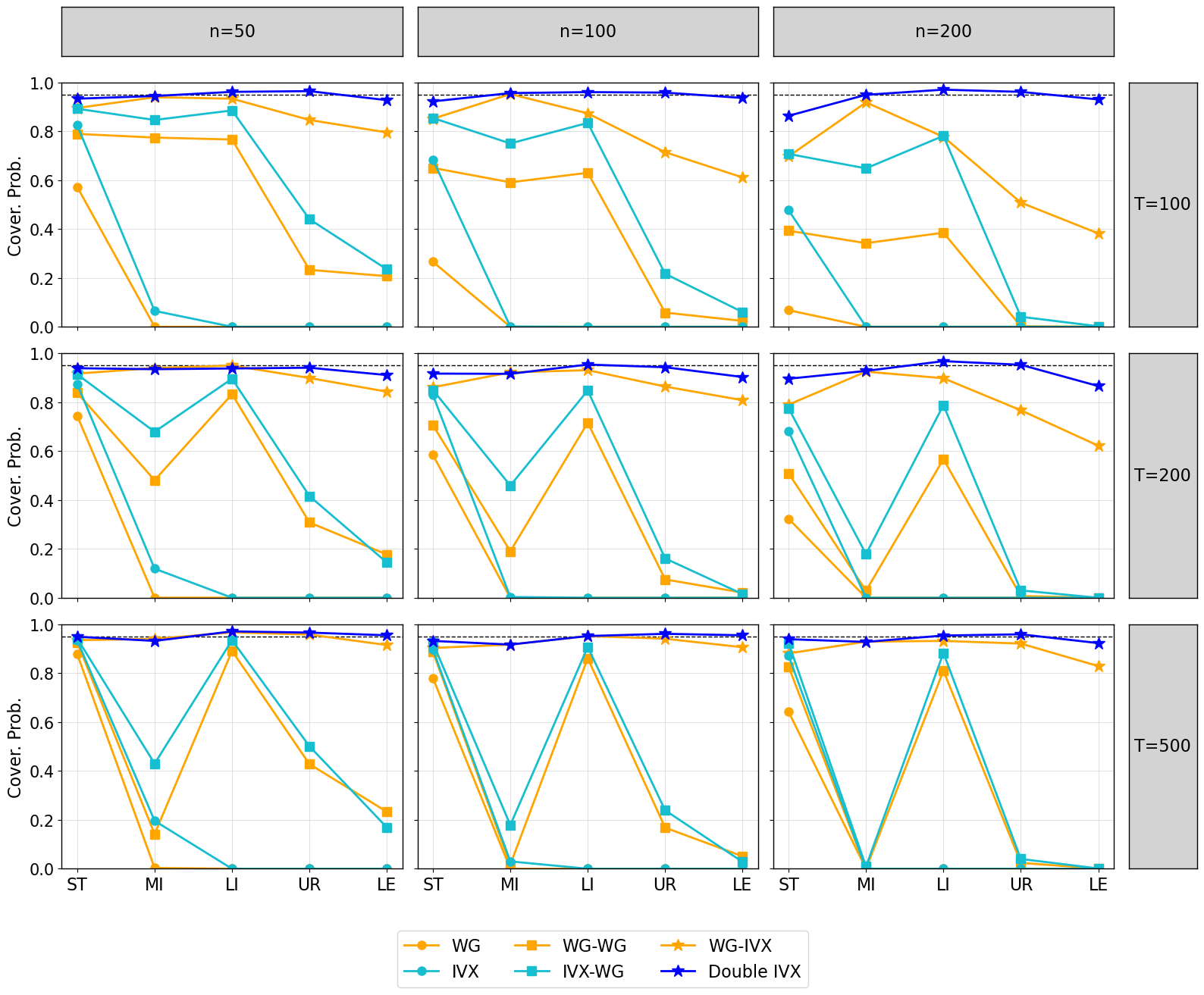}\smallskip{}
\par\end{centering}
\centering{}\caption{\label{fig:coverage} Coverage probabilities of 95\% confidence intervals
when $\omega_{12}^{*}=-0.95$}
\end{figure}

The necessity of bias correction is salient in statistical inference.
Figure \ref{fig:coverage} plots the empirical coverage probability
of the 95\% confidence intervals. We first focus on $T=100$. Obviously,
the vanilla WG and IVX confidence intervals fail to work, and the
distortion is more severe as the relative sample size $n/T$ gets
larger. When the regressor is highly persistent, such distortion in
WG cannot be fixed by either the WG or IVX estimator of $\rho^{*}$
in the bias correction formula. When $T$ gets larger, the bias does
not vanish---even at $T=500$, the distortions of coverage probabilities
of the WG-based estimators are still severe, especially in the UR
and LE cases where the regressor is highly persistent. DIVX inference
stands out with the empirical coverage probabilities close to the
nominal 95\% level in all scenarios.

\begin{figure}[tp]
\begin{centering}
\includegraphics[width=1\columnwidth]{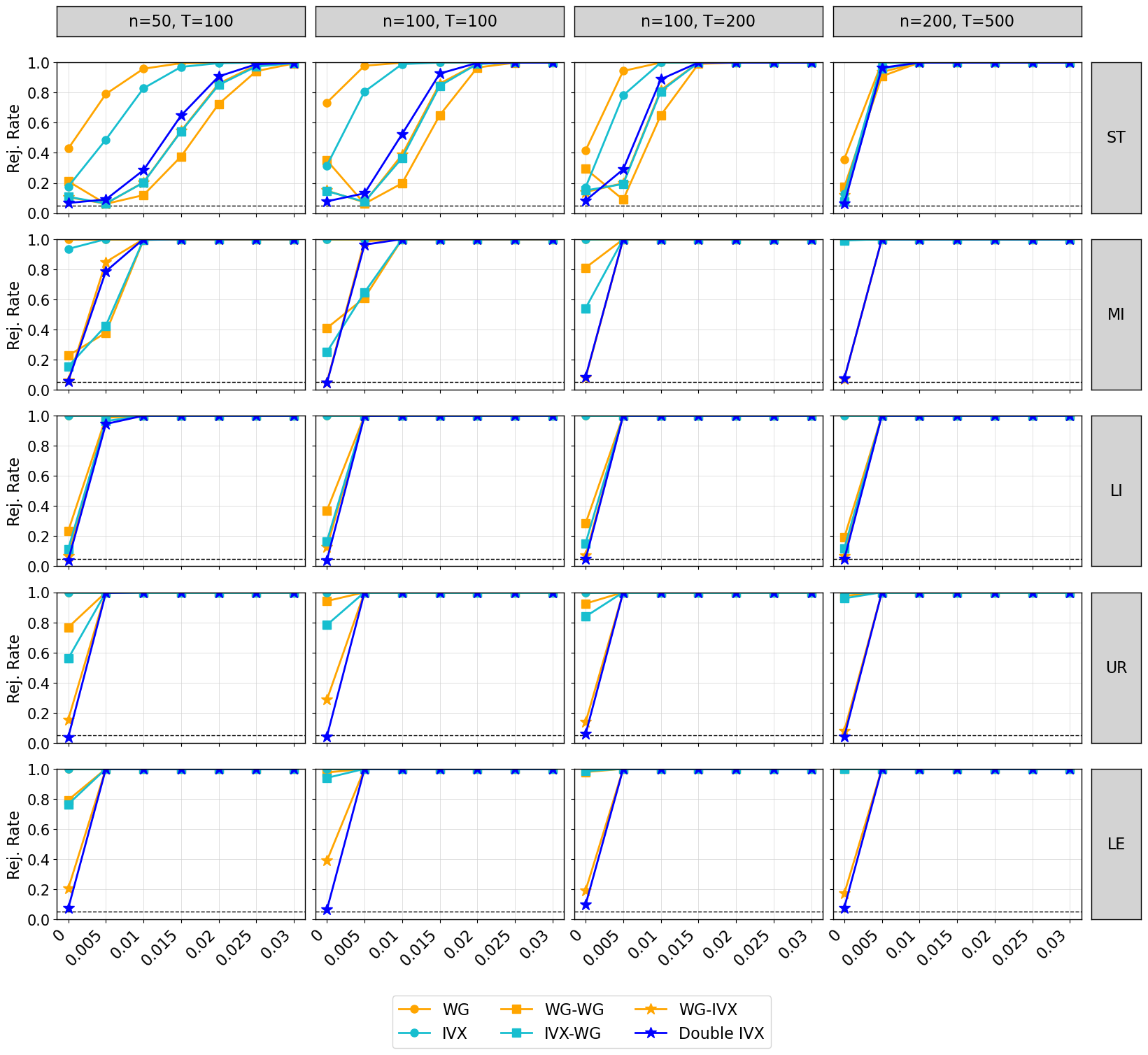}\smallskip{}
\par\end{centering}
\begin{centering}
\noindent\begin{minipage}[t]{1\linewidth}%
{\footnotesize Notes: To save space, this figure only exhibits the
rejection rates under the sample sizes $(n,T)\in\{(50,100),(100,100),(100,200),(200,500)\}$.}%
\end{minipage}
\par\end{centering}
\centering{}\caption{\label{fig:power} Rejection rates for $\mathbb{H}_{0}:\beta^{*}=0$
at the 5\% level when $\omega_{12}^{*}=-0.95$}
\end{figure}

We then turn to the test power. Figure \ref{fig:power} displays the
rejection rates for $\mathbb{H}_{0}:\beta^{*}=0$ at the 5\% level
when the true coefficient $\beta^{*}\in\{0,0.005,0.01,...,0.03\}$.
When $x_{i,t}$ is stationary with $\rho^{*}=0.6$, only WG-IVX and
DIVX inference exhibit correct empirical sizes when $\beta^{*}=0$.
These two estimators have similar empirical power that is competitive
even compared to the estimators with substantial biases. The rejection
rates achieve 100\% when $\beta^{*}$ reaches 0.025 with all sample
sizes under consideration. In the cases of MI, LI, UR, and LE where
the regressor is persistent, DIVX is the only method that achieves
accurate empirical sizes in all cases under $\beta^{*}=0$, with the
rejection rates achieving 100\% when $\beta^{*}\geq0.01$ in all scenarios.

In summary, DIVX boasts competitive performance in terms of point
estimation, and when it comes to coverage probabilities it is the
only estimator that demonstrates asymptotic validity in all settings,
with high empirical power to detect predictability of the outcome
$y_{i,t+1}$ using the regressor $x_{i,t}$ of different degrees of
persistence.

We carry out additional simulation studies in Section \ref{sec:Additional-Simulations}
of the Online Appendices to showcase the robustness of DIVX inference.
First, we evaluate DIVX under different degrees of endogeneity by
varying the $\omega_{12}^{*}$ in (\ref{eq:omega12 def}) in Section
\ref{subsec:various dgp error}. Second, we check the validity of
DIVX under conditional heteroskedasticity, and DIVX remains robust
(Section \ref{subsec:Conditional-Heteroskedasticity}). Third, we
also conduct simulation studies for the extensions of DIVX to address
two-way fixed effects (Section \ref{subsec:Two-way-Fixed-Effects}),
multiple regressions (Section \ref{subsec:Multivariate-Regressors}),
local projections (Section \ref{subsec:Local-Projections}), and latent
group structures (Section \ref{subsec:Latent-Group-Structures}).
Last, in addition to the WG- and IVX-based estimators, our Section
\ref{subsec:comparison} compares DIVX to other popular estimators
discussed in the Introduction. All these additional simulations consistently
showcase the excellent performance of DIVX.

\section{Predictability of Global Stock Returns \label{sec:Empirical-Application}}

The financial economics literature has long debated the predictive
power of valuation ratios for stock returns using time series \citep{campbell2006efficient,welch2008comprehensive,zhu2014predictive,goyal2024comprehensive}
and cross-country panel data \citep{hjalmarsson2008stambaugh,hjalmarsson2010predicting,westerlund2017testing}.
Most empirical applications of panel predictive regressions show that
valuation ratios, like dividend- and earnings-price ratios, exhibit
little predictive power for global stock returns. This section revisits
the return predictability using various valuation ratios, including
earnings-price ratio (EP), dividend-price ratio (DP), and sales-price
ratio (SP).

We collect monthly data of composite stock price indices, EP, DP,
and SP in 16 developed economies from September 2015 to July 2025.\footnote{Data source: Wind Information. The countries and regions include Australia,
Canada, Denmark, France, Germany, Hong Kong, Italy, Japan, Netherlands,
New Zealand, Singapore, Spain, Sweden, Switzerland, the United Kingdom,
and the United States.} First, we focus on the univariate panel predictive model
\[
r_{i,t}=\mu_{i}+\beta^{*}\log({\rm VR}_{i,t-1})+e_{i,t},
\]
 where $r_{i,t}$ is the log return of the composite stock price index,
and ${\rm VR}_{i,t}$ denotes the valuation ratio of either DP, EP,
or SP. Table \ref{tab:sum_stat} displays the summary statistics.
The sample correlation coefficients of the regression residuals $\hat{e}_{i,t}$
and $\hat{v}_{i,t}$ for all three valuation ratios are negative,
consistent with our baseline simulation setup in Section \ref{sec:simulation}.
The IVX estimates defined in \eqref{eq:rhoIVX} for $\rho^{*}$ are
0.998, 0.995, and 0.996 for log EP, log DP, and log SP, respectively.
In addition, we perform the panel unit root test using the $P_{m}$
test statistic \citep{choi2001unit}. The $p$-values are 0.935, 0.707,
0.504, suggesting no evidence to reject nonstationarity. The high
persistence calls for bias correction.

\begin{table}[htp]
\centering
\caption{Summary statistics}
\label{tab:sum_stat}
\begin{tabular}{lccccc}
\hline\hline
                 & Mean   & SD & $\widehat{r}_{ev}$ & $\widehat\rho^{\rm IVX}$ & $P_m$ test $p$-value \\
\hline
$\log({\rm EP}_{i,t-1})$ & -2.856 & 0.364              & -0.402  & 0.998                                               & 0.935                  \\
$\log({\rm DP}_{i,t-1})$ & 1.113  & 0.346              & -0.644  & 0.995                                               & 0.707                  \\
$\log({\rm SP}_{i,t-1})$ & -0.419 & 0.406              & -0.765  & 0.996                                               & 0.504              \\
\hline\hline
\end{tabular}
\begin{flushleft}
{\smaller Notes: ``Mean'' and ``SD'' represent the sample mean and standard deviation. $\hat{r}_{ev} = \frac{ \sum_{i=1}^n \sum_{t=1}^T  \widehat e_{i,t}\widehat v_{i,t}}{\sqrt{\sum_{i=1}^n \widehat e_{i,t}^2 \sum_{i=1}^n \sum_{t=1}^T  \widehat v_{i,t}^2}}$ is the sample correlation coefficient of the two residuals $\hat e_{i,t}$ and $\widehat{v}_{i,t}$ in (\ref{eq:residuals}). $\widehat\rho^{\rm IVX}$ is the IVX estimate defined in \eqref{eq:rhoIVX}  for $\rho^*$. ``$P_m$ test $p$-value'' reports the $p$-values of the $P_m$ test statistics in \citet{choi2001unit} for panel unit root test.}
\end{flushleft}
\end{table}

Figure \ref{fig:emp1} exhibits the core regression results. The vanilla
WG estimator suggests that all three valuation ratios have significant
predictive power for returns. According to our theoretical results,
such significance may be spurious due to the Nickell-Stambaugh bias.
The vanilla IVX also suffers from the bias. On the other hand, the
bias corrections are effective. Specifically, while bias corrections
through $\hat{\rho}^{{\rm WG}}$ nudge the point estimates of $\beta^{*}$
toward zero, the corrections by $\hat{\rho}^{{\rm IVX}}$ further
refine these point estimates, pushing them even closer to the origin.
DIVX showcases that all three valuation ratios display no significance
in predicting stock returns, echoing the literature of panel predictive
regressions \citep{hjalmarsson2008stambaugh,hjalmarsson2010predicting,westerlund2017testing}.
Our empirical results align with our theory: the WG-based estimators
can be misleading due to the bias, while DIVX corrects the bias effectively.

\begin{figure}[tp]
\begin{centering}
\includegraphics[width=1\columnwidth]{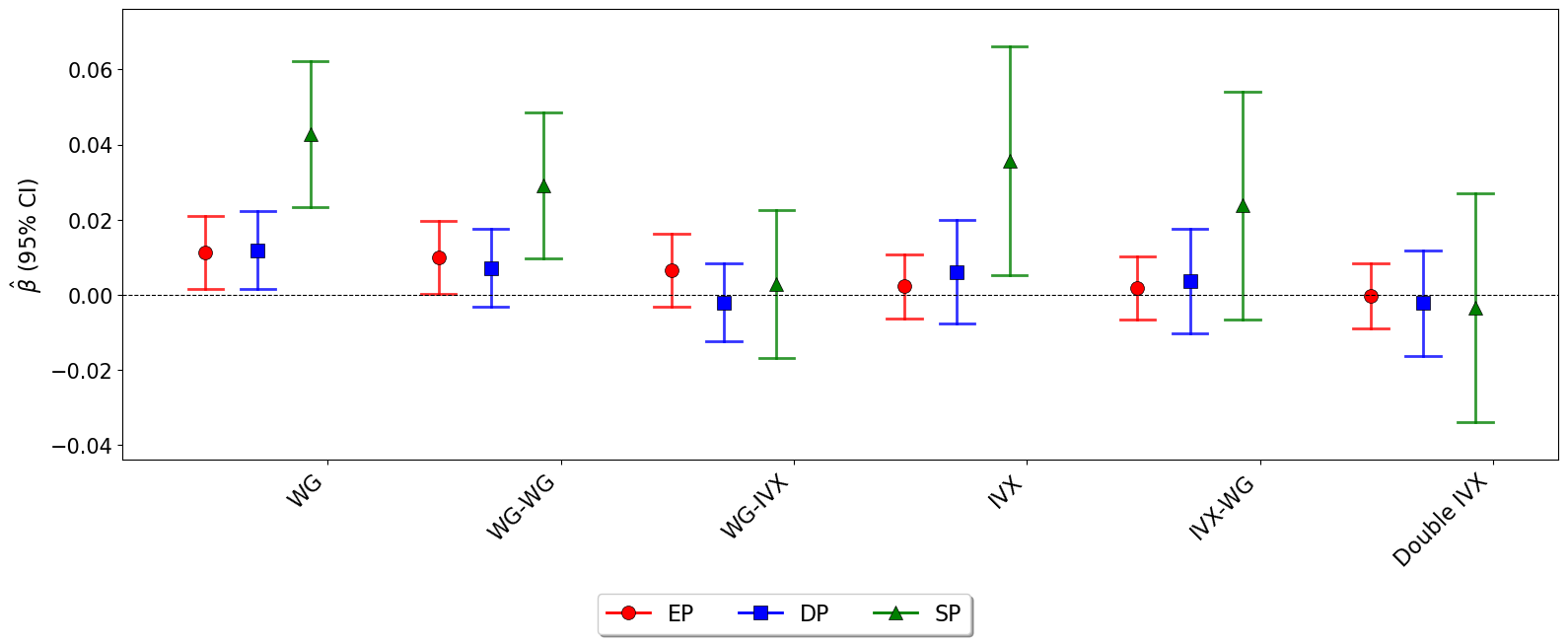}\smallskip{}
\par\end{centering}
\begin{centering}
\noindent\begin{minipage}[t]{1\linewidth}%
{\footnotesize Notes: The red circles, blue squares, and green triangles
represent the estimates of $\beta^{*}$ for EP, DP, and SP, respectively.
The bars represent the upper and lower bounds of the 95\% confidence
intervals.}%
\end{minipage}
\par\end{centering}
\centering{}\caption{\label{fig:emp1} Point estimates and 95\% confidence intervals of
return predictability using valuation ratios}
\end{figure}

Though the univariate regressions suggest that EP, DP, and SP \emph{individually}
lack significant predictive power, these results may stem from the
fact that each ratio captures only one dimension of market valuation.
It is possible that idiosyncratic noise in each individual metric
or the omitted variable bias obscures the underlying signal. To explore
whether these variables complement each other, our analysis turns
to a multivariate specification, incorporating all three valuation
ratios to examine their joint capacity in predicting future returns.

Specifically, we run the following multivariate regression:
\[
r_{i,t}=\mu_{i}+\beta_{1}^{*}\log({\rm EP}_{i,t-1})+\beta_{2}^{*}\log({\rm DP}_{i,t-1})+\beta_{3}^{*}\log({\rm SP}_{i,t-1})+e_{i,t}.
\]
Parallel to the univariate regressions, this multivariate regression
is also prone to Nickell-Stambaugh bias. Section \ref{sec:appdxMultivariate}
of the Online Appendices extends DIVX to multivariate predictive regressions,
supported with simulation evidence in Section \ref{subsec:Multivariate-Regressors}.
In this study, we will report the DIVX estimators, together with the
$t$-test for each coefficient, and the Wald test for the joint hypothesis
$\mathbb{H}_{0}:\beta_{1}^{*}=\beta_{2}^{*}=\beta_{3}^{*}=0$. This
Wald test has not been considered in existing empirical applications
of panel predictive regressions.

\begin{table}[htp]
\centering
\caption{Empirical results of univariate and multivariate regressions with DIVX}
\label{tab: multivariate}
\begin{tabular}{crrrr}
\hline\hline
                    & (1)        & (2)         & (3)         & (4)                 \\
\hline
\multirow{2}{*}{$\log(\mathrm{EP}_{i,t-1})$} & -0.0003   &           &           & -0.0179           \\
                    & (-0.0572) &           &           & (-0.9502)         \\
\multirow{2}{*}{$\log(\mathrm{DP}_{i,t-1})$} &           & -0.0022   &           & 0.0645            \\
                    &           & (-0.3142) &           & (1.3676)         \\
\multirow{2}{*}{$\log(\mathrm{SP}_{i,t-1})$} &           &           & -0.0035   & -0.0101           \\
                    &           &           & (-0.2232) & (-0.2177)         \\
\hline
                    &           &           &           & Wald = 7.7175     \\
                    &           &           &           & ($p$-value = 0.0522) \\
\hline\hline
\end{tabular}
\begin{flushleft}
{\smaller Notes: The table shows the DIVX point estimates of $\beta^*$, and the $t$-statistics in parentheses. ``Wald'' and ``$p$-value'' report the Wald statistic and its $p$-value. }
\end{flushleft}
\end{table}

Table \ref{tab: multivariate} displays the multivariate regression
results in Column (4), with Columns (1)-(3) showing the univariate
results for comparison. The multivariate results reveal more interesting
empirical findings. While the individual coefficients for EP, DP,
and SP in Column (4) remain statistically insignificant, the magnitudes
get larger, with the maximum absolute $t$-statistic reaching 1.3676
for DP, in contrast to 0.3142 in the univariate regressions. In addition,
the Wald test for their joint null has a $p$-value of 0.0522, significant
at the 10\% level. That is, although individual predictors appear
weak, together as an ensemble they boast predictive power marginally
significant at the usual levels. This discrepancy suggests that the
predictive signal may be diffused across various fundamental metrics,
rather than concentrated in a single valuation ratio. Therefore, looking
only at individual coefficients may overlook the joint signal from
multiple valuation indicators. DIVX fills the gap in the toolbox of
robust econometric inference for multivariate panel predictive regressions,
and produces novel empirical findings that complement the literature.

\section{Conclusion}

This paper investigates the problem of panel predictive regressions,
with a focus on valid inference based on the $t$-statistic. When
$n$ and $T$ are both large, WG and IVX incur the Nickell-Stambaugh
bias which distorts the size of the standard inferential procedure.
We propose to use IVX to estimate $\beta^{*}$ in the panel predictive
regression, and then plug in the tailored IVX estimator for $\rho^{*}$
to correct the bias. We show that this procedure provides unified
inference in various modes of dynamic regressors, including the stationary
case, the mildly integrated case, and the local unit root case. The
unified inference cannot be achieved if either the main regression
or the AR regression is estimated by WG.

\bigskip \bigskip

\bibliographystyle{chicagoa}
\bibliography{panel_predictive}

\clearpage{}

\setcounter{page}{1}
\renewcommand{\thesection}{O.\Alph{section}}
\setcounter{section}{0}
% for correct hyperref bookmarks
% https://tex.stackexchange.com/questions/6098
% Hyperref anchor names must also be unique
\renewcommand{\theHsection}{OA.\Alph{section}}
\renewcommand{\theHequation}{OA.\arabic{equation}}
\renewcommand{\theHfigure}{OA.\arabic{figure}}
\renewcommand{\theHtable}{OA.\arabic{table}}
\renewcommand{\theHfootnote}{OA.\arabic{footnote}}

\renewcommand{\theHthm}{OA.\arabic{thm}}
\renewcommand{\theHlem}{OA.\arabic{lem}}
\renewcommand{\theHrem}{OA.\arabic{rem}}
\renewcommand{\theHprop}{OA.\arabic{prop}}
\renewcommand{\theHassumption}{OA.\Alph{section}.\arabic{assumption}}
\onehalfspacing
\setcounter{footnote}{0}
\setcounter{table}{0}
\setcounter{figure}{0}
\setcounter{equation}{0}
\renewcommand{\thefootnote}{O.\arabic{footnote}}
\renewcommand{\theequation}{O.\arabic{equation}}
\renewcommand{\thefigure}{O.\arabic{figure}}
\renewcommand{\thetable}{O.\arabic{table}}
\setcounter{thm}{0}
\setcounter{lem}{0}
\setcounter{rem}{0}
\setcounter{prop}{0}
\renewcommand{\thethm}{O.\arabic{thm}}
\renewcommand{\thelem}{O.\arabic{lem}}
\renewcommand{\therem}{O.\arabic{rem}}
\renewcommand{\theprop}{O.\arabic{prop}}

\setcounter{assumption}{0}
\setcounter{lyxalgorithm}{0}
\renewcommand{\theassumption}{O.\arabic{assumption}}
\renewcommand{\thelyxalgorithm}{O.\arabic{lyxalgorithm}}
\begin{center}
{\LARGE Online Appendices to ``Nickell Meets Stambaugh: A Tale\\[5pt] of Two Biases in Panel Predictive Regression''}
\par\end{center}

\begin{center}
{\large Chengwang Liao$^a$, Ziwei Mei$^b$, Zhentao Shi$^a$ \\
$^a$The Chinese University of Hong Kong \\
$^b$University of Macau }\\
\par\end{center}

\noindent The Online Appendices include four extensions with additional
simulation results. Section \ref{sec:appdx two-way} extends DIVX
to two-way fixed effect models. Section \ref{sec:appdxMultivariate}
considers multivariate regressions. Section \ref{sec:Local-Projection}
discusses the extension to multiple-period--ahead predictive models,
which accommodates the local projections in panel data. Section \ref{sec:Grouped-Heterogeneity}
generalizes DIVX to cover cross-sectional heterogeneity with latent
group structures. Additional simulation results are collected in Section
\ref{sec:Additional-Simulations}.

\section{Two-way Fixed Effects\label{sec:appdx two-way} }

The main paper focuses on panel predictive regressions with individual
fixed effects only. In practice, two-way fixed effects (TWFE) ---
individual and time --- are widely adopted. In this section, we generalize
the predictive models (\ref{eq:predictive}) and (\ref{eq:state space})
to allow for TWFE:
\begin{align}
y_{i,t+1} & =\mu_{y,i}+f_{y,t}+\beta^{*}x_{i,t}+e_{i,t+1},\quad\text{for }i=1,\dots,n\text{ and }t=1,\dots,T-1,\label{eq:predictive-TW}
\end{align}
and the regressor $x_{i,t}$ follows a state space representation
\begin{equation}
\begin{split}x_{i,t} & =\alpha_{i}+\delta_{i,t},\\
\delta_{i,t+1} & =f_{x,t}+\rho^{*}\delta_{i,t}+v_{i,t+1}.
\end{split}
\label{eq:state space-TW}
\end{equation}
It implies that $x_{i,t}$ admits the AR(1) form $x_{i,t}=\mu_{x,i}+f_{x,t}+\rho^{*}x_{i,t-1}+v_{i,t}$
with $\mu_{x,i}=(1-\rho^{*})\alpha_{i}$. Similar to (\ref{eq:FE_AR}),
this specification of the individual fixed effect prevents an unrestricted
nonzero intercept in a local-to-unity process to avoid a drift that
dominates the stochastic trend and drastically complicates the asymptotics.

To remove the time effects, we follow the standard between-group (BG)
transformation in the literature. Specifically, for any generic panel
variable $w_{i,t}$, let $\check{w}_{i,t}=w_{i,t}-n^{-1}\sum_{j=1}^{n}w_{j,t}$
denote the BG transformed variant. Unlike the WG transformation that
causes the Nickell bias, the BG transformation removing time fixed
effects would not induce additional bias to the estimator. Therefore,
to address time fixed effects, our DIVX estimator replaces the original
$x_{i,t}$ and $y_{i,t}$ by the BG transformed variables $\check{x}_{i,t}$
and $\check{y}_{i,t}$.

Specifically, the instrument for IVX becomes
\begin{equation}
\check{z}_{i,t}=\sum_{s=1}^{t}\rho_{z}^{t-s}\Delta\check{x}_{i,s},\ \rho_{z}=1+c_{z}/T^{\theta},\label{eq:def IV-TW}
\end{equation}
where $c_{z}=-1$ and $\theta=0.95$ following the main paper. For
any generic $w_{i,t}$, let $\ddot{w}_{i,t}=\check{w}_{i,t}-\bar{\check{w}}_{i}$
denote the WG transformed $\check{w}_{i,t}$. Then we abuse the notation
$\bhativx$ to redefine IVX estimator of $\beta^{*}$ as
\begin{equation}
\bhativx=\frac{\sum_{i=1}^{n}\sum_{t=1}^{T-1}\ddot{z}_{i,t}\ddot{y}_{i,t+1}}{\sum_{i=1}^{n}\sum_{t=1}^{T-1}\ddot{z}_{i,t}\ddot{x}_{i,t}}.\label{eq:beta_IVX-TW}
\end{equation}

Since the BG transformation that addresses time effects does not induce
additional bias, the IVX bias formula still follows (\ref{eq:ivx_bias}).
To establish the IVX estimator of $\rho^{*}$, we generate the IV
as
\begin{equation}
\check{z}_{i,t}^{(1)}=\sum_{s=1}^{t}\left(1+\frac{c_{z}}{T^{\theta_{1}}}\right)^{t-s}\Delta\check{x}_{i,s},\label{eq:def IV (1)-1}
\end{equation}
where $c_{z}=-1$ and $\theta_{1}=0.975$. Define
\[
\hat{\rho}^{{\rm TW}}=\frac{\sum_{i=1}^{n}\sum_{t=1}^{T-1}\ddot{x}_{i,t}\ddot{x}_{i,t+1}}{\sum_{i=1}^{n}\sum_{t=1}^{T-1}\ddot{x}_{i,t}^{2}}
\]
as the estimator of $\rho^{*}$, and we abuse the notation $\hat{\rho}^{\mkern3mu \mathrm{IVX}}$
to re-define the IVX estimator of $\rho^{*}$ as
\begin{equation}
\hat{\rho}^{\mkern3mu \mathrm{IVX}}=\frac{\sum_{i=1}^{n}\sum_{t=1}^{T-1}\left(\check{z}_{i,t}^{(1)}\check{x}_{i,t+1}-\ddot{\varDelta}_{vv}\right)}{\sum_{i=1}^{n}\sum_{t=1}^{T-1}\check{z}_{i,t}^{(1)}\check{x}_{i,t}},\label{eq:rho IVX-1}
\end{equation}
where
\begin{equation}
\ddot{\varDelta}_{vv}=\frac{1}{nT}\sum_{h=1}^{G}\sum_{i=1}^{n}\sum_{t=h+1}^{T}\hat{v}_{i,t}^{{\rm TW}}\hat{v}_{i,t-h}^{{\rm TW}},\ G=\lfloor T^{1/4}\rfloor\label{eq:NW Deltavv-TW}
\end{equation}
and $\hat{v}_{i,t+1}^{{\rm TW}}=\ddot{x}_{i,t+1}-\hat{\rho}^{{\rm TW}}\ddot{x}_{i,t}$.
In addition, the intertemporal covariances are estimated by
\begin{equation}
\hat{\omega}_{ev,h}=\frac{1}{nT}\sum_{i=1}^{n}\sum_{t=1}^{T-h}\ddot{v}_{i,t+h}\ddot{e}_{i,t},\label{eq:omega ev hat-TW}
\end{equation}
where
\begin{equation}
\ddot{v}_{i,t}=\ddot{x}_{i,t}-\hat{\rho}^{{\rm IVX}}\ddot{x}_{i,t-1},\ \ \ddot{e}_{i,t}=\ddot{y}_{i,t}-\hat{\beta}^{{\rm IVX}}\ddot{x}_{i,t-1}.\label{eq:residuals-TW}
\end{equation}
With the newly defined estimators, the DIVX estimator $\hat{\beta}^{{\rm DIVX}}$
still follows (\ref{eq:DIVX}), where the truncated bias formula follows
(\ref{eq:truncated bias IVX}). The standard error is
\begin{equation}
\widehat{\varsigma}^{{\rm IVX}}=\frac{\sqrt{\sum_{i=1}^{n}\left(\sum_{t=1}^{T-1}\ddot{z}_{i,t}^{2}\ddot{e}_{i,t+1}^{2}-T\hat{\lambda}\bar{z}_{i}^{2}\hat{\omega}_{ee}\right)}}{\bigl|\sum_{i=1}^{n}\sum_{t=1}^{T-1}\ddot{z}_{i,t}\ddot{x}_{i,t}\bigr|},\label{eq:se IVX-1}
\end{equation}
where $\hat{\omega}_{ee}=(n(T-1))^{-1}\sum_{i=1}^{n}\sum_{t=1}^{T-1}\ddot{e}_{i,t+1}^{2}$
and $\widehat{\lambda}=(1-n\hat{\varrho}_{ev}^{\mkern3mu 2}/T^{3/2})_{+}$
with $\hat{\varrho}_{ev}=\hat{\omega}_{ev,0}/\hat{\omega}_{ee}$.
Theoretical justifications for DIVX with TWFE are essentially the
same as those in the main text with individual-specific intercepts
only, and therefore omitted to save space. In Section \ref{subsec:Two-way-Fixed-Effects}
of the Online Appendices, we conduct additional simulations to examine
the finite-sample performance of DIVX for panel predictive regressions
with TWFE. DIVX performs well.

\section{Multivariate Regression\label{sec:appdxMultivariate} }

Empirical research often includes multiple regressors as control variables.
Suppose that the target variable of interest $y_{i,t+1}$ is linked
with $k$ regressors $\boldsymbol{x}_{i,t}=(x_{j,i,t})_{j=1}^{k}$
in the linear form
\begin{equation}
y_{i,t+1}=\mu_{y,i}+\bm{x}_{i,t}^{\prime}\bm{\beta}^{*}+e_{i,t+1}\qquad\text{for }i=1,\dots,n\text{ and }t=1,\dots,T-1.\label{eq:pred_mul}
\end{equation}
The regressors are generated by a vector state space model:
\begin{align}
\bm{x}_{i,t} & =\bm{\alpha}_{i}+\bm{\delta}_{i,t},\label{eq:AR_mul}\\
\boldsymbol{\delta}_{i,t+1} & =\bm{R}_{T}^{*}\bm{\delta}_{i,t}+\bm{v}_{i,t+1},\nonumber
\end{align}
where
\[
\bm{R}^{*}=\bm{R}_{T}^{*}=\mathrm{diag}\bigl(\{\rho_{j}^{*}\}_{j=1}^{k}\bigr)\quad\text{with}\quad\rho_{j}^{*}=1+c_{j}^{*}/T^{\gamma_{j}},\quad c_{j}^{*}\in\mathbb{R}\text{ and }\gamma_{j}\in[0,1].
\]
The subscript $T$ in $\bm{R}_{T}^{*}$ is suppressed when there is
no ambiguity. We allow the degree of persistence measured by $\gamma_{j}$
to be heterogeneous across regressors. Let $\boldsymbol{w}_{i,t}=(e_{i,t},\bm{\varepsilon}_{i,t}^{\prime})^{\prime}$.
For a generic matrix $\bm{A}$, let $\|\bm{A}\|$ denote its Frobenius
norm. The following two assumptions generalize those in Section \ref{sec:framework}
to multivariate regressions.
\begin{assumption}[Initial values and drift]
\label{assump:mult-initval-1} The VAR(1) process $\boldsymbol{\delta}_{i,t}$
and the drift $\bm{\alpha}_{i}$ in \eqref{eq:AR_mul} satisfy the
following conditions uniformly across all $i$ and $j$:
\begin{enumerate}
\item $\mathbb{E}\left(\boldsymbol{\delta}_{i,0}|\boldsymbol{\alpha}_{i}\right)=\bm{0}_{k}$.
\item $\mathbb{E}(\delta_{j,i,0}^{2})=O(|1-\rho_{j}^{*}|^{-1}\wedge T^{1-\varepsilon})$
where $\varepsilon>0$ is an arbitrarily small constant.
\item $\mathbb{E}(\delta_{j,i,0}^{4}+\alpha_{i}^{4})=O(|1-\rho_{j}^{*}|^{-2}\wedge T^{2})$.
(The convention $1/\infty$ is invoked if $\rho_{j}^{*}=1$).
\item $\sup_{t\leq0}\|\mathbb{E}(\boldsymbol{\delta}_{i,0}\bm{\varepsilon}_{i,t})\|<\infty$.
\end{enumerate}
\end{assumption}
\begin{assumption}[Innovations]
\label{assump:mult-innov-1}\mbox{}
\begin{enumerate}
\item Define $\bm{w}_{i,t}=(e_{i,t},\bm{\varepsilon}_{i,t}^{\prime})^{\prime}$.
Suppose that $\{\bm{w}_{i,t}\}$ are i.i.d.\ across $i$. For each
$i$, the time series $\{\bm{w}_{i,t}\}$ is a strictly stationary
and ergodic m.d.s.\ with respect to the filtration $\{\mathcal{F}_{i,t}=\sigma(\bm{\alpha}_{i},\bm{\delta}_{i,0},\bm{w}_{i,t},\bm{w}_{i,t-1},\dots)\}$,
with absolutely summable fourth order cumulants: $\sup_{a,b,c,d\in\{1,\dots,k+1\}}\sum_{t,s,r=-\infty}^{\infty}|\kappa_{abcd}(0,t,s,r)|<\infty$.
\item \label{enu:garch-mul}Let $u_{i,t}$ be i.i.d.\ random variables
with $\mathbb{E}(u_{1,1})=0$, $\mathbb{E}(u_{1,1}^{2})=1$ and $\mathbb{E}(u_{1,1}^{4})<\infty$.
For each $i$, the sequence $\{e_{i,t}\}$ admits a $\mathrm{GARCH}(q,r)$
representation:
\[
e_{i,t}=h_{i,t}^{1/2}u_{i,t},\qquad h_{i,t}=\phi+\sum_{k=1}^{q}a_{k}e_{i,t-k}^{2}+\sum_{\ell=1}^{r}b_{\ell}h_{i,t-\ell},
\]
where $\phi>0$ is a constant, $a_{k},b_{\ell}\geq0$ and $0\leq\sum_{k=1}^{q}a_{k}+\sum_{\ell=1}^{r}b_{\ell}<1$.
\item \label{enu:LP-1}For each $i$, $\bm{v}_{i,t}$ is a strictly stationary
linear process
\[
\bm{v}_{i,t}=\sum_{s=0}^{\infty}\bm{G}_{s}\bm{\varepsilon}_{i,t-s},
\]
where $\|\bm{G}_{s}\|\leq C_{0}\exp(-C_{g}s)$ for some positive constants
$C_{0}$ and $C_{g}$ for each $s$.
\end{enumerate}
\end{assumption}
To implement IVX, we generate the instrumental variable as
\[
\bm{z}_{i,t}=\sum_{s=1}^{t}\rho_{z}^{t-s}\Delta\bm{x}_{i,s},\quad\text{for }t=1,\dots,T_{h},
\]
where $\rho_{z}=1+c_{z}/T^{\theta}$, and $\Delta\bm{x}_{i,s}=\bm{x}_{i,s}-\bm{x}_{i,s-1}$.
The IVX estimator for $\bm{\beta}^{*}$ is
\[
\hat{\bm{\beta}}^{\mathrm{IVX}}=\left(\sum_{i=1}^{n}\sum_{t=1}^{T-1}\tilde{\boldsymbol{z}}_{i,t}\bm{x}_{i,t}'\right)^{-1}\sum_{i=1}^{n}\sum_{t=1}^{T-1}\tilde{\boldsymbol{z}}_{i,t}y_{i,t+1}.
\]
The following proposition gives the bias formula induced from the
numerator.
\begin{prop}
\label{prop:ivx_multiple_bias}Under Assumptions \ref{assump:mult-initval-1}
and \ref{assump:mult-innov-1},
\[
\mathbb{E}\left(\sum_{i=1}^{n}\sum_{t=1}^{T-1}\tilde{\bm{z}}_{i,t}e_{i,t+1}\right)=-\bm{\xi}_{n,T}(\{\bm{\omega}_{ev,\ell}^{*}\},\bm{R}^{*},\rho_{z}),
\]
where $\bm{\xi}_{n,T}(\{\bm{\omega}_{ev,\ell}^{*}\},\bm{R}^{*},\rho_{z})$
is a $k$-dimensional vector with its $j$-th entry
\[
\xi_{j,n,T}(\{\bm{\omega}_{ev,\ell}^{*}\},\rho_{j}^{*},\rho_{z})=\frac{n}{T-1}\sum_{\ell=0}^{T-3}\Psi_{\ell,T-1}(\rho_{j}^{*},\rho_{z})\omega_{ev,j,\ell}^{*},
\]
where
\[
\Psi_{\ell,T-1}(\rho_{j}^{*},\rho_{z})=\sum_{k=\ell+2}^{T-1}\frac{\rho_{z}^{T-k}-\rho_{j}^{*T-k}}{\rho_{z}-\rho_{j}^{*}}\quad\text{and}\quad\omega_{ev,j,\ell}^{*}=\mathbb{E}(v_{j,i,t+\ell}e_{i,t}).
\]
\end{prop}
The bias formula is given by
\begin{equation}
\bm{b}_{n,T}^{\mathrm{IVX}}(\{\bm{\omega}_{ev,\ell}^{*}\},\bm{R}^{*},\rho_{z})=\left(\sum_{i=1}^{n}\sum_{t=1}^{T-1}\tilde{\boldsymbol{z}}_{i,t}\bm{x}_{i,t}'\right)^{-1}\bm{\xi}_{n,T}(\{\bm{\omega}_{ev,\ell}^{*}\},\bm{R}^{*},\rho_{z}),\label{eq:ivx_bias_multiple}
\end{equation}
which is parallel to \eqref{eq:ivx_bias} for the simple regression.

To convert \eqref{eq:ivx_bias_multiple} into a feasible bias formula,
we need to replace the parameters with their respective estimators.
Let $\hat{\bm{R}}^{\mathrm{IVX}}$ denote the IVX estimator for $\boldsymbol{R}^{*}$,
namely,
\begin{equation}
\hat{\bm{R}}^{\text{IVX}}={\rm diag}\bigl(\{\hat{\rho}_{j}^{\text{IVX}}\}_{j=1}^{k}\bigr),\label{eq:hat R XJ-1}
\end{equation}
where $\hat{\rho}_{j}^{\text{IVX}}$ follows \eqref{eq:rhoIVX} for
the simple regression where $x_{i,t}$ is replaced by $x_{j,i,t}$
for $j=1,2,\ldots,k$. Let $\hat{\bm{\omega}}_{ev,\ell}$ denote the
estimators for the covariances of innovations based on the IVX estimators
$\hat{\bm{\beta}}^{\mathrm{IVX}}$ and $\hat{\bm{R}}^{\text{IVX}}$:
\begin{equation}
\hat{\bm{\omega}}_{ev,\ell}=\dfrac{1}{n(T-\ell)}\sum_{i=1}^{n}\sum_{t=1}^{T-\ell}\hat{\boldsymbol{v}}_{i,t+\ell}\hat{e}_{i,t},\label{eq:hat_omg_ell_mixed}
\end{equation}
where the residuals are calculated as
\[
\hat{\boldsymbol{v}}_{i,t}=\tilde{\boldsymbol{x}}_{i,t}-\hat{\bm{R}}^{\text{IVX}}\tilde{\boldsymbol{x}}_{i,t-1},\ \hat{e}_{i,t}=\tilde{y}_{i,t}-\tilde{\boldsymbol{x}}_{i,t-1}^{\prime}\hat{\bm{\beta}}^{\mathrm{IVX}}
\]
Moreover, we need to truncate the long-run covariance to get the feasible
bias formula:
\begin{equation}
\hat{\bm{b}}_{n,T}^{\mathrm{IVX}}(\{\hat{\bm{\omega}}_{ev,\ell}\},\hat{\bm{R}}^{\text{IVX}},\rho_{z})=\left(\sum_{i=1}^{n}\sum_{t=1}^{T-1}\tilde{\boldsymbol{z}}_{i,t}\bm{x}_{i,t}^{\prime}\right)^{-1}\hat{\bm{\xi}}_{n,T}(\{\hat{\bm{\omega}}_{ev,\ell}\},\hat{\bm{R}}^{\text{IVX}},\rho_{z}),\label{eq:hat_bias_IVX}
\end{equation}
where the $j$th entry of $\hat{\bm{\xi}}_{n,T}(\{\hat{\bm{\omega}}_{ev,\ell}\},\hat{\bm{R}}^{\text{IVX}},\rho_{z})$
is given by
\[
\hat{\xi}_{j,n,T}(\{\hat{\omega}_{ev,j,\ell}\},\hat{\rho}_{j}^{\mathrm{IVX}},\rho_{z})=\frac{n}{T-1}\sum_{\ell=0}^{G}\Psi_{\ell,T-1}(\hat{\rho}_{j}^{\mathrm{IVX}},\rho_{z})\hat{\omega}_{ev,j,\ell}\quad\text{with }G=\lfloor T^{1/4}\rfloor.
\]
 The DIVX estimator for multivariate panel predictive regression is
constructed by
\[
\hat{\bm{\beta}}^{\mathrm{DIVX}}=\hat{\bm{\beta}}^{\mathrm{IVX}}+\hat{\bm{b}}_{n,T}^{\mathrm{IVX}}(\{\hat{\bm{\omega}}_{ev,\ell}\},\hat{\bm{R}}^{\text{IVX}},\rho_{z}).
\]
where $\hat{\bm{b}}_{n,T}^{\mathrm{IVX}}$ is defined by \eqref{eq:hat_bias_IVX}.

The next proposition is about the asymptotic variance of $\hat{\bm{\beta}}^{\mathrm{DIVX}}$.
Let $\boldsymbol{D}_{T}={\rm diag}\bigl(\{T^{1+(\theta\wedge\gamma_{j})}\}_{j=1}^{k}\bigr)$.
Note that the infeasible matrix $\boldsymbol{D}_{T}$ measures the
convergence rates and only facilitates theoretical analysis; it does
not appear in the practical implementation.
\begin{prop}
\label{prop:ivx_variance_multiple}Under \prettyref{assump:mult-initval-1}
and \ref{assump:mult-innov-1}, as $(n,T)\to\infty$, we have
\[
(n\boldsymbol{D}_{T})^{1/2}\left(\hat{\bm{\beta}}^{\mathrm{IVX}}-\bm{\beta}^{*}+\bm{b}_{n,T}^{\mathrm{IVX}}(\{\bm{\omega}_{ev,\ell}^{*}\},\bm{R}^{*},\rho_{z})\right)\to_{d}\mathcal{N}\bigl(\bm{0},\bm{\Sigma}^{\mathrm{IVX}}\bigr),
\]
where
\[
\bm{\Sigma}^{\mathrm{IVX}}=\omega_{ee}^{*}\cdot\lim_{T\to\infty}\boldsymbol{D}_{T}^{1/2}\left[\mathbb{E}\left(\sum_{t=1}^{T-1}\tilde{\boldsymbol{z}}_{i,t}\bm{x}_{i,t}'\right)\right]^{-1}\cdot\left[\mathbb{E}\left(\sum_{t=1}^{T-1}\bm{z}_{i,t}\bm{z}_{i,t}'\right)\right]\cdot\left[\mathbb{E}\left(\sum_{t=1}^{T-1}\tilde{\boldsymbol{z}}_{i,t}\bm{x}_{i,t}'\right)\right]^{-1}\boldsymbol{D}_{T}^{1/2}.
\]
\end{prop}
\bigskip
\begin{prop}
\label{prop:feasible_ivx_normal}Under \prettyref{assump:mult-initval-1}
and \ref{assump:mult-innov-1}, as $(n,T)\to\infty$ and $n/T\to c\in[0,\infty)$,
we have
\[
(n\boldsymbol{D}_{T})^{1/2}\left(\hat{\bm{\beta}}^{\mathrm{DIVX}}-\bm{\beta}^{*}\right)\to_{d}\mathcal{N}\bigl(\bm{0},\bm{\Sigma}^{\mathrm{IVX}}\bigr).
\]
\end{prop}
By the asymptotic normality shown in \prettyref{prop:feasible_ivx_normal},
we estimate the asymptotic variance of $\hat{\bm{\beta}}^{\mathrm{DIVX}}$
as
\begin{equation}
\hat{\bm{\Theta}}^{\mathrm{DIVX}}=\left(\sum_{i=1}^{n}\sum_{t=1}^{T_{h}}\tilde{\bm{z}}_{i,t}\bm{x}_{i,t}'\right)^{-1}\hat{\bm{\Sigma}}_{n,T}\left(\sum_{i=1}^{n}\sum_{t=1}^{T_{h}}\tilde{\bm{z}}_{i,t}\bm{x}_{i,t}^{\prime}\right)^{-1},\label{eq:Theta hat DIVX}
\end{equation}
where\textbf{\textcolor{red}{{} }}
\begin{align}
\hat{\bm{\Sigma}}_{n,T} & =\frac{1}{n}\sum_{i=1}^{n}\left(\sum_{t=1}^{T-1}\bm{z}_{i,t}\bm{z}_{i,t}^{\prime}\hat{e}_{i,t+1}^{\mkern3mu 2}-T\hat{\omega}_{ee}\cdot\bm{\Lambda}_{n,T}\bm{\bar{z}}_{i}\bm{\bar{z}}_{i}^{\prime}\bm{\Lambda}_{n,T}\right)\qquad\text{with}\label{eq:CovMatMul}\\
\bm{\Lambda}_{n,T} & =\mathrm{diag}\Bigl(\bigl\{\hat{\lambda}_{j}\bigr\}_{j=1}^{k}\Bigr),\ \hat{\lambda}_{j}=\left(1-\hat{\varrho}_{ev,j}\frac{n}{T^{3/2}}\right)_{+},\nonumber
\end{align}
with $\hat{\omega}_{ee}$ and $\hat{\varrho}_{ev,j}$ following the
definitions in (\ref{eq:se IVX}). Parallel to (\ref{eq:se IVX})
for the univariate regression, the second term \textbf{$T\hat{\omega}_{ee}\cdot\bm{\Lambda}_{n,T}\bm{\bar{z}}_{i}\bm{\bar{z}}_{i}^{\prime}\bm{\Lambda}_{n,T}$}
is merely a correction term under finite samples; it does not affect
the asymptotics as its stochastic order is dominated by the first
term $\sum_{t}\bm{z}_{i,t}\bm{z}_{i,t}^{\prime}\hat{e}_{i,t+1}^{\mkern3mu 2}$.
By the argument we have used in the proof of the simple regression,
we can show that, as $(n,T)\to\infty$,
\begin{equation}
(n\boldsymbol{D}_{T})^{1/2}\hat{\bm{\Theta}}^{\mathrm{DIVX}}(n\boldsymbol{D}_{T})^{1/2}\to_{p}\bm{\Sigma}^{\mathrm{IVX}}.\label{eq:DIVX_Theta_hat_converge}
\end{equation}

Suppose we are interested in testing a linear joint null hypothesis
$\mathbb{H}_{0}\colon\bm{A}\bm{\beta}^{*}=\bm{q}$, where $\boldsymbol{A}$
is an $m\times k$ constant matrix of full row rank accommodating
$m$ linear restrictions, and $\boldsymbol{q}$ is an $m\times1$
constant vector. We reject $\mathbb{H}_{0}$ at the significance level
$\alpha$ if the Wald statistic
\begin{equation}
\text{Wald}^{\mathrm{DIVX}}=\bigl(\bm{A}\hat{\bm{\beta}}^{\mathrm{DIVX}}-\bm{q}\bigr)'\bigl(\bm{A}\hat{\bm{\Theta}}^{\mathrm{DIVX}}\bm{A}'\bigr)^{-1}\bigl(\bm{A}\hat{\bm{\beta}}^{\mathrm{DIVX}}-\bm{q}\bigr)\label{eq:Wald_DIVX}
\end{equation}
is larger than the $(1-\alpha)$th quantile of a $\chi^{2}$ distribution
of degree of freedom $m$, where $\hat{\bm{\Theta}}^{\mathrm{DIVX}}$
is defined in \eqref{eq:Theta hat DIVX}.
\begin{thm}
\label{thm:wald-chi2}Suppose that \prettyref{assump:mult-initval-1}
and \ref{assump:mult-innov-1} hold. Under the null hypothesis $\mathbb{H}_{0}\colon\bm{A}\bm{\beta}^{*}=\bm{q}$,
the Wald statistic
\[
\mathrm{Wald}^{\mathrm{DIVX}}\to_{d}\chi^{2}(m)
\]
as $(n,T)\to\infty$ and $n/T\to c\in[0,\infty)$.
\end{thm}
We conduct simulation studies for multivariate regressions in Section
\ref{subsec:Multivariate-Regressors} of the Online Appendices, where
we allow the regressors to have various degrees of persistence. The
results show that DIVX inference remains robust for multiple regressions
with mixed roots.

\section{Local Projection\label{sec:Local-Projection}}

Since \citet{jorda2005estimation} proposed local projection in the
time series context, it has been widely used to estimate the impulse
response functions that characterize the dynamic relations in economic
and financial systems. In empirical studies, local projection is naturally
introduced to panel data applications, where the WG estimator is the
default estimator and suffers from potential Nickell bias \citep{mei2023implicit}.
As local projection applies a sequence of predictive regressions,
the Nickell-Stambaugh bias in panel predictive regressions discussed
in the main text naturally retains in panel local projections. This
section is the first discussion of \emph{panel} local projection with
persistent regressors, complementing \citet{mei2023implicit} who
cover panel with a stationary time dimension.

Let $H$ be the maximum horizon of interest specified by the user.
By Equations \eqref{eq:pred_mul} and \eqref{eq:AR_mul} we can deduce
the following $h$-period predictive model
\begin{equation}
y_{i,t+h}=\mu_{y,i}^{(h)}+\bm{x}_{i,t}^{\prime}\bm{\beta}^{(h)*}+e_{i,t+h}^{(h)},\quad\text{for }h\in\left\{ 1,2,\ldots,H\right\} ,\label{eq:pred_h}
\end{equation}
where
\begin{align}
\mu_{y,i}^{(h)} & =\mu_{y,i}+\bm{\beta}^{*\prime}\left(\sum_{\tau=0}^{h-2}\bm{R}^{*\tau}\right)\bm{\mu}_{x,i}\quad\text{with }\bm{\mu}_{x,i}=(\bm{I}-\bm{R}^{*})\bm{\alpha}_{i},\nonumber \\
\bm{\beta}^{(h)*} & =\bm{R}^{*h-1}\bm{\beta}^{*},\nonumber \\
e_{i,t+h}^{(h)} & =e_{i,t+h}+\bm{\beta}^{*\prime}\left(\sum_{\tau=1}^{h-1}\bm{R}^{*h-1-\tau}\bm{v}_{i,t+\tau}\right).\label{eq:error_h}
\end{align}
The coefficients $\bm{\beta}^{(h)*}$ are the impulse response functions
of central interest. Note that the effective number of time periods
is $T_{h}=T-h$. We follow the literature of local projections \citep{jorda2005estimation,montiel2021local,mei2023implicit}
to impose that $e_{i,t}$ is conditionally homoskedastic, and the
AR(1) innovations $\bm{v}_{i,t}$ in \eqref{eq:AR_mul} are m.d.s.,
formalized as the following assumption.\footnote{If $\bm{v}_{i,t}$ has nonzero autocorrelations, the regressors $\bm{x}_{i,t}$
and the error term $e_{i,t+h}^{(h)}$ will be correlated, which causes
endogeneity.}
\begin{assumption}[Innovations]
\label{assump:mult-innov-1-1}\mbox{}The conditions in \prettyref{assump:mult-innov-1}
hold except that $a_{k}=0$ for $k=1,2,...,q$ and $b_{\ell}=0$ for
$\ell=1,2,...,r$ in Condition \ref{enu:garch-mul}, and $\boldsymbol{G}_{0}=\boldsymbol{I}$,
$\boldsymbol{G}_{s}=\boldsymbol{O}$ for all $s\geq1$ in Condition
\ref{enu:LP-1}.
\end{assumption}
The IVX estimator for the $h$-period-ahead predictive regression
is
\[
\hat{\bm{\beta}}^{(h)\mathrm{IVX}}=\left(\sum_{i=1}^{n}\sum_{t=1}^{T_{h}}\tilde{\boldsymbol{z}}_{i,t}\bm{x}_{i,t}^{\prime}\right)^{-1}\sum_{i=1}^{n}\sum_{t=1}^{T_{h}}\tilde{\boldsymbol{z}}_{i,t}y_{i,t+h}.
\]
The following proposition shows the bias formula induced from the
numerator. Its \proofref[proof]{prop:ivx_bias_h} is relegated to
Section~\ref{subsec:Proofs-for-mult}.
\begin{prop}
\label{prop:ivx_bias_h}Under Assumptions~\ref{assump:mult-initval-1}
and \ref{assump:mult-innov-1-1}, we have
\[
\mathbb{E}\left(\sum_{i=1}^{n}\sum_{t=1}^{T_{h}}\tilde{\boldsymbol{z}}_{i,t}e_{i,t+h}^{(h)}\right)=-\bm{\xi}_{n,T}^{(h)}(\bm{R}^{*},\bm{\omega}_{ev}^{*},\bm{\Omega}_{vv}^{*},\bm{\beta}^{*}),
\]
where
\[
\bm{\xi}_{n,T}^{(h)}(\bm{R}^{*},\bm{\omega}_{ev}^{*},\bm{\Omega}_{vv}^{*},\bm{\beta}^{*})=\frac{n}{T_{h}}\Biggl[\begin{array}{c}
\sum_{t=h+1}^{T_{h}}\sum_{s=h+1}^{t}\rho_{z}^{t-s}\bm{R}^{*s-h-1}\bm{\omega}_{ev}^{*}\\
+\sum_{\tau=1}^{h-1}\sum_{t=\tau+1}^{T_{h}}\sum_{s=\tau+1}^{t}\rho_{z}^{t-s}\bm{R}^{*s-\tau-1}\bm{\Omega}_{vv}^{*}\bm{R}^{*h-1-\tau}\bm{\beta}^{*}
\end{array}\Biggr].
\]
\end{prop}
Given \prettyref{prop:ivx_bias_h}, we have the following bias formula
for $\hat{\bm{\beta}}^{(h){\rm IVX}}$:
\begin{equation}
\bm{b}_{n,T}^{(h)\mathrm{IVX}}(\bm{R}^{*},\bm{\omega}_{ev}^{*},\bm{\Omega}_{vv}^{*},\bm{\beta}^{*})=\left(\sum_{i=1}^{n}\sum_{t=1}^{T_{h}}\tilde{\boldsymbol{z}}_{i,t}\bm{x}_{i,t}'\right)^{-1}\bm{\xi}_{n,T}^{(h)}(\bm{R}^{*},\bm{\omega}_{ev}^{*},\bm{\Omega}_{vv}^{*},\bm{\beta}^{*}),\label{eq:ivx_bias_h}
\end{equation}
which is parallel to \eqref{eq:ivx_bias} for the simple regression.
To make \eqref{eq:ivx_bias_h} a feasible bias formula, we need to
replace the parameters with their respective estimators. The IVX estimator
for $\bm{R}^{*}$ follows \eqref{eq:hat R XJ-1}. Let $\hat{\bm{\omega}}_{ev}$
and $\hat{\bm{\Omega}}_{vv}$ denote the estimators for the covariances
of innovations based on the IVX estimators $\hat{\bm{\beta}}^{\mathrm{IVX}}$
and $\hat{\bm{R}}^{{\rm IVX}}$:
\begin{align}
\hat{\bm{\omega}}_{ev} & =\dfrac{1}{nT}\sum_{i=1}^{n}\sum_{t=1}^{T-1}(\tilde{\bm{x}}_{i,t+1}-\hat{\bm{R}}^{{\rm IVX}}\tilde{\bm{x}}_{i,t})(\tilde{y}_{i,t+1}-\tilde{\bm{x}}_{i,t}'\hat{\bm{\beta}}^{\mathrm{IVX}}),\label{eq:hat omg 12 IVXJ}\\
\hat{\bm{\Omega}}_{vv} & =\dfrac{1}{nT}\sum_{i=1}^{n}\sum_{t=1}^{T-1}(\tilde{\bm{x}}_{i,t+1}-\hat{\bm{R}}^{{\rm IVX}}\tilde{\bm{x}}_{i,t})(\tilde{\bm{x}}_{i,t+1}-\hat{\bm{R}}^{{\rm IVX}}\tilde{\bm{x}}_{i,t})^{\prime}.\label{eq:hat Omg 22 XJ}
\end{align}
The DIVX estimator for multivariate panel predictive regression is
constructed by
\[
\hat{\bm{\beta}}^{(h){\rm DIVX}}=\hat{\bm{\beta}}^{(h)\mathrm{IVX}}+\bm{b}_{n,T}^{(h)\mathrm{IVX}}(\hat{\bm{R}}^{{\rm IVX}},\hat{\boldsymbol{\omega}}_{ev},\hat{\bm{\Omega}}_{vv},\hat{\bm{\beta}}^{\mathrm{IVX}}),
\]
where $\bm{b}_{n,T}^{(h)\mathrm{IVX}}$ is defined by \eqref{eq:ivx_bias_h}.

The next proposition relates to the asymptotic variance of $\hat{\bm{\beta}}^{(h){\rm DIVX}}$.
Given a fixed horizon $h$, we pass $(n,T)\to\infty.$ Let $\boldsymbol{D}_{T}={\rm diag}\bigl(\{T^{1+(\theta\wedge\gamma_{j})}\}_{j=1}^{k}\bigr)$.
Note that the infeasible matrix $\boldsymbol{D}_{T}$ measures the
convergence rates and only facilitates theoretical analysis; it does
not appear in the practical implementation of DIVX. The \proofref[proof]{prop:ivx_variance_h}
is relegated to Section~\ref{subsec:Proofs-for-mult}.
\begin{prop}
\label{prop:ivx_variance_h}Under \prettyref{assump:mult-initval-1}
and \ref{assump:mult-innov-1-1}, as $(n,T)\to\infty$,
\[
{\rm var}\Biggl((n\boldsymbol{D}_{T})^{-1/2}\sum_{i=1}^{n}\sum_{t=1}^{T_{h}}\tilde{\bm{z}}_{i,t}e_{i,t+h}^{(h)}\Biggr)-\boldsymbol{D}_{T}^{-1/2}\boldsymbol{\Sigma}_{T}^{(h)}\boldsymbol{D}_{T}^{-1/2}\to\bm{0}_{k\times k},
\]
where
\begin{gather*}
\bm{\Sigma}_{T}^{(h)}=\bm{\Pi}_{T}(0)+\sum_{\ell=1}^{h-1}\bigl[\bm{\Pi}_{T}(\ell)+\bm{\Pi}_{T}(\ell)'\bigr]\qquad\text{with}\\
\bm{\Pi}_{T}(\ell)=\Gamma_{ee}^{(h)}(\ell)\cdot\mathbb{E}\left(\sum_{t=1}^{T_{h}}\bm{z}_{i,t}\bm{z}_{i,t}'\right)\bm{R}^{*\ell}\\
\quad\text{and}\quad\Gamma_{ee}^{(h)}(\ell)=\mathbb{E}\bigl(e_{i,t+h}^{(h)}e_{i,t+h-\ell}^{(h)}\bigr)\ \text{for }\ell=0,1,\dots,h-1.
\end{gather*}
\end{prop}
\begin{rem}
If all regressors are mildly integrated or local to unity, then $\boldsymbol{R}^{*}\to\boldsymbol{I}_{k}$
as $T\to\infty$, so the formula of $\boldsymbol{\Sigma}_{T}^{(h)}$
can be simplified as
\[
\boldsymbol{\Sigma}_{T}^{(h)}=\left[\sum_{\ell=-(h-1)}^{h-1}\mathbb{E}\bigl(e_{i,t+h}^{(h)}e_{i,t+h+\ell}^{(h)}\bigr)\right]\mathbb{E}\left(\sum_{t=1}^{T_{h}}\bm{z}_{i,t}\bm{z}_{i,t}^{\prime}\right),
\]
where $\sum_{\ell=-(h-1)}^{h-1}\mathbb{E}(e_{i,t+h}^{(h)}e_{i,t+h+\ell}^{(h)})$
is the long-run variance of the error term $e_{i,t+h}^{(h)}.$
\end{rem}
Following the justification of \prettyref{prop:ivx_variance_h} and
the asymptotic normality \eqref{eq:IVXJ_normal}, we estimate the
asymptotic variance of $\hat{\bm{\beta}}^{(h){\rm DIVX}}$ by
\begin{equation}
\hat{\bm{\Theta}}^{(h)}=\left(\sum_{i=1}^{n}\sum_{t=1}^{T_{h}}\tilde{\bm{z}}_{i,t}\tilde{\bm{x}}_{i,t}^{\prime}\right)^{-1}\hat{\bm{\Sigma}}^{(h)}\left(\sum_{i=1}^{n}\sum_{t=1}^{T_{h}}\tilde{\bm{x}}_{i,t}\tilde{\bm{z}}_{i,t}^{\prime}\right)^{-1},\label{eq:Theta hat IVXJ}
\end{equation}
with
\begin{align*}
\hat{\bm{\Sigma}}^{(h)} & =n^{-1}\sum_{i=1}^{n}\left(\hat{\bm{\Sigma}}_{i}^{(h)}-T\hat{\omega}_{ee}\cdot\bm{\Lambda}_{n,T}\bm{\bar{z}}_{i}\bm{\bar{z}}_{i}^{\prime}\bm{\Lambda}_{n,T}\right),
\end{align*}
where
\[
\hat{\bm{\Sigma}}_{i}^{(h)}=\hat{\bm{\Pi}}_{i,T}(0)+\sum_{\ell=1}^{h-1}\bigl[\hat{\bm{\Pi}}_{i,T}(\ell)+\hat{\bm{\Pi}}_{i,T}(\ell)^{\prime}\bigr],\ \hat{\bm{\Pi}}_{i,T}(\ell)=\hat{\Gamma}_{ee}^{(h)}(\ell)\cdot\sum_{t=1}^{T_{h}}\bm{z}_{i,t}\bm{z}_{i,t}'(\hat{\bm{R}}^{{\rm IVX}})^{\ell}
\]
with $\hat{\Gamma}_{ee}^{(h)}(\ell)=\frac{1}{nT_{h}}\sum_{i=1}^{n}\sum_{t=1}^{T_{h}}\hat{e}_{i,t+h}^{(h)}\hat{e}_{i,t+h-\ell}^{(h)}$,
and the second term is the finite-sample correction term following
\eqref{eq:CovMatMul}. By the argument we have used in the proof of
the simple regression, we can show that, as $(n,T)\to\infty$,
\begin{align}
\boldsymbol{D}_{T}^{-1/2}\hat{\bm{\Sigma}}_{n,T}^{(h)}\boldsymbol{D}_{T}^{-1/2} & \to_{p}\lim_{T\to\infty}\bigl(\boldsymbol{D}_{T}^{-1/2}\bm{\Sigma}_{T}^{(h)}\boldsymbol{D}_{T}^{-1/2}\bigr)=:\boldsymbol{\Lambda}^{(h)},\label{eq:Sigma pto}\\
\boldsymbol{D}_{T}^{-1/2}\sum_{i=1}^{n}\sum_{t=1}^{T_{h}}\tilde{\bm{z}}_{i,t}\tilde{\bm{x}}_{i,t}^{\prime}\boldsymbol{D}_{T}^{-1/2} & \to_{p}\boldsymbol{Q}^{*},\quad\text{say}.\label{eq:zeta x pto}
\end{align}
It follows that
\begin{equation}
(n\boldsymbol{D}_{T})^{1/2}\hat{\bm{\Theta}}^{(h)}(n\boldsymbol{D}_{T})^{1/2}\to_{p}(\boldsymbol{Q}^{*})^{-1}\boldsymbol{\Lambda}^{(h)}(\boldsymbol{Q}^{*\prime})^{-1}.\label{eq:IVXJ var est pto}
\end{equation}

Adapting the proofs in Section~\ref{sec:Proofs} for the simple regression,
we deduce
\[
(n\boldsymbol{D}_{T})^{-1/2}\left(\sum_{i=1}^{n}\sum_{t=1}^{T_{h}}\tilde{\boldsymbol{z}}_{i,t}e_{i,t+h}^{(h)}+\bm{\xi}_{n,T}^{(h)}(\boldsymbol{R}^{*},\boldsymbol{\omega}_{ev}^{*},\boldsymbol{\Omega}_{vv}^{*},\boldsymbol{\beta}^{*})\right)\to_{d}\mathcal{N}\bigl(\boldsymbol{0}_{k},\boldsymbol{\Lambda}^{(h)}\bigr),
\]
where $\boldsymbol{\Lambda}^{(h)}$ is defined in \eqref{eq:Sigma pto}.
By virtue of the argument used in the proof of \prettyref{thm:DIVX},
the estimators $\hat{\bm{R}}^{{\rm IVX}}$, $\hat{\bm{\omega}}_{ev}$,
$\hat{\bm{\Omega}}_{vv}$, and $\hat{\bm{\beta}}^{\mathrm{IVX}}$
produce a consistent estimator of the bias as $(n,T)\to\infty$ with
$n/T\to c\in[0,\infty)$, and thus
\begin{equation}
(n\boldsymbol{D}_{T})^{-1/2}\left(\sum_{i=1}^{n}\sum_{t=1}^{T_{h}}\tilde{\boldsymbol{z}}_{i,t}e_{i,t+h}^{(h)}+\bm{\xi}_{n,T}^{(h)}(\hat{\bm{R}}^{{\rm IVX}},\hat{\bm{\omega}}_{ev},\hat{\bm{\Omega}}_{vv},\hat{\bm{\beta}}^{\mathrm{IVX}})\right)\to_{d}\mathcal{N}\bigl(\boldsymbol{0}_{k},\boldsymbol{\Lambda}^{(h)}\bigr).\label{eq:zeta e dto}
\end{equation}
Consequently, \eqref{eq:zeta x pto} and \eqref{eq:zeta e dto} jointly
imply that, as $(n,T)\to\infty$ with $n/T\to c\in[0,\infty)$, we
have
\begin{align}
 & (n\boldsymbol{D}_{T})^{1/2}\left(\hat{\bm{\beta}}^{(h){\rm DIVX}}-\bm{\beta}^{(h)*}\right)\nonumber \\
 & \quad=\Biggl[(n\boldsymbol{D}_{T})^{-1/2}\sum_{i=1}^{n}\sum_{t=1}^{T_{h}}\tilde{\bm{z}}_{i,t}\tilde{\bm{x}}_{i,t}^{\prime}(n\boldsymbol{D}_{T})^{-1/2}\Biggr]^{-1}\nonumber \\
 & \qquad\quad(n\boldsymbol{D}_{T})^{-1/2}\Biggl[\sum_{i=1}^{n}\sum_{t=1}^{T_{h}}\tilde{\boldsymbol{z}}_{i,t}e_{i,t+h}^{(h)}+\bm{\xi}_{n,T}^{(h)}(\hat{\bm{R}}^{{\rm IVX}},\hat{\bm{\omega}}_{ev},\hat{\bm{\Omega}}_{vv},\hat{\bm{\beta}}^{\mathrm{IVX}})\Biggr]\nonumber \\
 & \quad\to_{d}\mathcal{N}\bigl(\boldsymbol{0}_{k},(\boldsymbol{Q}^{*})^{-1}\boldsymbol{\Lambda}^{(h)}(\boldsymbol{Q}^{*\prime})^{-1}\bigr).\label{eq:IVXJ_normal}
\end{align}
Suppose we are interested in testing a linear joint null hypothesis
$\mathbb{H}_{0}\colon\bm{A}\bm{\beta}^{(h)*}=\bm{q}$, where $\boldsymbol{A}$
is an $m\times k$ constant matrix of full row rank accommodating
$m$ linear restrictions, and $\boldsymbol{q}$ is an $m\times1$
constant vector. We reject $\mathbb{H}_{0}$ under the significance
level $\alpha$ if the Wald statistic
\begin{equation}
\text{Wald}^{(h){\rm DIVX}}=\bigl(\bm{A}\hat{\bm{\beta}}^{(h){\rm DIVX}}-\bm{q}\bigr)'\bigl(\bm{A}\hat{\bm{\Theta}}^{(h)}\bm{A}'\bigr)^{-1}\bigl(\bm{A}\hat{\bm{\beta}}^{(h){\rm DIVX}}-\bm{q}\bigr)\label{eq:Wald_IVXJ_h}
\end{equation}
is greater than the $(1-\alpha)$-th quantile of a $\chi^{2}$ distribution
of degree of freedom $m$, where $\hat{\bm{\Theta}}^{(h)}$ is defined
in \eqref{eq:Theta hat IVXJ}. We can then deduce from \eqref{eq:IVXJ_normal}
that the Wald statistic constructed by \eqref{eq:Wald_IVXJ_h} is
asymptotically $\chi^{2}(m)$ distributed under the null hypothesis
$\mathbb{H}_{0}\colon\bm{A}\bm{\beta}^{(h)}=\bm{q}$.

Numerical studies for local projections with DIVX are available in
Section \ref{subsec:Local-Projections} of the Online Appendices.

\section{Latent Group Structure \label{sec:Grouped-Heterogeneity}}

\subsection{Setup}

Our framework can also be extended to allow for heterogenous coefficients
across individuals. Allowing for full parameter heterogeneity across
all individuals in a panel data model would be too challenging, as
it significantly inflates the variance of the estimators. To balance
generality and feasibility, we follow the literature \citep{su2016identifying,su2018identifying,wang2021identifying}
to allow for latent group structures, where the coefficients are different
across groups while remain homogeneous in the same group.

We start with the univariate heterogeneous panel predictive models:
\begin{align}
y_{i,t+1} & =\mu_{y,i}+\beta_{i}^{*}x_{i,t}+e_{i,t+1},\quad\text{for }i=1,\dots,n\text{ and }t=1,\dots,T-1,\label{eq:predictive-group}
\end{align}
and
\begin{equation}
\begin{split}x_{i,t} & =\alpha_{i}+\delta_{i,t},\\
\delta_{i,t+1} & =\rho_{i}^{*}\delta_{i,t}+v_{i,t+1}.
\end{split}
\label{eq:state space-group}
\end{equation}
Different from \eqref{eq:predictive} and \eqref{eq:AR1}, the coefficients
$\beta_{i}^{*}$ and $\rho_{i}^{*}$ are heterogeneous across individuals
with latent group structures. Specifically, assume that $\{1,\dots,n\}$
is partitioned into $K$ disjoint groups $\mathcal{G}_{1},\mathcal{G}_{2},\dots,\mathcal{G}_{K}$,
so that
\[
\beta_{i}^{*}=\sum_{k=1}^{K}\beta_{[k]}^{*}\textbf{1}\{i\in\mathcal{G}_{k}\},\quad\rho_{i}^{*}=\sum_{k=1}^{K}\rho_{[k]}^{*}\textbf{1}\{i\in\mathcal{G}_{k}\}.
\]
To simplify the analysis, we assume that there are $K-1$ distinct
stationary groups $\mathcal{G}_{1},\mathcal{G}_{2},\dots,\mathcal{G}_{K-1}$
and one unit-root group $\mathcal{G}_{K}$. For $k=1,\dots,K-1$,
the AR coefficients $\rho_{[k]}^{*}\in(-1,1)$, while for $k=K$,
$\rho_{[k]}^{*}=1$. When there exist multiple local-to-unity groups,
the true AR(1) coefficients $\rho_{i}^{*}$ are too close to each
other and thus hard to separate in finite samples. Let $n_{k}=|\mathcal{G}_{k}|$
be the number of individuals belonging to group $\mathcal{G}_{k}$.
Let $\mathcal{G}^{*}=\{\mathcal{G}_{1},\dots,\mathcal{G}_{K}\}$ collect
all groups.

Define $\boldsymbol{\theta}_{[k]}^{*}=(\beta_{[k]}^{*},\rho_{[k]}^{*})^{\prime}$.
The following assumption characterizes the group pattern.
\begin{assumption}[Group structure]
\mbox{}\label{assump:group}
\begin{enumerate}[label=(\alph*)]
\item \label{enu:gap}There exists an absolute constant $\underline{c}$
such that $\min_{1\leq k_{1}\neq k_{2}\leq K}\|\boldsymbol{\theta}_{[k_{1}]}^{*}-\boldsymbol{\theta}_{[k_{2}]}^{*}\|>\underline{c}$.
\item \label{enu:proportion}$K$ is fixed and $n_{k}/n\to\tau_{k}\in(0,1)$
as $n\to\infty$ for each $k=1,\dots,K$.
\end{enumerate}
\end{assumption}
This assumption essentially follows the literature on grouped heterogeneity
in panel data models; see, e.g., \citet[Assumption~2]{bonhommeGroupedPatternsHeterogeneity2015}
and \citet[Assumption~A2]{wang2021identifying}. Condition~\ref{enu:gap}
ensures that groups are well-separated. Condition \ref{enu:proportion}
means that each group accounts for a nontrivially large proportion.

The general idea to handle heterogeneity with latent group structures
includes two steps. First, we identify the latent group structure
using the sequential binary segmentation algorithm (SBSA) proposed
by \citet{wang2021identifying}. Second, we perform DIVX for each
estimated group for inference of the group-wise coefficients $\beta_{[k]}^{*}$.

\subsection{Grouping with SBSA}

Let $\hat{\beta}_{i}=\sum_{t=1}^{T-1}\tilde{x}_{i,t}y_{i,t+1}/\sum_{t=1}^{T-1}\tilde{x}_{i,t}^{2}$
and $\hat{\rho}_{i}=\sum_{t=1}^{T-1}\tilde{x}_{i,t}x_{i,t+1}/\sum_{t=1}^{T-1}\tilde{x}_{i,t}^{2}$
be the time-series estimates for each cross-section unit $i=1,\dots,n$.
Denote $\theta_{i,1}^{*}=\beta_{i}^{*}$, $\theta_{i,2}^{*}=\rho_{i}^{*}$,
and similarly $\hat{\theta}_{i,1}=\hat{\beta}_{i}$ and $\hat{\theta}_{i,2}=\hat{\rho}_{i}$.
We sort the coefficients in ascending order and denote the order statistics
by
\[
\hat{\theta}_{\pi_{p}(1),p}\leq\hat{\theta}_{\pi_{p}(2),p}\leq\dots\leq\hat{\theta}_{\pi_{p}(n),p},\ p=1,2
\]
where $\{\pi_{p}(1),\dots,\pi_{p}(n)\}$ is a permutation of $\{1,\dots,n\}$
determined by the order relation. For any $1\leq i\leq j\leq n$,
let
\[
\mathcal{S}_{[i:j]}(p)=\{\hat{\theta}_{\pi_{p}(i),p},\dots,\hat{\theta}_{\pi_{p}(j),p}\},\ p=1,2
\]
be the set of $i$-th to $j$-th ordered estimates. SBSA leverages
the variance of $\mathcal{S}_{[i:j]}(p)$ to identify break points.
Intuitively, if $\theta_{\pi_{p}(i),p}^{*}=\dots=\theta_{\pi_{p}(j),p}^{*}$
and the estimates are $\sqrt{T}$ consistent, then we should observe
that the variance of $\mathcal{S}_{[i:j]}(p)$ is proportional to
$T^{-1}$ in large samples. On the other hand, if there is any break
point between $i$ and $j$, then the variance of $\mathcal{S}_{[i:j]}(p)$
should be bounded away from zero.

In view of this intuition, let
\[
\bar{\theta}_{[i:j],p}=\frac{1}{j-i+1}\sum_{\ell=i}^{j}\hat{\theta}_{\pi_{p}(\ell),p}\qquad\text{and}\qquad\hat{V}_{[i:j]}^{0}(p)=\frac{1}{j-i}\sum_{\ell=i}^{j}\left[\hat{\theta}_{\pi_{p}(\ell),p}-\bar{\theta}_{[i:j],p}\right]^{2}
\]
be the sample mean and variance of $\mathcal{S}_{[i:j]}(p)$. Let
$\hat{\sigma}_{i}^{2}(p)$ denote a consistent estimator of the asymptotic
variance of $\hat{\theta}_{\pi_{p}(i),p}$. Let
\[
\bar{\sigma}_{[i:j]}^{2}(\theta)=(j-i+1)^{-1}\sum_{\ell=i}^{j}\hat{\sigma}_{\ell}^{2}(\theta)
\]
 be the sample average of these variance estimates. Divide $\hat{V}_{[i:j]}^{0}(p)$
by $\bar{\sigma}_{[i:j]}^{2}(p)$ to get a standardized version:
\[
\hat{V}_{[i:j]}(p)=\hat{V}_{[i:j]}^{0}(p)\big/\bar{\sigma}_{[i:j]}^{2}(p).
\]
Define
\[
\hat{S}_{[i:j]}(p,m)=\frac{1}{j-i+1}\left(\sum_{\ell=i}^{m}\left[\hat{\theta}_{\pi_{p}(\ell),p}-\bar{\theta}_{[i:m],p}\right]^{2}+\sum_{\ell=m+1}^{j}\left[\hat{\theta}_{\pi_{p}(\ell),p}-\bar{\theta}_{[m+1:j],p}\right]^{2}\right),
\]
which measures the \emph{within-segment} variation in $\mathcal{S}_{[i:j]}(p)$
when a conjectured break happens at $m$ ($i\leq m\leq j$). Intuitively,
if $m$ is the true break point, the ordered estimates should exhibit
little variation over both segments $[i:m]$ and $[m+1:j]$. We therefor
estimate $m$ by minimizing $\hat{S}_{[i:j]}(p,m)$. This is how SBSA
works, as formalized by the following algorithm.
\begin{lyxalgorithm}[SBSA]
\mbox{}
\begin{enumerate}[label={\textnormal{Step \arabic*.}},leftmargin=*,itemsep=0pt,topsep=0pt,parsep=0pt]
\item Starting from $K=1$, there is no break point to determine.
\item When $K=2$, let $\hat{p}_{1}=\arg\max_{p\in\{1,2\}}\hat{V}_{[1:n]}(p)$.
We estimate the break point by
\[
\hat{m}_{1}=\arg\min_{1\leq m\leq n}\hat{S}_{[1:n]}(\hat{p}_{1},m).
\]
This gives rise to two segments: $\mathcal{S}_{[1:\hat{m}_{1}]}(\hat{p}_{1})$
and $\mathcal{S}_{[\hat{m}_{1}+1:n]}(\hat{p}_{1})$.
\item When $K\geq3$, we use $\hat{m}_{1}<\dots<\hat{m}_{K-2}$ to denote
the break points (perhaps relabeled) detected in previous steps. Let
\[
\hat{p}_{K-1}=\arg\max_{p\in\{1,2\}}\sum_{k=1}^{K}\hat{V}_{[(\hat{m}_{k-1}+1):\hat{m}_{k}]}(p).
\]
Define
\[
\hat{m}_{K-1}(k)=\arg\min_{\hat{m}_{k-1}+1\leq m\leq\hat{m}_{k}}\hat{S}_{[\hat{m}_{k-1}+1:\hat{m}_{k}]}(\hat{p}_{K-1},m)\qquad\text{for }k=1,\dots,K-1,
\]
where we prescribe $\hat{m}_{0}=0$ and $\hat{m}_{K-1}=n$. Then $\hat{m}_{K-1}(k)$
divides $\mathcal{S}_{[\hat{m}_{k-1}+1:\hat{m}_{k}]}$ into two segments.
For $k=1,\dots,K-1$, calculate
\begin{align*}
\hat{S}_{K-1}(k) & =\sum_{\ell=\hat{m}_{k-1}+1}^{\hat{m}_{K-1}(k)}\left[\hat{\theta}_{\pi_{\hat{p}_{K-1}}(\ell),\hat{p}_{K-1}}-\bar{\theta}_{[\hat{m}_{k-1}+1:\hat{m}_{K-1}(k)],\hat{p}_{K-1}}\right]^{2}\\
 & \qquad+\sum_{\ell=\hat{m}_{K-1}(k)+1}^{\hat{m}_{k}}\left[\hat{\theta}_{\pi_{\hat{p}_{K-1}}(\ell),\hat{p}_{K-1}}-\bar{\theta}_{[\hat{m}_{K-1}(k)+1:\hat{m}_{k}],\hat{p}_{K-1}}\right]^{2}\\
 & \qquad+\sum_{\substack{1\leq\tau\leq K-1\\
\tau\neq k
}
}\sum_{\ell=\hat{m}_{\tau-1}+1}^{\hat{m}_{\tau}}\left[\hat{\theta}_{\pi_{\hat{p}_{K-1}}(\ell),\hat{p}_{K-1}}-\bar{\theta}_{[\hat{m}_{\tau-1}+1:\hat{m}_{\tau}],\hat{p}_{K-1}}\right]^{2}.
\end{align*}
It measures the (within-segment) variation in all $\hat{\theta}_{i,\hat{p}_{K-1}}$
when an additional break point $\hat{m}_{K-1}(k)$ is detected. Let
\[
\hat{k}=\arg\min_{1\leq k\leq K-1}\hat{S}_{K-1}(k),
\]
which is the \emph{estimated segment number} based on which the new
break point $\hat{m}_{K-1}(\hat{k})$ is found. We now have $K-1$
break points $\{\hat{m}_{1},\dots,\hat{m}_{K-2},\hat{m}_{K-1}(\hat{k})\}$.
Relabel them to obtain $\hat{m}_{1}<\dots<\hat{m}_{K-1}$.
\item Repeat the last step until $K$ reaches the specified number of groups.
\end{enumerate}
\end{lyxalgorithm}

\subsection{Classification Consistency}

We impose the following additional assumptions. Suppose that the AR(1)
errors $\{v_{i,t}\}$ follow the heterogeneous linear processes
\begin{equation}
v_{i,t}=\sum_{s=0}^{\infty}g_{i,s}\varepsilon_{i,t-s}.\label{eq:linearProcesshetero}
\end{equation}

\begin{assumption}[Innovations]
\label{assump:innov-group} \mbox{}
\begin{enumerate}
\item \label{enu:w_cumu-1}For each $i$, let $\bm{w}_{i,t}=(e_{i,t},\varepsilon_{i,t})^{\prime}$,
with $e_{i,t}$ as in \eqref{eq:predictive}, denote a two-dimensional
strictly stationary and ergodic martingale difference sequence (m.d.s.)
adaptive to the filtration $\{\mathcal{F}_{i,t}=\sigma(\delta_{i,0},\alpha_{i},\bm{w}_{i,t},\bm{w}_{i,t-1},\dots)\}$.
$\{\bm{w}_{i,t}\}$ are independent~across $i$. In addition, we
assume absolutely summable fourth order cumulants: $\sup_{a,b,c,d\in\{1,2\}}\sum_{t,s,r=-\infty}^{\infty}|\kappa_{abcd}(0,t,s,r)|<\infty$,
where
\begin{align*}
\kappa_{abcd}(t_{1},t_{2},t_{3},t_{4}) & =\mathbb{E}(w_{a,i,t_{1}}w_{b,i,t_{2}}w_{c,i,t_{3}}w_{d,i,t_{4}})-\mathbb{E}(w_{a,i,t_{1}}w_{b,i,t_{2}})\mathbb{E}(w_{c,i,t_{3}}w_{d,i,t_{4}})\\
 & \qquad-\mathbb{E}(w_{a,i,t_{1}}w_{c,i,t_{3}})\mathbb{E}(w_{b,i,t_{2}}w_{d,i,t_{4}})-\mathbb{E}(w_{a,i,t_{1}}w_{d,i,t_{4}})\mathbb{E}(w_{b,i,t_{2}}w_{c,i,t_{3}}),
\end{align*}
 with $w_{a,i,t}$ being the $a$-th element of $\bm{w}_{i,t}$.
\item \label{enu:hetero_e-1}Let $u_{i,t}$ be i.i.d.\ random variables
with $\mathbb{E}(u_{i,t})=0$, $\mathbb{E}(u_{i,t}^{2})=1$ and $\mathbb{E}(u_{i,t}^{4})<\infty$.
For each $i$, the sequence $\{e_{i,t}\}$ admits $\mathrm{GARCH}(q,r)$
representation:
\[
e_{i,t}=h_{i,t}^{1/2}u_{i,t},\qquad h_{i,t}=\phi_{i}+\sum_{m=1}^{q}a_{i,m}e_{i,t-m}^{2}+\sum_{\ell=1}^{r}b_{i,\ell}h_{i,t-\ell},
\]
where the constant coefficients satisfy $\phi_{i}>0,$ $a_{i,m},b_{i,\ell}\geq0$
and $0\leq\sum_{m=1}^{q}a_{i,m}+\sum_{\ell=1}^{r}b_{i,\ell}<1$. In
addition,
\[
\phi_{i}=\sum_{k=1}^{K}\phi_{[k]}{\rm {\bf 1}}\{i\in\mathcal{G}_{k}\},\ a_{i,m}=\sum_{k=1}^{K}a_{[k],m}{\rm {\bf 1}}\{i\in\mathcal{G}_{k}\},\ b_{i,\ell}=\sum_{k=1}^{K}b_{[k],\ell}{\rm {\bf 1}}\{i\in\mathcal{G}_{k}\}.
\]
\item \label{enu:LP-2} The innovations $\{\varepsilon_{i,t}\}$ in the
linear process \eqref{eq:linearProcesshetero} are i.i.d.~across
$i$. The coefficients in the linear process \eqref{eq:linearProcesshetero}
satisfy $|g_{i,s}|\leq C_{0}\exp(-C_{g}s)$ for any $s$ with positive
constants $C_{0}$ and $C_{g}$. Besides,
\[
g_{i,s}=\sum_{k=1}^{K}g_{[k],s}{\rm {\bf 1}}\{i\in\mathcal{G}_{k}\}.
\]
\end{enumerate}
\end{assumption}
Assumption \ref{assump:innov-group} follows Assumption \ref{assump:innov}
for homogeneous models in the main text. The only difference is that
Assumption \ref{assump:innov-group} allows the coefficients in the
GARCH processes of $e_{i,t}$ and the linear processes $v_{i,t}$
to be heterogeneous with latent group structures.
\begin{assumption}
\mbox{}\label{assump:epsilon}
\begin{enumerate}
\item \label{enu:mixing} For each $i$, the error processes $\{e_{i,t}\}$
in \eqref{eq:predictive-group} and $\{v_{i,t}\}$ in \eqref{eq:state space-group}
are $\alpha$-mixing with geometric rates.
\item \label{enu:sub-exp}For all $i,t$ and each $b>0$, there exist absolute
constants $C_{e}$ and $K_{e}$ such that
\[
\mathbb{P}(|e_{i,t}|>b)+\mathbb{P}(|v_{i,t}|>b)\leq C_{e}\exp(-b/K_{e}).
\]
\end{enumerate}
\end{assumption}
Condition~\ref{enu:mixing} imposes the $\alpha$-mixing assumption
for the error terms $e_{i,t}$ which follows a GARCH process and $v_{i,t}$
which follows a stationary linear process. The literature has studied
the conditions under which a stationary GARCH model is strongly mixing
with geometric rates; see \citet{lindner2009stationarity} for a survey
and Theorem 8 in the same reference for a formal theoretical result.
The low-level sufficient conditions for a linear process to be strong-mixing
can be found in \citet{gorodetskiiStrongMixingProperty1978}. Condition~\ref{enu:sub-exp}
assumes the error terms to have sub-exponential tails, which is a
standard assumption in high-dimension settings. These two conditions
lead to concentration inequalities on weakly dependent processes with
thin tails \citep{merlevedeBernsteinTypeInequality2011}, which is
critical for establishing the probabilistic bounds of $\max_{1\leq i\leq n}|\sum_{t=1}^{T}x_{i,t}e_{i,t+1}|$
and $\max_{1\leq i\leq n}|\sum_{t=1}^{T}x_{i,t}v_{i,t+1}|$; see \citet{mei2022lasso}.

The following two lemmas demonstrate the uniform consistency property
of the time series least-squares estimators when $x_{i,t}$ is stationary
and unit root, respectively. The proofs are relegated to Section~\ref{subsec:Proofs-for-SBSA}.
\begin{lem}
\label{lem:sup_rho_hat_bound}Suppose Assumptions~\ref{assump:innov-group}
and \ref{assump:epsilon} hold. Then
\begin{equation}
\max_{1\leq i\leq n}|\hat{\beta}_{i}-\beta_{i}^{*}|=O_{p}\biggl(\sqrt{\frac{\log n}{T}}\biggr),\label{eq:beta uniform}
\end{equation}
In addition, assume that $\rho^{*}\in(-1,1)$ is a constant. Then,
there exists a function $R(\cdot)$ that is continuous, strictly increasing
on $(-1,1)$ and $R(\rho)<1$ for all $\rho\in(-1,1)$ such that
\begin{equation}
\max_{1\leq i\leq n}\left|\hat{\rho}_{i}-R(\rho^{*})\right|=O_{p}\biggl(\sqrt{\frac{\log n}{T}}\biggr),\label{eq:rho uniform}
\end{equation}
where $\hat{\rho}_{i}$ is the WG estimator for $\rho^{*}$ using
only the time series for individual $i$.
\end{lem}
\begin{lem}
\label{lem:sup_rho_ur}Suppose Assumptions~\ref{assump:innov-group}
and \ref{assump:epsilon} hold and $\rho^{*}=1$. Then
\[
\max_{1\leq i\leq n}\left|\hat{\rho}_{i}-1\right|=O_{p}\biggl(\frac{(\log n)^{5/2}}{T}\biggr).
\]
\end{lem}
In Lemma \ref{lem:sup_rho_hat_bound}, (\ref{eq:beta uniform}) establishes
the uniform convergence of $\hat{\beta}_{i}$. Due to weak dependence
of the AR(1) error $v_{i,t}$, the individual-specific estimator $\hat{\rho}_{i}$
is not consistent for $\rho^{*}$ when $x_{i,t}$ is stationary. Interestingly,
(\ref{eq:rho uniform}) indicates that the probability limit of $\hat{\rho}_{i}$
falls within the unit circle strictly, which is sufficient for SBSA
to separate groups. Lemma \ref{lem:sup_rho_ur} shows that when $x_{i,t}$
is unit root, the estimator $\hat{\rho}_{i}$ enjoys super consistency
with a convergence rate faster than $T^{-1/2}$.

The classification consistency is established in the following proposition.

\begin{prop}
\label{prop:classify}Suppose Assumptions~\ref{assump:group}, \ref{assump:innov-group}
and \ref{assump:epsilon} hold. Let $\hat{\mathcal{G}}=\{\hat{\mathcal{G}}_{1},\dots,\hat{\mathcal{G}}_{K}\}$
be the estimated group structure by the SBSA procedure. Then,
\[
\mathbb{P}(\hat{\mathcal{G}}=\mathcal{G}^{*})\to1
\]
 as $(n,T)\to\infty$.
\end{prop}

\subsection{Post-Classification DIVX Inference}

After classification by SBSA, we conduct group-wise DIVX inference.
The individual-specific instrument still follows (\ref{eq:def IV}).
The group-wise IVX estimator of $\beta_{[k]}^{*}$ is
\begin{equation}
\hat{\beta}_{[k]}^{{\rm IVX}}=\frac{\sum_{i\in\hat{\mathcal{G}}_{k}}\sum_{t=1}^{T-1}\tilde{z}_{i,t}y_{i,t+1}}{\sum_{i\in\hat{\mathcal{G}}_{k}}\sum_{t=1}^{T-1}\tilde{z}_{i,t}x_{i,t}},\label{eq:beta_IVX-group}
\end{equation}
with the (truncated) bias formula
\begin{equation}
\hat{b}_{[k],n,T}^{\mathrm{IVX}}(\boldsymbol{\omega}_{[k],ev,G}^{*},\rho^{*})=\frac{|\hat{\mathcal{G}}_{k}|\cdot\sum_{h=0}^{G}\Psi_{h,T-1}(\rho_{[k]}^{*},\rho_{z})\omega_{[k],ev,h}^{*}}{(T-1)\sum_{i\in\hat{\mathcal{G}}_{k}}\sum_{t=1}^{T-1}\tilde{z}_{i,t}x_{i,t}},\ G=\lfloor T^{1/4}\rfloor,\label{eq:truncated bias IVX-group}
\end{equation}
where $\Psi_{h,T-1}(\cdot,\cdot)$ follows (\ref{eq:ivx_bias}), and
\[
\omega_{[k],ev,h}^{*}=\sum_{k=1}^{K}\mathbb{E}(e_{i,t}v_{i,t+h})\textbf{1}\{i\in\hat{\mathcal{G}}_{k}\}.
\]
 To establish the IVX estimator of $\rho^{*}$, we generate the IV
$z_{i,t}^{(1)}$ following (\ref{eq:def IV (1)}), and define the
group-wise estimator of $\rho^{*}$ as
\[
\hat{\rho}_{[k]}^{{\rm WG}}=\frac{\sum_{i\in\hat{\mathcal{G}}_{k}}\sum_{t=1}^{T-1}\tilde{x}_{i,t}x_{i,t+1}}{\sum_{i\in\hat{\mathcal{G}}_{k}}\sum_{t=1}^{T-1}\tilde{x}_{i,t}^{2}}.
\]
Then the group-wise IVX estimator of $\rho^{*}$ is
\[
\hat{\rho}_{[k]}^{\mathrm{IVX}}=\frac{\sum_{i\in\hat{\mathcal{G}}_{k}}\sum_{t=1}^{T-1}\left(z_{i,t}^{(1)}x_{i,t+1}-\hat{\varDelta}_{[k],vv}\right)}{\sum_{i\in\hat{\mathcal{G}}_{k}}\sum_{t=1}^{T-1}z_{i,t}^{(1)}x_{i,t}},
\]
where $\hat{\omega}_{ee}=(n(T-1))^{-1}\sum_{i=1}^{n}\sum_{t=1}^{T-1}\hat{e}_{i,t+1}^{\mkern3mu 2}$,
\[
\hat{\omega}_{ev,h}=\frac{1}{n(T-h)}\sum_{i=1}^{n}\sum_{t=1}^{T-h}\hat{v}_{i,t+h}\hat{e}_{i,t},
\]
 and
\begin{equation}
\hat{\varDelta}_{[k],vv}=\frac{1}{|\hat{\mathcal{G}}_{k}|T}\sum_{h=1}^{G}\sum_{i\in\hat{\mathcal{G}}_{k}}\sum_{t=h+1}^{T}\hat{v}_{i,t}^{{\rm WG}}\hat{v}_{i,t-h}^{{\rm WG}},\ G=\lfloor T^{1/4}\rfloor.\label{eq:NW Deltavv-group}
\end{equation}
In (\ref{eq:NW Deltavv-group}), we have $\hat{v}_{i,t+1}^{{\rm WG}}=\tilde{x}_{i,t+1}-\hat{\rho}_{[k]}^{{\rm WG}}\tilde{x}_{i,t}$
for $i\in\hat{\mathcal{G}}_{k}$. In addition, the intertemporal covariances
are estimated by
\[
\hat{\omega}_{[k],ev,h}=\frac{1}{|\hat{\mathcal{G}}_{k}|T}\sum_{i\in\hat{\mathcal{G}}_{k}}\sum_{t=1}^{T-h}\hat{v}_{i,t+h}\hat{e}_{i,t},
\]
where for $i\in\hat{\mathcal{G}}_{k}$
\[
\hat{v}_{i,t+h}=\tilde{x}_{i,t+1}-\hat{\rho}_{[k]}^{{\rm IVX}}\tilde{x}_{i,t},\ \ \hat{e}_{i,t}=\tilde{y}_{i,t+1}-\hat{\beta}_{[k]}^{{\rm IVX}}\tilde{x}_{i,t}.
\]
With the newly defined estimators, the group-wise DIVX estimator is
\[
\hat{\beta}_{[k]}^{{\rm DIVX}}=\hat{\beta}_{[k]}^{{\rm IVX}}+\hat{b}_{[k],n,T}^{\mathrm{IVX}}(\hat{\boldsymbol{\omega}}_{[k],ev,G}^{*},\hat{\rho}_{[k]}^{\mathrm{IVX}})
\]
with its standard error
\[
\widehat{\varsigma}_{[k]}^{{\rm IVX}}=\frac{\left(\sum_{i\in\hat{\mathcal{G}}_{k}}\left(\sum_{t=1}^{T-1}z_{i,t}^{2}\hat{e}_{i,t+1}^{{\rm 2}}-T\hat{\lambda}_{[k]}(\bar{z}_{i})^{2}\hat{\omega}_{[k],ee}\right)\right)^{1/2}}{|\sum_{i\in\hat{\mathcal{G}}_{k}}\sum_{t=1}^{T-1}\tilde{z}_{i,t}x_{i,t}|},\quad\hat{\lambda}_{[k]}=\left(1-\dfrac{n}{T^{3/2}}\hat{\varrho}_{[k],ev}^{2}\right)_{+},
\]
where $\hat{\omega}_{[k],ee}=(|\hat{\mathcal{G}}_{k}|(T-1))^{-1}\sum_{i\in\hat{\mathcal{G}}_{k}}\sum_{t=1}^{T-1}\hat{e}_{i,t+1}^{\mkern3mu 2}$
and $\hat{\varrho}_{[k],ev}^{2}=\hat{\omega}_{[k],ev,0}^{2}/\hat{\omega}_{[k],ee}^{2}$.
Accordingly, the $t$-statistic is
\[
t_{[k]}^{\mathrm{DIVX}}=\frac{\hat{\beta}_{[k]}^{\mathrm{DIVX}}-\beta_{[k]}^{*}}{\hat{\varsigma}_{[k]}^{{\rm IVX}}}.
\]
We call this procedure for panel predictive regressions with latent
group structures SBSA-DIVX.

Parallel to \prettyref{thm:DIVX}, we establish the following theorem.
The \proofref[proof]{thm:divx_under_group} is relegated to Section~\ref{subsec:Proofs-for-SBSA}.
\begin{thm}
\label{thm:divx_under_group}Under \prettyref{assump:initval}, \ref{assump:innov},
\ref{assump:group} and \ref{assump:epsilon}, if $(n,T)\to\infty$
and $n/T\to c\in[0,\infty)$, we have
\[
t_{[k]}^{\mathrm{DIVX}}\to_{d}\mathcal{N}(0,1)
\]
 for each $k=1,2,...,K$.
\end{thm}
Simulation studies for heterogeneous panel predictive regressions
with group structures are available in Section \ref{subsec:Latent-Group-Structures}
of the Online Appendices.

\section{Additional Simulations\label{sec:Additional-Simulations}}

This section collects additional simulation studies. Section \ref{subsec:various dgp error}
enriches the settings in Section \ref{sec:simulation} of the main
text by considering various degrees of endogeneity and AR(1) coefficients.
Section \ref{subsec:Conditional-Heteroskedasticity} examines the
finite sample results under conditional heteroskedasticity. Section
\ref{subsec:Two-way-Fixed-Effects} conducts simulations with TWFE.
Section \ref{subsec:Multivariate-Regressors} considers multivariate
panel predictive regressions, evaluating both the $t$-test for an
individual coefficient and the Wald test for significance of multiple
coefficients. Section \ref{subsec:Local-Projections} studies panel
local projection, which is represented by a sequence of panel predictive
regressions. Section \ref{subsec:Latent-Group-Structures} simulates
heterogeneous panel predictive models with latent group structures,
showing the usefulness of our proposed SBSA-DIVX in Section \ref{sec:Grouped-Heterogeneity}.
Section \eqref{subsec:comparison} compares DIVX to alternative methodologies,
including the split-panel jackknife estimator, the X-differencing
estimator, and forwards and backwards recursive detrending.

\subsection{Robustness Check with Various Degrees of Endogeneity and AR coefficients\label{subsec:various dgp error}}

In Section \ref{sec:simulation} of the main text, we set $\omega_{12}^{*}=-0.95$
following the literature \citep{Kostakis2015,phillips2016robust}
to capture the typical case in stock-return predictive regressions.
To further evaluate the validity of DIVX in general DGPs, we vary
$\omega_{12}^{*}$ in $\{\pm0.95,\pm0.7,\pm0.5,\pm0.3,0\}$. In addition,
we add $\rho^{*}=0.2$ and 0.4 to examine the performance of DIVX
when the autocorrelation of $x_{i,t}$ is relatively weak. All other
settings follow Section \ref{sec:simulation}.

Figure \ref{fig:cover-omega12} exhibits the coverage probabilities
of the 95\% confidence intervals under various sample sizes and $\omega_{12}^{*}$.
The coverage probabilities are close to the nominal 95\% level. These
results witness the robustness of DIVX under different endogeneity
levels and AR(1) coefficients.

\begin{figure}[t]
\begin{centering}
\includegraphics[width=1\columnwidth]{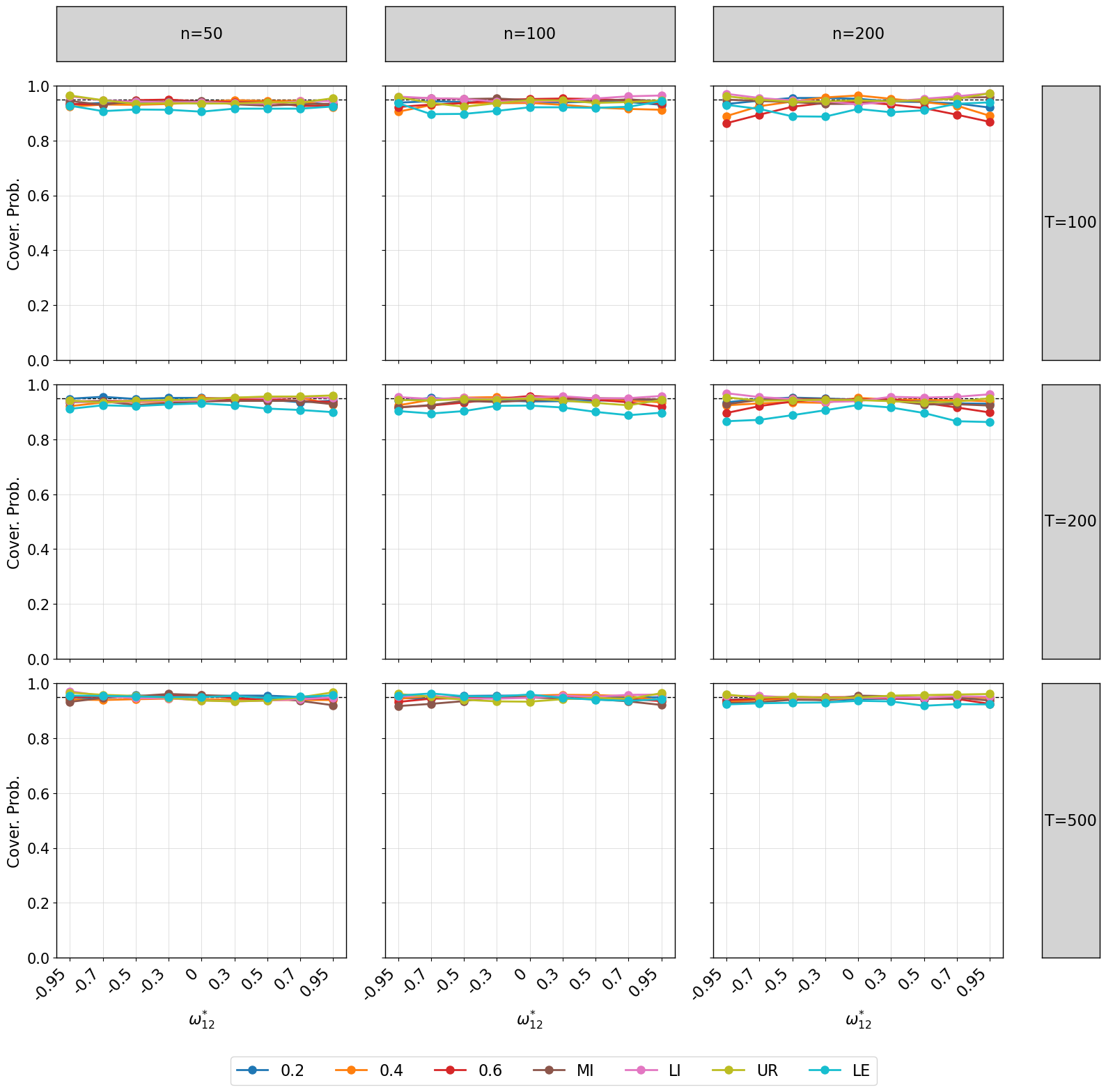}\smallskip{}
\noindent\begin{minipage}[t]{1\linewidth}%
{\footnotesize Notes: In the line graphs, each color represents a value
of $\rho^{*}$. In the legend, ``0.2'', ``0.4'', and ``0.6''
represents the stationary case when $\rho^{*}=0.2,0.4,0.6$. Other
abbreviations ``MI'', ``LI'', ``UR'', and ``LE'' mean that
$x_{i,t}$ is mildly integrated with $\rho^{*}=1-1/T^{0.75}$, locally
integrated with $\rho^{*}=1-1/T$, unit root with $\rho^{*}=1$, and
locally explosive with $\rho^{*}=1+1/T$, respectively.}%
\end{minipage}
\par\end{centering}
\centering{}\caption{\label{fig:cover-omega12} Coverage probabilities under various $\omega_{12}^{*}$}
\end{figure}

\subsection{Conditional Heteroskedasticity\label{subsec:Conditional-Heteroskedasticity} }

DIVX allows for conditional heteroskedasticity characterized in Assumption
\ref{assump:innov}\ref{enu:hetero_e}. To examine the robustness
of DIVX under conditional heteroskedasticity, we conduct additional
simulations with the error term $e_{i,t}$ in the main predictive
regression (\ref{eq:predictive}) following a GARCH(1,1) model
\[
e_{i,t}=h_{i,t}^{1/2}u_{i,t},\quad h_{i,t}=\phi+a_{1}h{}_{i,t-1}+b_{1}u_{i,t-1}^{2}.
\]
Following (\ref{eq:omega12 def}), the i.i.d.~innovations $\varepsilon_{i,t}$
for the ARMA(1,1) process $v_{i,t}$ and the error term in the GARCH(1,1)
process $u_{i,t}$ are generated from the bivariate normal distribution
\begin{equation}
(u_{i,t},\varepsilon_{i,t})^{\prime}\sim{\rm i.i.d.}\ \mathcal{N}\left(\begin{pmatrix}0\\
0
\end{pmatrix},\begin{pmatrix}1 & \omega_{12}^{*}\\
\omega_{12}^{*} & 1
\end{pmatrix}\right).\label{eq:omega12 def-hetero}
\end{equation}
We follow \citet{Kostakis2015} to set $a_{1}=0.13$, $b_{1}=0.85$,
and $\phi=1-a_{1}-b_{1}$. All other simulations settings follow Section
\ref{sec:simulation} in the main text.

Figures \ref{fig:bias_sd-hetero} to \ref{fig:power-hetero} show
the bias, size, and power of DIVX and the other five competitors under
conditional heteroskedasticity, parallel to Figures \ref{fig:bias_sd}
to \ref{fig:power} in Section \ref{sec:simulation} of the main text.
The simulation results share similarities to those in Section \ref{sec:simulation},
indicating that DIVX is robust to conditional heteroskedasticity.

\begin{figure}[h]
\begin{centering}
\includegraphics[width=0.95\columnwidth]{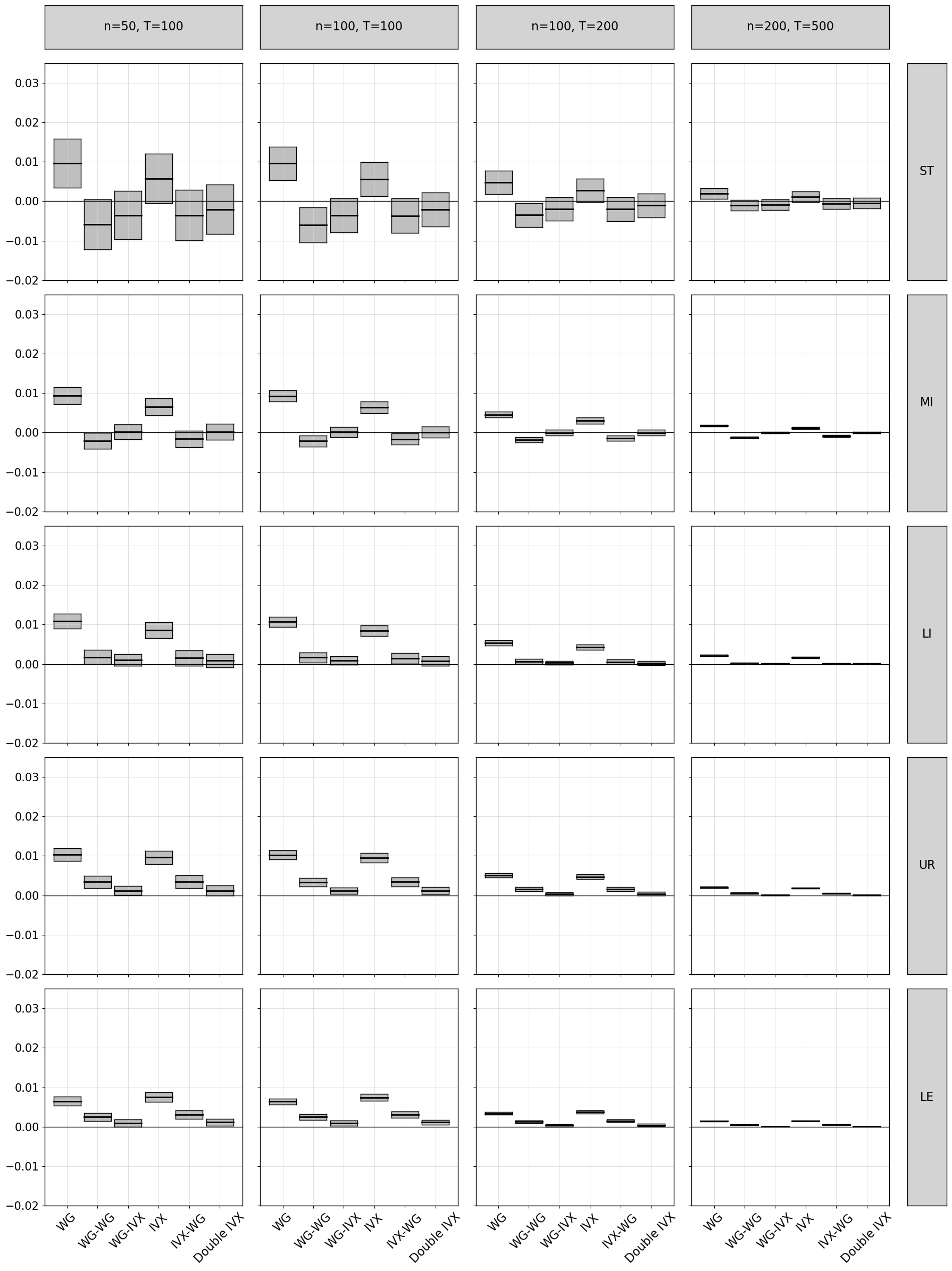}\smallskip{}
\noindent\begin{minipage}[t]{1\linewidth}%
{\footnotesize Notes: In each small box, the central line indicates
the empirical bias of $\hat{\beta}$, and the total height is twice
the empirical standard deviation, marking the lower and upper limits
$(\hat{\beta}-{\rm s.d.},\hat{\beta}+{\rm s.d.})$. To save space,
this figure only exhibits the results under $(n,T)\in\{(50,100),(100,100),(100,200),(200,500)\}$.}%
\end{minipage}
\par\end{centering}
\centering{}\caption{\label{fig:bias_sd-hetero} Bias and standard deviation under conditional
heteroskedasticity}
\end{figure}

\begin{figure}[h]
\begin{centering}
\includegraphics[width=1\columnwidth]{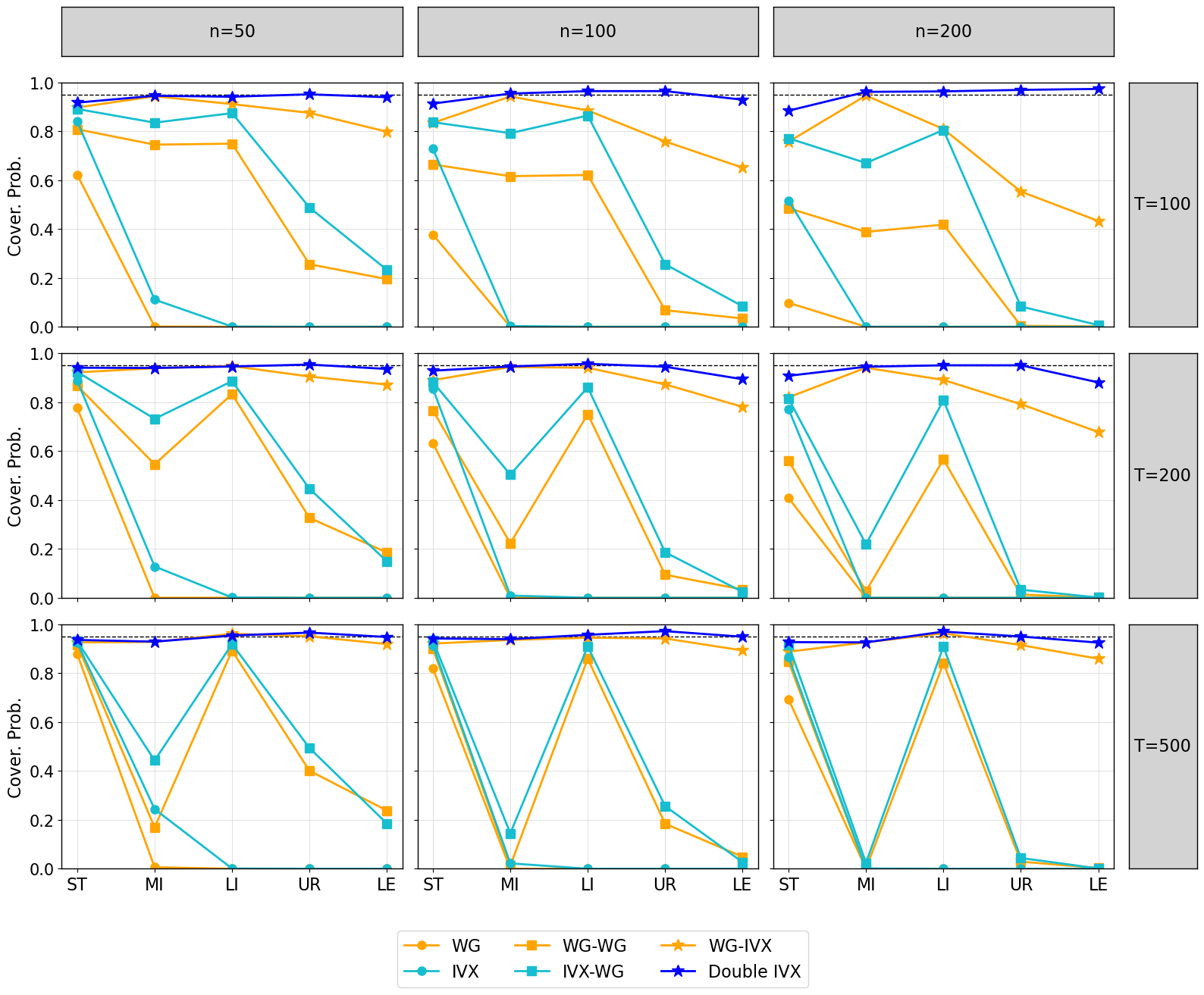}\smallskip{}
\par\end{centering}
\centering{}\caption{\label{fig:coverage-hetero} Coverage Probabilities of 95\% confidence
intervals when $\omega_{12}^{*}=-0.95$ with conditional heteroskedasticity}
\end{figure}

\begin{figure}[h]
\begin{centering}
\includegraphics[width=1\columnwidth]{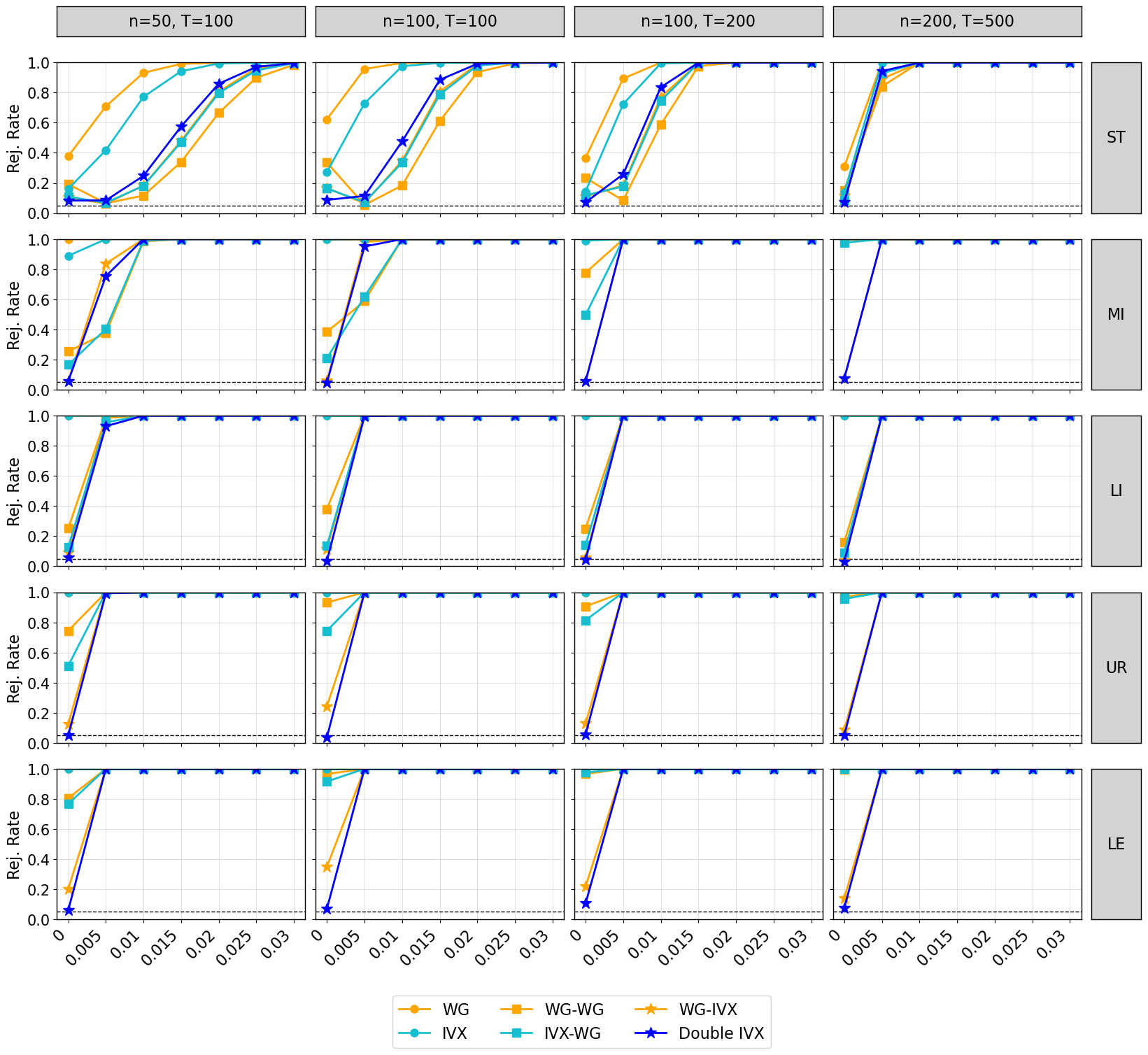}\smallskip{}
\par\end{centering}
\begin{centering}
\noindent\begin{minipage}[t]{1\linewidth}%
{\footnotesize Notes: To save space, this figure only exhibits the
rejection rates under the sample sizes $(n,T)\in\{(50,100),(100,100),(100,200),(200,500)\}$.}%
\end{minipage}
\par\end{centering}
\centering{}\caption{\label{fig:power-hetero} Rejection rates for $\mathbb{H}_{0}:\beta^{*}=0$
at the 5\% level when $\omega_{12}^{*}=-0.95$ with conditional heteroskedasticity}
\end{figure}

\subsection{Two-way Fixed Effects\label{subsec:Two-way-Fixed-Effects}}

Section \ref{sec:appdx two-way} discusses DIVX estimation and inference
for panel predictive regressions with two-way fixed effects. For simulation
studies in this scenario, we consider the predictive model (\ref{eq:predictive-TW})
with a quadratic time trend $f_{y,t}=0.025t+0.001t^{2}$. All other
simulation settings follow Section \ref{sec:simulation} of the main
text. The DIVX estimator and standard error follow Section \ref{sec:appdx two-way}.
Similar to DIVX, all other five methodologies WG, WG-WG, WG-IVX, IVX,
and IVX-WG follow the same path as in Section \ref{sec:appdx two-way},
which replaces the raw data $y_{i,t}$ and $x_{i,t}$ by the between-group
transformed data $\check{y}_{i,t}=y_{i,t}-n^{-1}\sum_{j=1}^{n}y_{j,t}$
and $\check{x}_{i,t}=x_{i,t}-n^{-1}\sum_{j=1}^{n}x_{j,t}$ from the
very beginning before conducting the estimation and inference.

Figures \ref{fig:bias_sd-twoway} and \ref{fig:coverage-twoway} show
the bias and size of DIVX and the other five competitors under conditional
heteroskedasticity, parallel to Figures \ref{fig:bias_sd} and \ref{fig:coverage}
in Section \ref{sec:simulation} of the main text. They are of similar
patterns to those in Section \ref{sec:simulation}, which showcases
DIVX's robustness.

\begin{figure}[h]
\begin{centering}
\includegraphics[width=0.95\columnwidth]{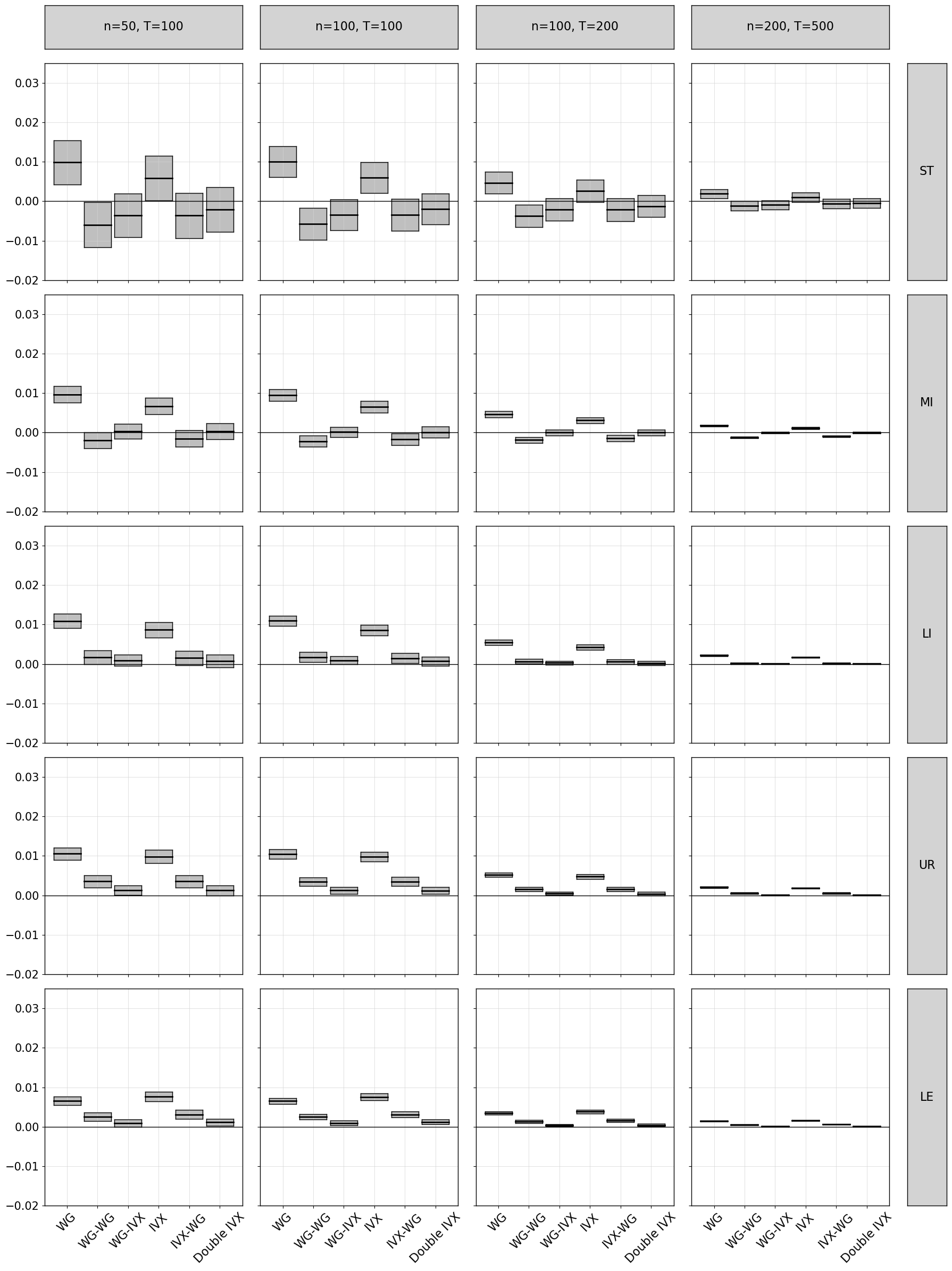}\smallskip{}
\noindent\begin{minipage}[t]{1\linewidth}%
{\footnotesize Notes: In each small box, the central line indicates
the empirical bias of $\hat{\beta}$, and the total height is twice
the empirical standard deviation, marking the lower and upper limits
$(\hat{\beta}-{\rm s.d.},\hat{\beta}+{\rm s.d.})$. To save space,
this figure only exhibits the results under $(n,T)\in\{(50,100),(100,100),(100,200),(200,500)\}$.}%
\end{minipage}
\par\end{centering}
\centering{}\caption{\label{fig:bias_sd-twoway} Bias and standard deviation under two-way
fixed effects}
\end{figure}

\begin{figure}[h]
\begin{centering}
\includegraphics[width=1\columnwidth]{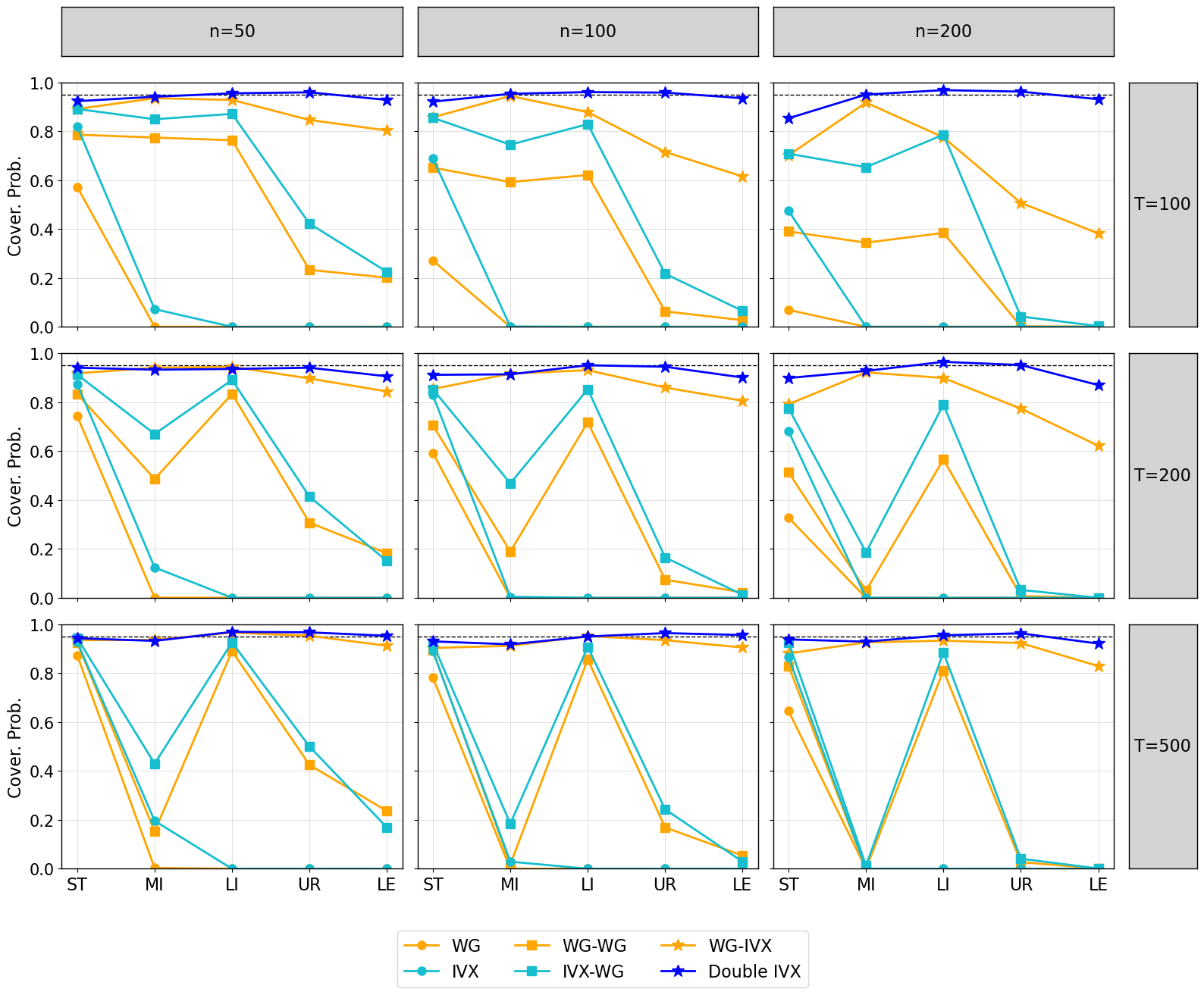}\smallskip{}
\par\end{centering}
\centering{}\caption{\label{fig:coverage-twoway} Coverage Probabilities of 95\% confidence
intervals when $\omega_{12}^{*}=-0.95$ with two-way fixed effects}
\end{figure}

\subsection{Multivariate Regressors\label{subsec:Multivariate-Regressors}}

Section \ref{sec:appdxMultivariate} formally establishes DIVX inference
for multivariate regressions. To examine its finite sample performance,
we conduct simulation studies with the data generating process (\ref{eq:pred_mul})
and (\ref{eq:AR_mul}). For the regressors (\ref{eq:AR_mul}) we set
$\boldsymbol{R}_{T}^{*}=\text{diag}\{(0.6,1-1/T^{0.75},1-1/T,1,1+1/T)^{\prime}\}$,
which includes all five categories of persistence including stationary,
MI, LI, UR, and LE. We set the coefficients of (\ref{eq:pred_mul})
as $\boldsymbol{\beta}^{*}=(b_{0},0,0,b_{0},0)^{\prime}$. We examine
the size and power of the null hypotheses $\mathbb{H}_{0}:\beta_{1}^{*}=0$
for the stationary regressor, $\mathbb{H}_{0}:\beta_{4}^{*}=0$ for
the unit root regressor, and the joint hypothesis $\mathbb{H}_{0}:\beta_{1}^{*}=\beta_{4}^{*}=0$
with the Wald test for mixed roots. We set $b_{0}=0$ when we analyze
the empirical size, and very $b_{0}$ in $\{0.005,0.01,0.015\}$ for
the empirical power. To construct the error terms, we generate the
innovations from the following i.i.d.~multivariate normal distributions
\begin{equation}
(e_{i,t},\boldsymbol{\varepsilon}_{i,t}^{\prime})^{\prime}\sim{\rm i.i.d.}\ \mathcal{N}\left(\boldsymbol{0}_{6},\boldsymbol{\Sigma}\right),\ \boldsymbol{\Sigma}=((-0.5)^{|i-j|})_{1\leq i,j\leq6}.\label{eq:omega12 def-1}
\end{equation}
The AR(1) errors are generated as the ARMA(1,1) processes $\boldsymbol{v}_{i,t}=0.5\boldsymbol{v}_{i,t-1}+\boldsymbol{\varepsilon}_{i,t}+0.4\boldsymbol{\varepsilon}_{i,t-1}$.
The fixed effects follow the settings in Section \ref{sec:simulation}
of the main paper.

Figures \ref{fig:bias-multiple} and \ref{fig:rmse_multiple} exhibit
the estimation biases and RMSEs of DIVX. These two indicators shrink
toward zero as the sample size gets larger, and the regressor gets
more persistent. These results not only reflect the consistency of
the DIVX estimator, but also echo the analysis in Remark \ref{rem:local power}
that highly persistent regressors enjoy faster convergence than stationary
regressors.

In terms of inference, Figure \ref{fig:rej-multiple} displays the
empirical rejection probabilities of the 5\% $t$-tests using DIVX
over the 1000 replications. When $b_{0}=0$, the $t$-tests for both
the stationary and the unit root regressors reject the null hypotheses
with probability close to the nominal level 5\%, suggesting that DIVX
produces correct empirical sizes for multivariate panel predictive
regressions. The empirical power increases if $b_{0}$ or the sample
sizes are larger, and the rejection rate of both $t$-tests achieves
100\% in all cases when $b_{0}$ reaches 0.015. In addition, the empirical
power for the unit root regressor is higher than that for the stationary
regressor, which again verifies that unit roots enjoy super-consistency
and therefore produce higher power against the alternative.

Figure \ref{fig:rej-Wald} shows the empirical rejection probabilities
of 5\% Wald test for the joint null $\mathbb{H}_{0}:\beta_{1}^{*}=\beta_{4}^{*}=0$.
Similar to the $t$-tests, the Wald test exhibits correct empirical
sizes when $b_{0}=0$, and high empirical power when $b_{0}>0$. These
results demonstrate the usefulness of DIVX in testing the significance
of an individual coefficient as well as the joint test for multiple
coefficients.

\begin{figure}[h]
\begin{centering}
\includegraphics[width=1\columnwidth]{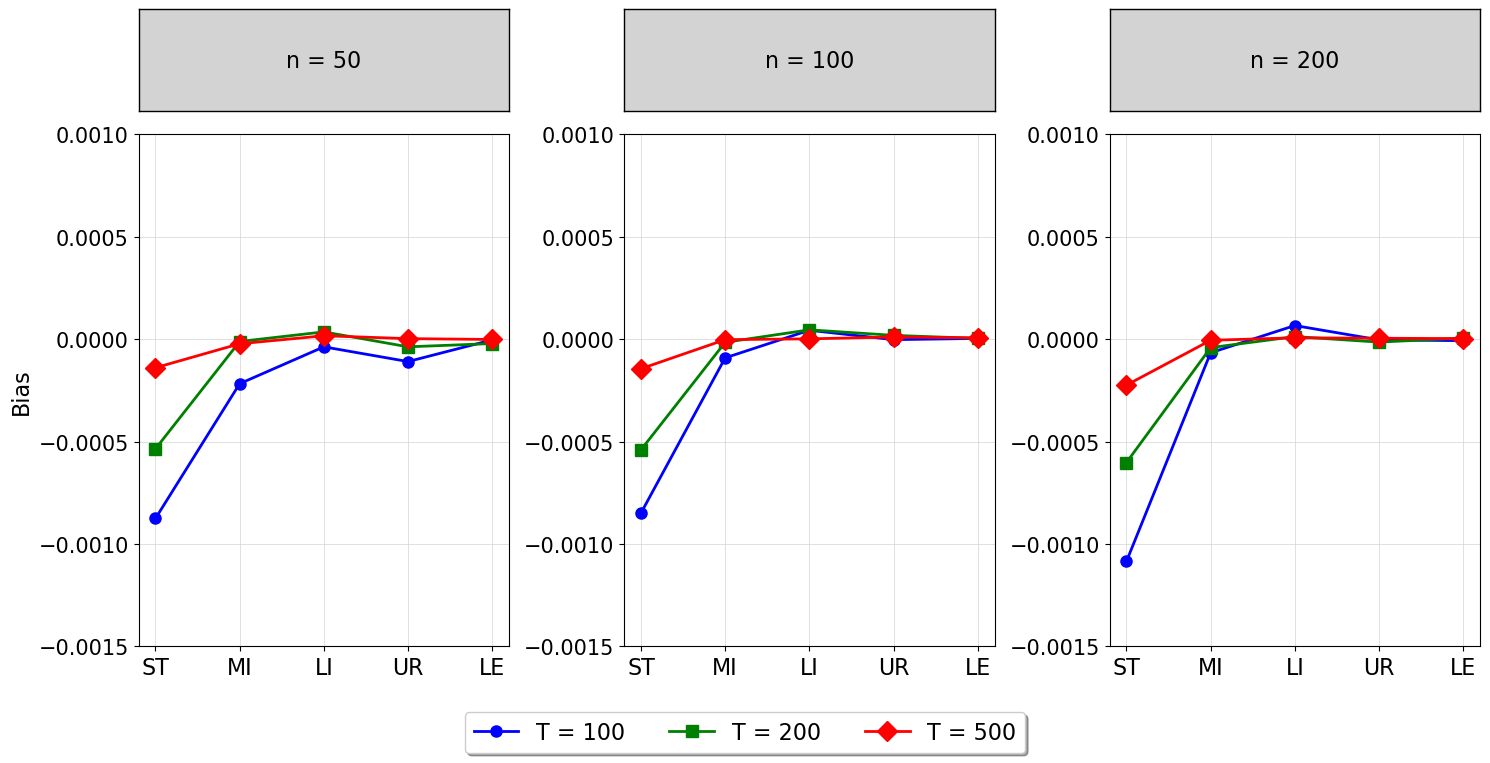}\smallskip{}
\par\end{centering}
\centering{}\caption{\label{fig:bias-multiple} Biases of estimation with multiple regressors}
\end{figure}

\begin{figure}[h]
\begin{centering}
\includegraphics[width=1\columnwidth]{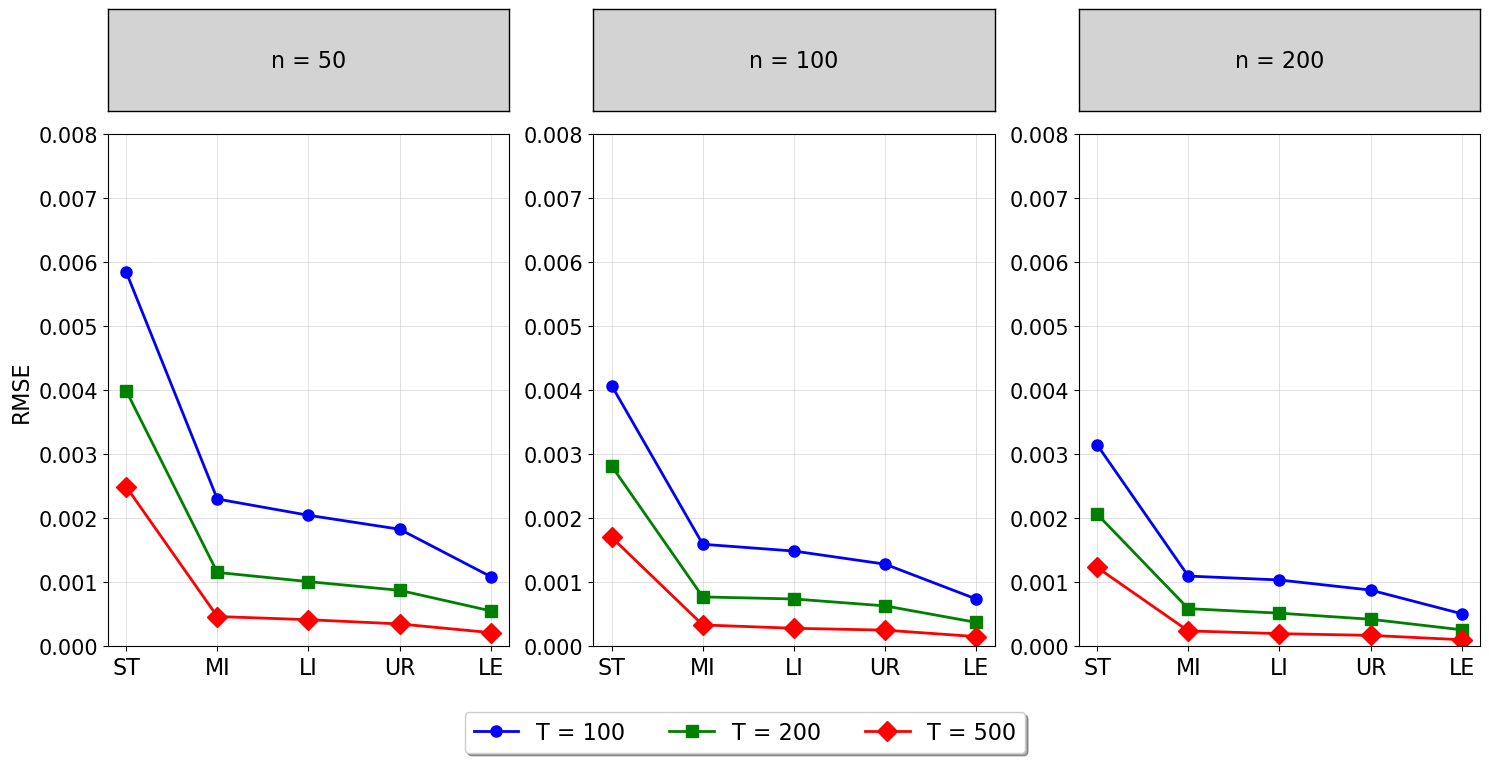}\smallskip{}
\par\end{centering}
\centering{}\caption{\label{fig:rmse_multiple} RMSEs with multiple regressors}
\end{figure}

\begin{figure}[h]
\begin{centering}
\includegraphics[width=1\columnwidth]{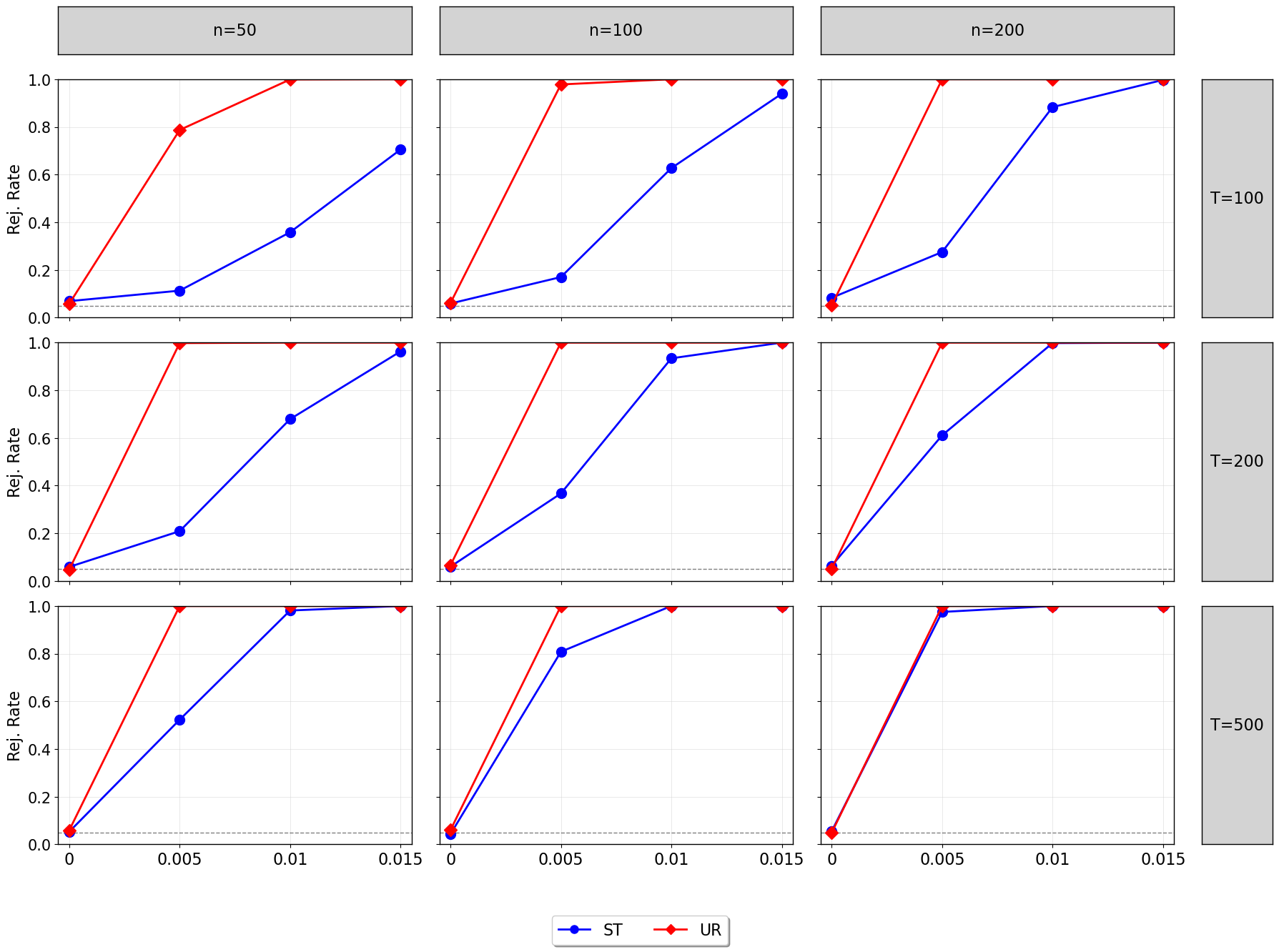}\smallskip{}
\par\end{centering}
\begin{centering}
\noindent\begin{minipage}[t]{1\linewidth}%
{\footnotesize Notes: ``ST'' and ``UR'' represent the rejection
rate of 5\% $t$-tests for} {\footnotesize$\mathbb{H}_{0}:\beta_{1}^{*}=0$
for the stationary regressor, and $\mathbb{H}_{0}:\beta_{4}^{*}=0$
for the unit root regressor, respectively. The x-axis in each graph
represents the value of $b_{0}$.}%
\end{minipage}
\par\end{centering}
\centering{}\caption{\label{fig:rej-multiple} Rejection rate of 5\% $t$-test under multiple
regressors}
\end{figure}

\begin{figure}[h]
\begin{centering}
\includegraphics[width=1\columnwidth]{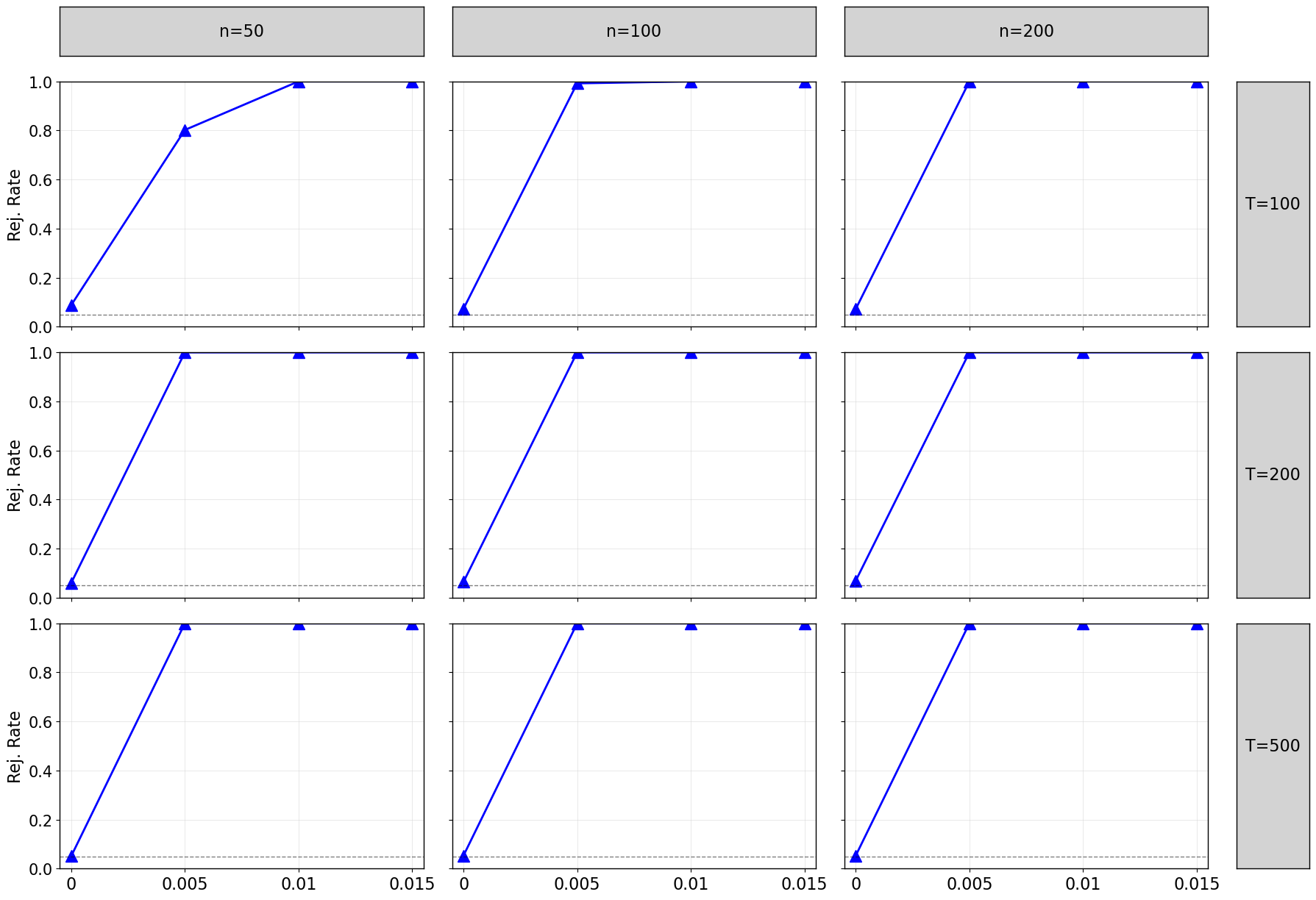}\smallskip{}
\par\end{centering}
\centering{}%
\noindent\begin{minipage}[t]{1\linewidth}%
{\footnotesize Notes: The x-axis in each graph represents the value
of $b_{0}$.}%
\end{minipage}\caption{\label{fig:rej-Wald} Rejection rate of 5\% Wald test with multiple
regressors}
\end{figure}

\subsection{Local Projections\label{subsec:Local-Projections}}

Section \ref{sec:Local-Projection} discusses the applications of
DIVX in panel local projections. Here we conduct simulation studies
to examine the performance of DIVX for panel local projections with
highly persistent regressors. For the data generating process, we
still consider the univariate models Recall that local projections
require an m.d.s.~error of the AR(1) model. Therefore, different
from the ARMA(1,1) process as in Section \ref{sec:simulation} of
the main text, in this section we set $v_{i,t}=\varepsilon_{i,t}$
as i.i.d.~normal variables. In addition, we set $\beta^{*}=-0.01$
and $\omega_{12}^{*}=-0.3$, which mimics the negative impact of financial
crises on economic growth that is widely discussed in empirical studies
\citep{mei2023implicit}. All other settings follow Section \ref{sec:simulation}.
We focus on the inference of the slope coefficient in the $h$-period
ahead predictive model $y_{i,t+h}=\mu_{y,i}^{(h)}+\beta^{(h)*}x_{i,t}+e_{i,t+h}^{(h)}$
for $h=1,2,3$.

Figure \ref{fig:lp-coverage} exhibits the coverage probabilities
of the DIVX 95\% confidence intervals for the impulse response function
$\{\beta^{(h)*}\}_{h=1,2,3}$. In all cases, DIVX confidence intervals
cover the truth with probabilities close to the nominal level.

\begin{figure}[h]
\begin{centering}
\includegraphics[width=1\columnwidth]{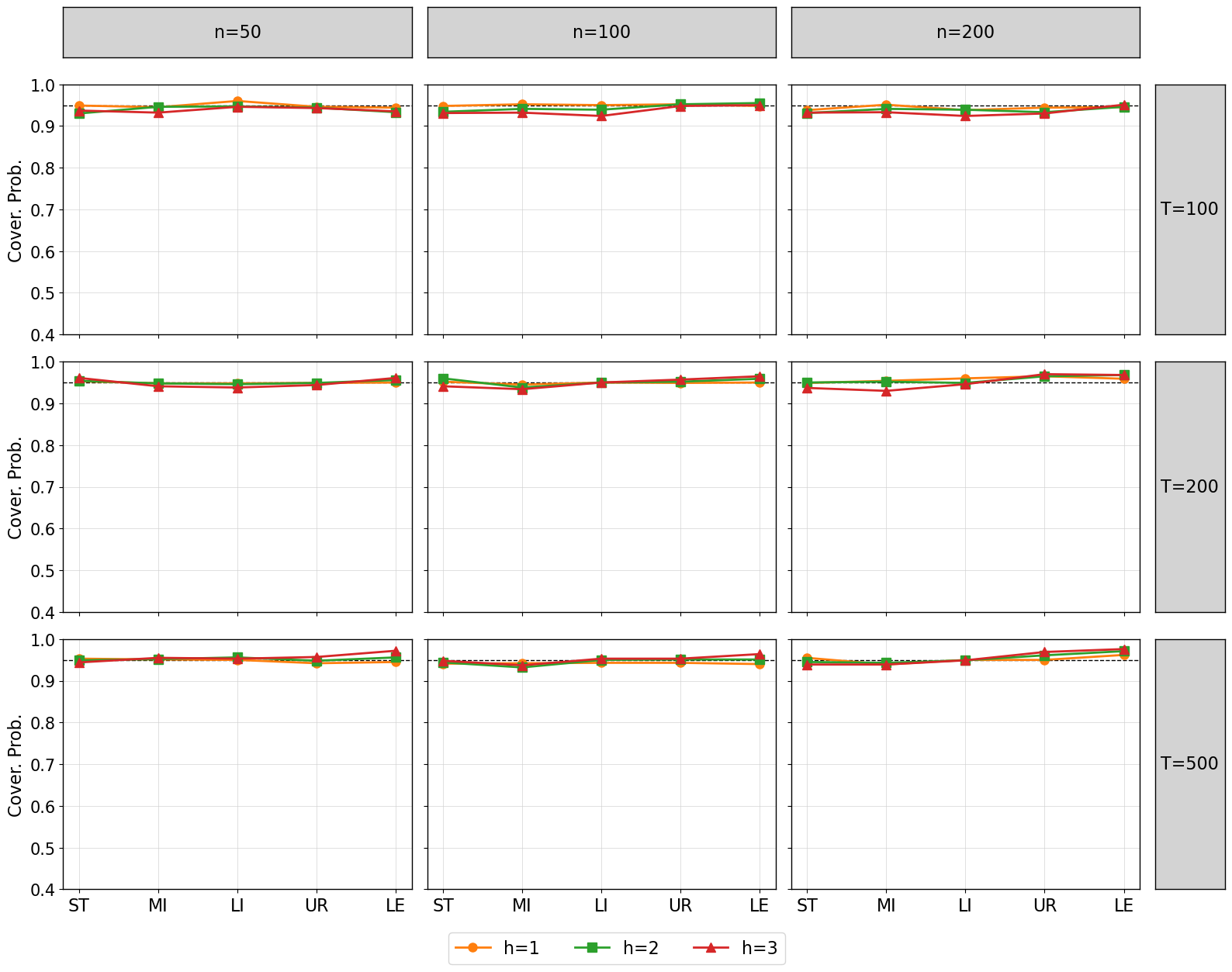}\smallskip{}
\par\end{centering}
\centering{}%
\noindent\begin{minipage}[t]{1\linewidth}%
{\footnotesize Notes: The x-axis represents various categories of persistence.
``ST'', ``MI'', ``LI'', ``UR'', and ``LE'' abbreviate stationary,
mildly integrated, locally integrated, unit roots, and locally explosive,
respectively.}%
\end{minipage}\caption{\label{fig:lp-coverage} Coverage probabilities of DIVX's 95\% confidence
intervals for local projections}
\end{figure}

\subsection{Latent Group Structures\label{subsec:Latent-Group-Structures}}

Section \ref{sec:Grouped-Heterogeneity} studies panel predictive
regressions with latent group structures and proposes SBSA-DIVX. Here
we perform simulations under data generating processes with latent
group structures. We consider the data generating processes (\ref{eq:predictive-group})
and (\ref{eq:state space-group}). We consider $K=3$ groups, including
two groups of stationary panels and one group of unit roots. Specifically,
we set $(\rho_{[1]}^{*},\rho_{[2]}^{*},\rho_{[3]}^{*})=(0.2,0.7,1)$
and $(\beta_{[1]}^{*},\beta_{[2]}^{*},\beta_{[3]}^{*})=(-0.1,-0.05,0)$,
with the group sizes $|\mathcal{G}_{1}|:|\mathcal{G}_{2}|:|\mathcal{G}_{3}|=3:3:4$.
For example, if $n=100$, we have $|\mathcal{G}_{1}|=|\mathcal{G}_{2}|=30$
and $|\mathcal{G}_{3}|=40$. The settings for the fixed effects and
error terms follow Section \ref{sec:simulation} in the main text.

Table \ref{tab: group} displays the simulation results. The column
``Accuracy'' lists the grouping accuracy by SBSA, measured by the
proportion of individuals in each group that are correctly classified
into the its original group. When $T$ is sufficiently large ($T=200$
or 500), the proportions of individuals that are correctly classified
exceed 98\% in all three groups, suggesting that SBSA effectively
identifies the latent group structures in panel predictive regressions
even at the presence of highly persistent regressors. In terms of
estimation error and coverage probabilities, we compare the results
of SBSA-DIVX to the ``Oracle'' estimator, where the group membership
is known a priori and the DIVX estimation and inference are performed
on each group respectively. In Table \ref{tab: group}, ``RMSE''
represents the root mean squared errors of the DIVX point estimate
for the true coefficient in each group, while the column ``Coverage
Probabilities'' displays the coverage of the 95\% confidence interval
of each group. The RMSEs of the Oracle DIVX approach towards zero
as the sample sizes grow, and the coverage probabilities of the 95\%
confidence intervals are close to the nominal level in all scenarios.
Also, the RMSE is the smallest in Group 3 with a unit root regressor,
which echoes the super-consistency of highly persistent regressors.
When $T=200$ or 500, with high classification accuracy, the RMSEs
of SBSA-DIVX are comparable to the Oracle DIVX, with correct empirical
coverage probabilities as well.

\subsection{Comparison to Alternative Estimators\label{subsec:comparison}}

To demonstrate our central idea, Section \ref{sec:Grouped-Heterogeneity}
involves the WG- and IVX-based estimators only. As mentioned in our
introduction, the vast literature of dynamic panel regressions has
proposed alternative estimators. Among recently developed methods,
the split-panel jackknife estimator \citep[SPJ,][]{dhaene2015split,chudik2018half},
the X-differencing estimator \citep[XDiff,][]{han2014x}, and forwards
and backwards recursive detrending \citep[FB,][]{westerlund2017testing},
are most relevant to our setting. SPJ removes the Nickell bias by
splitting the panel over the time dimension with an easy formula
\[
\hat{\beta}^{{\rm SPJ}}=2\hat{\beta}^{{\rm WG}}-0.5(\hat{\beta}_{a}^{{\rm WG}}+\hat{\beta}_{b}^{{\rm WG}}),
\]
 where $\hat{\beta}_{a}^{{\rm WG}}$ and $\hat{\beta}_{b}^{{\rm WG}}$
denote the WG estimator using the time periods $\mathcal{T}_{a}=\{1,2,...,\lfloor T/2\rfloor\}$
and $\mathcal{T}_{b}=\{\lfloor T/2\rfloor+1,\lfloor T/2\rfloor+2,...,T\}$,
respectively. Nevertheless, SPJ does not cover nonstationary panels.
The other two (XDiff and FB) rigorously accommodate nonstationary
panels with theoretical justifications. However, the former works
for stationary and pure unit root regressors only, and requires correct
specification of the AR(1) model. The latter exhibits power loss when
the regressor becomes highly persistent; see \citet[Section 4]{westerlund2017testing}
for details.

This section focuses on the settings of Section \ref{sec:simulation}
of the main text with $(n,T)\in\{(100,100),(100,200)\}$, and compare
the size and power of DIVX to SPJ, IVX-XDiff, and FB. We vary the
true $\beta^{*}$ in $\{0,0.005,0.01,0.015,0.02,0.025,0.3\}$. Similar
to IVX-WG discussed in the main text, IVX-XDiff adopts IVX to estimate
$\beta^{*}$ in the first step, and applies XDiff to estimate $\rho^{*}$
for bias correction. Figure \ref{fig:comparison} plots the empirical
rejection probabilities of the 5\% $t$-test for $\mathbb{H}_{0}:\beta^{*}=0$
under different settings. When $x_{i,t}$ is stationary, all methods
produce accurate empirical sizes when $\beta^{*}=0$ except the IVX-XDiff
estimator that requires correctly specified AR(1) regressors. The
power of DIVX is competitive among all estimators under the alternatives.
In the scenarios with $\rho^{*}=1$ where $x_{i,t}$ is highly persistent,
both SPJ and IVX-XDiff inferences are evidently biased, with the empirical
size far from the nominal level 5\% when $\beta^{*}=0$. Both FB and
DIVX exhibit correct empirical sizes. Under the alternative, the power
of DIVX dominates FB. The power loss of FB is the most severe in the
UR case where $x_{i,t}$ is a pure unit root. These findings are consistent
with \citet[Section 4]{westerlund2017testing}.

\medskip

To summarize, the additional simulations in this section verify the
excellent performance of DIVX for panel predictive regressions in
a variety of scenarios, including various degrees of endogeneity,
relatively small AR(1) coefficients (Section \ref{subsec:various dgp error}),
and conditional heteroskedasticity (\prettyref{subsec:Conditional-Heteroskedasticity}).
With necessary refinements, DIVX can be extended to TWFE, multivariate
regressions, local projections, and heterogeneous panels with latent
groups. The simulation studies in Sections \ref{subsec:Two-way-Fixed-Effects}
to \ref{subsec:Latent-Group-Structures} witness the robustness of
DIVX in diverse DGPs. Lastly, Simulation results in Section \ref{subsec:comparison}
demonstrates that DIVX outperforms alternative methods.

\begin{table}[t]
\centering
\caption{Simulation results for panel predictive regressions with latent group structures}
\label{tab: group}
\begin{tabular}{c|c|c|c|cc|cc}
\hline \hline
\multirow{2}{*}{$n$}   & \multirow{2}{*}{$T$}   & \multirow{2}{*}{Group} & \multirow{2}{*}{Accuracy} & \multicolumn{2}{c|}{RMSE}     & \multicolumn{2}{c}{Coverage Probabilities} \\
\cline{5-8}
                     &                      &                        &                           & SBSA-DIVX & Oracle & SBSA-DIVX        & Oracle        \\
                     \hline
\multirow{9}{*}{50}  & \multirow{3}{*}{100} & 1                      & 0.932                     & 0.037               & 0.017  & 0.802                      & 0.937         \\
                     &                      & 2                      & 0.825                     & 0.020               & 0.009  & 0.808                      & 0.946         \\
                     &                      & 3                      & 0.962                     & 0.003               & 0.002  & 0.884                      & 0.938         \\
\cline{2-8}
                     & \multirow{3}{*}{200} & 1                      & 1.000                     & 0.012               & 0.011  & 0.932                      & 0.943         \\
                     &                      & 2                      & 0.986                     & 0.006               & 0.006  & 0.931                      & 0.94          \\
                     &                      & 3                      & 0.997                     & 0.001               & 0.001  & 0.938                      & 0.939         \\
\cline{2-8}
                     & \multirow{3}{*}{500} & 1                      & 1.000                     & 0.007               & 0.007  & 0.954                      & 0.954         \\
                     &                      & 2                      & 1.000                     & 0.004               & 0.004  & 0.95                       & 0.95          \\
                     &                      & 3                      & 1.000                     & 0.000               & 0.000  & 0.956                      & 0.956         \\
\cline{1-8}
\multirow{9}{*}{100} & \multirow{3}{*}{100} & 1                      & 0.971                     & 0.024               & 0.011  & 0.781                      & 0.944         \\
                     &                      & 2                      & 0.874                     & 0.013               & 0.007  & 0.825                      & 0.922         \\
                     &                      & 3                      & 0.959                     & 0.002               & 0.002  & 0.904                      & 0.938         \\
\cline{2-8}
                     & \multirow{3}{*}{200} & 1                      & 1.000                     & 0.008               & 0.008  & 0.941                      & 0.945         \\
                     &                      & 2                      & 0.988                     & 0.005               & 0.004  & 0.935                      & 0.936         \\
                     &                      & 3                      & 0.998                     & 0.001               & 0.001  & 0.939                      & 0.94          \\
\cline{2-8}
                     & \multirow{3}{*}{500} & 1                      & 1.000                     & 0.005               & 0.005  & 0.948                      & 0.947         \\
                     &                      & 2                      & 1.000                     & 0.003               & 0.003  & 0.944                      & 0.944         \\
                     &                      & 3                      & 1.000                     & 0.000               & 0.000  & 0.945                      & 0.945         \\
\cline{1-8}
\multirow{9}{*}{200} & \multirow{3}{*}{100} & 1                      & 0.990                     & 0.014               & 0.009  & 0.703                      & 0.924         \\
                     &                      & 2                      & 0.900                     & 0.007               & 0.004  & 0.841                      & 0.943         \\
                     &                      & 3                      & 0.959                     & 0.001               & 0.001  & 0.92                       & 0.924         \\
\cline{2-8}
                     & \multirow{3}{*}{200} & 1                      & 1.000                     & 0.006               & 0.006  & 0.933                      & 0.936         \\
                     &                      & 2                      & 0.988                     & 0.003               & 0.003  & 0.922                      & 0.919         \\
                     &                      & 3                      & 0.998                     & 0.001               & 0.001  & 0.928                      & 0.922         \\
\cline{2-8}
                     & \multirow{3}{*}{500} & 1                      & 1.000                     & 0.004               & 0.004  & 0.941                      & 0.941         \\
                     &                      & 2                      & 1.000                     & 0.002               & 0.002  & 0.941                      & 0.941         \\
                     &                      & 3                      & 1.000                     & 0.000               & 0.000  & 0.937                      & 0.937 \\
\hline   \hline
\end{tabular}
\end{table}

\begin{figure}[h]
\begin{centering}
\includegraphics[scale=0.5]{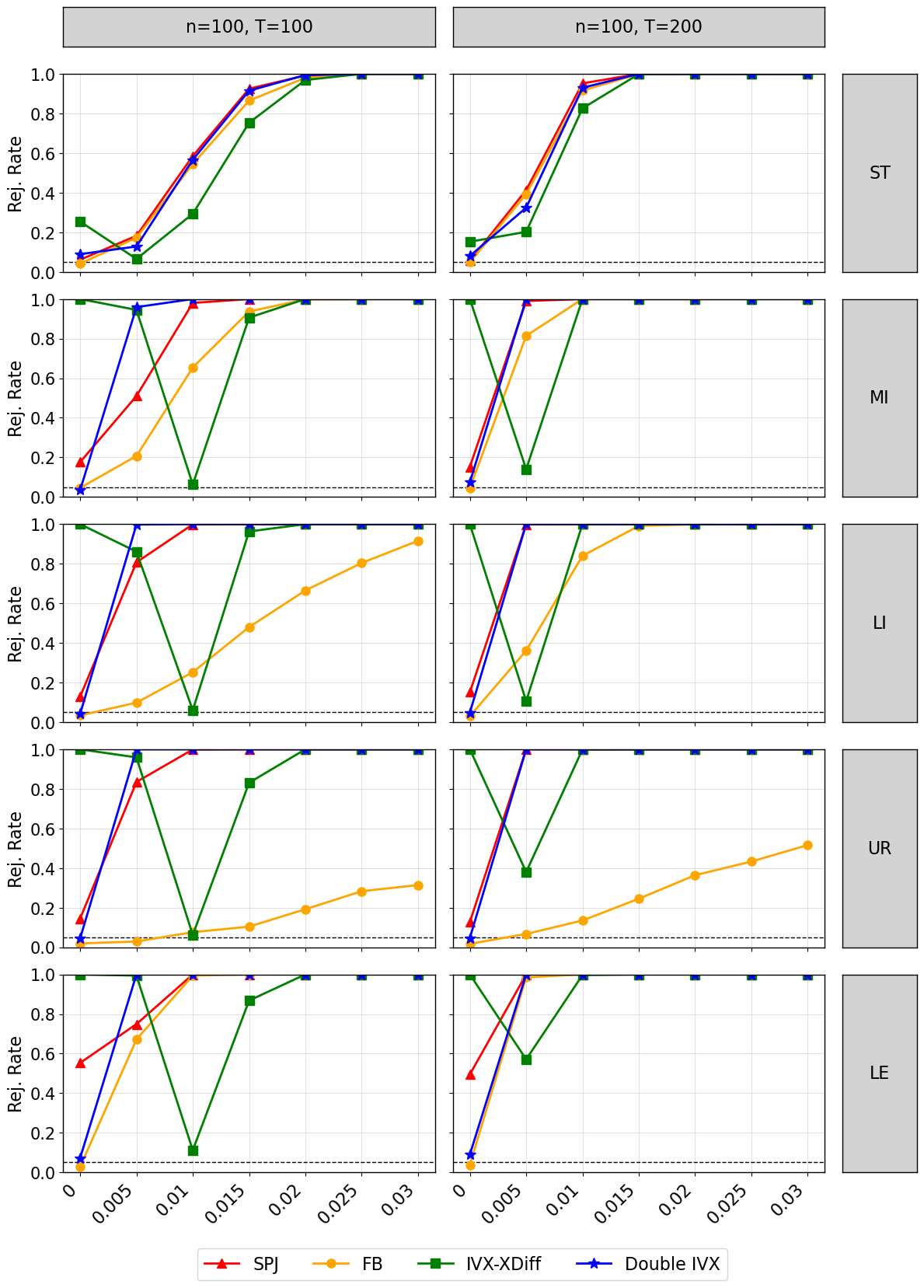}\smallskip{}
\par\end{centering}
\centering{}%
\noindent\begin{minipage}[t]{1\linewidth}%
{\footnotesize Notes: The x-axis in each graph represents the value
of true $\beta^{*}$. ``ST'', ``MI'', ``LI'', ``UR'', and
``LE'' abbreviate stationary, mildly integrated, locally integrated,
unit roots, and locally explosive, respectively.}%
\end{minipage}\caption{\label{fig:comparison} Rejection Rates of 5\% $t$-test by SPJ, FB,
IVX-XDiff, and DIVX}
\end{figure}

\clearpage{}

\setcounter{page}{1}
\renewcommand{\thesection}{S.\Alph{section}}
\setcounter{section}{0}
% for correct hyperref bookmarks
% https://tex.stackexchange.com/questions/6098
\renewcommand{\theHsection}{Supplement.\thesection}
\renewcommand{\theHequation}{Supplement.\arabic{equation}}
\setcounter{footnote}{0}
\setcounter{table}{0}
\setcounter{figure}{0}
\setcounter{equation}{0}
\renewcommand{\thefootnote}{S.\arabic{footnote}}
\renewcommand{\theequation}{S.\arabic{equation}}
\renewcommand{\thefigure}{S.\arabic{figure}}
\renewcommand{\thetable}{S.\arabic{table}}
\setcounter{thm}{0}
\setcounter{lem}{0}
\setcounter{rem}{0}
\setcounter{prop}{0}
\renewcommand{\thethm}{S.\arabic{thm}}
\renewcommand{\thelem}{S.\arabic{lem}}
\renewcommand{\therem}{S.\arabic{rem}}
\renewcommand{\theprop}{S.\arabic{prop}}
\begin{center}
{\LARGE Supplementary Materials to ``Nickell Meets Stambaugh:\\[5pt] A Tale of Two Biases in Panel Predictive Regression''}
\par\end{center}

\begin{center}
{\large Chengwang Liao$^a$, Ziwei Mei$^b$, Zhentao Shi$^a$ \\
$^a$The Chinese University of Hong Kong \\
$^b$University of Macau }\\
\par\end{center}

The Supplementary Materials collect all theoretical proofs. Section
\ref{sec:Proofs-of-Main} includes the proofs of main results. Section
\ref{sec:Proofs} shows the proofs of extensive theoretical results
about multivariate regressions, multiple-period-ahead prediction,
and heterogeneity with a group structure. Section \ref{sec:Supporting-lemmas}
collects technical lemmas used in the proofs. Section~\ref{sec:Proof-Pre-Lemmas}
proves the preliminary lemmas in Section \ref{subsec:Pre-Lemmas}.
Section~\ref{sec:Proofs-for-WG-1} proves the lemmas for WG in Section
\ref{subsec:lemmas-for-WG}. Section~\ref{sec:Proofs-for-IVX-1}
proves the lemmas for IVX in Section \ref{subsec:lemmas-for-IVX}.
Section~\ref{sec:Proof-of-Moments} derives the analytic formulae
for moments of stochastic integrals in Section \ref{subsec:stoch_integral}.
We use $\mathbb{E}_{s}(\cdot):=\mathbb{E}(\cdot|\mathcal{F}_{s})$
to denote conditional expectation with respect to a sigma-field $\mathcal{F}_{s}$
where $\{\mathcal{F}_{s}\}_{s}$ is a filtration. For two nonnegative
sequences $a_{T}$ and $b_{T}$, we write $a_{T}\lesssim b_{T}$ if
$a_{T}\le Cb_{T}$ for some constant $C>0$.

\section{Proofs of Main Theoretical Results\label{sec:Proofs-of-Main}\protect
}
\subsection{Proofs for IVX Estimator}
\begin{proof}[Proof of \prettyref{prop:rho_convergence}]
Note that $\hat{\rho}^{\mathrm{IVX}}-\rho^{*}=U_{n,T}/D_{n,T}$ where
(letting $\theta_{1}=(1+\theta)/2$)
\begin{align*}
U_{n,T} & :=\frac{1}{nT^{1+(\theta_{1}\wedge\gamma)}}(1-\rho^{*})\sum_{i=1}^{n}\sum_{t=1}^{T-1}z_{i,t}^{(1)}\alpha_{i}+\frac{1}{nT^{1+(\theta_{1}\wedge\gamma)}}\sum_{i=1}^{n}\sum_{t=1}^{T-1}(z_{i,t}^{(1)}v_{i,t+1}-\hat{\Delta}_{vv})\\
 & =:U_{1,n,T}+U_{2,n,T},\\
D_{n,T} & :=\frac{1}{nT^{1+(\theta_{1}\wedge\gamma)}}\sum_{i=1}^{n}\sum_{t=1}^{T-1}z_{i,t}^{(1)}\alpha_{i}+\frac{1}{nT^{1+(\theta_{1}\wedge\gamma)}}\sum_{i=1}^{n}\sum_{t=1}^{T-1}z_{i,t}^{(1)}\delta_{i,t}\\
 & =:D_{1,n,T}+D_{2,n,T}.
\end{align*}
The proof consists of two steps. \textbf{Step I:} show that $D_{n,T}$
converges in probability to a positive constant; \textbf{Step~II:}
show that $U_{n,T}=O_{p}\bigl(1/\sqrt{nT^{2\gamma}}+1/T^{1+\gamma}\bigr)$.
Then the stochastic order of $\hat{\rho}^{\mathrm{IVX}}-\rho^{*}$
is the same as that of $U_{n,T}$.

\textbf{Step I.} For $D_{1,n,T}$, first by \prettyref{lem:ivxepct}\ref{enu:z4}
we have
\begin{align*}
\mathbb{E}\left[\left(\sum_{t=1}^{T-1}z_{i,t}^{(1)}\alpha_{i}\right)^{2}\right] & \leq\left\{ \mathbb{E}(\alpha_{i}^{4})\mathbb{E}\left[\left(\sum_{t=1}^{T-1}z_{i,t}^{(1)}\right)^{4}\right]\right\} ^{1/2}\\
 & =O(T^{\gamma}\cdot T^{1+\theta_{1}+(\theta_{1}\wedge\gamma)})=O(T^{1+\theta_{1}+\gamma+(\theta_{1}\wedge\gamma)}).
\end{align*}
It follows by independence across $i$ that
\[
\mathrm{var}\left(\frac{1}{nT^{1+(\theta_{1}\wedge\gamma)}}\sum_{i=1}^{n}\sum_{t=1}^{T-1}z_{i,t}^{(1)}\alpha_{i}\right)=\frac{1}{nT^{2[1+(\theta_{1}\wedge\gamma)]}}\mathbb{E}\left[\left(\sum_{t=1}^{T-1}z_{i,t}^{(1)}\alpha_{i}\right)^{2}\right]=\frac{1}{nT^{1-(\theta_{1}\vee\gamma)}}.
\]
Note that $z_{i,t}$ is a linear transformation of $\{\Delta x_{i,t}\}=\{\delta_{i,t}\}$.
Since $\alpha_{i}$ is uncorrelated with $\delta_{i,t}$ for any $t$
(\prettyref{assump:initval}), it is also uncorrelated with $z_{i,t}^{(1)}$.
Then by Markov's inequality,
\begin{equation}
U_{1,n,T}=O\left(\frac{1}{nT^{1+(\theta_{1}\wedge\gamma)}}\sum_{i=1}^{n}\sum_{t=1}^{T-1}z_{i,t}^{(1)}\alpha_{i}\right)=O_{p}\left(\frac{1}{\sqrt{nT^{1-(\theta_{1}\vee\gamma)}}}\right).\label{eq:order_x_alpha}
\end{equation}

For $D_{2,n,T}$, using \prettyref{lem:ivxepct}\ref{enu:zx}, we
have
\[
\mathbb{E}\left[\left(\frac{1}{T^{1+(\theta_{1}\wedge\gamma)}}\sum_{t=1}^{T-1}z_{i,t}^{(1)}\delta_{i,t}\right)^{2}\right]=O(1).
\]
This shows the uniform integrability in $T$ of $T^{-[1+(\theta_{1}\wedge\gamma)]}\sum_{t=1}^{T-1}z_{i,t}^{(1)}\delta_{i,t}$.
By the same argument as in \prettyref{lem:ivxjoint}\ref{enu:Q_pto},
there is some $Q_{zx}^{(1)}>0$ such that
\begin{equation}
\frac{1}{nT^{1+(\theta_{1}\wedge\gamma)}}\sum_{i=1}^{n}\sum_{t=1}^{T-1}z_{i,t}^{(1)}\delta_{i,t}\to_{p}Q_{zx}^{(1)}\quad\text{as }(n,T)\to\infty.\label{eq:order_zx}
\end{equation}
Combining \eqref{eq:order_x_alpha} and \eqref{eq:order_zx} yields
\[
D_{n,T}\to Q_{zx}^{(1)}\quad\text{as }(n,T)\to\infty.
\]

\textbf{Step II.} By \eqref{eq:order_x_alpha} we have
\[
U_{1,n,T}=O_{p}\left(\frac{1}{\sqrt{nT^{1+2\gamma-(\theta_{1}\vee\gamma)}}}\right).
\]
By the Beveridge-Nelson decomposition $v_{i,t}=G(1)\varepsilon_{i,t}-\Delta\tilde{\varepsilon}_{i,t}$
where $G(1):=\sum_{\ell=0}^{\infty}g_{\ell}$, $\tilde{\varepsilon}_{i,t}:=\sum_{s=0}^{\infty}\tilde{g}_{s}\varepsilon_{i,t-s}\text{ and }\tilde{g}_{s}:=\sum_{\tau=s+1}^{\infty}g_{\tau}$,
we use summation by parts to deduce
\begin{align*}
\sum_{t=1}^{T-1}z_{i,t}^{(1)}v_{i,t+1} & =G(1)\sum_{t=1}^{T-1}z_{i,t}^{(1)}\varepsilon_{i,t+1}-\sum_{t=1}^{T-1}z_{i,t}^{(1)}\Delta\tilde{\varepsilon}_{i,t+1}\\
 & =G(1)\sum_{t=1}^{T-1}z_{i,t}^{(1)}\varepsilon_{i,t+1}+\sum_{t=1}^{T-1}\Delta z_{i,t}^{(1)}\tilde{\varepsilon}_{i,t+1}-z_{i,t}^{(1)}\tilde{\varepsilon}_{i,T}.
\end{align*}
For the first term, the stationary m.d.s.\ condition and \prettyref{lem:generic_AR}\ref{enu:E_sum_eta2}
lead to
\begin{equation}
\mathbb{E}\left[\left(\sum_{t=1}^{T-1}z_{i,t}^{(1)}\varepsilon_{i,t+1}\right)^{2}\right]=\mathbb{E}(\varepsilon_{i,1}^{2})\sum_{t=1}^{T-1}\mathbb{E}\left[(z_{i,t}^{(1)})^{2}\right]=O\bigl(T^{1+(\theta_{1}\wedge\gamma)}\bigr).\label{eq:sum_z_eps}
\end{equation}
The third term can be bounded as
\begin{equation}
\mathbb{E}\left[(z_{i,t}^{(1)}\tilde{\varepsilon}_{i,T})^{2}\right]\leq\sqrt{\mathbb{E}\left[(z_{i,t}^{(1)})^{4}\right]\mathbb{E}(\tilde{\varepsilon}_{i,T}^{4})}=O(T^{\theta_{1}\wedge\gamma}).\label{eq:z_epsT}
\end{equation}
For the second term, let $X_{i,t}:=\Delta z_{i,t}^{(1)}\tilde{\varepsilon}_{i,t+1}-\Delta_{vv}$,
then
\[
X_{i,t}=\underbrace{v_{i,t}\tilde{\varepsilon}_{i,t+1}-\Delta_{vv}}_{A_{i,t}}+\underbrace{\frac{c_{z}}{T^{\theta_{1}}}\zeta_{i,t-1}^{(1)}\tilde{\varepsilon}_{i,t+1}}_{B_{i,t}}+\underbrace{\frac{c^{*}c_{z}}{T^{\theta_{1}+\gamma}}\psi_{i,t-1}^{(1)}\tilde{\varepsilon}_{i,t+1}}_{C_{i,t}}+\underbrace{\frac{c^{*}}{T^{\gamma}}\delta_{i,t}\tilde{\varepsilon}_{i,t+1}}_{D_{i,t}}.
\]
Clearly, $A_{i,t}$ has zero mean and thus
\[
\mathbb{E}\left[\left(\sum_{t=1}^{T-1}A_{i,t}\right)^{2}\right]=\sum_{t,s}^{T-1}\mathbb{E}(A_{i,t}A_{i,s})\lesssim T\sum_{h\geq0}|\Gamma_{A}(h)|
\]
where
\[
\Gamma_{A}(h)=\mathrm{cov}(A_{i,0},A_{i,h})=\sum_{j,s}\sum_{j',s'}g_{j}\tilde{g}_{s}g_{j'}\tilde{g}_{s'}\mathrm{cov}(\varepsilon_{i,-j}\varepsilon_{i,1-s},\varepsilon_{i,h-j'}\varepsilon_{i,h+1-s'}).
\]
By the relationship between fourth cumulant and covariance and using
the absolutely summable cumulant condition we can deduce that $\sum_{h\geq0}|\Gamma_{A}(h)|=O(1)$,
which gives
\[
\mathbb{E}\left[\left(\sum_{t=1}^{T-1}A_{i,t}\right)^{2}\right]=O(T).
\]
For the other three terms, by the same argument as in the proof of
\prettyref{lem:generic_two_ARs}\ref{enu:Esum2}, we can get
\[
\mathbb{E}\left[\left(\sum_{t=1}^{T-1}B_{i,t}\right)^{2}\right]=O(T^{2-2\theta_{1}}),\quad\mathbb{E}\left[\left(\sum_{t=1}^{T-1}C_{i,t}\right)^{2}\right]=O(T^{2-2(\theta_{1}\wedge\gamma)}),\quad\mathbb{E}\left[\left(\sum_{t=1}^{T-1}D_{i,t}\right)^{2}\right]=O(T^{2-2\gamma}).
\]
It follows that
\begin{equation}
\mathbb{E}\left[\left(\sum_{t=1}^{T-1}X_{i,t}\right)^{2}\right]=\mathbb{E}\left[\left(\sum_{t=1}^{T-1}(\Delta z_{i,t}^{(1)}\tilde{\varepsilon}_{i,t+1}-\Delta_{vv})\right)^{2}\right]=O(T)+O(T^{2-2(\theta_{1}\wedge\gamma)}).\label{eq:sum_X}
\end{equation}
By \eqref{eq:sum_z_eps}, \eqref{eq:z_epsT} and \eqref{eq:sum_X},
we conclude
\[
\mathbb{E}\left[\left(\sum_{t=1}^{T-1}[z_{i,t}^{(1)}v_{i,t+1}-\Delta_{vv}]\right)^{2}\right]=O\bigl(T^{1+(\theta_{1}\wedge\gamma)}\bigr)+O\bigl(T^{2[1-(\theta_{1}\wedge\gamma)]}\bigr).
\]
Thus, by independence across $i$ we have
\begin{align}
\mathrm{var}\left(\frac{1}{nT^{1+(\theta_{1}\wedge\gamma)}}\sum_{i=1}^{n}\sum_{t=1}^{T-1}(z_{i,t}^{(1)}v_{i,t+1}-\Delta_{vv})\right) & \leq\frac{1}{nT^{2[1+(\theta_{1}\wedge\gamma)]}}\mathbb{E}\left[\left(\sum_{t=1}^{T-1}[z_{i,t}^{(1)}v_{i,t+1}-\Delta_{vv}]\right)^{2}\right]\nonumber \\
 & =O\left(\frac{1}{nT^{1+(\theta_{1}\wedge\gamma)}}\right)+O\left(\frac{1}{nT^{4(\theta_{1}\wedge\gamma)}}\right).\label{eq:var_zv_center}
\end{align}
Since by \prettyref{assump:innov}\ref{enu:LP}, $|\tilde{g}_{s}|\leq\tilde{C}q_{\nu}^{s}$
for some constant $\tilde{C}>0$ (where $q_{\nu}:=\exp(-C_{g})$),
then
\[
|\mathbb{E}(v_{i,t-m}\tilde{\varepsilon}_{i,t+1})|\lesssim\Biggl|\sum_{j\geq0}g_{j}\tilde{g}_{j+m+1}\Biggr|\lesssim\sum_{j\geq0}q_{\nu}^{j}\cdot q_{\nu}^{j+m+1}\lesssim q_{\nu}^{m+1}.
\]
Let $\rho_{z}^{(1)}:=1+c_{z}/T^{\theta_{1}}$. It then follows that
uniformly for all $t\leq T$,
\[
\left|\mathbb{E}(\zeta_{i,t-1}^{(1)}\tilde{\varepsilon}_{i,t+1})\right|\leq\sum_{k=0}^{t-1}|\rho_{z}^{(1)}|^{k}|\mathbb{E}(v_{i,t-1-k}\tilde{\varepsilon}_{i,t+1})|\lesssim\sum_{k\geq0}|\rho_{z}^{(1)}|^{k}q_{\nu}^{k+2}=O(1).
\]
Likewise, we can show that uniformly for all $t\leq T$,
\begin{align*}
\left|\mathbb{E}(\psi_{i,t-1}^{(1)}\tilde{\varepsilon}_{i,t+1})\right| & \lesssim\sum_{a,k\geq0}|\rho_{z}^{(1)}|^{a}|\rho^{*}|^{k}q_{\nu}^{a+k+2}=O(1)\quad\text{and}\\
\left|\mathbb{E}(\delta_{i,t}\tilde{\varepsilon}_{i,t+1})\right|\lesssim & \sum_{k\geq0}|\rho^{*}|^{k}q_{\nu}^{k+1}=O(1).
\end{align*}
The expectation of $\sum_{t}B_{i,t}$, $\sum_{t}C_{i,t}$ and $\sum_{t}D_{i,t}$
can be bounded as
\begin{align*}
\left|\sum_{t=1}^{T-1}\mathbb{E}(B_{i,t})\right|=\left|\frac{c_{z}}{T^{\theta_{1}}}\sum_{t=1}^{T-1}\mathbb{E}\left[\zeta_{i,t-1}^{(1)}\tilde{\varepsilon}_{i,t+1}\right]\right| & =O(T^{1-\theta_{1}}),\\
\left|\sum_{t=1}^{T-1}\mathbb{E}(C_{i,t})\right|=\left|\frac{c^{*}c_{z}}{T^{\theta_{1}+\gamma}}\sum_{t=1}^{T-1}\mathbb{E}\left[\psi_{i,t-1}^{(1)}\tilde{\varepsilon}_{i,t+1}\right]\right| & =O(T^{1-(\theta+\gamma)}),\\
\left|\sum_{t=1}^{T-1}\mathbb{E}(D_{i,t})\right|=\left|\frac{c^{*}}{T^{\gamma}}\sum_{t=1}^{T-1}\mathbb{E}\left(\delta_{i,t}\tilde{\varepsilon}_{i,t+1}\right)\right| & =O(T^{1-\gamma}).
\end{align*}
It follows that
\begin{equation}
\sum_{t=1}^{T-1}\mathbb{E}\bigl[\Delta z_{i,t}^{(1)}\tilde{\varepsilon}_{i,t+1}-\Delta_{vv}\bigr]=O(T^{1-\theta_{1}})+O(T^{1-\gamma})+O(T^{1-(\theta+\gamma)})=O(T^{1-(\theta_{1}\wedge\gamma)}).\label{eq:E_z_eps}
\end{equation}
By \eqref{eq:z_epsT} and \eqref{eq:E_z_eps} we have
\begin{equation}
\mathbb{E}\left[\frac{1}{nT^{1+(\theta_{1}\wedge\gamma)}}\sum_{i=1}^{n}\sum_{t=1}^{T-1}(z_{i,t}^{(1)}v_{i,t+1}-\Delta_{vv})\right]=O\left(\frac{1}{T^{2(\theta_{1}\wedge\gamma)}}\right).\label{eq:E_zv_center}
\end{equation}
It follows by \eqref{eq:var_zv_center} and \eqref{eq:E_zv_center}
that
\begin{equation}
\frac{1}{nT^{1+(\theta_{1}\wedge\gamma)}}\sum_{i=1}^{n}\sum_{t=1}^{T-1}(z_{i,t}^{(1)}v_{i,t+1}-\Delta_{vv})=O_{p}\left(\frac{1}{\sqrt{nT^{1+(\theta_{1}\wedge\gamma)}}}+\frac{1}{T^{2(\theta_{1}\wedge\gamma)}}\right).\label{eq:zv_center_total}
\end{equation}
The next step is to show that
\[
\hat{\Delta}_{vv}-\Delta_{vv}=O_{p}\left(\frac{G}{\sqrt{nT}}+\frac{G}{T}\right).
\]
Using the fact that $\hat{v}_{i,t}=v_{i,t}-\delta_{\rho}^{\mathrm{WG}}\tilde{x}_{i,t-1}$
where $\delta_{\rho}^{\mathrm{WG}}:=\hat{\rho}^{\mathrm{WG}}-\rho^{*}$,
we can write
\[
\hat{v}_{i,t}\hat{v}_{i,t-h}=v_{i,t}v_{i,t-h}-\delta_{\rho}^{\mathrm{WG}}(\tilde{x}_{i,t-1}v_{i,t-h}+\tilde{x}_{i,t-h}v_{i,t})+(\delta_{\rho}^{\mathrm{WG}})^{2}\tilde{x}_{i,t-1}\tilde{x}_{i,t-h-1}.
\]
We thus have (since the linear process coefficient is exponentially
decaying, the truncation error is exponential $O(q^{G})$ so that
the bandwidth $G$ can be log rate)
\begin{align*}
\hat{\Delta}_{vv}-\Delta_{vv} & =\sum_{h=1}^{G}[\hat{\Gamma}_{vv}(h)-\Gamma_{vv}^{*}(h)]-\sum_{h=G+1}^{\infty}\Gamma_{vv}^{*}(h)\\
 & =\sum_{h=1}^{G}\frac{1}{nT}\sum_{i=1}^{n}\sum_{t=h+1}^{T}(\hat{v}_{i,t}\hat{v}_{i,t-h}-\Gamma_{vv}^{*}(h))+O(q^{G})\\
 & =\frac{1}{nT}\sum_{i=1}^{n}\sum_{t=2}^{T}\sum_{h=1}^{G\wedge(t-1)}[v_{i,t}v_{i,t-h}-\Gamma_{vv}^{*}(h)]\\
 & \qquad-\frac{1}{nT}\sum_{i=1}^{n}\sum_{t=2}^{T}\sum_{h=1}^{G\wedge(t-1)}\delta_{\rho}^{\mathrm{WG}}(\tilde{x}_{i,t-1}v_{i,t-h}+\tilde{x}_{i,t-h}v_{i,t})\\
 & \qquad+\frac{1}{nT}\sum_{i=1}^{n}\sum_{t=2}^{T}\sum_{h=1}^{G\wedge(t-1)}(\delta_{\rho}^{\mathrm{WG}})^{2}\tilde{x}_{i,t-1}\tilde{x}_{i,t-h-1}+O(q^{G}).
\end{align*}
The first term can be shown, by standard argument, to be of $O_{p}(G/\sqrt{nT})$.
For the second term, by the same argument as in the proof of \prettyref{lem:generic_two_ARs}\ref{enu:Esum2}
we can get
\[
\sup_{h\leq T}\mathbb{E}\left[\left(\sum_{t=1}^{T}x_{i,t}v_{i,t-h}\right)^{2}\right]=O(T^{2}),
\]
which, together with $\delta_{\rho}^{\mathrm{WG}}=O_{p}\bigl((nT^{1+\gamma})^{1/2}+T^{-1}\bigr)$,
leads to
\[
\frac{1}{nT}\sum_{i=1}^{n}\sum_{t=2}^{T}\sum_{h=1}^{G\wedge(t-1)}\delta_{\rho}^{\mathrm{WG}}(\tilde{x}_{i,t-1}v_{i,t-h}+\tilde{x}_{i,t-h}v_{i,t})=O_{p}(G\delta_{\rho}^{\mathrm{WG}})=O_{p}\left(\frac{G}{\sqrt{nT^{1+\gamma}}}+\frac{G}{T}\right).
\]
Likewise we have
\[
\frac{1}{nT}\sum_{i=1}^{n}\sum_{t=2}^{T}\sum_{h=1}^{G\wedge(t-1)}(\delta_{\rho}^{\mathrm{WG}})^{2}\tilde{x}_{i,t-1}\tilde{x}_{i,t-h-1}=O_{p}\bigl(GT^{\gamma}(\delta_{\rho}^{\mathrm{WG}})^{2}\bigr)=O_{p}\left(\frac{G}{nT}+\frac{G}{T^{2-\gamma}}\right).
\]
It then follows that
\begin{equation}
\hat{\Delta}_{vv}-\Delta_{vv}=O_{p}\left(\frac{G}{\sqrt{nT}}+\frac{G}{T}\right).\label{eq:delta_vv_rate}
\end{equation}

Hence, by \eqref{eq:zv_center_total} and \eqref{eq:delta_vv_rate}
we obtain
\begin{align*}
U_{2,n,T} & =\frac{1}{nT^{1+(\theta_{1}\wedge\gamma)}}\sum_{i=1}^{n}\sum_{t=1}^{T-1}(z_{i,t}^{(1)}v_{i,t+1}-\Delta_{vv})+\frac{T-1}{T^{1+(\theta_{1}\wedge\gamma)}}(\hat{\Delta}_{vv}-\Delta_{vv})\\
 & =O_{p}\left(\frac{1}{\sqrt{nT^{1+(\theta_{1}\wedge\gamma)}}}+\frac{1}{T^{2(\theta_{1}\wedge\gamma)}}\right)+O_{p}\left(\frac{G}{\sqrt{nT^{1+2(\theta_{1}\wedge\gamma)}}}+\frac{G}{T^{1+(\theta_{1}\wedge\gamma)}}\right).
\end{align*}
We then conclude
\[
U_{n,T}=O_{p}\left(\frac{1}{\sqrt{nT^{1+(\theta_{1}\wedge\gamma)}}}+\frac{1}{T^{2(\theta_{1}\wedge\gamma)}}+\frac{G}{\sqrt{nT^{1+2(\theta_{1}\wedge\gamma)}}}+\frac{G}{T^{1+(\theta_{1}\wedge\gamma)}}\right).
\]
The orders of $U_{n,T}$ and $D_{n,T}$ imply that
\[
\hat{\rho}^{{\rm IVX}}-\rho^{*}=O_{p}\left(\frac{1}{\sqrt{nT^{1+(\theta_{1}\wedge\gamma)}}}+\frac{1}{T^{2(\theta_{1}\wedge\gamma)}}+\frac{G}{\sqrt{nT^{1+2(\theta_{1}\wedge\gamma)}}}+\frac{G}{T^{1+(\theta_{1}\wedge\gamma)}}\right).
\]
\end{proof}
\begin{rem}
\label{rem:verify_IVX_Cond}This remark verifies condition (\ref{eq:cond IVX bias correct})
in the leading asymptotic case $n/T\to c\in(0,\infty)$ when $G\lfloor T^{1/4}\rfloor$
and $\theta_{1}=(1+\theta)/2\in(\theta,1)$. Then
\[
\hat{\rho}^{{\rm IVX}}-\rho^{*}=O_{p}\left(a_{T}\right),\ a_{T}:=\frac{1}{T^{1+(\theta_{1}\wedge\gamma)/2}}+\frac{1}{T^{2(\theta_{1}\wedge\gamma)}}+\frac{1}{T^{3/4+(\theta_{1}\wedge\gamma)}}.
\]
Also, condition (\ref{eq:cond IVX bias correct}) becomes
\[
\hat{\rho}-\rho^{*}=o_{p}\left(b_{T}\right),\ b_{T}:=\frac{1}{\sqrt{T^{\theta+3\gamma+(\theta\vee\gamma)-2}}}.
\]
It thus suffices to show that that $a_{T}/b_{T}\to0$, which will
be proved in the following.
\end{rem}
\begin{proof}[Proof of $a_{T}/b_{T}\to0$ in Remark \ref{rem:verify_IVX_Cond}]
Note that
\[
\dfrac{a_{T}}{b_{T}}=\sqrt{\frac{1}{T^{4+(\theta_{1}\wedge\gamma)-\theta-3\gamma-(\theta\vee\gamma)}}}+\sqrt{\dfrac{1}{T^{2+4(\theta_{1}\wedge\gamma)-\theta-3\gamma-(\theta\vee\gamma)}}}+\sqrt{\dfrac{1}{T^{7/2+2(\theta_{1}\wedge\gamma)-\theta-3\gamma-(\theta\vee\gamma)}}}.
\]

\textbf{Case I}. $\gamma\leq\theta_{1}$. Then $(\theta_{1}\wedge\gamma)=\gamma$
and $(\theta\vee\gamma)\leq\theta_{1}$. Therefore,
\begin{align*}
\dfrac{a_{T}}{b_{T}} & \leq\sqrt{\frac{1}{T^{4+\gamma-\theta-3\gamma-\theta_{1}}}}+\sqrt{\dfrac{1}{T^{2+4\gamma-\theta-3\gamma-\theta_{1}}}}+\sqrt{\dfrac{1}{T^{7/2+2\gamma-\theta-3\gamma-\theta_{1}}}}.\\
 & =\sqrt{\frac{1}{T^{4-\theta-2\gamma-\theta_{1}}}}+\sqrt{\dfrac{1}{T^{2+\gamma-\theta-\theta_{1}}}}+\sqrt{\dfrac{1}{T^{7/2-\theta-\gamma-\theta_{1}}}}\to0,
\end{align*}
given that $\gamma\leq\theta_{1}<1$ and $\theta<\theta_{1}<1$.

\textbf{Case II}. $\gamma>\theta_{1}$. Then $(\theta_{1}\wedge\gamma)=\theta_{1}=(1+\theta)/2$
and $(\theta\vee\gamma)=\gamma$. Therefore,
\begin{align*}
\dfrac{a_{T}}{b_{T}} & \leq\sqrt{\frac{1}{T^{4+\theta_{1}-\theta-3\gamma-\gamma}}}+\sqrt{\dfrac{1}{T^{2+4(\theta_{1}\wedge\gamma)-\theta-3\gamma-\gamma}}}+\sqrt{\dfrac{1}{T^{7/2+2\theta_{1}-\theta-3\gamma-\gamma}}}\\
 & =\sqrt{\frac{1}{T^{4(1-\gamma)+(1+\theta)/2-\theta}}}+\sqrt{\dfrac{1}{T^{2+2(1+\theta)-\theta-3\gamma-\gamma}}}+\sqrt{\dfrac{1}{T^{7/2+1+\theta-\theta-3\gamma-\gamma}}}\\
 & =\sqrt{\frac{1}{T^{4(1-\gamma)+(1-\theta)/2}}}+\sqrt{\dfrac{1}{T^{4(1-\gamma)+\theta}}}+\sqrt{\dfrac{1}{T^{9/2-4\gamma}}}\to0,
\end{align*}
given that $\gamma\leq1$ and $\theta<1$. We complete the proof of
$a_{T}/b_{T}\to0$ in Remark \ref{rem:verify_IVX_Cond}.
\end{proof}
\medskip{}

\begin{delayedproof}{prop:bhativx}
We decompose
\[
\sqrt{nT^{1+(\theta\wedge\gamma)}}\bigl[\hat{\beta}^{\mathrm{IVX}}-\beta^{*}+b_{n,T}^{\mathrm{IVX}}(\{\omega_{ev,h}^{*}\},\rho^{*},\rho_{z})\bigr]=\sqrt{T^{\theta\wedge\gamma}[1-(\rho^{*}\rho_{z})^{2}]}\frac{n^{-1/2}\sum_{i=1}^{n}L_{i,T}}{n^{-1}\sum_{i=1}^{n}(Q_{i,T}-R_{i,T})},
\]
where $Q_{i,T}$, $R_{i,T}$, and $L_{i,T}$ are defined by \eqref{eq:def Q IVX},
\eqref{eq:def R IVX}, and \eqref{eq:def L IVX}, respectively. It
follows by \prettyref{lem:ivxjoint} that, as $(n,T)\to\infty$,
\[
\frac{n^{-1/2}\sum_{i=1}^{n}L_{i,T}}{n^{-1}\sum_{i=1}^{n}(Q_{i,T}-R_{i,T})}\to_{d}\mathcal{N}\left(0,\frac{{\rm S}_{xe}}{\left[\mathbb{E}(Q_{zx}-R_{zx})\right]^{2}}\right),
\]
where $Q_{zx}$, $R_{zx}$, and ${\rm S}_{xe}$ are defined in \prettyref{lem:ivx_marginal}.
In addition,
\begin{equation}
T^{\theta\wedge\gamma}[1-(\rho^{*}\rho_{z})^{2}]\to c^{{\rm IVX}}:=\begin{cases}
-2c_{z}, & \theta<\gamma\leq1,\\
-2(c_{z}+c^{*}), & \gamma=\theta,\\
-2c^{*}, & 0<\gamma<\theta,\\
1-\rho^{*2}, & \gamma=0.
\end{cases}\label{eq:Civx}
\end{equation}
As a result,
\begin{equation}
\sqrt{nT^{1+(\theta\wedge\gamma)}}\bigl[\hat{\beta}^{\mathrm{IVX}}-\beta^{*}+b_{n,T}^{\mathrm{IVX}}(\{\omega_{ev,h}^{*}\},\rho^{*},\rho_{z})\bigr]\to_{d}\mathcal{N}\left(0,\Sigma^{{\rm IVX}}\right),\label{eq:bhat ivx asym dist}
\end{equation}
where
\begin{equation}
\Sigma^{{\rm IVX}}:=\dfrac{c^{{\rm IVX}}{\rm S}_{xe}}{\left[\mathbb{E}(Q_{zx}-R_{zx})\right]^{2}}.\label{eq:SigmaIVX}
\end{equation}

By \eqref{eq:1 over zeta x} and \eqref{eq:E sum zeta sum e}, the
order of $b_{n,T}^{\mathrm{IVX}}(\{\omega_{ev,h}^{*}\},\rho^{*},\rho_{z})$
is given by
\begin{align*}
b_{n,T}^{\mathrm{IVX}}(\{\omega_{ev,h}^{*}\},\rho^{*},\rho_{z}) & =\dfrac{\sum_{h=0}^{T-2}\Psi_{h,T}(\rho^{*},\rho_{z})\omega_{ev,h}^{*}}{n^{-1}T\sum_{i=1}^{n}\sum_{i=1}^{T}\tilde{z}_{i,t}x_{i,t}}\\
 & =O_{p}(T^{\theta+\gamma})\cdot O_{p}\left(\dfrac{1}{T^{2+(\theta\wedge\gamma)}}\right)=O_{p}\left(\dfrac{1}{T^{2-(\theta\vee\gamma)}}\right).
\end{align*}
We have completed the proof of \prettyref{prop:bhativx}.
\end{delayedproof}
\medskip{}

\begin{delayedproof}{thm:DIVX}
We first show $nT^{1+(\theta\wedge\gamma)}\bigl(\hat{\varsigma}^{\mathrm{IVX}}\bigr)^{2}\to_{p}\Sigma^{{\rm IVX}}$.
Define
\[
S_{i,T}^{e}=\frac{1-(\rho^{*}\rho_{z})^{2}}{T}\sum_{t=1}^{T}z_{i,t}^{2}e_{i,t+1}^{2}.
\]
By \eqref{eq:E_Z2_converge} it holds that
\begin{equation}
\frac{1}{n}\sum_{i=1}^{n}S_{i,T}^{e}\to_{p}{\rm S}_{xe}.\label{eq:S_iT_e_plim}
\end{equation}

Let $\Delta_{\beta}=\hat{\beta}^{\mathrm{IVX}}-\beta^{*}$. Note that
$\hat{e}_{i,t+1}=\tilde{y}_{i,t+1}-\hat{\beta}^{\mathrm{IVX}}\tilde{x}_{i,t}=(e_{i,t+1}-\bar{e}_{i})-\Delta_{\beta}\tilde{x}_{i,t}$,
then
\[
\hat{e}_{i,t+1}^{2}-e_{i,t+1}^{2}=-2e_{i,t+1}\bar{e}_{i}+\bar{e}_{i}^{2}-2\Delta_{\beta}\tilde{x}_{i,t}(e_{i,t+1}-\bar{e}_{i})+\Delta_{\beta}^{2}\tilde{x}_{i,t}^{2}.
\]
We now show
\begin{align}
 & \frac{1}{n}\frac{1-(\rho^{*}\rho_{z})^{2}}{T}\sum_{i=1}^{n}\sum_{t=1}^{T}z_{i,t}^{2}(\hat{e}_{i,t+1}^{2}-e_{i,t+1}^{2})\nonumber \\
={} & \frac{1}{n}\frac{1-(\rho^{*}\rho_{z})^{2}}{T}\sum_{i=1}^{n}\sum_{t=1}^{T}z_{i,t}^{2}\Bigl[-2e_{i,t+1}\bar{e}_{i}+\bar{e}_{i}^{2}-2\Delta_{\beta}\tilde{x}_{i,t}(e_{i,t+1}-\bar{e}_{i})+\Delta_{\beta}^{2}\tilde{x}_{i,t}^{2}\Bigr]=o_{p}(1).\label{eq:hat_e_replace_e}
\end{align}
By the Cauchy-Schwarz inequality,
\begin{align}
\ensuremath{} & \left|\frac{1}{n}\frac{1-(\rho^{*}\rho_{z})^{2}}{T}\sum_{i}\sum_{t}z_{i,t}^{2}e_{i,t+1}\bar{e}_{i}\right|\nonumber \\
\leq{} & \left(\frac{1}{n}\frac{1-(\rho^{*}\rho_{z})^{2}}{T}\sum_{i}\sum_{t}z_{i,t}^{2}e_{i,t+1}^{2}\right)^{1/2}\left(\frac{1}{n}\frac{1-(\rho^{*}\rho_{z})^{2}}{T}\sum_{i}\sum_{t}z_{i,t}^{2}\bar{e}_{i}^{2}\right)^{1/2}=o_{p}(1),\label{eq:z2eebar}
\end{align}
where the first factor is $O_{p}(1)$ due to \eqref{eq:S_iT_e_plim}
and the second factor is $o_{p}(1)$ since by the proof of \prettyref{lem:ivxjoint}\ref{enu:S_pto_joint}
we can deduce
\[
\mathbb{E}\left|\frac{1}{n}\frac{1-(\rho^{*}\rho_{z})^{2}}{T}\sum_{i}\sum_{t}z_{i,t}^{2}\bar{e}_{i}^{2}\right|\le\frac{1}{n}\sum_{i=1}^{n}\mathbb{E}|S_{i,T}\bar{e}_{i}^{2}|\leq\{\mathbb{E}(S_{i,T}^{2})\}^{1/2}\{\mathbb{E}(\bar{e}_{i}^{4})\}^{1/2}=O\left(\frac{1}{T}\right).
\]
Using the same argument we can show
\begin{equation}
\frac{1}{n}\frac{1-(\rho^{*}\rho_{z})^{2}}{T}\sum_{i}\sum_{t}z_{i,t}^{2}\bar{e}_{i}^{2}=o_{p}(1).\label{eq:z2_bar_e2}
\end{equation}
For the term containing $\Delta_{\beta}$, by the Cauchy-Schwarz we
have
\begin{align}
 & \left|\Delta_{\beta}\frac{1}{n}\frac{1-(\rho^{*}\rho_{z})^{2}}{T}\sum_{i}\sum_{t}z_{i,t}^{2}\tilde{x}_{i,t}(e_{i,t+1}-\bar{e}_{i})\right|\nonumber \\
\leq{} & |\Delta_{\beta}|\left(\frac{1}{n}\frac{1-(\rho^{*}\rho_{z})^{2}}{T}\sum_{i}\sum_{t}z_{i,t}^{2}\tilde{x}_{i,t}^{2}\right)^{1/2}\left(\frac{1}{n}\frac{1-(\rho^{*}\rho_{z})^{2}}{T}\sum_{i}\sum_{t}z_{i,t}^{2}(e_{i,t+1}-\bar{e}_{i})^{2}\right)^{1/2}=o_{p}(1),\label{eq:delta_long}
\end{align}
where the second factor is $O_{p}(1)$ by \eqref{eq:S_iT_e_plim}
and the first factor is $o_{p}(1)$ by \prettyref{prop:bhativx} and
\prettyref{lem:ivxepct}\ref{enu:zx}:
\[
|\Delta_{\beta}|\left(\frac{1}{n}\frac{1-(\rho^{*}\rho_{z})^{2}}{T}\sum_{i}\sum_{t}z_{i,t}^{2}\tilde{x}_{i,t}^{2}\right)^{1/2}=O_{p}\left(\frac{1}{\sqrt{nT^{1+(\theta\wedge\gamma)}}}+\frac{1}{T^{2-(\theta\vee\gamma)}}\right)\cdot O_{p}(\sqrt{T^{\theta\vee\gamma}})=o_{p}(1).
\]
Similarly,
\begin{equation}
\Delta_{\beta}^{2}\cdot\frac{1}{n}\frac{1-(\rho^{*}\rho_{z})^{2}}{T}\sum_{i}\sum_{t}z_{i,t}^{2}\tilde{x}_{i,t}^{2}=o_{p}(1).\label{eq:delta_short}
\end{equation}
By \eqref{eq:z2eebar}\eqref{eq:z2_bar_e2}\eqref{eq:delta_long}\eqref{eq:delta_short},
we conclude that \eqref{eq:hat_e_replace_e} holds.

Let
\[
A_{n,T}=\frac{1-(\rho^{*}\rho_{z})^{2}}{n}\sum_{i=1}^{n}\bar{z}_{i}^{2}.
\]
By \prettyref{lem:ivxepct}\ref{enu:z4},
\begin{align*}
\mathbb{E}|A_{n,T}| & =|1-(\rho^{*}\rho_{z})^{2}|\cdot\mathbb{E}\left(\bar{z}_{i}^{2}\right)\leq|1-(\rho^{*}\rho_{z})^{2}|\cdot\frac{1}{T^{2}}\left\{ \mathbb{E}\left[\left(\sum_{t=1}^{T}z_{i,t}\right)^{4}\right]\right\} ^{1/2}\\
 & =O\bigl(T^{-(\theta\wedge\gamma)}\bigr)\cdot\frac{1}{T^{2}}\cdot O\bigl(T^{1+\theta+(\theta\wedge\gamma)}\bigr)=O\left(\frac{1}{T^{1-\theta}}\right)\to0,
\end{align*}
which implies $A_{n,T}\to_{p}0.$ By \prettyref{lem:Omega_hat}, $\hat{\omega}_{ee}\to_{p}\omega_{ee}^{*}$.
Together with the fact that $\hat{\lambda}\in[0,1]$, it follows that
\begin{equation}
\hat{\lambda}\hat{\omega}_{ee}A_{n,T}=o_{p}(1).\label{eq:lambda_w_A}
\end{equation}
As a result of \eqref{eq:S_iT_e_plim}, \eqref{eq:hat_e_replace_e}
and \eqref{eq:lambda_w_A}, the numerator of $\hat{\varsigma}^{\mathrm{IVX}}$
satisfy
\begin{align*}
 & \frac{1-(\rho^{*}\rho_{z})^{2}}{nT}\sum_{i=1}^{n}\left[\sum_{t=1}^{T}z_{i,t}^{2}\hat{e}_{i,t+1}^{{\rm 2}}-T\hat{\lambda}\bar{z}_{i}^{2}\hat{\omega}_{ee}\right]\\
={} & \frac{1}{n}\sum_{i=1}^{n}S_{i,T}^{e}-\hat{\lambda}\hat{\omega}_{ee}A_{n,T}+o_{p}(1)\to_{p}{\rm S}_{xe}.
\end{align*}
Thus, as $(n,T)\to\infty$,
\begin{align}
 & nT^{1+(\theta\wedge\gamma)}\bigl(\hat{\varsigma}^{\mathrm{IVX}}\bigr)^{2}\nonumber \\
= & T^{\theta\wedge\gamma}[1-(\rho^{*}\rho_{z})^{2}]\cdot\frac{(nT)^{-1}[1-(\rho^{*}\rho_{z})^{2}]\sum_{i=1}^{n}\bigl[\sum_{t=1}^{T}z_{i,t}^{2}\hat{e}_{i,t+1}^{{\rm 2}}-T\hat{\lambda}\bar{z}_{i}^{2}\hat{\omega}_{ee}\bigr]}{\bigl((nT)^{-1}[1-(\rho^{*}\rho_{z})^{2}]\sum_{i=1}^{n}\sum_{t=1}^{T}\tilde{z}_{i,t}x_{i,t}\bigr)^{2}}\nonumber \\
\to_{p} & c^{{\rm IVX}}\cdot\dfrac{{\rm S}_{xe}}{\left[\mathbb{E}(Q_{zx}-R_{zx})\right]^{2}}=\Sigma^{{\rm IVX}}.\label{eq:rate sigmaIVX}
\end{align}
By \prettyref{lem:ivxerror} and \eqref{eq:rate sigmaIVX} we have
\begin{align}
r_{n,T}^{\mathrm{IVX}}(\hat{\rho}) & :=\dfrac{\hat{b}_{n,T}^{\mathrm{IVX}}(\{\hat{\omega}_{ev,h}\},\hat{\rho},\rho_{z})-b_{n,T}^{\mathrm{IVX}}(\{\omega_{ev,h}^{*}\},\rho^{*},\rho_{z})}{\hat{\varsigma}^{\mathrm{IVX}}}\nonumber \\
 & =O_{p}\Bigl(\sqrt{nT^{1+(\theta\wedge\gamma)}}\Bigr)\cdot O_{p}\left(\dfrac{G}{\sqrt{nT^{5-2(\theta\vee\gamma)}}}+\frac{G}{T^{3-(\theta\vee\gamma)}}+\frac{|\hat{\rho}-\rho^{*}|}{T^{2-(\theta\vee\gamma)-\gamma}}+\frac{q_{\nu}^{G}}{T^{2-(\theta\vee\gamma)}}\right)\nonumber \\
 & =O_{p}\left(\dfrac{G}{\sqrt{T^{4-2(\theta\vee\gamma)-(\theta\wedge\gamma)}}}+\dfrac{\sqrt{n}G}{\sqrt{T^{5-2(\theta\vee\gamma)-(\theta\wedge\gamma)}}}+\dfrac{\sqrt{n}|\hat{\rho}-\rho^{*}|}{\sqrt{T^{3-2(\theta\vee\gamma)-(\theta\wedge\gamma)-2\gamma}}}\right),\label{eq:r IVX rate}
\end{align}
where the last line uses the fact that the exponential rate $q_{\nu}^{G}$
grows faster than any polynomial rate of $n$ and $T$. By \eqref{eq:r IVX rate}
and \prettyref{prop:rho_convergence}, when using $\hat{\rho}^{\text{IVX}}$
for $\rho$ we have
\begin{align*}
r_{n,T}^{\mathrm{IVX}}(\hat{\rho}^{\text{IVX}}) & =O_{p}\left(\dfrac{G}{\sqrt{T^{4-2(\theta\vee\gamma)-(\theta\wedge\gamma)}}}+\dfrac{\sqrt{n}G}{\sqrt{T^{5-2(\theta\vee\gamma)-(\theta\wedge\gamma)}}}\right)\\
 & \quad\,+O_{p}\left(\dfrac{1}{\sqrt{T^{4+(\theta_{1}\wedge\gamma)-(\theta\vee\gamma)-\theta-3\gamma}}}+\dfrac{\sqrt{n}}{\sqrt{T^{3+4(\theta_{1}\wedge\gamma)-2(\theta\vee\gamma)-(\theta\wedge\gamma)-2\gamma}}}\right)\\
 & \quad\,+O_{p}\left(\dfrac{G}{\sqrt{T^{4+2(\theta_{1}\wedge\gamma)-(\theta\vee\gamma)-\theta-3\gamma}}}+\dfrac{\sqrt{n}G}{\sqrt{T^{5+2(\theta_{1}\wedge\gamma)-2(\theta\vee\gamma)-(\theta\wedge\gamma)-2\gamma}}}\right).
\end{align*}
The following shows that $r_{n,T}^{\mathrm{IVX}}(\hat{\rho}^{\text{IVX}})=o_{p}(1)$
under the conditions of Theorem \ref{thm:DIVX}.

\textbf{Case I}. $\gamma\leq\theta_{1}$. Then $(\theta_{1}\wedge\gamma)=\gamma$,
$(\theta\wedge\gamma)\leq\gamma$, and $(\theta\vee\gamma)\leq\theta_{1}=(1+\theta)/2$.
Therefore, when $G=O(T^{1/4})$ and $n=o(T^{2-\theta}),$
\begin{align*}
r_{n,T}^{\mathrm{IVX}}(\hat{\rho}^{\text{IVX}}) & =O_{p}\left(\dfrac{G}{\sqrt{T^{4-2\gamma-\theta_{1}}}}+\dfrac{\sqrt{n}G}{\sqrt{T^{5-1-\theta-\gamma}}}+\dfrac{1}{\sqrt{T^{4+\gamma-\theta_{1}-\theta-3\gamma}}}+\dfrac{\sqrt{n}}{\sqrt{T^{3+4\gamma-1-\theta-\gamma-2\gamma}}}\right)\\
 & \quad\,+O_{p}\left(\dfrac{G}{\sqrt{T^{4+2\gamma-\theta_{1}-\theta-3\gamma}}}+\dfrac{\sqrt{n}G}{\sqrt{T^{5+2\gamma-1-\theta-\gamma-2\gamma}}}\right)\\
 & =O_{p}\left(\dfrac{1}{\sqrt{T^{7/2-2\gamma-\theta_{1}}}}+\dfrac{\sqrt{n}}{\sqrt{T^{7/2-\theta-\gamma}}}+\dfrac{1}{\sqrt{T^{4-2\gamma-\theta_{1}-\theta}}}+\dfrac{\sqrt{n}}{\sqrt{T^{2-\theta}}}\right)\\
 & \quad\,+O_{p}\left(\dfrac{1}{\sqrt{T^{7/2-\gamma-\theta_{1}-\theta}}}+\dfrac{\sqrt{n}}{\sqrt{T^{7/2-\theta-\gamma}}}\right)=o_{p}(1),
\end{align*}
given that $\gamma\leq\theta_{1}<1$ and $\theta<\theta_{1}<1$.

\textbf{Case II}. $\gamma>\theta_{1}$. Then $(\theta_{1}\wedge\gamma)=\theta_{1}=(1+\theta)/2$,
$(\theta\wedge\gamma)=\theta$, $(\theta\vee\gamma)=\gamma$. Therefore,
when $G=O(T^{1/4})$ and $n=o(T^{3/2}),$
\begin{align}
r_{n,T}^{\mathrm{IVX}}(\hat{\rho}^{\text{IVX}}) & =O_{p}\left(\dfrac{G}{\sqrt{T^{4-2\gamma-\theta}}}+\dfrac{\sqrt{n}G}{\sqrt{T^{5-2\gamma-\theta}}}+\dfrac{1}{\sqrt{T^{4+(1+\theta)/2-\gamma-\theta-3\gamma}}}+\dfrac{\sqrt{n}}{\sqrt{T^{3+2(1+\theta)-2\gamma-\theta-2\gamma}}}\right)\nonumber \\
 & \quad\,+O_{p}\left(\dfrac{G}{\sqrt{T^{4+1+\theta-\gamma-\theta-3\gamma}}}+\dfrac{\sqrt{n}G}{\sqrt{T^{5+(1+\theta)-2\gamma-\theta-2\gamma}}}\right)\nonumber \\
 & =O_{p}\left(\dfrac{1}{\sqrt{T^{7/2-2\gamma-\theta}}}+\dfrac{\sqrt{n}}{\sqrt{T^{3/2+(1-\theta)+2(1-\gamma)}}}+\dfrac{1}{\sqrt{T^{11/2-4\gamma-\theta/2}}}+\dfrac{\sqrt{n}}{\sqrt{T^{1+\theta+4(1-\gamma)}}}\right)\nonumber \\
 & \quad\,+O_{p}\left(\dfrac{1}{\sqrt{T^{9/2-4\gamma}}}+\dfrac{\sqrt{n}}{\sqrt{T^{3/2+4(1-\gamma)}}}\right)=o_{p}(1),\label{eq:rIVX Case 2}
\end{align}
given that $\gamma\leq1$ and $\theta<1$.

By \prettyref{prop:bhativx} and \eqref{eq:rate sigmaIVX} we have
\begin{equation}
\frac{\hat{\beta}^{\mathrm{IVX}}-\beta^{*}+b_{n,T}^{\mathrm{IVX}}(\{\omega_{ev,h}^{*}\},\rho^{*},\rho_{z})}{\hat{\varsigma}^{\mathrm{IVX}}}\to_{d}\mathcal{N}(0,1).\label{eq:oracle-t}
\end{equation}
Therefore
\begin{align*}
t^{{\rm DIVX}} & =\frac{\hat{\beta}^{\mathrm{DIVX}}-\beta^{*}}{\hat{\varsigma}^{\mathrm{IVX}}}\\
 & =\frac{\hat{\beta}^{\mathrm{IVX}}-\beta^{*}+b_{n,T}^{\mathrm{IVX}}(\{\omega_{ev,h}^{*}\},\rho^{*},\rho_{z})}{\hat{\varsigma}^{\mathrm{IVX}}}+r_{n,T}^{\mathrm{IVX}}(\hat{\rho}^{\text{IVX}})\to_{d}\mathcal{N}(0,1).
\end{align*}
We complete the proof of Theorem \ref{thm:DIVX}.
\end{delayedproof}
\medskip{}

\begin{proof}[Proof of \prettyref{cor:DIVX_uniform}]
 The polynomial rate $\rho_{T}^{*}=1+c^{*}/T^{\gamma}$ is devised
for simplicity of exposition and categorization. As pointed out in
the third paragraph of \citet[p.~19]{phillips2009econometric}, the
asymptotic results for IVX still hold with a general (convergent)
$\rho_{T}^{*}\to\rho\in(-1,1]$; see also \citet{magdalinosUniformDistributionfreeInference2024}
for a uniform inference using time series IVX.

We proceed with this proof of uniform asymptotic normality by verifying
\citet[p.~504]{andrews2020generic}'s Assumption~B\textsuperscript{*}.
Notice that all earlier proofs in this section depend on the following
two limits.
\end{proof}
\begin{enumerate}[label=(\arabic*)]
\item  $c_{0}^{*}=\lim_{T\to\infty}T(1-\rho_{T}^{*})$ determines whether
$x_{i,t}$ is LUR or MI. Specifically, $|c_{0}^{*}|<\infty$ corresponds
to $\gamma=1$, and $|c_{0}^{*}|=\infty$ corresponds to $\gamma<1$.
\begin{enumerate}
\item $\vartheta_{0}^{*}=\lim_{T\to\infty}T^{\theta}(1-\rho_{T}^{*})$ determines
the relative persistence of $x_{i,t}$ and the IV. Specifically, $\vartheta_{0}^{*}=0$,
$\vartheta_{0}^{*}\in(0,\infty)$, and $\vartheta_{0}^{*}=\infty$
are associated with $\theta<\gamma$, $\theta=\gamma$, and $\theta>\gamma$,
respectively.
\end{enumerate}
\end{enumerate}
\begin{proof}
The above two indices govern the family of distributions to which
the test statistic belong. Following \citet{andrews2020generic}'s
notations in Assumption~B\textsuperscript{*}, we set $\lambda_{T}=\rho_{T}^{*}$
and $h_{T}(\lambda_{T})=(T(1-\rho_{T}^{*}),T^{\theta}(1-\rho_{T}^{*}))$
as in our context.\footnote{In \prettyref{cor:DIVX_uniform}, the parameter space for $\lambda_{T}$
is $\Lambda_{T}=[-1+m_{1}^{*},1+m_{2}^{*}/T]$, where the right endpoint
$1+m_{2}^{*}/T$ allows for locally explosive regressors. Although
\citet{andrews2020generic} specify a fixed parameter space, it is
straightforward to adapt their results to our context where the right
endpoint is a convergent sequence.} Under the conditions of \prettyref{thm:DIVX} where $h_{T}(\lambda_{T})\to h_{0}\in[-m_{2}^{*},\infty]\times[0,\infty]$
and $n/T\to c\in[0,\infty)$, we invoke \citet{andrews2020generic}'s
Corollary~2.1(c) to conclude that for any $\alpha\in[0,1]$ we have
\[
\left|\Pr\left\{ t^{{\rm DIVX}}<\Phi^{-1}\left(\alpha\right)\right\} -\alpha\right|\to0
\]
 uniformly in the specified regime of $\rho_{T}^{*}$.
\end{proof}
\medskip{}

\begin{proof}[Proof of \prettyref{cor:ivx-wg}]
To show \eqref{eq:rho WG rate}, note that the denominator of $\rohatfe$
is the same as that in $\bhatfe$, while the numerator changes from
$\sum_{i=1}^{n}\sum_{t=1}^{T}\tilde{x}_{i,t}e_{i,t+1}$ to $\sum_{i=1}^{n}\sum_{t=1}^{T}\tilde{x}_{i,t}v_{i,t+1}$;
both are cross-products between the regressor $x_{i,t}$ and an m.d.s.
Thus, $\rohatfe$ behaves asymptotically the same way as $\bhatfe$.
By \prettyref{prop:bhatfe} we have $\bhatfe-\beta^{*}=O_{p}\bigl((nT^{1+\gamma})^{-1/2}+T^{-1}\bigr)$,
so that \eqref{eq:rho WG rate} holds for $\rohatfe$.

\eqref{eq:r IVX rate} and \eqref{eq:rho WG rate} yield
\[
r_{n,T}^{\ivx}(\rohatfe)=\dfrac{b_{n,T}^{{\rm IVX}}(\rohatfe)-b_{n,T}^{{\rm IVX}}(\rho^{*})}{\hat{\varsigma}^{{\rm IVX}}}=O_{p}\left(\dfrac{1}{\sqrt{T^{4-(\theta\vee\gamma)-\theta-2\gamma}}}+\sqrt{\dfrac{n}{T^{5-(\theta\vee\gamma)-\theta-3\gamma}}}\right).
\]
As $\joto$ and $n/T^{5-(\theta\vee\gamma)-\theta-3\gamma}\to0$,
it follows that $r_{n,T}^{\ivx}(\rohatfe)\pto0$, and hence by \prettyref{prop:bhativx}
the (infeasible) $t$-statistic using $\rohatfe$ is asymptotically
normal:
\[
\frac{\hat{\beta}^{\,\textup{IVX-WG}}-\beta^{*}}{\hat{\varsigma}^{{\rm IVX}}}=\frac{\bhativx-\beta^{*}+\omega_{ev}^{*}b_{n,T}^{\ivx}(\rho^{*})}{\hat{\varsigma}^{{\rm IVX}}}+\omega_{ev}^{*}r_{n,T}^{\ivx}(\rohatfe)\dto\mathcal{N}(0,1).\qedhere
\]
\end{proof}

\subsection{Proofs for WG Estimator}
\begin{proof}[Proof of \prettyref{prop:bhatfe}]
 By definition of $\bhatfe$ and $b_{n,T}^{\fe}(\rho^{*})$, we have
\[
\sqrt{nT^{1+\gamma}}\bigl[\bhatfe-\beta^{*}+\omega_{ev}^{*}b_{n,T}^{\fe}(\rho^{*})\bigr]=\frac{n^{-1/2}\sum_{i=1}^{n}L_{i,T}^{\fe}}{n^{-1}\sum_{i=1}^{n}\left(Q_{i,T}^{\fe}-R_{i,T}^{\fe}\right)},
\]
where $Q_{i,T}^{\fe}$, $R_{i,T}^{\fe}$ and $L_{i,T}^{\fe}$ are
respectively defined in \eqref{eq:def_Q}, \eqref{eq:def_R}, and
\eqref{eq:def_L}. By \prettyref{lem:joint}(iv) and \ref{lem:lurjoint}(iii),
the numerator
\begin{equation}
\frac{1}{\sqrt{n}}\sum_{i=1}^{n}L_{i,T}^{{\rm WG}}\dto\mathcal{N}(0,\Sigma_{\xtd e})\qquad\text{as }\joto,\label{eq:lim L general}
\end{equation}
where
\[
\Sigma_{\xtd e}=\begin{cases}
\omega_{ee}^{*}\omega_{vv}^{*}/(1-\rho^{*2}), & \text{if \ensuremath{\gamma=0} (stationary)},\\
\omega_{ee}^{*}\omega_{vv}^{*}/(-2c^{*}), & \text{if \ensuremath{0<\gamma<1} (MI)},\\
\omega_{ee}^{*}\mathrm{var}\left[\int_{0}^{1}\bigl(J_{2,c^{*}}(r)-\int_{0}^{1}J_{2,c^{*}}(\tau)\,d\tau\bigr)\,dB_{1}(r)\right], & \text{if \ensuremath{\gamma=1} (LUR)}.
\end{cases}
\]
In addition, \prettyref{lem:joint}(i)(ii) and \ref{lem:lurjoint}(i)(ii)
imply
\begin{equation}
\frac{1}{n}\sum_{i=1}^{n}\left(Q_{i,T}^{{\rm WG}}-R_{i,T}^{{\rm WG}}\right)\pto\lim_{T\to\infty}\mathbb{E}\left(Q_{i,T}^{{\rm WG}}-R_{i,T}^{{\rm WG}}\right)=Q_{\xtd\xtd}\qquad\text{as }\joto,\label{eq:limit Q m R}
\end{equation}
where
\[
Q_{\xtd\xtd}=\begin{cases}
\omega_{22}^{*}/(1-\rho^{*2}), & \text{if \ensuremath{\gamma=0} (stationary)},\\
\omega_{22}^{*}/(-2c^{*}), & \text{if \ensuremath{0<\gamma<1} (MI)},\\
\mathbb{E}\left[\int_{0}^{1}\bigl(J_{2,c^{*}}(r)-\int_{0}^{1}J_{2,c^{*}}(\tau)\,d\tau\bigr)^{2}\,dr\right], & \text{if \ensuremath{\gamma=1} (LUR)}.
\end{cases}
\]
Then by the Slutsky's theorem:
\begin{equation}
\sqrt{nT^{1+\gamma}}\bigl[\bhatfe-\beta^{*}+b_{n,T}^{\fe}(\rho^{*},\omega_{ev}^{*})\bigr]\dto\mathcal{N}(0,\Sigma^{\fe}\text{)}\qquad\text{as }\joto,\label{eq:bhatfe dto}
\end{equation}
where
\begin{equation}
\Sigma^{\fe}:=\Sigma_{\xtd e}/Q_{\xtd\xtd}^{2}.\label{eq:fe_var}
\end{equation}

Next we show $nT^{1+\gamma}\left(\varsigma^{{\rm WG}}\right)^{2}\pto\Sigma^{\fe}.$
By \eqref{eq:bhatfe dto} it follows that
\[
\left[\bhatfe-\beta^{*}+\omega_{ev}^{*}\cdot b_{n,T}^{\fe}(\rho^{*})\right]\Big/\varsigma^{\mathrm{WG}}\dto\ncal(0,1).
\]
Note that by the definition in \eqref{eq:fe_se-1}
\begin{equation}
nT^{1+\gamma}\left(\varsigma^{{\rm WG}}\right)^{2}=\frac{T^{-(1+\gamma)}\mathrm{var}\bigl(\sum_{t=1}^{T}\tilde{x}_{i,t}e_{i,t+1}\bigr)}{\bigl[(nT)^{-(1+\gamma)}\sum_{i=1}^{n}\sum_{t=1}^{T}\td{x}_{i,t}^{2}\bigr]^{2}}=\frac{\mathrm{var}\bigl(L_{i,T}^{{\rm WG}}\bigr)}{\bigl[(nT)^{-(1+\gamma)}\sum_{i=1}^{n}\sum_{t=1}^{T}\td{x}_{i,t}^{2}\bigr]^{2}}.\label{eq:se sq WG}
\end{equation}
By \eqref{eq:limit Q m R} and the continuous mapping theorem, the
denominator of the right hand side of \eqref{eq:se sq WG}
\[
\left(\frac{1}{nT^{1+\gamma}}\sum_{i=1}^{n}\sum_{t=1}^{T}\td{x}_{i,t}^{2}\right)^{2}\pto Q_{\xtd\xtd}^{2}.
\]
By the limiting distribution in \eqref{eq:limit dist L} and the uniform
integrability of $\bigl(L_{i,T}^{{\rm WG}}\bigr)^{2}$ deduced in
the proof of \prettyref{lem:lurjoint}(iii), we have by \eqref{eq:limit EXT}
in \prettyref{lem:u.i.} that, as $T\to\infty$,
\begin{align*}
\mathrm{var}\bigl(L_{i,T}^{{\rm WG}}\bigr)=\epct\left[(L_{i,T}^{{\rm WG}})^{2}\right] & \to\epct\left[\left\{ \int_{0}^{1}\biggl(J_{2,c^{*}}(r)-\int_{0}^{1}J_{2,c^{*}}(\tau)\,d\tau\biggr)\,dB_{1}(r)\right\} ^{2}\right]-\mathbb{E}\bigl[(H_{i,T}^{{\rm WG}})^{2}\bigr]\\
 & ={\rm var}\left[\int_{0}^{1}\biggl(J_{2,c^{*}}(r)-\int_{0}^{1}J_{2,c^{*}}(\tau)\,d\tau\biggr)\,dB_{1}(r)\right]=\Sigma_{\xtd e}.
\end{align*}
Then
\begin{equation}
nT^{1+\gamma}\left(\varsigma^{{\rm WG}}\right)^{2}\pto\Sigma_{\xtd e}/Q_{\xtd\xtd}^{2}=\Sigma^{\fe}\qquad\text{as }\joto.\label{eq:sigmaWG lim}
\end{equation}

Finally, note that
\[
b_{n,T}^{\fe}(\rho^{*})=\frac{T^{-(1+\gamma)/2}\mathbb{E}(H_{i,T}^{\fe})}{n^{-1}\sum_{i=1}^{n}(Q_{i,T}^{\fe}-R_{i,T}^{\fe})},
\]
where $H_{i,T}^{\fe}$ is defined in \eqref{eq:def_H}. By \prettyref{lem:epct}\ref{enu:E_sum_x_sum_e_2}
and \eqref{eq:limit Q m R},
\[
b_{n,T}^{\fe}(\rho^{*})=O\bigl(T^{-\frac{1}{2}(1+\gamma)}\bigr)\cdot O\bigl(T^{-\frac{1}{2}(1-\gamma)}\bigr)=O(T^{-1}).\qedhere
\]
\end{proof}
\begin{proof}[Proof of \prettyref{prop:FE-fail}]
By \prettyref{lem:esterror}, \eqref{eq:rho WG rate} in \prettyref{cor:ivx-wg},
and \prettyref{prop:rho_convergence}, we have
\begin{align*}
r_{n,T}^{\mathrm{WG}}(\rohatfe) & =O_{p}\left(\sqrt{\frac{n}{T^{1-3\gamma}}}|\rohatfe-\rho^{*}|\right)=O_{p}\left(\dfrac{1}{T^{1-\gamma}}+\sqrt{\dfrac{n}{T^{3(1-\gamma)}}}\right),\\
r_{n,T}^{\mathrm{WG}}(\hat{\rho}^{\mathrm{IVX}}) & =O_{p}\left(\sqrt{\frac{n}{T^{1-3\gamma}}}|\hat{\rho}^{\mathrm{IVX}}-\rho^{*}|\right)\\
 & =O_{p}\left(\dfrac{1}{\sqrt{T^{2-3\gamma+(\theta_{1}\wedge\gamma)}}}+\sqrt{\dfrac{n}{T^{1-3\gamma+4(\theta_{1}\wedge\gamma)}}}+\frac{G}{\sqrt{T^{2-3\gamma+2(\theta_{1}\wedge\gamma)}}}+\frac{\sqrt{n}G}{\sqrt{T^{3-3\gamma+2(\theta_{1}\wedge\gamma)}}}\right).
\end{align*}
Therefore, if $n/T^{3(1-\gamma)}\to0$, then $r_{n,T}^{\mathrm{WG}}(\rohatfe)\to_{p}0$;
if on the other hand, $\theta_{1}>3/4$, $n/T\to c\in[0,\infty)$
and $1/T^{1-\gamma}\to0$, then $r_{n,T}^{\mathrm{WG}}(\hat{\rho}^{\mathrm{IVX}})\to_{p}0$.
By \prettyref{prop:bhatfe} we conclude that
\[
\frac{\hat{\beta}^{\,\textup{WG-WG}}-\beta^{*}}{\varsigma^{\mathrm{WG}}}=\frac{\bhatfe-\beta^{*}+\omega_{ev}^{*}b_{n,T}^{\mathrm{WG}}(\rho^{*})}{\varsigma^{\mathrm{WG}}}+\omega_{ev}^{*}r_{n,T}^{\mathrm{WG}}(\rohatfe)\to_{d}\mathcal{N}(0,1)
\]
as $n/T^{3(1-\gamma)}\to0$, and
\[
\frac{\hat{\beta}^{\,\textup{WG-IVX}}-\beta^{*}}{\varsigma^{\mathrm{WG}}}=\frac{\bhatfe-\beta^{*}+\omega_{ev}^{*}b_{n,T}^{\mathrm{WG}}(\rho^{*})}{\varsigma^{\mathrm{WG}}}+\omega_{ev}^{*}r_{n,T}^{\mathrm{WG}}(\hat{\rho}^{\mathrm{IVX}})\to_{d}\mathcal{N}(0,1)
\]
as $n/T\to c\in[0,\infty)$ and $1/T^{1-\gamma}\to0$ provided $\theta_{1}>3/4$.
\end{proof}

\section{Proofs of Extensions of Theory \label{sec:Proofs}}

\subsection{Proofs for Multivariate Regression and Local Projection\label{subsec:Proofs-for-mult}}

\begin{delayedproof}{prop:ivx_multiple_bias}
This directly follows by \eqref{eq:E sum zeta sum e}.
\end{delayedproof}
\begin{delayedproof}{prop:ivx_variance_multiple}
 The proof is essentially the same as the \proofref{prop:bhativx}.
\end{delayedproof}
\begin{delayedproof}{prop:feasible_ivx_normal}
 The argument used in the \proofref{thm:DIVX} applies here.
\end{delayedproof}
\begin{proof}[Proof of \prettyref{thm:wald-chi2}]
Without loss of generality, assume that $\bm{A}$ is in reduced row
echelon form (since otherwise for the linear restriction $\bm{A}\bm{\beta}^{*}=\bm{q}$
we can left-multiply both sides elementary matrix to convert $\bm{A}$
into educed row echelon form) and he diagonal entries of $\bm{D}_{T}$
are arranged in ascending order. Let $\bm{L}_{T}$ be the $m\times m$
principal submatrix of $\bm{D}_{T}^{1/2}$ that corresponds to the
$m$ pivot columns of $\bm{A}$. For example, if $\bm{A}=[0,1]$,
then $L_{T}$ is the second column of $\bm{D}_{T}^{1/2}$. Consider
$\bm{B}_{T}:=\bm{L}_{T}\bm{A}\bm{D}_{T}^{-1/2}$. The $i$-th row
of $\bm{A}$ is scaled by $L_{T,i}$ where $L_{T,i}$ denotes of $i$-th
diagonal entry of $\bm{A}$ and the $j$-th column of $\bm{A}$ is
scaled by $D_{T,j}^{-1/2}$. Thus, for each row, the pivot entry is
unscaled while the others are either zero of has zero as limit so
that $\bm{B}_{T}$ is convergent and its limit has full row rank;
that is,
\begin{equation}
\lim_{T\to\infty}\underset{(m\times m)}{\bm{L}_{T}}\underset{(m\times k)}{\bm{A}}\underset{(k\times k)}{\bm{D}_{T}^{-1/2}}=\underset{(m\times k)}{\bm{B}},\label{eq:GAD}
\end{equation}
where $\bm{B}$ has full row rank. Then, under $\mathbb{H}_{0}$,
by \eqref{eq:IVXJ_normal} we have, as $(n,T)\to\infty$ with $n/T\to c\in[0,\infty)$,
\begin{align}
\sqrt{n}\bm{L}_{T}\bigl(\bm{A}\hat{\bm{\beta}}^{\mathrm{DIVX}}-\bm{q}\bigr) & =\bm{L}_{T}\bm{A}\bm{D}_{T}^{-1/2}(n\boldsymbol{D}_{T})^{1/2}\bigl(\hat{\bm{\beta}}^{\mathrm{DIVX}}-\bm{\beta}^{*}\bigr)\nonumber \\
 & \to_{d}\mathcal{N}\bigl(\boldsymbol{0}_{k},\bm{B}\bm{\Sigma}^{\mathrm{IVX}}\bm{B}^{\prime}\bigr).\label{eq:GAbeta}
\end{align}
By \ref{eq:DIVX_Theta_hat_converge}, \eqref{eq:GAD}, and \eqref{eq:GAbeta},
it follows that, as $(n,T)\to\infty$ with $n/T\to c\in[0,\infty)$,
\begin{align*}
 & \text{Wald}^{{\rm DIVX}}=\bigl(\bm{A}\hat{\bm{\beta}}^{\mathrm{DIVX}}-\bm{q}\bigr)'\bigl(\bm{A}\hat{\bm{\Theta}}^{\mathrm{DIVX}}\bm{A}'\bigr)^{-1}\bigl(\bm{A}\hat{\bm{\beta}}^{\mathrm{DIVX}}-\bm{q}\bigr)\\
={} & \left\Vert \left[\bm{L}_{T}\bm{A}\bm{D}_{T}^{-1/2}\left((n\boldsymbol{D}_{T})^{1/2}\hat{\bm{\Theta}}^{\mathrm{DIVX}}(n\boldsymbol{D}_{T})^{1/2}\right)\bm{D}_{T}^{-1/2}\bm{A}^{\prime}\bm{L}_{T}^{\prime}\right]^{-1/2}\left[\sqrt{n}\bm{L}_{T}\bigl(\bm{A}\hat{\bm{\beta}}^{\mathrm{DIVX}}-\bm{q}\bigr)\right]\right\Vert \\
\to_{d}{} & \chi^{2}(m).
\end{align*}
This justifies our DIVX estimator in the hypothesis testing for panel
predictive regression allowing for multivariate regressors.
\end{proof}
\begin{delayedproof}{prop:ivx_bias_h}
Let $\bm{v}_{i,0}:=\bm{x}_{i,0}$. Note that by construction, $\bm{z}_{i,t}$
can be written as
\begin{align*}
\boldsymbol{z}_{i,t} & =\sum_{j=1}^{t}\rho_{z}^{t-j}(\boldsymbol{x}_{i,j}-\boldsymbol{x}_{i,j-1})=\sum_{j=1}^{t}\rho_{z}^{t-j}[(\boldsymbol{R}^{*}-\boldsymbol{I}_{k})\boldsymbol{x}_{i,j-1}+\boldsymbol{v}_{i,j}]\\
 & =\sum_{j=1}^{t}\rho_{z}^{t-j}\left[(\boldsymbol{R}^{*}-\boldsymbol{I}_{k})\sum_{k=0}^{j-1}\boldsymbol{R}^{*j-1-k}\boldsymbol{v}_{i,k}+\boldsymbol{v}_{i,j}\right]\\
 & =\sum_{k=0}^{t-1}\Biggl[\sum_{j=k+1}^{t}(\boldsymbol{R}^{*}-\boldsymbol{I}_{k})\rho_{z}^{t-j}\boldsymbol{R}^{*j-1-k}\Biggr]\boldsymbol{v}_{i,k}+\sum_{j=1}^{t}\rho_{z}^{t-j}\boldsymbol{v}_{i,j}.
\end{align*}
It follows that
\begin{align*}
\sum_{t=1}^{T_{h}}\boldsymbol{z}_{i,t} & =\sum_{t=1}^{T_{h}}\left\{ \sum_{k=0}^{t-1}\Biggl[\sum_{j=k+1}^{t}(\boldsymbol{R}^{*}-\boldsymbol{I}_{k})\rho_{z}^{t-j}\boldsymbol{R}^{*j-1-k}\Biggr]\boldsymbol{v}_{i,k}+\sum_{j=1}^{t}\rho_{z}^{t-j}\boldsymbol{v}_{i,j}\right\} \\
 & =\sum_{k=0}^{T_{h}-1}\Biggl[\sum_{t=k+1}^{T_{h}}\sum_{j=k+1}^{t}(\boldsymbol{R}^{*}-\boldsymbol{I}_{k})\rho_{z}^{t-j}\boldsymbol{R}^{*j-1-k}\Biggr]\boldsymbol{v}_{i,k}+\sum_{j=1}^{T_{h}}\Biggl(\sum_{t=j}^{T_{h}}\rho_{z}^{t-j}\Biggr)\boldsymbol{v}_{i,j}.
\end{align*}
Since $\mathbb{E}(\boldsymbol{v}_{i,k}e_{i,s+h})=\boldsymbol{\omega}_{ev}^{*}$
only if $k=s+h$ and equals 0 otherwise, following the same argument
in \eqref{eq:E sum zeta sum e}, we deduce that
\begin{equation}
\mathbb{E}\left(\sum_{t=1}^{T_{h}}\boldsymbol{z}_{i,t}\sum_{s=1}^{T_{h}}e_{i,s+h}\right)=\sum_{t=h+1}^{T_{h}}\sum_{j=h+1}^{t}\rho_{z}^{t-j}\bm{R}^{*j-h-1}\bm{\omega}_{ev}^{*}.\label{eq:E zeta e 1}
\end{equation}
Likewise, since $\mathbb{E}(\boldsymbol{v}_{i,k}\bm{v}_{i,s+\tau}^{\prime})=\bm{\Omega}_{vv}^{*}$
only if $k=s+\tau$, we have for $\tau=1,\dots,h-1$ that
\begin{equation}
\mathbb{E}\left(\sum_{t=1}^{T_{h}}\boldsymbol{z}_{i,t}\sum_{s=1}^{T_{h}}\boldsymbol{v}_{i,s+\tau}^{\prime}\right)=\sum_{t=\tau+1}^{T_{h}}\sum_{j=\tau+1}^{t}\rho_{z}^{t-j}\bm{R}^{*j-h-1}\bm{\Omega}_{vv}^{*}.\label{eq:E zeta v vec}
\end{equation}
By \eqref{eq:error_h}, \eqref{eq:E zeta e 1}, and \eqref{eq:E zeta v vec},
we have
\begin{align*}
 & \mathbb{E}\left(\sum_{t=1}^{T_{h}}\boldsymbol{z}_{i,t}\sum_{s=1}^{T_{h}}e_{i,s+h}^{(h)}\right)=\mathbb{E}\left(\sum_{t=1}^{T_{h}}\boldsymbol{z}_{i,t}\sum_{s=1}^{T_{h}}e_{i,s+h}\right)+\sum_{\tau=1}^{h-1}\mathbb{E}\left(\sum_{t=1}^{T_{h}}\boldsymbol{z}_{i,t}\sum_{s=1}^{T_{h}}\boldsymbol{v}_{i,s+\tau}^{\prime}\bm{R}^{*h-1-\tau}\boldsymbol{\beta}^{*}\right)\\
 & \qquad=\sum_{t=h+1}^{T_{h}}\sum_{j=h+1}^{t}\rho_{z}^{t-j}\bm{R}^{*j-h-1}\bm{\omega}_{ev}^{*}+\sum_{\tau=1}^{h-1}\sum_{t=\tau+1}^{T_{h}}\sum_{j=\tau+1}^{t}\rho_{z}^{t-j}\bm{R}^{*j-h-1}\bm{\Omega}_{vv}^{*}\bm{R}^{*h-1-\tau}\boldsymbol{\beta}^{*}\\
 & \qquad=\dfrac{T_{h}}{n}\bm{\xi}_{n,T}(\boldsymbol{R}^{*},\boldsymbol{\omega}_{ev}^{*},\boldsymbol{\Omega}_{vv}^{*},\boldsymbol{\beta}^{*}).
\end{align*}
This completes the proof of \prettyref{prop:ivx_bias_h}.
\end{delayedproof}
\begin{proof}[Proof of \prettyref{prop:ivx_variance_h}]
For convenience, define \textbf{
\begin{align*}
\boldsymbol{Z}_{i,T} & :=\boldsymbol{D}_{T}^{-1/2}\sum_{t=1}^{T_{h}}\bm{z}_{i,t}e_{i,t+h}^{(h)},\ \ \ \boldsymbol{H}_{i,T}:=\boldsymbol{D}_{T}^{-1/2}T_{h}^{-1}\sum_{t=1}^{T_{h}}\bm{z}_{i,t}\sum_{t=1}^{T_{h}}e_{i,t+h}^{(h)}.
\end{align*}
}Then we have
\[
(n\boldsymbol{D}_{T})^{-1/2}\sum_{i=1}^{n}\sum_{t=1}^{T_{h}}\tilde{\bm{z}}_{i,t}e_{i,t+h}^{(h)}=\frac{1}{\sqrt{n}}\sum_{i=1}^{n}(\boldsymbol{Z}_{i,T}-\boldsymbol{H}_{i,T}).
\]
By the independence across $i$, we only need to show
\[
\mathrm{var}(\boldsymbol{Z}_{i,T}-\boldsymbol{H}_{i,T})-\boldsymbol{D}_{T}^{-1/2}\bm{\Sigma}_{T}^{(h)}\boldsymbol{D}_{T}^{-1/2}\to\bm{0}_{k\times k}\qquad\text{as }T\to\infty.
\]

This will follow as soon as we show: (i) $\mathbb{E}\bigl(\|\boldsymbol{H}_{i,T}\|^{2}\bigr)\to0$;
(ii) $\bigl\|\boldsymbol{D}_{T}^{-1/2}\boldsymbol{\Sigma}_{T}^{(h)}\boldsymbol{D}_{T}^{-1/2}\bigr\|=O(1)$;
and (iii) ${\rm var}(\boldsymbol{Z}_{i,T})-\boldsymbol{D}_{T}^{-1/2}\boldsymbol{\Sigma}_{T}^{(h)}\boldsymbol{D}_{T}^{-1/2}\to\bm{0}_{k\times k}$.
Note that (i) implies\footnote{We use the following fact about the Frobenius norm. For any two random
vectors $\bm{a}$ and $\bm{b}$ of the same length, by the Jensen's
inequality $\|\mathbb{E}(\bm{a}\bm{b}^{\prime})\|\leq\mathbb{E}(\|\bm{a}\bm{b}^{\prime}\|)=\mathbb{E}(\|\bm{a}\|\cdot\|\bm{b}\|)$
because $\|\bm{a}\bm{b}^{\prime}\|=\sqrt{\mathrm{tr}(\bm{b}\bm{a}^{\prime}\bm{a}\bm{b}^{\prime})}=\sqrt{(\bm{a}^{\prime}\bm{a})\cdot(\bm{b}^{\prime}\bm{b})}=\|\bm{a}\|\cdot\|\bm{b}\|$.}
\[
\left\Vert {\rm var}\left(\boldsymbol{H}_{i,T}\right)\right\Vert \leq\mathbb{E}\bigl(\left\Vert \boldsymbol{H}_{i,T}\right\Vert ^{2}\bigr)+\left\Vert \mathbb{E}\left(\boldsymbol{H}_{i,T}\right)\right\Vert ^{2}\leq\mathbb{E}\bigl(\left\Vert \boldsymbol{H}_{i,T}\right\Vert ^{2}\bigr)\to0.
\]
In addition, (ii) and (iii) imply $\mathbb{E}(Z_{j,i,T}^{2})=O(1)$
for any of the entry of $\bm{Z}_{i,T}$, and thus $\mathbb{E}\bigl(\left\Vert \boldsymbol{Z}_{i,T}\right\Vert ^{2}\bigr)=\sum_{j=1}^{k}\mathbb{E}(Z_{j,i,T}^{2})=O(1).$
It then follows that
\begin{align*}
 & \ \left\Vert {\rm var}(\boldsymbol{Z}_{i,T}-\boldsymbol{H}_{i,T})-{\rm var}(\boldsymbol{Z}_{i,T})\right\Vert \\
= & \ \left\Vert {\rm var}\left(\boldsymbol{H}_{i,T}\right)-\mathbb{E}\left(\boldsymbol{H}_{i,T}\boldsymbol{Z}_{i,T}^{\prime}\right)-\mathbb{E}\left(\boldsymbol{Z}_{i,T}\boldsymbol{H}_{i,T}^{\prime}\right)\right\Vert \leq\left\Vert {\rm var}\left(\boldsymbol{H}_{i,T}\right)\right\Vert +2\left\Vert \mathbb{E}\left(\boldsymbol{H}_{i,T}\boldsymbol{Z}_{i,T}^{\prime}\right)\right\Vert \\
\leq & \ \left\Vert {\rm var}\left(\boldsymbol{H}_{i,T}\right)\right\Vert +2\mathbb{E}\left(\left\Vert \boldsymbol{H}_{i,T}\right\Vert \left\Vert \boldsymbol{Z}_{i,T}\right\Vert \right)\leq\left\Vert {\rm var}\left(\boldsymbol{H}_{i,T}\right)\right\Vert +2\sqrt{\mathbb{E}\bigl(\left\Vert \boldsymbol{H}_{i,T}\right\Vert ^{2}\bigr)\cdot\mathbb{E}\bigl(\left\Vert \boldsymbol{Z}_{i,T}\right\Vert ^{2}\bigr)}\\
\to & \ 0\qquad\text{as }T\to\infty.
\end{align*}
This together with (iii) implies $\bigl\|\mathrm{var}(\boldsymbol{Z}_{i,T}-\boldsymbol{H}_{i,T})-\boldsymbol{D}_{T}^{-1/2}\bm{\Sigma}_{T}^{(h)}\boldsymbol{D}_{T}^{-1/2}\bigr\|\to0$
as $T\to\infty.$ \prettyref{prop:ivx_variance_h} is thus established.

Now we elaborate the proof of each step. \textbf{Step I.} Showing
$\mathbb{E}\bigl(\|\boldsymbol{H}_{i,T}\|^{2}\bigr)\to0$. By the
same argument in \eqref{eq: var H o(1)} for any $j$-th entry of
$\boldsymbol{H}_{i,T}$, we have $\mathbb{E}(H_{j,i,t}^{2})=o(1)$.
Thus $\mathbb{E}\bigl(\|\boldsymbol{H}_{i,T}\|^{2}\bigr)=o(1)$.

\textbf{Step II. }Showing $\bigl\|\boldsymbol{D}_{T}^{-1/2}\boldsymbol{\Sigma}_{T}^{(h)}\boldsymbol{D}_{T}^{-1/2}\bigr\|=O(1)$.
Since $\sup_{\ell}|\Gamma_{ee}^{(h)}(\ell)|$ and $\sup_{\ell}\|\bm{R}^{*\ell}\|$
are all bounded, it suffices to show $\bigl\|\boldsymbol{D}_{T}^{-1/2}\sum_{t=1}^{T_{h}}\mathbb{E}(\bm{z}_{i,t}\bm{z}_{i,t}^{\prime})\boldsymbol{D}_{T}^{-1/2}\bigr\|=O(1)$.
By the same argument used in the proof of \prettyref{lem:ivxjoint}\ref{enu:S_pto_joint},
for the $j$-th entry of $\bm{z}_{i,t}$ we have $\sum_{t=1}^{T_{h}}\mathbb{E}(z_{j,i,t}^{2})=O\bigl(T^{1+(\theta\wedge\gamma_{j})}\bigr),$
which further implies that for $z_{j,i,t}$ and $z_{m,i,t}$:
\begin{align*}
\sum_{t=1}^{T_{h}}\mathbb{E}(z_{j,i,t}z_{m,i,t}) & \leq\mathbb{E}\left[\left(\sum_{t=1}^{T_{h}}z_{j,i,t}^{2}\right)^{1/2}\left(\sum_{t=1}^{T_{h}}z_{m,i,t}^{2}\right)^{1/2}\right]\\
 & \leq\left[\mathbb{E}\left(\sum_{t=1}^{T_{h}}z_{j,i,t}^{2}\right)\cdot\mathbb{E}\left(\sum_{t=1}^{T_{h}}z_{m,i,t}^{2}\right)\right]^{1/2}=O\bigl(T^{\frac{1}{2}[1+(\theta\wedge\gamma_{j})+(\theta\wedge\gamma_{m})]}\bigr).
\end{align*}
Thus we have $\bigl\|\boldsymbol{D}_{T}^{-1/2}\sum_{t=1}^{T_{h}}\mathbb{E}(\bm{z}_{i,t}\bm{z}_{i,t}^{\prime})\boldsymbol{D}_{T}^{-1/2}\bigr\|=O(1)$.

\textbf{Step III.} Showing ${\rm var}\left(\boldsymbol{Z}_{i,T}\right)-\boldsymbol{D}_{T}^{-1/2}\boldsymbol{\Sigma}_{T}^{(h)}\boldsymbol{D}_{T}^{-1/2}\to\bm{0}_{k\times k}$.
Since $e_{i,t+h+\ell}^{(h)}$ is $\mathcal{F}_{t+\ell+1}$-measurable,
for $\ell\geq h$ we have
\[
\mathbb{E}\left(\bm{z}_{i,t}e_{i,t+h}^{(h)}\bm{z}_{i,t+\ell}e_{i,t+h+\ell}^{(h)}\right)=\mathbb{E}\left[\bm{z}_{i,t}e_{i,t+h}^{(h)}\bm{z}_{i,t+\ell}\mathbb{E}_{t+\ell}\bigl(e_{i,t+h+\ell}^{(h)}\bigr)\right]=0,
\]
and for $0\leq\ell<h$:
\[
\mathbb{E}\left(\bm{z}_{i,t}\bm{z}_{i,t+\ell}^{\prime}e_{i,t+h}^{(h)}e_{i,t+h+\ell}^{(h)}\right)=\mathbb{E}\left[\bm{z}_{i,t}\bm{z}_{i,t+\ell}^{\prime}\mathbb{E}_{t+\ell}\bigl(e_{i,t+h}^{(h)}e_{i,t+h+\ell}^{(h)}\bigr)\right]=\Gamma_{ee}^{(h)}(\ell)\cdot\mathbb{E}\left(\bm{z}_{i,t}\bm{z}_{i,t+\ell}^{\prime}\right).
\]
For the symmetric case where $\ell$ is negative, we can also get
for $\ell\leq-h$,
\[
\mathbb{E}\left(\bm{z}_{i,t}e_{i,t+h}^{(h)}\bm{z}_{i,t+\ell}e_{i,t+h+\ell}^{(h)}\right)=\mathbb{E}\left[\bm{z}_{i,t}e_{i,t+h+\ell}^{(h)}\bm{z}_{i,t+\ell}\mathbb{E}_{t}\bigl(e_{i,t+h}^{(h)}\bigr)\right]=0,
\]
and for $-h<\ell\leq0$:
\[
\mathbb{E}\left(\bm{z}_{i,t}\bm{z}_{i,t+\ell}^{\prime}e_{i,t+h}^{(h)}e_{i,t+h+\ell}^{(h)}\right)=\mathbb{E}\left[\bm{z}_{i,t}\bm{z}_{i,t+\ell}^{\prime}\mathbb{E}_{t}\left(e_{i,t+h}^{(h)}e_{i,t+h+\ell}^{(h)}\right)\right]=\Gamma_{ee}^{(h)}(\ell)\cdot\mathbb{E}\left(\bm{z}_{i,t}\bm{z}_{i,t+\ell}^{\prime}\right).
\]
These results give
\begin{align}
{\rm var}\left(\boldsymbol{Z}_{i,T}\right) & =\boldsymbol{D}_{T}^{-1/2}\mathbb{E}\left(\sum_{t=1}^{T_{h}}\bm{z}_{i,t}e_{i,t+h}^{(h)}\sum_{t=1}^{T_{h}}\bm{z}_{i,t}^{\prime}e_{i,t+h}^{(h)}\right)\boldsymbol{D}_{T}^{-1/2}\nonumber \\
 & =\boldsymbol{D}_{T}^{-1/2}\sum_{\ell=-(h-1)}^{h-1}\sum_{t=1}^{T_{h}}\mathbb{E}\left(\bm{z}_{i,t}\bm{z}_{i,t+\ell}^{\prime}e_{i,t+h}^{(h)}e_{i,t+h+\ell}^{(h)}\right)\boldsymbol{D}_{T}^{-1/2}\nonumber \\
 & =\boldsymbol{D}_{T}^{-1/2}\sum_{\ell=-(h-1)}^{h-1}\left[\Gamma_{ee}^{(h)}(\ell)\sum_{t=1}^{T_{h}}\mathbb{E}\left(\bm{z}_{i,t}\bm{z}_{i,t+\ell}^{\prime}\right)\right]\boldsymbol{D}_{T}^{-1/2}.\label{eq:var Z mult start}
\end{align}
Let $(\bm{A})_{j,m}$ be the $(j,m)$-entry of matrix $\bm{A}$. We
have
\[
\left({\rm var}\left(\boldsymbol{Z}_{i,T}\right)-\boldsymbol{D}_{T}^{-1/2}\boldsymbol{\Sigma}^{(h)}\boldsymbol{D}_{T}^{-1/2}\right)_{j,m}=\sum_{\ell=-(h-1)}^{h-1}\left[T^{-\frac{1}{2}[2+(\theta\wedge\gamma_{j})+(\theta\wedge\gamma_{m})]}\sum_{t=1}^{T_{h}}\mathbb{E}\left(z_{j,i,t}z_{m,i,t+\ell}-z_{j,i,t}z_{m,i,t}\rho_{m}^{*\ell}\right)\right].
\]
It thus suffices to show that for any $j,m\in\{1,\dots,k\}$ and $-(h-1)\leq\ell\leq h-1$,
as $T\to\infty$,
\begin{equation}
T^{-\frac{1}{2}[2+(\theta\wedge\gamma_{j})+(\theta\wedge\gamma_{m})]}\sum_{t=1}^{T_{h}}\mathbb{E}\left(z_{j,i,t}z_{m,i,t+\ell}-z_{j,i,t}z_{m,i,t}\rho_{m}^{*\ell}\right)\to0.\label{eq:target mult var}
\end{equation}

For any $j\in\{1,\dots,k\}$, define $\psi_{j,i,t}:=\sum_{\tau=1}^{t}\rho_{z}^{t-\tau}x_{j,i,\tau-1}$
as in \eqref{eq:psi}. Following \eqref{eq:psi eq Pv}, we have
\[
\psi_{j,i,t}=\sum_{s=0}^{t-1}P_{j,t,s}v_{j,i,s}\qquad\text{where}\qquad P_{j,t,s}:=\dfrac{\rho_{z}^{t-s}-\rho_{j}^{*t-s}}{\rho_{z}-\rho_{j}^{*}}\text{ and }v_{j,i,0}:=x_{j,i,0}.
\]
In addition, \prettyref{assump:mult-initval-1} gives $\mathbb{E}(v_{j,i,0}^{2})=O(T^{\gamma_{j}})$.
Therefore,
\begin{equation}
\sup_{t\le T}\mathbb{E}(\psi_{j,i,t}^{2})=\sup_{t\leq T}{}\sum_{s=0}^{t-1}P_{j,t,s}^{2}\mathbb{E}(v_{j,i,s}^{2})=O\bigl(T^{2(\theta\wedge\gamma_{j})+(\theta\vee\gamma_{j})}\bigr)=O\bigl(T^{\theta+\gamma_{j}+(\theta\wedge\gamma_{j})}\bigr).\label{eq:E psi j 2}
\end{equation}
For a fixed $\ell$, we have
\begin{align}
 & \sup_{t\le T}{}\mathbb{E}\left[(\psi_{j,i,t+\ell}-\psi_{j,i,t})^{2}\right]=\sup_{t\leq T}{}\mathbb{E}\left[\left(\sum_{s=0}^{t-1}(P_{j,t+\ell,s}-P_{j,t,s})v_{j,i,s}+\sum_{s=t}^{t+\ell}P_{j,t+\ell,s}v_{j,i,s}\right)^{2}\right]\nonumber \\
\leq{} & 2\sup_{t\leq T}{}\left[\sum_{s=0}^{t-1}(P_{j,t+\ell,s}-P_{j,t,s})^{2}\mathbb{E}(v_{j,i,s}^{2})+\sum_{s=t}^{t+\ell}P_{j,t+\ell,s}^{2}\mathbb{E}(v_{j,i,s}^{2})\right]\nonumber \\
={} & O\bigl(T^{2(\theta\wedge\gamma_{j})}\bigr)+O\bigl(T^{2(\theta\wedge\gamma_{j})}\bigr)=O\bigl(T^{2(\theta\wedge\gamma_{j})}\bigr).\label{eq:E d_psi_j2}
\end{align}
Let $\zeta_{j,i,t}$ be an AR(1) process such that $\zeta_{j,i,0}=0$
and
\[
\zeta_{j,i,t}=\rho_{z}\zeta_{j,i,t-1}+v_{j,i,t},\qquad t=1,\dots,T.
\]
Then $\zeta_{j,i,t}$ can be decomposed as indicated by \eqref{eq:zeta decom 1}
and \eqref{eq:zeta decom 2}:
\begin{equation}
z_{j,i,t}=\zeta_{j,i,t}-(1-\rho_{j}^{*})\psi_{j,i,t},\label{eq:zeta j decom 1}
\end{equation}
and
\begin{equation}
z_{j,i,t}=x_{j,i,t}-\rho_{z}^{t}x_{j,i,0}-(1-\rho_{z})\psi_{j,i,t},\label{eq:zeta j decom 2}
\end{equation}

By \prettyref{lem:generic_two_ARs} we have
\begin{equation}
\mathbb{E}(\zeta_{j,i,t}^{2})=O(T^{\theta})\qquad\text{and}\qquad\mathbb{E}(x_{j,i,t}^{2})=O(T^{\gamma_{j}}).\label{eq:E z2 and E x2}
\end{equation}
If $\theta\leq\gamma_{j}$, we employ the decomposition \eqref{eq:zeta j decom 1},
then \eqref{eq:E z2 and E x2} and \eqref{eq:E psi j 2} are combined
to yield
\[
\sup_{t\leq T}\mathbb{E}(z_{j,i,t}^{2})\leq\sup_{t\leq T}2\left[\mathbb{E}(\zeta_{j,i,t}^{2})+(1-\rho_{j}^{*})^{2}\mathbb{E}(\psi_{j,i,t}^{2})\right]=O(T^{\theta})+O\bigl(T^{\theta+(\theta\wedge\gamma)-\gamma}\bigr)=O(T^{\theta}).
\]
If $\gamma_{j}<\theta$, we use the other decomposition \eqref{eq:zeta j decom 2},
under which \eqref{eq:E z2 and E x2}, \eqref{eq:E psi j 2}, and
\prettyref{assump:mult-initval-1} leads to
\begin{align*}
\sup_{t\leq T}\mathbb{E}(z_{j,i,t}^{2}) & =\sup_{t\leq T}2\left[\mathbb{E}(x_{j,i,t}^{2})+\rho_{z}^{2t}x_{j,i,0}^{2}+(1-\rho_{z})^{2}\mathbb{E}(\psi_{j,i,t}^{2})\right]\\
 & =O(T^{\gamma_{j}})+O(T^{\gamma_{j}})+O\bigl(T^{\gamma+(\theta\wedge\gamma)-\theta}\bigr)=O(T^{\gamma_{j}}).
\end{align*}
To sum up, we have
\begin{equation}
\sup_{t\leq T}\mathbb{E}(z_{j,i,t}^{2})=O(T^{\theta\wedge\gamma_{j}}).\label{eq:E zeta j 2}
\end{equation}
We now are ready to prove \eqref{eq:target mult var}.

By \eqref{eq:zeta j decom 2} we have
\begin{align*}
z_{m,i,t+\ell} & =x_{m,i,t+\ell}-\rho_{z}^{t+\ell}x_{m,i,0}-(1-\rho_{z})\psi_{m,i,t+\ell}\\
 & =\rho_{m}^{*\ell}x_{m,i,t}+\sum_{s=t+1}^{t+\ell}\rho_{m}^{*t+\ell-s}v_{m,i,s}-\rho_{z}^{t+\ell}x_{m,i,0}-(1-\rho_{z})\psi_{m,i,t+\ell}\\
 & =\rho_{m}^{*\ell}z_{m,i,t}+\sum_{s=t+1}^{t+\ell}\rho_{m}^{*t+\ell-s}v_{m,i,s}+\rho_{z}^{t}(1-\rho_{z}^{\ell})x_{m,i,0}-(1-\rho_{z})(\psi_{m,i,t+\ell}-\psi_{m,i,t}).
\end{align*}
It follows that
\begin{align}
 & \mathbb{E}(z_{j,i,t}z_{m,i,t+\ell})-\mathbb{E}(z_{j,i,t}z_{m,i,t})\rho_{m}^{*\ell}\nonumber \\
 & \qquad=\rho_{z}^{t}(1-\rho_{z}^{\ell})\mathbb{E}(z_{j,i,t}x_{m,i,0})+(1-\rho_{z})\mathbb{E}[z_{j,i,t}(\psi_{m,i,t+\ell}-\psi_{m,i,t})].\label{eq:mult var 2-1}
\end{align}
By \eqref{eq:E zeta j 2} and \prettyref{assump:mult-initval-1},
we have
\begin{align}
 & T^{-\frac{1}{2}[2+(\theta\wedge\gamma_{j})+(\theta\wedge\gamma_{m})]}\sum_{t=1}^{T_{h}}\rho_{z}^{t}(1-\rho_{z}^{\ell})\mathbb{E}(z_{j,i,t}x_{m,i,0})\nonumber \\
\leq{} & T^{-\frac{1}{2}[2+(\theta\wedge\gamma_{j})+(\theta\wedge\gamma_{m})]}(1-\rho_{z}^{\ell})\left(\sum_{t=1}^{T_{h}}\rho_{z}^{t}\right)\sup_{t\leq T}\sqrt{\mathbb{E}(z_{j,i,t}^{2})\cdot\mathbb{E}(x_{m,i,0}^{2})}\nonumber \\
={} & T^{-\frac{1}{2}[2+(\theta\wedge\gamma_{j})+(\theta\wedge\gamma_{m})]}\cdot O(T^{-\theta})\cdot O(T^{\theta})\cdot O\bigl(T^{\frac{1}{2}[(\theta\wedge\gamma_{j})+\gamma_{m}]}\bigr)\nonumber \\
={} & O\bigl(T^{-\frac{1}{2}[2+(\theta\wedge\gamma_{m})-\gamma_{m}]}\bigr)\to0.\label{eq:mult var 2-2}
\end{align}
In addition, by \eqref{eq:E d_psi_j2} and \eqref{eq:E zeta j 2},
\begin{align}
 & T^{-\frac{1}{2}[2+(\theta\wedge\gamma_{j})+(\theta\wedge\gamma_{m})]}(1-\rho_{z})\sum_{t=1}^{T_{h}}\mathbb{E}[z_{j,i,t}(\psi_{m,i,t+\ell}-\psi_{m,i,t})]\nonumber \\
\leq{} & T^{-\frac{1}{2}[2+(\theta\wedge\gamma_{j})+(\theta\wedge\gamma_{m})]}(1-\rho_{z})T_{h}\sup_{t\leq T}\sqrt{\mathbb{E}(z_{j,i,t}^{2})\cdot\mathbb{E}\bigl[(\psi_{m,i,t+\ell}-\psi_{m,i,t})^{2}\bigr]}\nonumber \\
={} & O\bigl(T^{-\frac{1}{2}[2\theta-(\theta\wedge\gamma_{m})]}\bigr)\to0.\label{eq:mult var 2-3}
\end{align}
Equations \eqref{eq:mult var 2-1}, \eqref{eq:mult var 2-2}, and
\eqref{eq:mult var 2-3} establish \eqref{eq:target mult var} and
hence complete the proof of \prettyref{prop:ivx_variance_h}.
\end{proof}

\subsection{Proofs for SBSA\label{subsec:Proofs-for-SBSA}}
\begin{proof}[Proof of \prettyref{lem:sup_rho_hat_bound}]
We first show the part of $\hat{\rho}_{i}$. Without loss of generality,
assume $\alpha_{i}=0$. Letting $\bar{x}_{i}:=T^{-1}\sum_{t=1}^{T}x_{i,t}$,
we have
\begin{align*}
\hat{\rho}_{i} & =\frac{\sum_{t=1}^{T}(x_{i,t}-\bar{x}_{i})x_{i,t+1}}{\sum_{t=1}^{T}(x_{i,t}-\bar{x}_{i})^{2}}\\
 & =\rho^{*}+\frac{\sum_{t=1}^{T}\mathbb{E}(x_{i,t}v_{i,t+1})}{\sum_{t=1}^{T}(x_{i,t}-\bar{x}_{i})^{2}}+\frac{\sum_{t=1}^{T}[x_{i,t}v_{i,t+1}-\mathbb{E}(x_{i,t}v_{i,t+1})]}{\sum_{t=1}^{T}(x_{i,t}-\bar{x}_{i})^{2}}+\frac{\bar{x}_{i}\sum_{t=1}^{T}x_{i,t+1}}{\sum_{t=1}^{T}(x_{i,t}-\bar{x}_{i})^{2}}.
\end{align*}
The last two terms are asymptotically negligible while the second
is not. Let us first investigate the limit of $T^{-1}\sum_{t=1}^{T}\mathbb{E}(x_{i,t}v_{i,t+1})$
as $T\to\infty$. Using the fact that $x_{i,t}=\sum_{j=1}^{t}\rho^{*t-j}v_{i,j}+\rho^{*t}x_{i,0}$,
we have
\begin{align*}
\frac{1}{T}\sum_{t=1}^{T}\mathbb{E}(x_{i,t}v_{i,t+1}) & =\frac{1}{T}\sum_{t=1}^{T}\sum_{j=1}^{t}\rho^{*t-j}\Gamma_{vv}(t+1-j)+\frac{1}{T}\sum_{t=1}^{T}\rho^{*t}\mathbb{E}(x_{i,0}v_{i,t+1})\\
 & =\frac{1}{T}\sum_{t=1}^{T}\sum_{j=1}^{t}\rho^{*j-1}\Gamma_{vv}(j)+\frac{1}{T}\sum_{t=1}^{T}\rho^{*t}\mathbb{E}(x_{i,0}v_{i,t+1}).
\end{align*}
The second term
\[
\left|\frac{1}{T}\sum_{t=1}^{T}\rho^{*t}\mathbb{E}(x_{i,0}v_{i,t+1})\right|\lesssim\frac{1}{T}\sum_{t=1}^{T}|\rho^{*}|^{t}=\frac{|\rho^{*}|(1-|\rho^{*}|^{T})}{T(1-|\rho^{*}|)}\to0
\]
would vanish as $T\to\infty$ by the fact that $\sup_{t}|\mathbb{E}(x_{i,0}v_{i,t+1})|\leq\sup_{t}\sum_{s=-\infty}^{0}|g_{t-s}||\mathbb{E}(x_{i,0}\varepsilon_{i,s})|<\infty$.
For the first term, since $|\rho^{*j-1}\Gamma_{vv}(j)|\leq|\Gamma_{vv}(j)|$
and $\sum_{j=1}^{\infty}\Gamma_{vv}(j)<\infty$, $\sum_{j=1}^{\infty}\rho^{*j-1}\Gamma_{vv}(j)$
must exist, and then by the theorem of Cesàro mean, we have
\[
\lim_{T\to\infty}\frac{1}{T}\sum_{t=1}^{T}\sum_{j=1}^{t}\rho^{*j-1}\Gamma_{vv}(j)=\sum_{j=1}^{\infty}\rho^{*j-1}\Gamma_{vv}(j).
\]
Hence,
\begin{equation}
\lim_{T\to\infty}\frac{1}{T}\sum_{t=1}^{T}\mathbb{E}(x_{i,t}v_{i,t+1})=\sum_{j=1}^{\infty}\rho^{*j-1}\Gamma_{vv}(j)=:\sigma_{xv}.\label{eq:sigma_xv}
\end{equation}
Next we calculate the limit of $T^{-1}\sum_{t=1}^{T}\mathbb{E}(x_{i,t}^{2})$.
We can show that the term involving initial values $x_{i,0}$ must
be negligible so we simply assume $x_{i,0}=0$ here. Then we can write
\begin{align*}
\frac{1}{T}\sum_{t=1}^{T}\mathbb{E}(x_{i,t}^{2}) & =\frac{1}{T}\sum_{t=1}^{T}\sum_{j=1}^{t}\rho^{*2(t-j)}\Gamma_{vv}(0)+\frac{2}{T}\sum_{t=1}^{T}\sum_{j=1}^{t}\sum_{k=1}^{j-1}\rho^{*t-j}\rho^{*t-k}\Gamma_{vv}(j-k)\\
 & =\left[\frac{1}{1-\rho^{*2}}-\frac{\rho^{*2}(1-\rho^{*2T})}{T(1-\rho^{*2})}\right]\Gamma_{vv}(0)+\frac{2}{T}\sum_{k=1}^{T-1}\sum_{t=k+1}^{T}\sum_{j=k+1}^{t}\rho^{*2(t-j)}\rho^{*k}\Gamma_{vv}(k).
\end{align*}
The first term has limit
\[
\lim_{T\to\infty}\left[\frac{1}{1-\rho^{*2}}-\frac{\rho^{*2}(1-\rho^{*2T})}{T(1-\rho^{*2})}\right]\Gamma_{vv}(0)=\frac{1}{1-\rho^{*2}}\Gamma_{vv}(0).
\]
For the second term, since for each fixed $k$,
\[
\frac{1}{T}\sum_{t=k+1}^{T}\sum_{j=k+1}^{t}\rho^{*2(t-j)}\rho^{*k}=\frac{T-k}{T(1-\rho^{*2})}-\frac{\rho^{*2}\big(1-\rho^{*2(T-k)}\big)}{T(1-\rho^{*2})^{2}}\to\frac{1}{1-\rho^{*2}}\qquad\text{as }T\to\infty,
\]
then by \prettyref{lem:series} we have
\[
\lim_{T\to\infty}\frac{2}{T}\sum_{k=1}^{T-1}\sum_{t=k+1}^{T}\sum_{j=k+1}^{t}\rho^{*2(t-j)}\rho^{*k}\Gamma_{vv}(k)=\frac{2}{1-\rho^{*2}}\sum_{k=1}^{\infty}\rho^{*k}\Gamma_{vv}(k).
\]
It follows that
\begin{equation}
\lim_{T\to\infty}\frac{1}{T}\sum_{t=1}^{T}\mathbb{E}(x_{i,t}^{2})=\frac{1}{1-\rho^{*2}}\left[\Gamma_{vv}(0)+2\sum_{k=1}^{\infty}\rho^{*k}\Gamma_{vv}(k)\right]=:\sigma_{xx}.\label{eq:sigma_xx}
\end{equation}
By \prettyref{lem:generic_AR}\ref{enu:E_sum_eta_4}, we have
\[
\mathbb{E}(\bar{x}_{i}^{2})=\frac{1}{T^{2}}\mathbb{E}\left[\left(\sum_{t=1}^{T}x_{i,t}\right)^{2}\right]=O\left(\frac{1}{T}\right)\to0.
\]
Thus,
\[
\lim_{T\to\infty}\frac{1}{T}\mathbb{E}\left[\sum_{t=1}^{T}(x_{i,t}-\bar{x}_{i})^{2}\right]=\sigma_{xx}.
\]
Now, we show that the following map
\[
R(\rho):=\frac{\sum_{j=1}^{\infty}\rho^{j-1}\Gamma_{vv}(j)}{\frac{1}{1-\rho^{2}}\left[\Gamma_{vv}(0)+2\sum_{k=1}^{\infty}\rho^{k}\Gamma_{vv}(k)\right]}+\rho=\frac{\sigma_{xv}}{\sigma_{xx}}+\rho
\]
is strictly increasing on $(-1,1)$ and $R(\rho)<1$ for each $\rho\in(-1,1)$.

By the spectral representation $\Gamma_{vv}(k)=\frac{1}{2\pi}\int_{-\pi}^{\pi}\exp(\mathrm{i}k\theta)f_{v}(\theta)\,d\theta$
where $f_{v}(\theta)\geq0$ is the spectral density supported on $[-\pi,\pi]$
and the geometric-series identities
\[
1+2\sum_{k=1}^{\infty}r^{k}\cos(k\theta)=\frac{1-r^{2}}{1-2r\cos\theta+r^{2}},\qquad\sum_{j=1}^{\infty}r^{j-1}\cos(j\theta)=\frac{\cos\theta-r}{1-2r\cos\theta+r^{2}},
\]
valid for $|r|<1$, we obtain for $\rho^{*}\in(-1,1)$:
\[
\Gamma_{vv}(0)+2\sum_{k=1}^{\infty}\rho^{k}\Gamma_{vv}(k)=\frac{1}{2\pi}\int_{-\pi}^{\pi}\frac{1-\rho^{2}}{(\rho-\cos\theta)^{2}+(\sin\theta)^{2}}f_{v}(\theta)\,d\theta,
\]
and
\[
\sum_{j=1}^{\infty}\rho^{j-1}\Gamma_{vv}(j)=\frac{1}{2\pi}\int_{-\pi}^{\pi}\frac{\cos\theta-\rho}{(\rho-\cos\theta)^{2}+(\sin\theta)^{2}}f_{v}(\theta)\,d\theta.
\]
Therefore $R(\rho^{*})$ can be written as
\[
R(\rho)=\frac{\int_{-\pi}^{\pi}(\cos\theta)w_{\rho}(\theta)\,d\theta}{\int_{-\pi}^{\pi}w_{\rho}(\theta)\,d\theta}\qquad\text{where }w_{\rho^{*}}(\theta)=\frac{f_{v}(\theta)}{(\rho-\cos\theta)^{2}+(\sin\theta)^{2}}\geq0.
\]
Clearly, $\sup_{\theta\in[-\pi,\pi]}\cos\theta=1$ and the supremum
is attained only at $\theta=0$. For linear process $v_{i,t}$, the
spectral density is given by $f_{v}(\theta)=\frac{1}{2\pi}\mathbb{E}(\varepsilon_{i,0}^{2})|G(\exp(-\mathrm{i}\theta))|^{2}$
where $G(z)=\sum_{j=0}^{\infty}g_{j}z^{j}$, which is continuous and
strictly positive on a set of positive Lebesgue measure. This indicates
that $w_{\rho}(\theta)$ is not concentrated on $\theta=0$. It follows
that $\int_{-\pi}^{\pi}(\cos\theta)w_{\rho}(\theta)\,d\theta<\int_{-\pi}^{\pi}w_{\rho}(\theta)\,d\theta$
and thus $R(\rho)<1$ for all $\rho\in(-1,1)$.

Next we show that $R(\rho^{*})$ is strictly increasing (and hence
injective). Put $L(\theta;\rho):=(\rho^{*}-\cos\theta)^{2}+(\sin\theta)^{2}$
and define $A(\rho):=\int(\cos\theta)L(\theta;\rho)^{-1}f_{v}(\theta)\,d\theta$
and $B(\rho):=\int L(\theta;\rho)^{-1}f_{v}(\theta)\,d\theta$, so
that $R(\rho)=A(\rho)/B(\rho)$. Differentiating under the integral
sign,
\[
\frac{d}{d\rho}\left(\frac{1}{L}\right)=\frac{2(\cos\theta-\rho)}{L^{2}},\quad A'(\rho)=2\int\frac{(\cos\theta)(\cos\theta-\rho)}{L^{2}}f_{v}\,d\theta,\quad B'(\rho)=2\int\frac{\cos\theta-\rho}{L^{2}}f_{v}\,d\theta.
\]
Thus,
\begin{align*}
R'(\rho) & =\frac{A'(\rho)B(\rho)-A(\rho)B'(\rho)}{B(\rho)^{2}}\\
 & =\frac{2}{B(\rho)^{2}}\left[\left(\int\frac{(\cos\theta)(\cos\theta-\rho)}{L^{2}}f_{v}\,d\theta\right)B(\rho)-\left(\int\frac{\cos\theta}{L}f_{v}\,d\theta\right)\left(\int\frac{\cos\theta-\rho}{L^{2}}f_{v}\,d\theta\right)\right]\\
 & =2\left[\int(\cos\theta)(\cos\theta-\rho)\ell_{\rho}(\theta)\frac{\ell_{\rho}(\theta)f_{v}(\theta)}{B(\rho)}\,d\theta\right.\\
 & \qquad-\left.\left(\int(\cos\theta)\frac{\ell_{\rho}(\theta)f_{v}(\theta)}{B(\rho)}\,d\theta\right)\left(\int\ell_{\rho}(\theta)(\cos\theta-\rho)\frac{\ell_{\rho}(\theta)f_{v}(\theta)}{B(\rho)}\,d\theta\right)\right],
\end{align*}
where we let $\ell_{\rho}(\theta):=L(\theta;\rho)^{-1}$. Define the
probability measure $\mu_{\rho}$ on $[-\pi,\pi]$ by $d\mu_{\rho}(\theta)=\frac{\ell_{\rho}(\theta)f_{v}(\theta)}{B(\rho)}\,d\theta$.
Then the identity above rewrites as
\[
R'(\rho)=2\mathrm{Cov}_{\mu_{\rho}}\bigl(\cos\theta,\ell_{\rho}(\theta)(\cos\theta-\rho)\bigr).
\]
Consider the function $h_{\rho}(x)=\frac{x-\rho}{1-2x\rho+\rho^{2}}.$
It is easy to show that $dh_{\rho}(x)/dx>0$ for all $x\in[-1,1]$
and $\rho\in(-1,1)$. By Chebyshev's association inequality, we have
\[
\mathrm{Cov}_{\mu_{\rho}}\bigl(\cos\theta,\ell_{\rho}(\theta)(\cos\theta-\rho)\bigr)=\mathrm{Cov}_{\mu_{\rho}}\bigl(\cos\theta,h_{\rho}(\cos\theta)\bigr)\geq0,
\]
with strict inequality unless $\cos\theta$ is $\mu_{\rho}$-a.s.\ constant.
However, the latter cannot occur because the density of $\mu_{\rho}$,
$\frac{\ell_{\rho}(\theta)f_{v}(\theta)}{B(\rho)}$, inheriting properties
from $f_{v}(\theta)$, must assign positive mass to sets where $\cos\theta$
varies. Therefore, $F'(\rho)>0$ for all $\rho\in(-1,1)$, proving
monotonicity and injectivity.

Therefore,
\begin{align}
\hat{\rho}_{i}-R(\rho^{*}) & =\frac{T^{-1}\sum_{t=1}^{T}[x_{i,t}v_{i,t+1}-\mathbb{E}(x_{i,t}v_{i,t+1})]}{T^{-1}\sum_{t=1}^{T}(x_{i,t}-\bar{x}_{i})^{2}}\nonumber \\
 & \qquad+\frac{T^{-1}\sum_{t=1}^{T}\mathbb{E}(x_{i,t}v_{i,t+1})-\sigma_{xv}\sigma_{xx}^{-1}T^{-1}\sum_{t=1}^{T}(x_{i,t}-\bar{x}_{i})^{2}}{T^{-1}\sum_{t=1}^{T}(x_{i,t}-\bar{x}_{i})^{2}}\nonumber \\
 & \qquad+\frac{\bar{x}_{i}T^{-1}\sum_{t=1}^{T}x_{i,t+1}}{T^{-1}\sum_{t=1}^{T}(x_{i,t}-\bar{x}_{i})^{2}},\label{eq:rho_hat}
\end{align}
where $\sigma_{xv}$ and $\sigma_{xx}$ are defined in \eqref{eq:sigma_xv}
and \eqref{eq:sigma_xx}, respectively. First, from \eqref{eq:sigma_xx}
we see that $\frac{1}{T}\sum_{t=1}^{T}\mathbb{E}(x_{i,t}^{2})$ converges
to $\sigma_{xx}$ in an exponential rate and is thus faster than any
polynomial rate. We have
\begin{align}
\min_{1\leq i\leq n}\frac{1}{T}\sum_{t=1}^{T}(x_{i,t}-\bar{x}_{i})^{2} & \geq\mathbb{E}\left[\frac{1}{T}\sum_{t=1}^{T}x_{i,t}^{2}\right]-\max_{1\leq i\leq n}{}\left|\frac{1}{T}\sum_{t=1}^{T}\left[x_{i,t}^{2}-\mathbb{E}(x_{i,t}^{2})\right]\right|-\left(\max_{1\leq i\leq n}|\bar{x}_{i}|\right)^{2}\nonumber \\
 & =\sigma_{xx}+o(T^{-\eta})-O_{p}\biggl(\sqrt{\frac{\log n}{T}}\biggr)-O_{p}\biggl(\frac{\log n}{T}\biggr)=\sigma_{xx}-O_{p}\biggl(\sqrt{\frac{\log n}{T}}\biggr).\label{eq:min_x2}
\end{align}
for arbitrary $\eta>0$, where we use Lemma~B.2 and Proposition~B.3
of \citet{mei2022lasso} to obtain
\[
\max_{1\leq i\leq n}{}\left|\frac{1}{T}\sum_{t=1}^{T}\left[x_{i,t}^{2}-\mathbb{E}(x_{i,t}^{2})\right]\right|=O_{p}\biggl(\sqrt{\frac{\log n}{T}}\biggr)\quad\text{and}\quad\max_{1\leq i\leq n}|\bar{x}_{i}|=O_{p}\biggl(\sqrt{\frac{\log n}{T}}\biggr).
\]
Similarly, we can deduce
\begin{align}
\max_{1\leq i\leq n}\frac{1}{T}\sum_{t=1}^{T}(x_{i,t}-\bar{x}_{i})^{2} & \leq\mathbb{E}\left[\frac{1}{T}\sum_{t=1}^{T}x_{i,t}^{2}\right]+\max_{1\leq i\leq n}{}\left|\frac{1}{T}\sum_{t=1}^{T}\left[x_{i,t}^{2}-\mathbb{E}(x_{i,t}^{2})\right]\right|\label{eq:max_x2}\\
 & =\sigma_{xx}+O_{p}\biggl(\sqrt{\frac{\log n}{T}}\biggr).\nonumber
\end{align}
Equations~\eqref{eq:min_x2} and \eqref{eq:max_x2} yields
\begin{equation}
\max_{1\leq i\leq n}{}\left|\frac{1}{T}\sum_{t=1}^{T}(x_{i,t}-\bar{x}_{i})^{2}-\sigma_{xx}\right|=O_{p}\biggl(\sqrt{\frac{\log n}{T}}\biggr).\label{eq:max_x2_abs}
\end{equation}
Moreover, by Proposition~B.3 of \citet{mei2022lasso} we also have
\begin{equation}
\max_{1\leq i\leq n}{}\left|\frac{1}{T}\sum_{t=1}^{T}[x_{i,t}v_{i,t+1}-\mathbb{E}(x_{i,t}v_{i,t+1})]\right|=O_{p}\biggl(\sqrt{\frac{\log n}{T}}\biggr).\label{eq:max_xv}
\end{equation}
Now we are ready to analyze the three terms in \eqref{eq:rho_hat}.
For the first term, by \eqref{eq:min_x2} and \eqref{eq:max_xv},
we have
\begin{align}
 & \max_{1\leq i\leq n}{}\left|\frac{T^{-1}\sum_{t=1}^{T}[x_{i,t}v_{i,t+1}-\mathbb{E}(x_{i,t}v_{i,t+1})]}{T^{-1}\sum_{t=1}^{T}(x_{i,t}-\bar{x}_{i})^{2}}\right|\nonumber \\
 & \leq\frac{{\displaystyle \max_{1\leq i\leq n}{}\left|\frac{1}{T}\sum_{t=1}^{T}[x_{i,t}v_{i,t+1}-\mathbb{E}(x_{i,t}v_{i,t+1})]\right|}}{{\displaystyle \min_{1\leq i\leq n}{}\frac{1}{T}\sum_{t=1}^{T}(x_{i,t}-\bar{x}_{i})^{2}}}=O_{p}\biggl(\sqrt{\frac{\log n}{T}}\biggr).\label{eq:xv_bound}
\end{align}
For the second term, by \eqref{eq:max_x2_abs} and the fact that $T^{-1}\sum_{t=1}^{T}\mathbb{E}(x_{i,t}v_{i,t+1})$
converges to $\sigma_{xv}$ in an exponential rate by \eqref{eq:sigma_xv},
we have
\begin{align*}
 & \max_{1\leq i\leq n}{}\left|\frac{T^{-1}\sum_{t=1}^{T}\mathbb{E}(x_{i,t}v_{i,t+1})-\sigma_{xv}\sigma_{xx}^{-1}T^{-1}\sum_{t=1}^{T}(x_{i,t}-\bar{x}_{i})^{2}}{T^{-1}\sum_{t=1}^{T}(x_{i,t}-\bar{x}_{i})^{2}}\right|\\
\leq{} & \max_{1\leq i\leq n}{}\left|\frac{T^{-1}\sum_{t=1}^{T}\mathbb{E}(x_{i,t}v_{i,t+1})-\sigma_{xv}}{T^{-1}\sum_{t=1}^{T}(x_{i,t}-\bar{x}_{i})^{2}}\right|+\left|\frac{\sigma_{xv}}{\sigma_{xx}}\right|\max_{1\leq i\leq n}{}\left|\frac{T^{-1}\sum_{t=1}^{T}(x_{i,t}-\bar{x}_{i})^{2}-\sigma_{xx}}{T^{-1}\sum_{t=1}^{T}(x_{i,t}-\bar{x}_{i})^{2}}\right|\\
\leq{} & \frac{{\displaystyle \left|\frac{1}{T}\sum_{t=1}^{T}\mathbb{E}(x_{i,t}v_{i,t+1})-\sigma_{xv}\right|}}{{\displaystyle \min_{1\leq i\leq n}{}\frac{1}{T}\sum_{t=1}^{T}(x_{i,t}-\bar{x}_{i})^{2}}}+\left|\frac{\sigma_{xv}}{\sigma_{xx}}\right|\frac{{\displaystyle \max_{1\leq i\leq n}{}\left|\frac{1}{T}\sum_{t=1}^{T}(x_{i,t}-\bar{x}_{i})^{2}-\sigma_{xx}\right|}}{{\displaystyle \min_{1\leq i\leq n}{}\frac{1}{T}\sum_{t=1}^{T}(x_{i,t}-\bar{x}_{i})^{2}}}\\
= & o(T^{-\eta})+O_{p}\biggl(\sqrt{\frac{\log n}{T}}\biggr)=O_{p}\biggl(\sqrt{\frac{\log n}{T}}\biggr).
\end{align*}
For the third term, by \eqref{eq:min_x2} we have
\begin{equation}
\max_{1\leq i\leq n}{}\left|\frac{\bar{x}_{i}T^{-1}\sum_{t=1}^{T}x_{i,t+1}}{T^{-1}\sum_{t=1}^{T}(x_{i,t}-\bar{x}_{i})^{2}}\right|\leq\frac{{\displaystyle \max_{1\leq i\leq n}{}|\bar{x}_{i}|\max_{1\leq i\leq n}{}\left|\frac{1}{T}\sum_{t=1}^{T}x_{i,t+1}\right|}}{{\displaystyle \min_{1\leq i\leq n}{}\frac{1}{T}\sum_{t=1}^{T}(x_{i,t}-\bar{x}_{i})^{2}}}=O_{p}\biggl(\frac{\log n}{T}\biggr).\label{eq:xbarx_bound}
\end{equation}
It follows that
\[
\max_{1\leq i\leq n}{}\left|\hat{\rho}_{i}-R(\rho^{*})\right|=O_{p}\biggl(\sqrt{\frac{\log n}{T}}\biggr).
\]

For $\hat{\beta}_{i}$, we have the decomposition
\[
\hat{\beta}_{i}-\beta^{*}=\frac{\sum_{t=1}^{T}x_{i,t}e_{i,t+1}}{\sum_{t=1}^{T}(x_{i,t}-\bar{x}_{i})^{2}}+\frac{\bar{x}_{i}\sum_{t=1}^{T}x_{i,t+1}}{\sum_{t=1}^{T}(x_{i,t}-\bar{x}_{i})^{2}}.
\]
Similar to \eqref{eq:xv_bound} and \eqref{eq:xv_bound}, $\max_{1\leq i\leq n}|\hat{\beta}_{i}-\beta_{i}^{*}|=O_{p}\bigl(\sqrt{\log(n)/T}\bigr)$
follows.
\end{proof}
\begin{proof}[Proof of \prettyref{lem:sup_rho_ur}]
Write $\hat{\rho}_{i}$ as
\[
\hat{\rho}_{i}-1=\frac{\sum_{t=1}^{T}(x_{i,t}-\bar{x}_{i})v_{i,t+1}}{\sum_{t=1}^{T}(x_{i,t}-\bar{x}_{i})^{2}}.
\]
For the numerator, by Proposition~1 of \citet{mei2022lasso},
\[
\max_{1\leq i\leq n}\left|\frac{1}{T}\sum_{t=1}^{T}(x_{i,t}-\bar{x}_{i})v_{i,t+1}\right|=O_{p}\bigl((\log n)^{3/2}\bigr).
\]
For the denominator, by Proposition~B.4(b) of \citet{mei2022lasso},
\[
\min_{1\leq i\leq n}{}\frac{1}{T}\sum_{t=1}^{T}(x_{i,t}-\bar{x}_{i})^{2}\gtrsim_{p}\frac{T}{\log n}.
\]
It thus follows that
\[
\max_{1\leq i\leq n}{}\left|\hat{\rho}_{i}-1\right|\leq\frac{{\displaystyle \max_{1\leq i\leq n}{}\left|\frac{1}{T}\sum_{t=1}^{T}(x_{i,t}-\bar{x}_{i})v_{i,t+1}\right|}}{{\displaystyle \min_{1\leq i\leq n}{}\frac{1}{T}\sum_{t=1}^{T}(x_{i,t}-\bar{x}_{i})^{2}}}=O_{p}\biggl(\frac{(\log n)^{5/2}}{T}\biggr).
\]
This completes the proof.
\end{proof}
\begin{proof}[Proof of \prettyref{prop:classify}]
By \prettyref{lem:sup_rho_hat_bound} and \ref{lem:sup_rho_ur},
we have
\[
\max_{1\leq i\leq n}{}\left|\hat{\rho}_{i}-r_{i}^{*}\right|=O_{p}\biggl(\sqrt{\frac{\log n}{T}}\biggr)=o_{p}(1),
\]
where $r_{i}^{*}=R(\rho_{i}^{*})$ if $\rho_{i}^{*}\in(-1,1)$ and
$r_{i}^{*}=1$ if $\rho_{i}^{*}=1$. Moreover, \prettyref{lem:sup_rho_hat_bound}
also indicates that for any $i,j$ not in the same group $|r_{i}^{*}-r_{j}^{*}|>\underline{c}_{r}$
for some absolute constant $\underline{c}_{r}$ since $R(\cdot)$
is strictly increasing. Similary, we can deduce $\max_{1\leq i\leq n}|\hat{\beta}_{i}-\beta_{i}^{*}|=o_{p}(1)$
as well. Then, the proof strategy of Theorem~3.2 of \citet{wang2021identifying}
goes through here.
\end{proof}
\begin{delayedproof}{thm:divx_under_group}
The theorem follows by \prettyref{prop:classify} and using the argument
in the \proofref{thm:DIVX}.
\end{delayedproof}

\section{Supporting Lemmas \label{sec:Supporting-lemmas}}

Supporting lemmas are collected in this Section. Section~\ref{subsec:Pre-Lemmas}
presents several preliminary lemmas about a generic AR(1) process
that are useful in the proofs of the main results. Section~\ref{subsec:lemmas-for-WG}
includes the lemmas for WG. Section~\ref{subsec:lemmas-for-IVX}
lists the lemmas for IVX. Section~\ref{subsec:stoch_integral} displays
several analytic formulae for moments of stochastic integrals which
appear in \prettyref{lem:lurjoint} as the probability limits of finite
sample moments, and will only be used to verify uniform integrability
in the proof of \prettyref{lem:lurjoint}. Due to space limitation,
the proofs of these technical lemmas are relegated to the Supplementary
Materials.

\subsection{Preliminary Lemmas\label{subsec:Pre-Lemmas}}

\prettyref{lem:u.i.} provides sufficient conditions for uniform integrability
of a generic random sequence, which is important in proving joint
laws of large numbers and CLTs (cf.\ \citealp{phillips1999linear}).
\prettyref{lem:sum_rho} bounds the partial sum of a geometric progression.
\prettyref{lem:generic_AR} and \ref{lem:generic_two_ARs} bound the
moments of generic AR(1) processes.
\begin{lem}
\label{lem:series}Suppose $\{a_{i}\}$ is absolutely summable: $\sum_{i=1}^{\infty}|a_{i}|=A<\infty$,
and $\{w_{n,i}\}$ satisfies $\lim_{n\to\infty}w_{n,i}=w_{i}$ for
each $i$ and $\sup_{n,i}|w_{n,i}|\leq M$ for some $M>0$, then $\lim_{n\to\infty}\sum_{i=1}^{n}w_{n,i}a_{i}=\sum_{i=1}^{\infty}w_{i}a_{i}$.
\end{lem}
Let $\{X_{T}\}_{T=1}^{\infty}$ be a sequence of random variables.
We say $X_{T}$ is \emph{uniformly integrable} (u.i.) in $T$ if
\[
\lim_{B\to\infty}\sup_{T\geq1}\mathbb{E}[|X_{T}|\cdot\bm{1}(|X_{T}|\geq B)]=0,
\]
where $\bm{1}(\cdot)$ is the indicator function for an event.
\begin{lem}
\label{lem:u.i.}The random sequence $\{X_{T}\}_{T=1}^{\infty}$ is
u.i.~if any of the following conditions holds:
\end{lem}
\begin{enumerate}
\item $\sup_{T}\mathbb{E}|X_{T}|^{1+\delta}<\infty$ for some fixed constant
$\delta>0$, or
\begin{enumerate}
\item $X_{T}\geq0$, $\mathbb{E}(X_{T})<\infty$, and there exists some
random variable $X\geq0$ with $\mathbb{E}(X)<\infty$ such that $X_{T}\to_{d}X$
and $\mathbb{E}(X_{T})\to\mathbb{E}(X)$ as $T\to\infty$, or
\item $X_{T}\geq0$, $\mathbb{E}(X_{T})<\infty$, and $\mathbb{E}(X_{T})\to0$
as $T\to\infty$.
\end{enumerate}
\end{enumerate}
\begin{lem}
On the other hand, if $\{X_{T}\}_{T=1}^{\infty}$ is u.i.~and $X_{T}\to_{d}X$,
then $X$ is integrable and
\begin{equation}
\mathbb{E}(X_{T})\to\mathbb{E}(X)\quad\text{as }T\to\infty.\label{eq:limit EXT}
\end{equation}
\end{lem}
\begin{lem}
\label{lem:sum_rho}Let $b\in(0,1)$ and $c>0$ be fixed constants,
and $\rho_{T}\in(-1+b,1+c/T]$ be dependent on $T$. Then, for any
positive integer $m$,
\[
\sum_{t=1}^{T}\rho_{T}^{mt}=O\left(\dfrac{1}{|1-\rho_{T}|}\wedge T\right)\quad\text{as }T\to\infty.
\]
\end{lem}
\begin{lem}
\label{lem:ma_represent}Suppose that $\eta_{t}$ follows $\mathrm{AR}(1)$:
$\eta_{t}=\rho_{\eta,T}\eta_{t-1}+\nu_{\eta,t}$ where $-1+b\leq\rho_{\eta,T}\leq1+c/T$
for some fixed $b\in(0,1)$ and $c>0$, and $\nu_{\eta,t}=\sum_{s=0}^{\infty}g_{s}\varepsilon_{\eta,t-s}$
with $|g_{s}|\leq Cq_{\nu}^{s}$, $q_{\nu}\in(0,1)$ and $C>0$ being
constants. Then the $\mathrm{MA}(\infty)$ representation of $\eta_{t}$
is
\begin{equation}
\eta_{t}=\sum_{\ell=-\infty}^{t}\pi_{\eta,T}(t,\ell)\varepsilon_{\eta,\ell}+\rho_{\eta,T}^{t}\eta_{0},\label{eq:eta_LP_decomp}
\end{equation}
where for $\text{for }\ell=-\infty,\dots,t$,
\[
\pi_{\eta,T}(t,\ell)=\sum_{s=\ell\vee1}^{t}\rho_{\eta,T}^{t-s}g_{s-\ell}.
\]
Furthermore, for $T$ large enough, it holds that
\begin{equation}
|\pi_{\eta,T}(t,\ell)|\lesssim\begin{cases}
(|\rho_{\eta,T}|\vee q_{\nu})^{t}q_{\nu}^{1-\ell} & \ell<1,\\
(|\rho_{\eta,T}|\vee q_{\nu})^{t-\ell} & \ell\geq1.
\end{cases}\label{eq:pi_bound}
\end{equation}
and for any positive integer $k$
\[
\sup_{t\leq T}{}\left|\sum_{\ell=-\infty}^{t}|\pi_{\eta,T}(t,\ell)|^{k}\right|\lesssim\sum_{j=0}^{T}(|\rho_{\eta,T}|\vee q_{\nu})^{kj}=O\left(\dfrac{1}{|1-\rho_{\eta,T}|}\wedge T\right).
\]
\end{lem}
\begin{lem}
\label{lem:generic_AR}Suppose $\eta_{t}$ is a generic $\mathrm{AR}(1)$
process following $\eta_{t}=\rho_{T}\eta_{t-1}+\nu_{t}$ for $t=1,\dots,T$
where
\end{lem}
\begin{enumerate}[label=(\alph*)]
\item $-1+b\leq\rho_{T}\leq1+c/T$ for some fixed $b\in(0,1)$ and $c>0$,
\begin{enumerate}
\item $\mathbb{E}(\eta_{0}^{4})=O(T^{2})$,
\item $\nu_{t}=\sum_{s=0}^{\infty}g_{s}\varepsilon_{t-s}$ with $|g_{s}|\leq Cq_{\nu}^{s}$,
$q_{\nu}\in(0,1)$ and $C>0$ being constants, and $\{\varepsilon_{t}\}$
is a stationary m.d.s.\ with absolutely summable fourth cumulants.
\end{enumerate}
Then, as $T\to\infty$,
\begin{enumerate}
\item \label{enu:E_sum_eta2}$\mathbb{E}\bigl(\sum_{t=1}^{T}\eta_{t}^{2}\bigr)=O\bigl(\frac{T}{|1-\rho_{T}|}\wedge T^{2}\bigr)$,
\item \label{enu:E_sum_eta_4}$\mathbb{E}\bigl[\bigl(\sum_{t=1}^{T}\eta_{t}\bigr)^{4}\bigr]=O\bigl(\frac{T^{2}}{(1-\rho_{T})^{4}}\wedge T^{6}\bigr)$.
\end{enumerate}
\end{enumerate}
\begin{lem}
\label{lem:generic_two_ARs}Suppose that $\xi_{t}$ and $\eta_{t}$
are generic $\mathrm{AR}(1)$ processes following $\xi_{t}=\rho_{\xi,T}\xi_{t-1}+\nu_{\xi,t}$
and $\eta_{t}=\rho_{\eta,T}\eta_{t-1}+\nu_{\eta,t}$ for $t=1,\dots,T$
where
\end{lem}
\begin{enumerate}[label=(\alph*)]
\item  $\rho_{\xi,T},\rho_{\eta,T}\in[-1+b,1+c/T]$ for some fixed $b\in(0,1)$
and $c>0$,
\begin{enumerate}
\item \label{enu:init_vals}$\mathbb{E}(\xi_{0}^{4})=O\bigl((1-\rho_{\xi,T})^{-2}\wedge T^{2}\bigr)$,
$\mathbb{E}(\eta_{0}^{4})=O\bigl((1-\rho_{\eta,T})^{-2}\wedge T^{2}\bigr)$,
and $\mathbb{E}(\xi_{0}\eta_{0})=O\bigl(|1-\rho_{\eta,T}\rho_{\xi,T}|^{-1}\wedge T\bigr)$,
\item \label{enu:error}$\nu_{\xi,t}=\sum_{s=0}^{\infty}g_{\xi,s}\varepsilon_{\eta,t-s}$,
$\nu_{\eta,t}=\sum_{s=0}^{\infty}g_{\eta,s}\varepsilon_{\eta,t-s}$,
where $|g_{\xi,s}|+|g_{\eta,s}|\leq Cq_{0}^{s}$ for some positive
constants $C$ and $q_{0}\in(0,1)$. $\{(\varepsilon_{\xi,t},\varepsilon_{\eta,t})\}$
is a stationary m.d.s.\ with absolutely summable fourth cumulants.
In addition, $\sup_{s\leq0}|\mathbb{E}(\xi_{0}\varepsilon_{\eta,s})|\leq\infty$
and $\sup_{s\leq0}|\mathbb{E}(\eta_{0}\varepsilon_{\xi,s})|\leq\infty$.
\end{enumerate}
Then, as $T\to\infty$,
\begin{enumerate}
\item \label{enu:supE}$\sup_{t\leq T}|\mathbb{E}(\xi_{t}\eta_{t})|=O\bigl(\frac{1}{|1-\rho_{\xi,T}\rho_{\eta,T}|}\wedge T\bigr)$,
\item \label{enu:supE2}$\sup_{t\leq T}\mathbb{E}(\xi_{t}^{2}\eta_{t}^{2})=O\Bigl(\bigl[\frac{1}{|1-\rho_{\xi,T}|}\wedge T\bigr]\cdot\bigr[\frac{1}{|1-\rho_{\eta,T}|}\wedge T\bigr]\Bigr)$,
\item \label{enu:Esum2}$\mathbb{E}\big[(\sum_{t=1}^{T}\xi_{t}\eta_{t})^{2}\bigr]=O\bigl(\frac{T^{2}}{(1-\rho_{\xi,T}\rho_{\eta,T})^{2}}\wedge T^{4}\bigr)$,
\item \label{enu:Esumsum2}$\mathbb{E}\bigl[(\sum_{t=1}^{T}\xi_{t}\sum_{t=1}^{T}\eta_{t})^{2}\bigr]=O\Bigl(\bigl[\frac{T}{(1-\rho_{\xi,T})^{2}}\wedge T^{3}\bigr]\cdot\bigr[\frac{T}{(1-\rho_{\eta,T})^{2}}\wedge T^{3}\bigr]\Bigr)$.
\end{enumerate}
\end{enumerate}

\subsection{Technical Lemmas for WG \label{subsec:lemmas-for-WG}}

For notational convenience, we simply use $T$ instead of $T-1$ to
denote the effective time length. In addition, we assume that the
drift $\alpha_{i}=0$, which is without loss of generality because
the within-group transformation implies that $\xtd_{i,t}=(x_{i,t}-\alpha_{i})-(\bar{x}_{i}-\alpha_{i})$
where $x_{i,t}-\alpha_{i}$ is by design an AR(1) process without
an intercept and its group mean is exactly $\bar{x}_{i}-\alpha_{i}$.
The covariances of $e_{i,t}$ and $v_{i,t}$ are defined as $\omega_{ee}^{*}=\mathbb{E}(e_{i,t}^{2})$,
$\omega_{vv}^{*}=\mathbb{E}(v_{i,t}^{2})$ and $\omega_{ev}^{*}=\mathbb{E}(e_{i,t}v_{i,t})$.
\begin{lem}
\label{lem:epct}Suppose Assumptions~\ref{assump:initval} and \ref{assump:innov}
hold. For for each $i$ and $\gamma\in[0,1]$, as $T\to\infty$ we
have
\end{lem}
\begin{enumerate}
\item \label{enu:E_sum_x_2}$\mathbb{E}\bigl[\bigl(\sum_{t=1}^{T}x_{i,t}\bigr)^{2}\bigr]=O(T^{1+2\gamma})$,
\begin{enumerate}
\item \label{enu:E_sum_x_4}$\mathbb{E}\bigl[\bigl(\sum_{t=1}^{T}x_{i,t}\bigr)^{4}\bigr]=O(T^{2+4\gamma})$,
\item \label{enu:E_sum_x2_2}$\mathbb{E}\bigl[\bigl(\sum_{t=1}^{T}x_{i,t}^{2}\bigr)^{2}\bigr]=O(T^{2+2\gamma})$,
\item \label{enu:E_sum_x_sum_e_2}$\mathbb{E}\bigl[\bigl(\sum_{t=1}^{T}x_{i,t}\sum_{t=1}^{T}e_{i,t+1}\bigr)^{2}\bigr]=O(T^{2(1+\gamma)})$.
\end{enumerate}
\end{enumerate}
Throughout Lemmas~\ref{lem:wg-ts-converge}--\ref{lem:esterror},
Assumptions~\ref{assump:initval} and \ref{assump:iid} are imposed.
To present the asymptotic theory for WG estimator, we first define
a few objects:
\begin{align}
Q_{i,T}^{{\rm WG}} & :=\frac{1}{T^{1+\gamma}}\sum_{t=1}^{T}x_{i,t}^{2},\label{eq:def_Q}\\
R_{i,T}^{{\rm WG}} & :=\Biggl(\frac{1}{T^{1+\frac{\gamma}{2}}}\sum_{t=1}^{T}x_{i,t}\Biggr)^{2},\label{eq:def_R}\\
Z_{i,T}^{{\rm WG}} & :=\frac{1}{T^{\frac{1}{2}(1+\gamma)}}\sum_{t=1}^{T}x_{i,t}e_{i,t+1},\label{eq:def_Z}\\
H_{i,T}^{{\rm WG}} & :=\frac{1}{T^{\frac{1}{2}(3+\gamma)}}\sum_{t=1}^{T}x_{i,t}\sum_{t=1}^{T}e_{i,t+1},\label{eq:def_H}
\end{align}
and
\begin{align}
L_{i,T}^{{\rm WG}} & :=Z_{i,T}^{{\rm WG}}-\left[H_{i,T}^{{\rm WG}}-\mathbb{E}\left(H_{i,T}^{{\rm WG}}\right)\right]=\frac{1}{T^{\frac{1}{2}(1+\gamma)}}\sum_{t=1}^{T}\tilde{x}_{i,t}e_{i,t+1}+\mathbb{E}\left(H_{i,T}^{{\rm WG}}\right).\label{eq:def_L}
\end{align}
Note that $L_{i,T}^{\mathrm{WG}}$ is centered and thus has zero mean.
The expectation of $H_{i,T}^{\mathrm{WG}}$ can be easily deduced
using the DGP formula $x_{i,t}=\sum_{j=1}^{t}\rho^{*t-j}v_{i,j}+\rho^{*t}x_{i,0}$
and the m.d.s.\ assumption:
\begin{align}
\mathbb{E}\left(H_{i,T}^{{\rm WG}}\right) & =\mathbb{E}\Biggl(\frac{1}{T^{\frac{1}{2}(3+\gamma)}}\sum_{t=1}^{T}\sum_{j=1}^{t}\rho^{*t-j}v_{i,j}\sum_{s=1}^{T}e_{i,s+1}\Biggr)=\frac{1}{T^{\frac{1}{2}(3+\gamma)}}\sum_{s=1}^{T}\sum_{j=1}^{T}\Biggl(\sum_{t=j}^{T}\rho^{*t-j}\Biggr)\mathbb{E}(v_{i,j}e_{i,s+1})\nonumber \\
 & =\frac{\omega_{ev}^{*}}{T^{\frac{1}{2}(3+\gamma)}}\sum_{j=2}^{T}\Biggl(\sum_{t=j}^{T}\rho^{*t-j}\Biggr)=\frac{\omega_{ev}^{*}}{T^{\frac{1}{2}(3+\gamma)}}\sum_{t=2}^{T}\sum_{j=2}^{t}\rho^{*t-j},\label{eq:E of H}
\end{align}
where the third equality applies the fact that $\mathbb{E}(v_{i,j}e_{i,s+1})=\omega_{ev}^{*}$
only if $j=s+1$ and 0 otherwise.

\begin{lem}
\label{lem:wg-ts-converge} If $\gamma\in[0,1)$, we have, for each
$i$, as $T\to\infty$,
\end{lem}
\begin{enumerate}
\item \label{enu:Q_wg}$Q_{i,T}^{{\rm WG}}\to_{p}V_{xx}$,
\begin{enumerate}
\item \label{enu:Z_wg}$Z_{i,T}^{{\rm WG}}\to_{d}\mathcal{N}(0,\omega_{ee}^{*}V_{xx})$,
\end{enumerate}
\end{enumerate}
\begin{lem}
where
\begin{equation}
V_{xx}:=\begin{cases}
\omega_{vv}^{*}/(1-\rho^{*2}) & \gamma=0,\\
\omega_{vv}^{*}/(-2c^{*}) & \gamma\in(0,1).
\end{cases}\label{eq:V_xx}
\end{equation}
\end{lem}
\begin{lem}
\label{lem:joint} If $\gamma\in[0,1)$, we have, as $(n,T)\to\infty$,
\end{lem}
\begin{enumerate}
\item $n^{-1}\sum_{i=1}^{n}Q_{i,T}^{{\rm WG}}\to_{p}V_{xx}$,
\begin{enumerate}
\item $n^{-1}\sum_{i=1}^{n}R_{i,T}^{{\rm WG}}\to_{p}0$,
\item $n^{-1/2}\sum_{i=1}^{n}Z_{i,T}^{{\rm WG}}\to_{d}\mathcal{N}(0,\omega_{ee}^{*}V_{xx})$,
\item \label{enu:L_wg_clt}$n^{-1/2}\sum_{i=1}^{n}L_{i,T}^{{\rm WG}}\to_{d}\mathcal{N}(0,\omega_{ee}^{*}V_{xx})$,
where $V_{xx}$ is defined in \eqref{eq:V_xx}.
\end{enumerate}
\end{enumerate}
\begin{lem}
\label{lem:lurjoint} If $\gamma=1$, we have, as $(n,T)\to\infty$,
\end{lem}
\begin{enumerate}
\item $n^{-1}\sum_{i=1}^{n}Q_{i,T}^{{\rm WG}}\to_{p}\Omega_{c^{*}}$,
\begin{enumerate}
\item $n^{-1}\sum_{i=1}^{n}R_{i,T}^{{\rm WG}}\to_{p}\Sigma_{c^{*}}$,
\item \label{enu:L_wg_clt_1}$n^{-1/2}\sum_{i=1}^{n}L_{i,T}^{{\rm WG}}\to_{d}\mathcal{N}(0,\Sigma_{\tilde{x}e})$,
\end{enumerate}
\end{enumerate}
\begin{lem}
where
\begin{align*}
\Omega_{c^{*}} & =\mathbb{E}\left[\int_{0}^{1}J_{2,c^{*}}(r)^{2}\,dr\right],\quad\Sigma_{c^{*}}=\mathbb{E}\left[\biggl(\int_{0}^{1}J_{2,c^{*}}(r)\,dr\biggr)^{2}\right],\\
\Sigma_{\tilde{x}e} & =\mathrm{var}\left[\int_{0}^{1}\biggl(J_{2,c^{*}}(r)-\int_{0}^{1}J_{2,c^{*}}(\tau)\,d\tau\biggr)\,dB_{1}(r)\right].
\end{align*}
The analytic formulae of $\Omega_{c^{*}}$ and $\Sigma_{c^{*}}$ are
respectively given by \eqref{eq:E int J2} and \eqref{eq:E of squared int J}
in \prettyref{lem:expectations}. The analytic formula of $\Sigma_{\tilde{x}e}$
is omitted as it is irrelevant to uniform integrability. Here, we
let $\bm{B}(r):=[B_{1}(r),B_{2}(r)]'$ be a 2-dimensional Brownian
motion with covariances $\mathbb{E}[B_{1}(1)^{2}]=\omega_{ee}^{*}$,
$\mathbb{E}[B_{2}(1)^{2}]=\omega_{vv}^{*}$ and $\mathbb{E}[B_{1}(1)B_{2}(1)]=\omega_{ev}^{*}$,
and define the  functional $J_{2,c^{*}}(r):=\int_{s=0}^{r}e^{(r-s)c^{*}}\,dB_{2}(s)$.
\end{lem}
The following lemma shows the estimation error of $b_{n,T}^{\fe}(\rho)$
(scaled by $\varsigma^{\mathrm{WG}}$) when we plug in a consistent
estimator $\hat{\rho}$ with convergence rate $\hat{\rho}-\rho^{*}=O_{p}\bigl(T^{-\frac{1}{2}(1+\gamma)}\bigr)$.
\begin{lem}
\label{lem:esterror}Suppose $\hat{\rho}-\rho^{*}=O_{p}\bigl(T^{-\frac{1}{2}(1+\gamma)}\bigr)$.
Under \prettyref{assump:initval} and \ref{assump:innov}, the bias
estimation error has the stochastic order
\[
r_{n,T}^{\mathrm{WG}}(\hat{\rho}):=\bigl[b_{n,T}^{{\rm WG}}(\hat{\rho})-b_{n,T}^{{\rm WG}}(\rho^{*})\bigr]\big/\varsigma^{{\rm WG}}=O_{p}\left(\sqrt{\frac{n}{T^{1-3\gamma}}}|\hat{\rho}-\rho^{*}|\right)
\]
as $(n,T)\to\infty$, where $b_{n,T}^{{\rm WG}}(\rho)$ is defined
in \eqref{eq:fe_bias} and $\varsigma^{\mathrm{WG}}$ is defined in
\eqref{eq:fe_se-1}.
\end{lem}

\subsection{Technical Lemmas for IVX\label{subsec:lemmas-for-IVX}}

As we do in Section~\ref{subsec:lemmas-for-WG}, we assume that $\alpha_{i}=0$,
which is without loss of generality because $\alpha_{i}$ is either
eliminated by within-group transformation or by the first differencing
when constructing the IV.

Define
\begin{equation}
\psi_{i,t}=\sum_{j=1}^{t}\rho_{z}^{t-j}x_{i,j-1}\label{eq:psi}
\end{equation}
and $\zeta_{i,t}$ as an AR(1) process such that $\zeta_{i,0}=0$
and
\begin{equation}
\zeta_{i,t}=\rho_{z}\zeta_{i,t-1}+v_{i,t},\qquad t=1,\dots,T.\label{eq:zt def}
\end{equation}
By Equations (13) and (23) in \citet{phillips2009econometric}, the
IV $z_{i,t}$ can be decomposed into one of the following ways:
\begin{equation}
z_{i,t}=\zeta_{i,t}-\left(1-\rho^{*}\right)\psi_{i,t},\label{eq:zeta decom 1}
\end{equation}
or
\begin{equation}
z_{i,t}=x_{i,t}-\rho_{z}^{t}x_{i,0}-\left(1-\rho_{z}\right)\psi_{i,t}.\label{eq:zeta decom 2}
\end{equation}
Throughout Lemmas~\ref{lem:ivxepct}--\ref{lem:ivxjoint}, Assumptions~\ref{assump:initval}
and \ref{assump:innov} are imposed.

\subsubsection{Bounds for Moments}
\begin{lem}
\label{lem:ivxepct}For $\gamma\in[0,1]$, we have, as $T\to\infty$,
\end{lem}
\begin{enumerate}
\item \label{enu:zeta_x}$\mathbb{E}\bigl[\bigl(\sum_{t=1}^{T}\zeta_{i,t}x_{i,t}\bigr)^{2}\bigr]=O\bigl(T^{2[1+(\theta\wedge\gamma)]}\bigr)$,\\
$\mathbb{E}\bigl[\bigl(\sum_{t=1}^{T}\zeta_{i,t}\sum_{t=1}^{T}x_{i,t}\bigr)^{2}\bigr]=O\bigl(T^{2(1+\theta+\gamma)}\bigr)$.
\begin{enumerate}
\item \label{enu:phi_x}$\mathbb{E}\bigl[\bigl(\sum_{t=1}^{T}\psi_{i,t}x_{i,t}\bigr)^{2}\bigr]=O\bigl(T^{2[1+(\theta\wedge\gamma)+\gamma]}\bigr)$,\\
$\mathbb{E}\bigl[\bigl(\sum_{t=1}^{T}\psi_{i,t}\sum_{t=1}^{T}x_{i,t}\bigr)^{2}\bigr]=O\bigl(T^{2(1+\theta+2\gamma)}\bigr)$.
\item \label{enu:zx}$\mathbb{E}\bigl[\bigl(\sum_{t=1}^{T}z_{i,t}x_{i,t}\bigr)^{2}\bigr]=O\bigl(T^{2[1+(\theta\wedge\gamma)]}\bigr)$,\\
$\mathbb{E}\bigl[\bigl(\sum_{t=1}^{T}z_{i,t}\sum_{t=1}^{T}x_{i,t}\bigr)^{2}\bigr]=O\bigl(T^{2(1+\theta+\gamma)}\bigr)$.
\item \label{enu:zeta_e}$\mathbb{E}\bigl[\bigl(\sum_{t=1}^{T}\zeta{}_{i,t}\sum_{t=1}^{T}e_{i,t+1}\bigr)^{2}\bigr]=O\bigl(T^{2(1+\theta)}\bigr)$.
\item \label{enu:phi_e}$\mathbb{E}\bigl[\bigl(\sum_{t=1}^{T}\psi_{i,t}\sum_{t=1}^{T}e_{i,t+1}\bigr)^{2}\bigr]=O\bigl(T^{2(1+\theta+\gamma)}\bigr)$.
\item \label{enu:ze}$\mathbb{E}\bigl[\bigl(\sum_{t=1}^{T}z_{i,t}\sum_{t=1}^{T}e_{i,t+1}\bigr)^{2}\bigr]=O\bigl(T^{2+\theta+(\theta\wedge\gamma)}\bigr)$.
\item \label{enu:z2}$\mathbb{E}\bigl[\bigl(\sum_{t=1}^{T}z_{i,t}^{2}\bigr)^{2}\bigr]=O\bigl(T^{2[1+(\theta\wedge\gamma)]}\bigr)$.
\item \label{enu:z4}$\mathbb{E}\bigl[\bigl(\sum_{t=1}^{T}z_{i,t}\bigr)^{4}\bigr]=O\left(T^{2[1+\theta+(\theta\wedge\gamma)]}\right)$.
\end{enumerate}
\end{enumerate}
\begin{lem}
In addition, we have
\begin{equation}
\mathbb{E}\left(\sum_{t=1}^{T}z_{i,t}\sum_{t=1}^{T}e_{i,t+1}\right)=\sum_{h=0}^{T-2}\Psi_{h,T}(\rho^{*},\rho_{z})\omega_{ev,h}^{*}=O(T^{\theta+\gamma}),\label{eq:E sum zeta sum e}
\end{equation}
where
\[
\Psi_{h,T}(\rho^{*},\rho_{z}):=\sum_{k=h+2}^{T}\frac{\rho_{z}^{T-k+1}-\rho^{*T-k+1}}{\rho_{z}-\rho^{*}}\quad\text{and}\quad\omega_{ev,h}^{*}:=\mathbb{E}(v_{i,t+h}e_{i,t}).
\]
\end{lem}

\subsubsection{Asymptotic Convergence}

To characterize the panel IVX estimator, we define the following quantities:
\begin{align}
Q_{i,T} & :=\frac{1-(\rho^{*}\rho_{z})^{2}}{T}\sum_{t=1}^{T}z_{i,t}x_{i,t},\label{eq:def Q IVX}\\
R_{i,T} & :=\frac{1-(\rho^{*}\rho_{z})^{2}}{T^{2}}\sum_{t=1}^{T}z_{i,t}\sum_{t=1}^{T}x_{i,t},\label{eq:def R IVX}\\
S_{i,T} & :=\frac{1-(\rho^{*}\rho_{z})^{2}}{T}\sum_{t=1}^{T}z_{i,t}^{2},\label{eq:def S IVX}\\
Z_{i,T} & :=\sqrt{\frac{1-(\rho^{*}\rho_{z})^{2}}{T}}\sum_{t=1}^{T}z_{i,t}e_{i,t+1},\label{eq:def Z IVX}\\
H_{i,T} & :=\sqrt{\frac{1-(\rho^{*}\rho_{z})^{2}}{T}}\cdot\dfrac{1}{T}\sum_{t=1}^{T}z_{i,t}\sum_{t=1}^{T}e_{i,t+1},\label{eq:def H IVX}\\
L_{i,T} & :=\sqrt{\frac{1-(\rho^{*}\rho_{z})^{2}}{T}}\sum_{t=1}^{T}\tilde{z}_{i,t}e_{i,t+1}=Z_{i,T}-H_{i,T}+\mathbb{E}(H_{i,T}).\label{eq:def L IVX}
\end{align}
For convenience, denote the long-run variances of innovations $e_{i,t}$
and $v_{i,t}$ as
\[
\omega_{ee}^{*}=\mathbb{E}\left(e_{i,t}^{2}\right),\quad\omega_{vv}^{*}=\sum_{\ell=-\infty}^{\infty}\Gamma_{vv}(\ell),\quad\Gamma_{vv}(\ell)=\mathbb{E}(v_{t}v_{t-\ell}).
\]
Note that $e_{i,t}$ is m.d.s.~and thus its long-run variance is
the vanilla unconditional variance.

Lemmas~\ref{lem:ivx_marginal} and \ref{lem:ivxjoint} present asymptotic
theory (as $T\to\infty$ and as $(n,T)\to\infty$, respectively) for
these quantities.
\begin{lem}
\label{lem:ivx_marginal}For each individual $i$, as $T\to\infty$,
\end{lem}
\begin{enumerate}
\item \label{enu:Q_dto}$Q_{i,T}\to_{d}Q_{zx}$, where
\[
Q_{zx}=\begin{cases}
2\bigl(\int_{0}^{1}J_{v,c^{*}}\,dB_{v}+\omega_{vv}^{*}+c^{*}\int_{0}^{1}J_{v,c^{*}}^{2}\bigr) & \gamma=1,\\
\omega_{vv}^{*} & \gamma\in[0,1).
\end{cases}
\]

\begin{enumerate}
\item \label{enu:R_dto}$R_{i,T}\to_{d}R_{zx}$, where
\[
R_{zx}=\begin{cases}
2\left[B_{v}(1)\int_{0}^{1}J_{v,c^{*}}+\bigl(\int_{0}^{1}J_{v,c^{*}}\bigr)^{2}\right] & \gamma=1,\\
0 & \gamma\in[0,1).
\end{cases}
\]
\item \label{enu:S_pto}$S_{i,T}\to_{p}\omega_{vv}^{*}$.
\item \label{enu:Z_dto}$Z_{i,T}\to_{d}\mathcal{N}(0,{\rm S}_{xe})$, where
\[
{\rm S}_{xe}=\begin{cases}
(1-\rho^{*2}){\rm S}_{0,xe} & \gamma=0,\\
\omega_{ee}^{*}\omega_{vv}^{*} & \gamma\in(0,1],
\end{cases}
\]
 with ${\rm S}_{0,xe}=\mathbb{E}(x_{i,t}^{2}e_{i,t+1}^{2})$.
\end{enumerate}
\end{enumerate}
\begin{lem}
Here $Q_{i,T}$, $R_{i,T}$, $S_{i,T}$, and $Z_{i,T}$ are defined
in \eqref{eq:def Q IVX}, \eqref{eq:def R IVX}, \eqref{eq:def S IVX},
and \eqref{eq:def Z IVX}, respectively.
\end{lem}
\begin{lem}
\label{lem:ivxjoint}As $(n,T)\to\infty$ we have,
\end{lem}
\begin{enumerate}
\item \label{enu:Q_pto}$n^{-1}\sum_{i=1}^{n}Q_{i,T}\to_{p}\mathbb{E}(Q_{zx})$,
\begin{enumerate}
\item \label{enu:R_pto}$n^{-1}\sum_{i=1}^{n}R_{i,T}\to_{p}\mathbb{E}(R_{zx})$,
\item \label{enu:S_pto_joint}$n^{-1}\sum_{i=1}^{n}S_{i,T}\to_{p}\omega_{vv}^{*}$,
\item \label{enu:L_dto}$n^{-1/2}\sum_{i=1}^{n}L_{i,T}\to_{d}\mathcal{N}(0,{\rm S}_{xe})$.
\end{enumerate}
\end{enumerate}
\begin{lem}
Here $Q_{i,T}$, $R_{i,T}$, $S_{i,T}$, and $L_{i,T}$ are defined
in \eqref{eq:def Q IVX}, \eqref{eq:def R IVX}, \eqref{eq:def S IVX},
and \eqref{eq:def L IVX}, respectively.
\end{lem}
The next lemma shows the estimation error for the (co)variances $\omega_{ev,h}^{*}$,
$\omega_{ee}^{*}$ and $\omega_{vv}^{*}$ when generic estimates $\hat{\beta}$
and $\hat{\rho}$ are plugged in. Define
\[
\hat{\omega}_{ev,h}(\hat{\beta},\hat{\rho})=\dfrac{1}{n(T-h-1)}\sum_{i=1}^{n}\sum_{t=1}^{T-h-1}(\tilde{y}_{i,t+h+1}-\hat{\beta}\tilde{x}_{i,t+h})(x_{i,t+1}-\hat{\rho}\tilde{x}_{i,t}),
\]
\[
\hat{\omega}_{ee}(\hat{\beta})=\dfrac{1}{n(T-h-1)}\sum_{i=1}^{n}\sum_{t=1}^{T-1}(\tilde{y}_{i,t+1}-\hat{\beta}\tilde{x}_{i,t})^{2},
\]
\[
\hat{\omega}_{vv}(\hat{\rho})=\dfrac{1}{n(T-h-1)}\sum_{i=1}^{n}\sum_{t=1}^{T-1}(x_{i,t+1}-\hat{\rho}\tilde{x}_{i,t})^{2}.
\]

\begin{lem}
\label{lem:Omega_hat} Under Assumptions~\ref{assump:initval} and
\ref{assump:innov}, we have, for $G=o(T)$, as $(n,T)\to\infty$,
\begin{align}
\sum_{h=1}^{G}|\hat{\omega}_{ev,h}(\hat{\beta},\hat{\rho})-\omega_{ev,h}^{*}|^{2} & =O_{p}\left(G\left[\frac{1}{\sqrt{nT}}+|\hat{\rho}-\rho^{*}|+|\hat{\beta}-\beta^{*}|+T^{\gamma}|\hat{\rho}-\rho^{*}||\hat{\beta}-\beta^{*}|\right]^{2}\right),\label{eq:hat omg 12 rate}\\
|\hat{\omega}_{ee}(\hat{\beta})-\omega_{ee}^{*}| & =O_{p}\left(\dfrac{1}{\sqrt{nT}}+|\hat{\beta}-\beta^{*}|+T^{\gamma}|\hat{\beta}-\beta^{*}|^{2}\right),\label{eq:hat omg 11 rate}\\
|\hat{\omega}_{vv}(\hat{\rho})-\omega_{vv}^{*}| & =O_{p}\left(\dfrac{1}{\sqrt{nT}}+|\hat{\rho}-\rho^{*}|+T^{\gamma}|\hat{\rho}-\rho^{*}|^{2}\right),\label{eq:hat omg 22 rate}
\end{align}
\end{lem}
An immediate corollary to the above lemma is the following, where
$\hat{\beta}-\beta^{*}=O_{p}\bigl((nT)^{-1/2}+T^{-1}\bigr)$ and $\hat{\rho}-\rho^{*}=O_{p}\bigl((nT)^{-1/2}+T^{-1}\bigr)$;
these rates are satisfied by WG.
\begin{cor}
\label{cor:omega_mixed}Suppose Assumptions~\ref{assump:initval}
and \ref{assump:innov} hold. If $\hat{\beta}-\beta^{*}=O_{p}\bigl((nT)^{-1/2}+T^{-1}\bigr)$
and $\hat{\rho}-\rho^{*}=O_{p}\bigl((nT)^{-1/2}+T^{-1}\bigr)$. Then
for all $\gamma\in[0,1]$, we have
\[
\sum_{h=1}^{G}|\hat{\omega}_{ev,h}(\hat{\beta},\hat{\rho})-\omega_{ev,h}^{*}|^{2}=O_{p}\left(\dfrac{G}{nT}+\frac{G}{T^{2}}\right).
\]

\end{cor}
The following lemma shows the estimation error for $b_{n,T}^{\mathrm{IVX}}(\{\omega_{ev,h}^{*}\},\rho^{*},\rho_{z})$
when we plug in a consistent estimator $\hat{\rho}$ with convergence
rate $\hat{\rho}-\rho^{*}=O_{p}\bigl(T^{-\eta}\bigr)$ and $\hat{\omega}_{h}$
with convergence rate given by \prettyref{cor:omega_mixed}. An estimator
for $b_{n,T}^{\mathrm{IVX}}(\rho)$ is given by
\[
\hat{b}_{n,T}^{\mathrm{IVX}}(\{\hat{\omega}_{ev,h}\},\hat{\rho},\rho_{z})=\frac{n\sum_{h=0}^{G}\Psi_{h,T}(\hat{\rho},\rho_{z})\hat{\omega}_{ev,h}}{T\sum_{i=1}^{n}\sum_{t=1}^{T}\tilde{z}_{i,t}x_{i,t}}.
\]

\begin{lem}
\label{lem:ivxerror} Suppose Assumptions~\ref{assump:initval} and
\ref{assump:innov} hold. If $\hat{\rho}-\rho^{*}=O_{p}\bigl(T^{-\eta}\bigr)$
for some $\eta\geq0$ and $\hat{\rho}\leq1+O_{p}(T^{-1})$ and $\hat{\omega}_{ev,h}$
satisfies the rate in \prettyref{cor:omega_mixed}, then as $(n,T)\to\infty$
we have
\begin{align*}
 & \hat{b}_{n,T}^{\mathrm{IVX}}(\{\hat{\omega}_{ev,h}\},\hat{\rho},\rho_{z})-b_{n,T}^{\mathrm{IVX}}(\{\omega_{ev,h}^{*}\},\rho^{*},\rho_{z})\\
 & \quad=O_{p}\left(\dfrac{G}{\sqrt{nT^{5-2(\theta\vee\gamma)}}}+\frac{G}{T^{3-(\theta\vee\gamma)}}+\frac{|\hat{\rho}-\rho^{*}|}{T^{2-(\theta\vee\gamma)-\gamma}}+\frac{q_{\nu}^{G}}{T^{2-(\theta\vee\gamma)}}\right),
\end{align*}
where $b_{n,T}^{\mathrm{IVX}}(\{\omega_{ev,h}^{*}\},\rho^{*},\rho_{z})$
is given by \eqref{eq:ivx_bias}.
\end{lem}

\subsection{Moments of Stochastic Integrals\label{subsec:stoch_integral}}

In this subsection, we present a few analytic formulae regarding moments
of stochastic integrals. They appear in \prettyref{lem:lurjoint}
as the probability limits of finite sample moments, and are needed
for verifying  uniform integrability (by \prettyref{lem:u.i.}) in
the proof of \prettyref{lem:lurjoint}.
\begin{lem}
\label{lem:expectations}Let $\bm{B}(r):=[B_{1}(r),B_{2}(r)]'$ be
a 2-dimensional Brownian motion with covariances $\mathbb{E}[B_{1}(1)^{2}]=\omega_{ee}^{*}$,
$\mathbb{E}[B_{2}(1)^{2}]=\omega_{vv}^{*}$ and $\mathbb{E}[B_{1}(1)B_{2}(1)]=\omega_{ev}^{*}$.
Define a functional $J_{2,c^{*}}(r):=\int_{s=0}^{r}e^{c^{*}(r-s)}\,dB_{2}(s)$
where $c^{*}\in\mathbb{R}$ is a constant. We have the following analytic
formulae:
\begin{align}
\mathbb{E}\left[\int_{0}^{1}J_{2,c^{*}}(r)^{2}\,dr\right] & =\omega_{vv}^{*}\frac{e^{2c^{*}}-2c^{*}-1}{4c^{*2}},\label{eq:E int J2}\\
\mathbb{E}\left[\biggl(\int_{0}^{1}J_{2,c^{*}}(r)\,dr\biggr)^{2}\right] & =\omega_{vv}^{*}\frac{2c^{*}+(1-e^{c^{*}})(3-e^{c^{*}})}{2c^{*3}},\label{eq:E of squared int J}\\
\mathbb{E}\left[\biggl(B_{1}(1)\int_{0}^{1}J_{2,c^{*}}(r)\,dr\biggr)^{2}\right] & =\omega_{ee}^{*}\omega_{vv}^{*}\frac{2c^{*}+(1-e^{c^{*}})(3-e^{c^{*}})}{2c^{*3}}+2\omega_{ev}^{*2}\frac{(e^{c^{*}}-c^{*}-1)^{2}}{c^{*4}}.\label{eq:E of squared B_int_J}
\end{align}
If $c^{*}=0$, then the expressions are obtained by taking limits
of the right-hand side formulae as $c^{*}\to0$, and we have
\begin{align*}
\mathbb{E}\left[\int_{0}^{1}B_{2}(r)^{2}\,dr\right] & =\frac{1}{2}\omega_{vv}^{*},\quad\mathbb{E}\left[\biggl(\int_{0}^{1}B_{2}(r)\,dr\biggr)^{2}\right]=\frac{1}{3}\omega_{vv}^{*},\\
\mathbb{E}\left[\biggl(B_{1}(1)\int_{0}^{1}B_{2}(r)\,d\tau\biggr)^{2}\right] & =\frac{1}{3}\omega_{ee}^{*}\omega_{vv}^{*}+\frac{1}{2}\omega_{ev}^{*2}.
\end{align*}
\end{lem}

\section{Proofs of Preliminary Lemmas\label{sec:Proof-Pre-Lemmas}}
\begin{proof}[Proof of \prettyref{lem:series}]
 Denote $S_{n}:=\sum_{i=1}^{n}w_{n,i}a_{i}$. We have
\[
\left|S_{n}-\sum_{i=1}^{\infty}w_{i}a_{i}\right|\leq\sum_{i=1}^{n}|w_{n,i}-w_{i}||a_{i}|+\sum_{i=n+1}^{\infty}|w_{i}||a_{i}|.
\]
Since $\sup_{n,i}|w_{n,i}|\leq M$, we have $\sup_{i}|w_{i}|\leq M$.
Given $\varepsilon>0$, there is some $N_{1}>0$ so that we can choose
$n>N_{1}$ to achieve $\sum_{i=n+1}^{\infty}|w_{i}||a_{i}|\leq M\varepsilon$.
For the first term on the right-hand side, we write
\[
\sum_{i=1}^{n}|w_{n,i}-w_{i}||a_{i}|=\sum_{i=1}^{N_{2}}|w_{n,i}-w_{i}||a_{i}|+\sum_{i=N_{2}+1}^{n}|w_{n,i}-w_{i}||a_{i}|.
\]
We can choose $N_{2}$ large enough so that $\sum_{i=N_{2}+1}^{\infty}|w_{i}||a_{i}|\leq M\varepsilon$
and thus $\sum_{i=N_{2}+1}^{n}|w_{n,i}-w_{i}||a_{i}|\leq2M\varepsilon$.
Moreover, for each $i\leq N_{2}$, there exists $N_{3}(i)$ such that
for all $n>N_{3}(i)$ we have $|w_{n,i}-w_{i}|\leq\varepsilon$. We
can choose $N_{3}=\max_{i\leq N_{2}}N_{3}(i)$, then for $n>N_{3}$
it holds that $\sum_{i=1}^{N_{2}}|w_{n,i}-w_{i}||a_{i}|\leq A\varepsilon$.
Ultimately, we choose $N=\max(N_{1},N_{2},N_{3})$, then for $n\geq N$,
we have
\[
\left|S_{n}-\sum_{i=1}^{\infty}w_{i}a_{i}\right|\leq(A+3M)\varepsilon.
\]
Since $\varepsilon$ is arbitrary and $A$ and $M$ are absolute constants,
this shows $\lim_{n\to\infty}S_{n}=\sum_{i=1}^{\infty}w_{i}a_{i}$.
\end{proof}
\begin{proof}[Proof of \prettyref{lem:u.i.}]
The result can be found in classical probability books, for instance,
\citet{billingsley1968convergence}. Part (i) follows from the argument
between Equations (5.1) and (5.2) in \citet[p.~32]{billingsley1968convergence};
(ii) and \eqref{eq:limit EXT} come from Theorem~5.4 in \citet{billingsley1968convergence};
and (iii) is a special case of (ii) because $\epct(X_{T})\to0$ together
with nonnegativity of $X_{T}$ implies $X_{T}\pto0$ by Markov's inequality.
\end{proof}
\begin{proof}[Proof of \prettyref{lem:sum_rho}]
Without loss of generality, we assume $\rho_{T}\ge0$ in this proof
for simplicity of exposition. When $\rho_{T}\leq1-T^{-1},$ it follows
\begin{align*}
\sum_{t=1}^{T}\rho_{T}^{mt} & =\dfrac{\rho_{T}^{m}(1-\rho_{T}^{mT})}{1-\rho_{T}^{m}}\leq\dfrac{1}{1-\rho_{T}}=O\biggl(\dfrac{1}{1-\rho_{T}}\wedge T\biggr).
\end{align*}
On the other hand, when $\rho_{T}>1-T^{-1}$, we have
\begin{align*}
\sum_{t=1}^{T}\rho_{T}^{mt} & \leq\sum_{t=1}^{T}\left(1+\frac{c}{T}\right){}^{mt}=\left(1+\frac{c}{T}\right)^{m}\dfrac{(1+\frac{c}{T})^{mT}-1}{(1+\frac{c}{T})^{m}-1}\leq\left(1+\frac{c}{T}\right)^{m}\dfrac{(1+\frac{c}{T})^{mT}-1}{c/T}\\
 & =O(T)=O\left(\dfrac{1}{|1-\rho_{T}|}\wedge T\right).
\end{align*}
This completes the proof.
\end{proof}
\begin{proof}[Proof of \prettyref{lem:ma_represent}]
Note that $\eta_{t}$ can be expressed as
\begin{align*}
\eta_{t} & =\sum_{s=1}^{t}\rho_{\eta,T}^{t-s}v_{\eta,s}+\rho_{\eta,T}^{t}\eta_{0}\\
 & =\sum_{s=1}^{t}\rho_{\eta,T}^{t-s}\sum_{\tau=0}^{\infty}g_{\tau}\varepsilon_{\eta,s-\tau}+\rho_{\eta,T}^{t}\eta_{0}\\
 & =\sum_{\ell=-\infty}^{t}\sum_{s=\ell\vee1}^{t}\rho_{\eta,T}^{t-s}g_{s-\ell}\varepsilon_{\eta,\ell}+\rho_{\eta,T}^{t}\eta_{0}\\
 & =:\sum_{\ell=-\infty}^{t}\pi_{\eta,T}(t,\ell)\varepsilon_{\eta,\ell}+\rho_{\eta,T}^{t}\eta_{0},
\end{align*}
Under the condition $|g_{s}|\leq C_{g}q_{\nu}^{s}$, we have for $\ell<1$:
\[
|\pi_{\eta,T}(t,\ell)|\leq\sum_{s=1}^{t}|\rho_{\eta,T}|^{t-s}|g_{s-\ell}|\leq C_{g}\frac{|\rho_{\eta,T}|^{t}q_{\nu}^{1-\ell}-q_{\nu}^{t-\ell+1}}{|\rho_{\eta,T}|-q_{\nu}}\lesssim(|\rho_{\eta,T}|\vee q_{\nu})^{t}q_{\nu}^{1-\ell}
\]
and for $\ell\geq1$:
\[
|\pi_{\eta,T}(t,\ell)|\leq\sum_{s=\ell}^{t}|\rho_{\eta,T}|^{t-s}|g_{s-\ell}|\leq C_{g}\frac{|\rho_{\eta,T}|^{t-\ell+1}-q_{\nu}^{t-\ell+1}}{|\rho_{\eta,T}|-q_{\nu}}\lesssim(|\rho_{\eta,T}|\vee q_{\nu})^{t-\ell}.
\]
For any positive integer $k$, it holds that
\begin{align*}
\left|\sum_{\ell=-\infty}^{t}|\pi_{\eta,T}(t,\ell)|^{k}\right| & \lesssim\sum_{\ell=-\infty}^{0}[(|\rho_{\eta,T}|\vee q_{\nu})^{t}q_{\nu}^{1-\ell}]^{k}+\sum_{\ell=1}^{t}[(|\rho_{\eta,T}|\vee q_{\nu})^{t-\ell}]^{k}\\
 & \lesssim\sum_{\ell=0}^{t}(|\rho_{\eta,T}|\vee q_{\nu})^{k(t-\ell)}=\sum_{j=0}^{t}(|\rho_{\eta,T}|\vee q_{\nu})^{kj}.
\end{align*}
It follows that
\begin{align*}
\sup_{t\leq T}\left|\sum_{\ell=-\infty}^{t}|\pi_{\eta,T}(t,\ell)|^{k}\right| & \lesssim\sum_{j=0}^{T}(|\rho_{\eta,T}|\vee q_{\nu})^{kj}\\
 & =O\left(\dfrac{1}{|1-(|\rho_{\eta,T}|\vee q_{\nu})|}\wedge T\right)=O\left(\dfrac{1}{|1-\rho_{\eta,T}|}\wedge T\right),
\end{align*}
where the last equality is based on this fact: if $\rho_{\eta,T}<q$,
then $|1-q|_{\nu}^{-1}=O(|1-\rho_{\eta,T}|^{-1})$.
\end{proof}
\begin{proof}[Proof of \prettyref{lem:generic_AR}]
Similar to the proof of \prettyref{lem:sum_rho}, we assume $\rho_{T}\ge0$
here for simplicity. Part~\ref{enu:E_sum_eta2}. Since $\epsilon_{t}$
is an m.d.s., using \prettyref{lem:ma_represent} we have
\begin{align*}
\mathbb{E}\Biggl(\sum_{t=1}^{T}\eta_{t}^{2}\Biggr) & =\sum_{t=1}^{T}\mathbb{E}\left[\left(\sum_{\ell=-\infty}^{t}\pi_{T}(t,\ell)\varepsilon_{\ell}+\rho_{T}^{t}\eta_{0}\right)^{2}\right]\\
 & \leq2\mathbb{E}(\epsilon_{1}^{2})\sum_{t=1}^{T}\sum_{\ell=-\infty}^{t}[\pi_{T}(t,\ell)]^{2}+2\mathbb{E}(\eta_{0}^{2})\sum_{t=1}^{T}\rho_{T}^{2t}\\
 & \leq2\mathbb{E}(\epsilon_{1}^{2})\cdot T\sup_{t\leq T}\sum_{\ell=-\infty}^{t}[\pi_{T}(t,\ell)]^{2}+2\sqrt{\mathbb{E}(\eta_{0}^{4})}\sum_{t=1}^{T}\rho_{T}^{2t}\\
 & =O(T)\cdot O\left(\dfrac{1}{|1-\rho_{T}|}\wedge T\right)+O(T)\cdot O\left(\dfrac{1}{|1-\rho_{T}|}\wedge T\right)=O\left(\dfrac{T}{|1-\rho_{T}|}\wedge T^{2}\right).
\end{align*}

Part~\ref{enu:E_sum_eta_4}. For generic scalars $a,b>0$, we have
$(a+b)^{4}\leq8(a^{4}+b^{4})$. Applying this fact, we have
\begin{align}
\bigl({\textstyle \sum_{t=1}^{T}\eta_{t}}\bigr)^{4} & =\Biggl(\sum_{t=1}^{T}\sum_{\ell=-\infty}^{t}\pi_{T}(t,\ell)\varepsilon_{\ell}+\sum_{t=1}^{T}\rho_{T}^{t}\xi_{0}\Biggr)^{4}\nonumber \\
 & \leq8\Biggl(\sum_{t=1}^{T}\sum_{\ell=-\infty}^{t}\pi_{T}(t,\ell)\varepsilon_{\ell}\Biggr)^{4}+8\left(\sum_{t=1}^{T}\rho_{T}^{t}\eta_{0}\right)^{4}.\label{eq:xi 4 bound}
\end{align}
The expectation of the second term in the rightmost expression has
the order
\begin{align}
\mathbb{E}\left[\left(\sum_{t=1}^{T}\rho_{T}^{t}\eta_{0}\right)^{4}\right] & =\mathbb{E}(\eta_{0}^{4})\cdot\left(\sum_{t=1}^{T}\rho_{T}^{t}\right)^{4}\nonumber \\
 & =O(T^{2})\cdot O\left(\dfrac{1}{(1-\rho_{T})^{4}}\wedge T^{4}\right)=O\left(\dfrac{T^{2}}{(1-\rho_{T})^{4}}\wedge T^{6}\right).\label{eq:xi 4 term 2}
\end{align}
It thus suffices to bound the first term. Note that
\[
\sum_{t=1}^{T}\sum_{\ell=-\infty}^{t}\pi_{T}(t,\ell)\varepsilon_{\ell}=\sum_{\ell=-\infty}^{T}\Biggl(\sum_{t=\ell\vee1}^{T}\pi_{T}(t,\ell)\Biggr)\varepsilon_{\ell}=:\sum_{\ell=-\infty}^{T}R_{T,\ell}\varepsilon_{\ell},
\]
where $R_{T,\ell}:=\sum_{t=\ell\vee1}^{T}\pi_{T}(t,\ell)$ and by
\prettyref{lem:ma_represent} it holds that
\begin{align*}
\sum_{\ell=-\infty}^{T}|R_{T,\ell}|^{m} & \leq\sum_{\ell=-\infty}^{0}\left|\sum_{t=1}^{T}\pi_{T}(t,\ell)\right|^{m}+\sum_{\ell=1}^{T}\left|\sum_{t=\ell}^{T}\pi_{T}(t,\ell)\right|^{m}\\
 & \lesssim\sum_{\ell=-\infty}^{0}q^{m(1-\ell)}\left|\sum_{t=1}^{T}(|\rho_{T}|\vee q_{\nu})^{t}\right|^{m}+\sum_{\ell=1}^{T}\left|\sum_{t=\ell}^{T}(|\rho_{T}|\vee q_{\nu})^{t-\ell}\right|^{m}\\
 & =O\left(\frac{1}{|1-\rho_{T}|^{m}}\wedge T^{m}\right)+O\left(\frac{T}{|1-\rho_{T}|^{m}}\wedge T^{m+1}\right)\\
 & =O\left(\frac{T}{|1-\rho_{T}|^{m}}\wedge T^{m+1}\right).
\end{align*}
Hence, using the fact that $\varepsilon_{\ell}$ is an m.d.s., we
have
\begin{align}
 & \mathbb{E}\left[\Biggl(\sum_{t=1}^{T}\sum_{\ell=-\infty}^{t}\pi_{T}(t,\ell)\varepsilon_{\ell}\Biggr)^{4}\right]\nonumber \\
 & \quad=\sum_{\ell=-\infty}^{T}\sum_{k\neq\ell}R_{T,k}^{3}R_{T,\ell}\mathbb{E}(\varepsilon_{k}^{3}\varepsilon_{\ell})+\sum_{\ell=-\infty}^{T}\sum_{k=-\infty}^{T}R_{T,k}^{2}R_{T,\ell}^{2}\mathbb{E}(\varepsilon_{k}^{2}\varepsilon_{\ell}^{2})\nonumber \\
 & \quad\quad+\sum_{k=-\infty}^{T}\sum_{\ell\neq j=-\infty}^{k-1}R_{T,k}^{2}R_{T,\ell}R_{T,j}\mathbb{E}(\varepsilon_{k}^{2}\varepsilon_{\ell}\varepsilon_{j})\nonumber \\
 & \quad\lesssim\sum_{\ell=-\infty}^{T}|R_{T,\ell}|\sum_{k=-\infty}^{T}|R_{T,k}|^{3}+\left(\sum_{\ell=-\infty}^{T}|R_{T,\ell}|^{2}\right)^{2}\nonumber \\
 & \quad\quad+\left(\sup_{\ell\leq T}R_{T,\ell}^{2}\right)\sum_{k=-\infty}^{\infty}\sum_{\ell\neq j=-\infty}^{k-1}\left|\mathbb{E}(\varepsilon_{k}^{2}\varepsilon_{\ell}\varepsilon_{j})\right|\nonumber \\
 & \quad=O\left(\frac{T^{2}}{(1-\rho_{T})^{4}}\wedge T^{6}\right)+O\left(\frac{T^{2}}{(1-\rho_{T})^{4}}\wedge T^{6}\right)+O\left(\frac{1}{|1-\rho_{T}|}\wedge T^{m}\right)\cdot O(1)\nonumber \\
 & \quad=O\left(\frac{T^{2}}{(1-\rho_{T})^{4}}\wedge T^{6}\right).\label{eq:sum_sum_pi_4}
\end{align}
where we use the fact that $\mathbb{E}(\epsilon_{k}^{2}\epsilon_{\ell}\varepsilon_{j})=\kappa(k,k,\ell,j)$
when $k\geq\ell\neq j$ and the cumulant condition to get
\[
\sum_{k=-\infty}^{\infty}\sum_{\ell\neq j=-\infty}^{k-1}\left|\mathbb{E}(\varepsilon_{k}^{2}\varepsilon_{\ell}\varepsilon_{j})\right|<\infty.
\]
Then Part~\ref{enu:E_sum_eta_4} is confirmed by \eqref{eq:xi 4 bound},
\eqref{eq:xi 4 term 2} and \eqref{eq:sum_sum_pi_4}.
\end{proof}
\begin{proof}[Proof of \prettyref{lem:generic_two_ARs}\ref{enu:supE}]
We assume $\rho_{\xi,T},\rho_{\eta,T}\geq0$ here for simplicity.
The negative coefficients generate stationary processes and can be
easily handled by parallel arguments.

Using \eqref{eq:eta_LP_decomp}, we have
\begin{align}
\xi_{t}\eta_{t} & =\sum_{k=-\infty}^{t}\pi_{\xi,T}(t,k)\varepsilon_{\xi,k}\sum_{\ell=-\infty}^{t}\pi_{\eta,T}(t,\ell)\varepsilon_{\eta,\ell}+\rho_{\xi,T}^{t}\xi_{0}\sum_{\ell=-\infty}^{t}\pi_{\eta,T}(t,\ell)\varepsilon_{\eta,\ell}\nonumber \\
 & \qquad+\rho_{\eta,T}^{t}\eta_{0}\sum_{k=-\infty}^{t}\pi_{\xi,T}(t,k)\varepsilon_{\xi,k}+(\rho_{\xi,T}\rho_{\eta,T})^{t}\xi_{0}\eta_{0}\nonumber \\
 & =:\Lambda_{1,t}+\Lambda_{2,t}+\Lambda_{3,t}+\Lambda_{4,t}.\label{eq:xi_eta}
\end{align}
For $\Lambda_{2,t}$, by \eqref{eq:pi_bound} and $\sup_{s\leq0}|\mathbb{E}(\xi_{0}\varepsilon_{\eta,s})|\leq\infty$
in Condition~\ref{enu:error}, we obtain
\[
\sup_{t\leq T}{}\bigl|\mathbb{E}(\Lambda_{2,t})\bigr|\leq\sup_{t\leq T}\rho_{\xi,T}^{t}\sum_{\ell=-\infty}^{0}|\pi_{\eta,T}(t,\ell)||\mathbb{E}(\xi_{0}\varepsilon_{\eta,\ell})|=O(1).
\]
We can deduce in the same manner that
\[
\sup_{t\leq T}{}\bigl|\mathbb{E}(\Lambda_{3,t})\bigr|\leq\sup_{t\leq T}\rho_{\eta,T}^{t}\sum_{k=-\infty}^{0}|\pi_{\xi,T}(t,k)||\mathbb{E}(\eta_{0}\varepsilon_{\xi,k})|=O(1).
\]
For $\Lambda_{1,t}$, again by \eqref{eq:pi_bound},
\begin{align*}
\sup_{t\leq T}{}\bigl|\mathbb{E}(\Lambda_{1,t})\bigr| & =\sup_{t\leq T}{}\Biggl|\sum_{\ell=1}^{t}\pi_{\xi,T}(t,\ell)\pi_{\eta,T}(t,\ell)\mathbb{E}(\epsilon_{\xi,\ell}\varepsilon_{\eta,\ell})\Biggr|\lesssim\sup_{t\leq T}{}\sum_{\ell=1}^{t}[(\rho_{\xi,T}\vee q_{0})(\rho_{\eta,T}\vee q_{0})]^{t-\ell}\\
 & \leq\sum_{j=0}^{T}[(\rho_{\xi,T}\vee q_{0})(\rho_{\eta,T}\vee q_{0})]^{j}=O\left(\dfrac{1}{|1-\rho_{\xi,T}\rho_{\eta,T}|}\wedge T\right),
\end{align*}
where the last equality follows from \prettyref{lem:sum_rho} since
$\rho_{\xi,T}\rho_{\eta,T}\leq1+\frac{2c+c^{2}}{T}$. By Condition~\ref{enu:init_vals},
\[
\sup_{t\leq T}{}\bigl|\mathbb{E}(\Lambda_{4,t})\bigr|=|\mathbb{E}(\xi_{0}\eta_{0})|\sup_{t\leq T}{}(\rho_{\xi,T}\rho_{\eta,T})^{t}=O\left(\frac{1}{|1-\rho_{\xi,T}\rho_{\eta,T}|}\wedge T\right).
\]
It then follows that $\sup_{t\leq T}|\mathbb{E}(\xi_{t}\eta_{t})|=O\bigl(\frac{1}{|1-\rho_{\xi,T}\rho_{\eta,T}|}\wedge T\bigr)$.
\end{proof}
\begin{proof}[Proof of \prettyref{lem:generic_two_ARs}\ref{enu:supE2}]
We first show that
\begin{equation}
\left(\frac{1}{|1-\rho_{\xi,T}\rho_{\eta,T}|}\wedge T\right)\lesssim\left(\dfrac{1}{|1-\rho_{\xi,T}|}\wedge T\right)\ensuremath{\wedge\left(\dfrac{1}{|1-\rho_{\eta,T}|}\wedge T\right)}.\label{eq:rho_ineq}
\end{equation}

Case I: $T|1-\rho_{\xi,T}|\to c_{\xi}\geq0$ and $T|1-\rho_{\eta,T}|\to c_{\eta}\geq0$.
Then $|1-\rho_{\xi,T}|^{-1}\gtrsim T$, $|1-\rho_{\eta,T}|^{-1}\gtrsim T$,
and $|1-\rho_{\xi,T}\rho_{\eta,T}|^{-1}=|1-\rho_{\xi,T}+1-\rho_{\eta,T}-(1-\rho_{\xi,T})(1-\rho_{\eta,T})|^{-1}\gtrsim T$.
The inequality holds.

Case II: $T|1-\rho_{\xi,T}|\to c_{\xi}\geq0$ and $T|1-\rho_{\eta,T}|\to\infty$.
Then $|1-\rho_{\xi,T}|^{-1}\gtrsim T$, $|1-\rho_{\eta,T}|^{-1}\lesssim T$,
and
\[
|1-\rho_{\xi,T}\rho_{\eta,T}|^{-1}=|1-\rho_{\eta,T}|^{-1}\left|\frac{1-\rho_{\xi,T}}{1-\rho_{\eta,T}}+\rho_{\xi,T}\right|^{-1}\lesssim|1-\rho_{\eta,T}|^{-1}.
\]
The inequality follows.

Case III: $T|1-\rho_{\xi,T}|\to\infty$ and $T|1-\rho_{\eta,T}|\to c_{\eta}\geq0$.
Same argument as Case II.

Case IV: $T|1-\rho_{\xi,T}|\to\infty$ and $T|1-\rho_{\eta,T}|\to\infty$.
Then $|1-\rho_{\xi,T}|^{-1}\lesssim T$, $|1-\rho_{\eta,T}|^{-1}\lesssim T$,
and
\[
|1-\rho_{\xi,T}\rho_{\eta,T}|^{-1}=|1-\rho_{\xi,T}|^{-1}|1-\rho_{\eta,T}|^{-1}\left|\frac{1}{1-\rho_{\xi,T}}+\frac{1}{1-\rho_{\eta,T}}-1\right|^{-1}\lesssim|1-\rho_{\xi,T}|^{-1}\wedge|1-\rho_{\eta,T}|^{-1}.
\]
The inequality follows.

Noting that for any generic $a_{i}\geq0$, we have $\frac{1}{n}\sum_{i=1}^{n}a_{i}\leq\bigl(\frac{1}{n}\sum_{i=1}^{n}a_{i}^{2}\bigr)^{1/2}$
for $a_{i}\geq0$. This inequality gives
\[
\mathbb{E}\left(\xi_{t}^{2}\eta_{t}^{2}\right)\leq4\mathbb{E}\left[\Lambda_{1,t}^{2}+\Lambda_{2,t}^{2}+\Lambda_{3,t}^{2}+\Lambda_{4,t}^{2}\right],
\]
where $\Lambda_{1,t}$, $\Lambda_{2,t}$, $\Lambda_{3,t}$, and $\Lambda_{4,t}$
are defined in \eqref{eq:xi_eta}.

Bound for $\Lambda_{1,t}$. Note that using the m.d.s.\ property
we can write
\begin{align*}
\mathbb{E}\left[\Lambda_{1,t}^{2}\right] & =\mathbb{E}\Biggl(\sum_{j,k=-\infty}^{t}\sum_{\ell,m=-\infty}^{t}\pi_{\xi,T}(t,j)\pi_{\xi,T}(t,k)\pi_{\eta,T}(t,\ell)\pi_{\eta,T}(t,m)\epsilon_{\xi,j}\epsilon_{\xi,k}\epsilon_{\eta,\ell}\epsilon_{\eta,m}\Biggr)\\
 & =\sum_{j=-\infty}^{t}\sum_{\ell=-\infty}^{t}[\pi_{\xi,T}(t,j)]^{2}[\pi_{\eta,T}(t,\ell)]^{2}\mathbb{E}(\epsilon_{\xi,j}^{2}\epsilon_{\eta,\ell}^{2})\\
 & \qquad+2\sum_{j=-\infty}^{t}\sum_{\substack{k=-\infty\\
k\neq j
}
}^{t}[\pi_{\xi,T}(t,j)\pi_{\eta,T}(t,j)][\pi_{\xi,T}(t,k)\pi_{\eta,T}(t,k)]\mathbb{E}(\epsilon_{\xi,j}\epsilon_{\eta,j}\epsilon_{\xi,k}\epsilon_{\eta,k})\\
 & \qquad+2\sum_{j=-\infty}^{t}\sum_{m=-\infty}^{j-1}[\pi_{\xi,T}(t,j)]^{2}\pi_{\eta,T}(t,j)\pi_{\eta,T}(t,m)\mathbb{E}(\epsilon_{\xi,j}^{2}\epsilon_{\eta,j}\epsilon_{\eta,m})\\
 & \qquad+2\sum_{j=-\infty}^{t}\sum_{k=-\infty}^{j-1}[\pi_{\eta,T}(t,j)]^{2}\pi_{\xi,T}(t,j)\pi_{\xi,T}(t,k)\mathbb{E}(\epsilon_{\eta,j}^{2}\epsilon_{\xi,j}\epsilon_{\xi,k})\\
 & \qquad+\sum_{j=-\infty}^{t}\sum_{\substack{\ell,m=-\infty\\
\ell\neq m
}
}^{j-1}[\pi_{\xi,T}(t,j)]^{2}\pi_{\eta,T}(t,\ell)\pi_{\eta,T}(t,m)\mathbb{E}(\epsilon_{\xi,j}^{2}\epsilon_{\eta,\ell}\epsilon_{\eta,m})\\
 & \qquad+\sum_{\ell=-\infty}^{t}\sum_{\substack{j,k=-\infty\\
j\neq k
}
}^{j-1}[\pi_{\eta,T}(t,\ell)]^{2}\pi_{\xi,T}(t,j)\pi_{\xi,T}(t,k)\mathbb{E}(\epsilon_{\eta,\ell}^{2}\epsilon_{\xi,j}\epsilon_{\xi,k})\\
 & \qquad+4\sum_{j=-\infty}^{t}\sum_{\substack{\ell,m=-\infty\\
\ell\neq m
}
}^{j-1}[\pi_{\xi,T}(t,j)\pi_{\eta,T}(t,j)]\pi_{\xi,T}(t,\ell)\pi_{\eta,T}(t,m)\mathbb{E}(\epsilon_{\xi,j}\epsilon_{\eta,j}\epsilon_{\xi,\ell}\epsilon_{\eta,m})\\
 & =:E_{1,t}+E_{2,t}+E_{3,t}+E_{4,t}+E_{5,t}+E_{6,t}+E_{7,t.}
\end{align*}
$E_{1,t}$ is from the case ($j=k,\ell=m$); $E_{2,t}$ is from ($j=\ell\neq k=m$)
or ($j=m\neq k=\ell$), so involves a factor of 2; $E_{3,t}$ is from
($j=k=\ell>m$) or ($j=k=m>\ell$); $E_{4,t}$ is from ($k=\ell=m>j$)
or ($j=\ell=m>k$); $E_{5,t}$ is from ($j=k>\ell\neq m$); $E_{6,t}$
is from ($\ell=m>j\neq k$); $E_{7,t}$ is from ($j=\ell>k\neq m$)
or ($j=m>k\neq\ell$) or ($k=\ell>j\ne m$) or ($k=m>j\neq\ell$).
For $E_{1,t}$, we have
\[
\sup_{t\leq T}|E_{1,t}|\lesssim\sup_{t\leq T}\sum_{j=-\infty}^{t}[\pi_{\xi,T}(t,j)]^{2}\cdot\sup_{t\leq T}\sum_{\ell=-\infty}^{t}[\pi_{\eta,T}(t,\ell)]^{2}=O\left(\left[\dfrac{1}{|1-\rho_{\xi,T}|}\wedge T\right]\cdot\left[\dfrac{1}{|1-\rho_{\eta,T}|}\wedge T\right]\right).
\]
Likewise we can show that $\sup_{t\leq T}|E_{2,t}|$ has order
\[
\sup_{t\leq T}|E_{2,t}|=O\left(\left[\dfrac{1}{|1-\rho_{\xi,T}\rho_{\eta,T}|}\wedge T\right]^{2}\right)=O\left(\left[\dfrac{1}{|1-\rho_{\xi,T}|}\wedge T\right]\cdot\left[\dfrac{1}{|1-\rho_{\eta,T}|}\wedge T\right]\right),
\]
where we use \eqref{eq:rho_ineq}. For $E_{3,t}$, we have
\begin{align*}
\sup_{t\leq T}|E_{3,t}| & \lesssim\sup_{t\leq T}\sum_{j=-\infty}^{t}[\pi_{\xi,T}(t,j)]^{2}|\pi_{\eta,T}(t,j)|\cdot\sup_{t\leq T}\sum_{m=-\infty}^{t}|\pi_{\eta,T}(t,m)|\\
 & =O\left(\dfrac{1}{|1-\rho_{\xi,T}\rho_{\eta,T}|}\wedge T\right)\cdot O\left(\frac{1}{|1-\rho_{\eta,T}|}\wedge T\right)\\
 & =O\left(\left[\dfrac{1}{|1-\rho_{\xi,T}|}\wedge T\right]\cdot\left[\dfrac{1}{|1-\rho_{\eta,T}|}\wedge T\right]\right),
\end{align*}
where the last line uses \eqref{eq:rho_ineq}. Similarly, we can show
that $\sup_{t\leq T}|E_{4,t}|$ has the same order. For $E_{5,t}$,
we have
\[
\sup_{t\leq T}|E_{5,t}|\leq\left(\sup_{t,j\leq T}[\pi_{\xi,T}(t,j)]^{2}\right)\biggl(\sup_{t,\ell\leq T}|\pi_{\eta,T}(t,\ell)|\biggr)\sum_{j=-\infty}^{\infty}\sum_{\substack{\ell,m=-\infty\\
\ell\neq m
}
}^{j-1}\left|\mathbb{E}\left(\epsilon_{\xi,j}^{2}\epsilon_{\eta,\ell}\epsilon_{\eta,m}\right)\right|=O(1),
\]
where the fact that $\mathbb{E}\bigl(\epsilon_{\xi,j}^{2}\epsilon_{\eta,\ell}\epsilon_{\eta,m}\bigr)=\kappa_{\xi\xi\eta\eta}(j,j,\ell,m)$
when $j\geq\ell\neq m$ and the cumulant condition imply
\[
\sum_{j=-\infty}^{\infty}\sum_{\substack{\ell,m=-\infty\\
\ell\neq m
}
}^{j-1}\left|\mathbb{E}\left(\epsilon_{\xi,j}^{2}\epsilon_{\eta,\ell}\epsilon_{\eta,m}\right)\right|<\infty.
\]
By the same argument we can show $\sup_{t\leq T}|E_{6,t}|$ and $\sup_{t\leq T}|E_{7,t}|$
are both bounded. The above discussion gives rise to
\[
\sup_{t\leq T}\mathbb{E}\left[\Lambda_{1,t}^{2}\right]=O\left(\left[\dfrac{1}{|1-\rho_{\xi,T}|}\wedge T\right]\cdot\left[\dfrac{1}{|1-\rho_{\eta,T}|}\wedge T\right]\right).
\]

Bounds for $\Lambda_{2,t}$ and $\Lambda_{3,t}$. Using the Cauchy-Schwarz
inequality and by the same argument in \eqref{eq:sum_sum_pi_4} (with
\prettyref{lem:ma_represent} invoked), we have
\begin{align*}
\sup_{t\leq T}\mathbb{E}\left[\Lambda_{2,t}^{2}\right] & \leq\left(\sup_{t\leq T}\rho_{\xi,T}^{2t}\right)\sqrt{\mathbb{E}(\xi_{0}^{4})}\left\{ \sup_{t\leq T}\mathbb{E}\left[\left(\sum_{\ell=-\infty}^{t}\pi_{\eta,T}(t,\ell)\epsilon_{\eta,\ell}\right)^{4}\right]\right\} ^{1/2}\\
 & =O\left(\frac{1}{|1-\rho_{\xi,T}|}\wedge T\right)\cdot O\left(\frac{1}{|1-\rho_{\eta,T}|}\wedge T\right).
\end{align*}
We can show in the same manner that $\mathbb{E}\bigl[\Lambda_{3,t}^{2}\bigr]$
has the same order.

Finally we bound $\mathbb{E}\bigl[\Lambda_{4,t}^{2}\bigr]$. The Cauchy-Schwarz
inequality and Condition~\ref{enu:init_vals} lead to
\begin{align*}
\sup_{t\leq T}\mathbb{E}\left[\Lambda_{4,t}^{2}\right] & =\mathbb{E}(\xi_{0}^{2}\eta_{0}^{2})\sup_{t\leq T}{}(\rho_{\xi,T}\rho_{\eta,T})^{2t}\leq\sqrt{\mathbb{E}(\xi_{0}^{4})\cdot\mathbb{E}(\eta_{0}^{4})}\sup_{t\leq T}{}(\rho_{\xi,T}\rho_{\eta,T})^{2t}\\
 & =O\left(\left[\dfrac{1}{|1-\rho_{\xi,T}|}\wedge T\right]\cdot\left[\dfrac{1}{|1-\rho_{\eta,T}|}\wedge T\right]\right).
\end{align*}
We complete the proof of \prettyref{lem:generic_two_ARs}\ref{enu:supE2}.
\end{proof}
\begin{proof}[Proof of \prettyref{lem:generic_two_ARs}\ref{enu:Esum2}]
 Let $\Xi_{\ell,T}:=\sum_{t=1}^{T}\Lambda_{\ell,t}$ for $\ell=1,2,3,4$
where $\Lambda_{\ell,t}$ are defined in \eqref{eq:xi_eta}. Then
we have
\[
\mathbb{E}\left[\left(\sum_{t=1}^{T}\xi_{t}\eta_{t}\right)^{2}\right]=\mathbb{E}\left[\left(\Xi_{1,T}+\Xi_{2,T}+\Xi_{3,T}+\Xi_{4,T}\right)^{2}\right]\leq4\mathbb{E}\left(\Xi_{1,T}^{2}+\Xi_{2,T}^{2}+\Xi_{3,T}^{2}+\Xi_{4,T}^{2}\right).
\]

Bound for $\mathbb{E}\left(\Xi_{1,T}^{2}\right)$. We have
\[
\mathbb{E}\left(\Xi_{1,T}^{2}\right)=\sum_{t=1}^{T}\mathbb{E}\left(\Lambda_{1,t}^{2}\right)+2\sum_{t=1}^{T}\sum_{s=1}^{t-1}\mathbb{E}\left(\Lambda_{1,s}\Lambda_{1,t}\right).
\]
The first term, by the proof of \ref{enu:supE2}, has order
\[
\sum_{t=1}^{T}\mathbb{E}\left(\Lambda_{1,t}^{2}\right)\leq T\cdot\sup_{t\le T}\mathbb{E}\left(\Lambda_{1,t}^{2}\right)=O\left(T\cdot\left[\dfrac{1}{|1-\rho_{\xi,T}|}\wedge T\right]\cdot\left[\dfrac{1}{|1-\rho_{\eta,T}|}\wedge T\right]\right).
\]
The second term, using the m.d.s.\ condition, can be decomposed as
\begin{align*}
 & \quad\sum_{t=1}^{T}\sum_{s=1}^{t-1}\mathbb{E}\left(\Lambda_{1,s}\Lambda_{1,t}\right)\\
 & =\sum_{t=1}^{T}\sum_{s=1}^{t-1}\sum_{j,\ell=-\infty}^{t}\sum_{k,m=-\infty}^{s}\pi_{\xi,T}(t,j)\pi_{\xi,T}(s,k)\pi_{\eta,T}(t,\ell)\pi_{\eta,T}(s,m)\mathbb{E}\left(\epsilon_{\xi,j}\epsilon_{\xi,k}\epsilon_{\eta,\ell}\epsilon_{\eta,m}\right)\\
 & =\sum_{t=1}^{T}\sum_{s=1}^{t-1}\sum_{j=-\infty}^{t}\sum_{k=-\infty}^{s}[\pi_{\xi,T}(t,j)\pi_{\eta,T}(t,j)][\pi_{\xi,T}(s,k)\pi_{\eta,T}(s,k)]\mathbb{E}\left(\varepsilon_{\xi,j}\varepsilon_{\eta,j}\varepsilon_{\xi,k}\varepsilon_{\eta,k}\right)\\
 & \qquad+\sum_{t=1}^{T}\sum_{s=1}^{t-1}\sum_{m=-\infty}^{s}\sum_{\substack{k=-\infty\\
k\neq m
}
}^{s}[\pi_{\xi,T}(t,m)\pi_{\eta,T}(s,m)][\pi_{\xi,T}(s,k)\pi_{\eta,T}(t,k)]\mathbb{E}\left(\varepsilon_{\xi,m}\varepsilon_{\eta,m}\varepsilon_{\xi,k}\varepsilon_{\eta,k}\right)\\
 & \qquad+\sum_{t=1}^{T}\sum_{s=1}^{t-1}\sum_{\substack{k,m=-\infty\\
k\neq m
}
}^{s}[\pi_{\xi,T}(t,k)\pi_{\xi,T}(s,k)][\pi_{\eta,T}(t,m)\pi_{\eta,T}(s,m)]\mathbb{E}\left(\epsilon_{\xi,k}^{2}\epsilon_{\eta,m}^{2}\right)\\
 & \qquad+\sum_{t=1}^{T}\sum_{s=1}^{t-1}\sum_{k=-\infty}^{s}\sum_{m=-\infty}^{k-1}[\pi_{\xi,T}(t,k)\pi_{\xi,T}(s,k)\pi_{\eta,T}(t,k)]\pi_{\eta,T}(s,m)\mathbb{E}\left(\varepsilon_{\xi,k}^{2}\varepsilon_{\eta,k}\varepsilon_{\eta,m}\right)\\
 & \qquad+\sum_{t=1}^{T}\sum_{s=1}^{t-1}\sum_{m=-\infty}^{s}\sum_{\ell=-\infty}^{m-1}[\pi_{\xi,T}(t,m)\pi_{\xi,T}(s,m)\pi_{\eta,T}(s,m)]\pi_{\eta,T}(t,\ell)\mathbb{E}\left(\varepsilon_{\xi,m}^{2}\varepsilon_{\eta,m}\varepsilon_{\eta,\ell}\right)\\
 & \qquad+\sum_{t=1}^{T}\sum_{s=1}^{t-1}\sum_{m=-\infty}^{s}\sum_{k=-\infty}^{m-1}[\pi_{\eta,T}(t,m)\pi_{\eta,T}(s,m)\pi_{\xi,T}(t,m)]\pi_{\xi,T}(s,k)\mathbb{E}\left(\varepsilon_{\eta,m}^{2}\varepsilon_{\xi,m}\varepsilon_{\eta,k}\right)\\
 & \qquad+\sum_{t=1}^{T}\sum_{s=1}^{t-1}\sum_{m=-\infty}^{s}\sum_{j=-\infty}^{m-1}[\pi_{\eta,T}(t,m)\pi_{\eta,T}(s,m)\pi_{\xi,T}(s,m)]\pi_{\xi,T}(t,j)\mathbb{E}\left(\varepsilon_{\eta,m}^{2}\varepsilon_{\xi,m}\varepsilon_{\eta,j}\right)\\
 & \qquad+\sum_{t=1}^{T}\sum_{s=1}^{t-1}\sum_{k=-\infty}^{s}\sum_{\substack{\ell,m=-\infty\\
\ell\neq m
}
}^{k-1}[\pi_{\xi,T}(t,k)\pi_{\xi,T}(s,k)]\pi_{\eta,T}(t,\ell)\pi_{\eta,T}(s,m)\mathbb{E}\left(\varepsilon_{\xi,k}^{2}\varepsilon_{\eta,\ell}\varepsilon_{\eta,m}\right)\\
 & \qquad+\sum_{t=1}^{T}\sum_{s=1}^{t-1}\sum_{m=-\infty}^{s}\sum_{\substack{j,k=-\infty\\
j\neq k
}
}^{m-1}[\pi_{\eta,T}(t,m)\pi_{\eta,T}(s,m)]\pi_{\xi,T}(t,j)\pi_{\xi,T}(s,k)\mathbb{E}\left(\varepsilon_{\eta,m}^{2}\varepsilon_{\xi,j}\varepsilon_{\xi,k}\right)\\
 & \qquad+\sum_{t=1}^{T}\sum_{s=1}^{t-1}\sum_{j=-\infty}^{t}\sum_{\substack{k,m=-\infty\\
k\neq m
}
}^{(j-1)\wedge s}\pi_{\xi,T}(t,j)\pi_{\eta,T}(t,j)\pi_{\xi,T}(s,k)\pi_{\eta,T}(s,m)\mathbb{E}\left(\varepsilon_{\xi,j}\varepsilon_{\eta,j}\varepsilon_{\xi,k,}\varepsilon_{\eta,m}\right)\\
 & \qquad+\sum_{t=1}^{T}\sum_{s=1}^{t-1}\sum_{m=-\infty}^{s}\sum_{\substack{k,\ell=-\infty\\
k\neq\ell
}
}^{m-1}\pi_{\xi,T}(t,m)\pi_{\eta,T}(s,m)\pi_{\xi,T}(s,k)\pi_{\eta,T}(t,\ell)\mathbb{E}\left(\varepsilon_{\xi,m}\varepsilon_{\eta,m}\varepsilon_{\xi,k,}\varepsilon_{\eta,\ell}\right)\\
 & \qquad+\sum_{t=1}^{T}\sum_{s=1}^{t-1}\sum_{m=-\infty}^{s}\sum_{\substack{j,\ell=-\infty\\
j\neq\ell
}
}^{m-1}\pi_{\xi,T}(s,m)\pi_{\eta,T}(s,m)\pi_{\xi,T}(t,k)\pi_{\eta,T}(t,\ell)\mathbb{E}\left(\varepsilon_{\xi,m}\varepsilon_{\eta,m}\varepsilon_{\xi,j,}\varepsilon_{\eta,\ell}\right)\\
 & \qquad+\sum_{t=1}^{T}\sum_{s=1}^{t-1}\sum_{k=-\infty}^{s}\sum_{\substack{j,m=-\infty\\
j\neq m
}
}^{k-1}\pi_{\xi,T}(s,m)\pi_{\eta,T}(t,m)\pi_{\xi,T}(t,j)\pi_{\eta,T}(s,m)\mathbb{E}\left(\varepsilon_{\xi,k}\varepsilon_{\eta,k}\varepsilon_{\xi,j,}\varepsilon_{\eta,m}\right)\\
 & =:F_{1,t}+\dots+F_{13,t}.
\end{align*}
$F_{1,t}$ is from the case $(j=\ell,k=m)$; $F_{2,t}$ is from $(j=m\neq k=\ell)$;
$F_{3,t}$ is from $(j=k\neq\ell=m)$; $F_{4,t}$ is from $(j=k=\ell>m)$;
$F_{5,t}$ is from $(j=k=m>\ell)$; $F_{6,t}$ is from $(\ell=m=j>k)$;
$F_{7,t}$ is from $(\ell=m=k>j)$; $F_{8,t}$ is from $(j=k>\ell\neq m)$;
$F_{9,t}$ is from $(\ell=m>j\ne k)$; $F_{10,t}$ is from $(j=\ell>k\ne m)$;
$F_{11,t}$ is from $(j=m>k\ne\ell)$; $F_{12,t}$ is from $(k=m>j\ne\ell)$;
$F_{13,t}$ is from $(k=\ell>j\ne m)$. For $F_{1,t}$, we have
\begin{align*}
\sup_{t\leq T}|F_{1,t}| & \lesssim\sum_{t=1}^{T}\sum_{s=1}^{t-1}\sum_{j=-\infty}^{t}\sum_{k=-\infty}^{s}|\pi_{\xi,T}(t,j)\pi_{\eta,T}(t,j)||\pi_{\xi,T}(s,k)\pi_{\eta,T}(s,k)|\\
 & \leq T^{2}\cdot\left(\sup_{t\leq T}\sum_{j=-\infty}^{t}|\pi_{\xi,T}(t,j)\pi_{\eta,T}(t,j)|\right)^{2}=O\left(\dfrac{T^{2}}{(1-\rho_{\xi,T}\rho_{\eta,T})^{2}}\wedge T^{4}\right).
\end{align*}
For $F_{2,t}$, we use \prettyref{lem:ma_represent} to get
\begin{align*}
\sup_{t\leq T}|F_{2,t}| & \lesssim\sum_{t=1}^{T}\sum_{s=1}^{t-1}\sum_{m=-\infty}^{s}\sum_{k=-\infty}^{s}|\pi_{\xi,T}(t,m)\pi_{\eta,T}(s,m)||\pi_{\xi,T}(s,k)\pi_{\eta,T}(t,k)|\\
 & \lesssim\sum_{t=1}^{T}\sum_{s=1}^{t-1}[(\rho_{\xi,T}\vee q_{0})(\rho_{\eta,T}\vee q_{0})]^{t-s}\left(\sum_{m=0}^{s}[(\rho_{\xi,T}\vee q_{0})(\rho_{\eta,T}\vee q_{0})]^{s-m}\right)^{2}\\
 & \leq T\cdot\left(\sup_{t\leq T}\sum_{s=0}^{t}[(\rho_{\xi,T}\vee q_{0})(\rho_{\eta,T}\vee q_{0})]^{s}\right)^{3}=O\left(\dfrac{T}{|1-\rho_{\xi,T}\rho_{\eta,T}|^{3}}\wedge T^{4}\right).
\end{align*}
For $F_{3,t}$, in the same manner,
\begin{align*}
\sup_{t\leq T}|F_{3,t}| & \lesssim\sum_{t=1}^{T}\sum_{s=1}^{t-1}\sum_{k,m=-\infty}^{s}|\pi_{\xi,T}(t,k)\pi_{\xi,T}(s,k)||\pi_{\eta,T}(t,m)\pi_{\eta,T}(s,m)|\\
 & \lesssim\sum_{t=1}^{T}\sum_{s=1}^{t-1}[(\rho_{\xi,T}\vee q_{0})(\rho_{\eta,T}\vee q_{0})]^{t-s}\left(\sum_{k=0}^{s}[(\rho_{\xi,T}\vee q_{0})^{2}]^{s-k}\right)\left(\sum_{m=0}^{s}[(\rho_{\eta,T}\vee q_{0})^{2}]^{s-m}\right)\\
 & \leq T\cdot\left(\sup_{t\leq T}\sum_{s=0}^{t}[(\rho_{\xi,T}\vee q_{0})(\rho_{\eta,T}\vee q_{0})]^{s}\right)\left(\sup_{t\leq T}\sum_{s=0}^{t}(\rho_{\xi,T}\vee q_{0})^{2s}\right)\left(\sup_{t\leq T}\sum_{s=0}^{t}(\rho_{\eta,T}\vee q_{0})^{2s}\right)\\
 & =O\left(\dfrac{T}{|1-\rho_{\xi,T}\rho_{\eta,T}|}\wedge T^{2}\right)\cdot O\left(\dfrac{1}{|1-\rho_{\xi,T}|}\wedge T\right)\cdot O\left(\dfrac{1}{|1-\rho_{\eta,T}|}\wedge T\right).
\end{align*}
For $F_{4,t}$, we have
\begin{align*}
 & \quad\sup_{t\leq T}|F_{4,t}|\lesssim\sum_{t=1}^{T}\sum_{s=1}^{t-1}\sum_{k=-\infty}^{s}\sum_{m=-\infty}^{k-1}|\pi_{\xi,T}(t,k)\pi_{\xi,T}(s,k)\pi_{\eta,T}(t,k)||\pi_{\eta,T}(s,m)|\\
 & \lesssim\sum_{t=1}^{T}\sum_{s=1}^{t-1}[(\rho_{\xi,T}\vee q_{0})(\rho_{\eta,T}\vee q_{0})]^{t-s}\left(\sum_{k=0}^{s}[(\rho_{\xi,T}\vee q_{0})^{2}(\rho_{\eta,T}\vee q_{0})]^{s-k}\right)\left(\sum_{m=0}^{k-1}(\rho_{\eta,T}\vee q_{0})^{s-m}\right)\\
 & \leq T\left(\sup_{t\leq T}\sum_{s=0}^{t}[(\rho_{\xi,T}\vee q_{0})(\rho_{\eta,T}\vee q_{0})]^{s}\right)\left(\sup_{t\leq T}\sum_{s=0}^{t}[(\rho_{\xi,T}\vee q_{0})^{2}(\rho_{\eta,T}\vee q_{0})]^{s}\right)\left(\sup_{t\leq T}\sum_{s=0}^{t}(\rho_{\eta,T}\vee q_{0})^{s}\right)\\
 & =O\left(\dfrac{T}{|1-\rho_{\xi,T}\rho_{\eta,T}|}\wedge T^{2}\right)\cdot O\left(\dfrac{1}{|1-\rho_{\xi,T}\rho_{\eta,T}|}\wedge T\right)\cdot O\left(\dfrac{1}{|1-\rho_{\eta,T}|}\wedge T\right).
\end{align*}
Similarly, we can show that $F_{5,t},F_{6,t},F_{7,t}$ all have the
same order as $F_{4,t}$. For $F_{8,t},\dots,F_{13,t}$, as in the
proof of \ref{enu:supE2}, we use the cumulant condition to conclude
that they are all $O(1)$. Furthermore, use the argument for proving
\eqref{eq:rho_ineq}, we can show
\begin{equation}
\left(\dfrac{1}{|1-\rho_{\xi,T}|}\wedge T\right)\cdot\left(\dfrac{1}{|1-\rho_{\eta,T}|}\wedge T\right)\lesssim T\left(\dfrac{1}{|1-\rho_{\xi,T}\rho_{\eta,T}|}\wedge T\right).\label{eq:rho_ineq2}
\end{equation}
The above discussion leads to
\[
\sum_{t=1}^{T}\sum_{s=1}^{t-1}\mathbb{E}\left(\Lambda_{1,s}\Lambda_{1,t}\right)=O\left(\dfrac{T^{2}}{(1-\rho_{\xi,T}\rho_{\eta,T})^{2}}\wedge T^{4}\right).
\]
It follows that
\begin{equation}
\mathbb{E}\left(\Xi_{1,T}^{2}\right)=O\left(\dfrac{T^{2}}{(1-\rho_{\xi,T}\rho_{\eta,T})^{2}}\wedge T^{4}\right).\label{eq:Xi1_bound}
\end{equation}

Bound for $\mathbb{E}\left(\Xi_{2,T}^{2}\right)$ and $\mathbb{E}\left(\Xi_{3,T}^{2}\right)$.
We have
\[
\mathbb{E}\left(\Xi_{2,T}^{2}\right)=\sum_{t=1}^{T}\mathbb{E}\left(\Lambda_{2,t}^{2}\right)+2\sum_{t=1}^{T}\sum_{s=1}^{t-1}\mathbb{E}\left(\Lambda_{2,s}\Lambda_{2,t}\right).
\]
The first term, by the proof of \ref{enu:supE2}, has order
\[
\sum_{t=1}^{T}\mathbb{E}\left(\Lambda_{2,t}^{2}\right)\leq T\cdot\sup_{t\le T}\mathbb{E}\left(\Lambda_{2,t}^{2}\right)=O\left(T\cdot\left[\dfrac{1}{|1-\rho_{\xi,T}|}\wedge T\right]\cdot\left[\dfrac{1}{|1-\rho_{\eta,T}|}\wedge T\right]\right).
\]
The second term can be written as
\begin{align*}
 & \quad\sum_{t=1}^{T}\sum_{s=1}^{t-1}\mathbb{E}\left(\Lambda_{2,s}\Lambda_{2,t}\right)\\
 & =\sum_{t=1}^{T}\sum_{s=1}^{t-1}\sum_{k=-\infty}^{s}\sum_{\ell=-\infty}^{t}\rho_{\xi,T}^{s+t}\pi_{\eta,T}(s,k)\pi_{\eta,T}(t,\ell)\mathbb{E}\left(\varepsilon_{\eta,k}\varepsilon_{\eta,\ell}\xi_{0}^{2}\right)\\
 & =\sum_{t=1}^{T}\sum_{s=1}^{t-1}\sum_{k=1}^{s}\rho_{\xi,T}^{s+t}\pi_{\eta,T}(s,k)\pi_{\eta,T}(t,k)\mathbb{E}\left(\varepsilon_{\eta,k}^{2}\xi_{0}^{2}\right)\\
 & \qquad+\sum_{t=1}^{T}\sum_{s=1}^{t-1}\sum_{k=-\infty}^{0}\sum_{\ell=-\infty}^{0}\rho_{\xi,T}^{s+t}\pi_{\eta,T}(s,k)\pi_{\eta,T}(t,\ell)\mathbb{E}\left(\varepsilon_{\eta,k}\varepsilon_{\eta,\ell}\xi_{0}^{2}\right)\\
 & \lesssim\sum_{t=1}^{T}\sum_{s=1}^{t-1}\rho_{\xi,T}^{s+t}(\rho_{\eta,T}\vee q_{0})^{t-s}\left(\sup_{t\leq T}\sum_{\ell=0}^{t}(\rho_{\eta,T}\vee q_{0})^{2\ell}\right)\sqrt{\mathbb{E}\big(\varepsilon_{\eta,k}^{4}\bigr)\mathbb{E}\left(\xi_{0}^{4}\right)}\\
 & \qquad+\sum_{t=1}^{T}\sum_{s=1}^{t-1}[\rho_{\xi,T}(\rho_{\eta,T}\vee q_{0})]^{s+t}\left(\sum_{\ell=-\infty}^{0}q_{0}^{\ell}\right)^{2}\sqrt{\mathbb{E}\big(\varepsilon_{\eta,k}^{4}\bigr)\mathbb{E}\left(\xi_{0}^{4}\right)}\\
 & =O\left(\dfrac{1}{|1-\rho_{\xi,T}|}\wedge T\right)\cdot O\left(\dfrac{1}{|\rho_{\xi,T}-\rho_{\eta,T}|}\wedge T\right)\cdot O\left(\dfrac{1}{|1-\rho_{\eta,T}|}\wedge T\right)\cdot O\left(\dfrac{1}{|1-\rho_{\xi,T}|}\wedge T\right)\\
 & \qquad+O\left(\dfrac{1}{|1-\rho_{\xi,T}\rho_{\eta,T}|^{2}}\wedge T^{2}\right)\cdot O(1)\cdot O\left(\dfrac{1}{|1-\rho_{\xi,T}|}\wedge T\right)\\
 & =O\left(\left[\dfrac{1}{|1-\rho_{\xi,T}|}\wedge T\right]^{2}\cdot\left[\dfrac{1}{|1-\rho_{\eta,T}|}\wedge T\right]^{2}\right).
\end{align*}
Thus, by \eqref{eq:rho_ineq2} we have
\begin{equation}
\mathbb{E}\left(\Xi_{2,T}^{2}\right)=O\left(\dfrac{T^{2}}{(1-\rho_{\xi,T}\rho_{\eta,T})^{2}}\wedge T^{4}\right).\label{eq:Xi2_bound}
\end{equation}
We can show that $\mathbb{E}\bigl(\Xi_{3,T}^{3}\bigr)$ has the same
order by the above argument.

Bound for $\mathbb{E}\bigl(\Xi_{4,T}^{2}\bigr)$. It is easy to deduce
\begin{align}
\mathbb{E}\left(\Xi_{4,T}^{2}\right) & =\mathbb{E}\left(\xi_{0}^{2}\eta_{0}^{2}\right)\left[\sum_{t=1}^{T}(\rho_{\xi,T}\rho_{\eta,T})^{t}\right]^{2}\nonumber \\
 & =O\left(\left[\dfrac{1}{|1-\rho_{\xi,T}|}\wedge T\right]\cdot\left[\dfrac{1}{|1-\rho_{\eta,T}|}\wedge T\right]\right)\cdot O\left(\dfrac{1}{(1-\rho_{\xi,T}\rho_{\eta,T})^{2}}\wedge T^{2}\right)\nonumber \\
 & =O\left(\dfrac{T}{(1-\rho_{\xi,T}\rho_{\eta,T})^{3}}\wedge T^{3}\right).\label{eq:Xi4_bound}
\end{align}
By \eqref{eq:Xi1_bound}--\eqref{eq:Xi4_bound}, it follows that
\[
\mathbb{E}\left[\left(\sum_{t=1}^{T}\xi_{t}\eta_{t}\right)^{2}\right]=O\left(\dfrac{T^{2}}{(1-\rho_{\xi,T}\rho_{\eta,T})^{2}}\wedge T^{4}\right).
\]
This completes the proof of \prettyref{lem:generic_two_ARs}\prettyref{enu:Esum2}.
\end{proof}
\begin{proof}[Proof of \prettyref{lem:generic_two_ARs}\prettyref{enu:Esumsum2}]
By \prettyref{lem:generic_AR}\ref{enu:E_sum_eta_4} we have
\begin{align*}
\mathbb{E}\left[\left(\sum_{t=1}^{T}\xi_{t}\sum_{t=1}^{T}\eta_{t}\right)^{2}\right] & \leq\left\{ \mathbb{E}\left[\left(\sum_{t=1}^{T}\xi_{t}\right)^{4}\right]\cdot\mathbb{E}\left[\left(\sum_{t=1}^{T}\eta_{t}\right)^{4}\right]\right\} ^{1/2}\\
 & =\left\{ O\left(\frac{T^{2}}{(1-\rho_{\xi,T})^{4}}\wedge T^{6}\right)\cdot O\left(\frac{T^{2}}{(1-\rho_{\eta,T})^{4}}\wedge T^{6}\right)\right\} ^{1/2}\\
 & =O\left(\left[\frac{T}{(1-\rho_{\xi,T})^{2}}\wedge T^{3}\right]\cdot\left[\frac{T}{(1-\rho_{\eta,T})^{2}}\wedge T^{3}\right]\right),
\end{align*}
 where the inequality invokes the Cauchy-Schwarz. This completes the
proof of \prettyref{lem:generic_two_ARs}\prettyref{enu:Esumsum2}.
\end{proof}

\section{Proofs of Technical Lemmas for WG \label{sec:Proofs-for-WG-1}}
\begin{proof}[Proof of \prettyref{lem:epct}]
Note that $\{x_{i,t}\}$ is an AR(1) process with coefficient $\rho^{*}=1+c^{*}/T^{\gamma}$
and $\{e_{i,t}\}$ is an m.d.s.\ that can also be considered as an
AR(1) process with $\rho_{e}=0$. Then Items \ref{enu:E_sum_x_2}\ref{enu:E_sum_x_4}
follow from \prettyref{lem:generic_AR}\ref{enu:E_sum_eta_4}, and
Items \ref{enu:E_sum_x2_2}\ref{enu:E_sum_x_sum_e_2} follow from
\prettyref{lem:generic_two_ARs}\ref{enu:Esum2}\ref{enu:Esumsum2},
respectively.
\end{proof}
\begin{proof}[Proof of \prettyref{lem:wg-ts-converge}]
 If $\gamma\in(0,1)$, Parts \ref{enu:Q_wg} and \ref{enu:Z_wg}
follow from Equation~(7) and Lemma~3.3 of \citet{magdalinos2009limit},
respectively. If $\gamma=0$, the ergodic theorem and the classic
martingale CLT give rise to the same results.
\end{proof}
\begin{proof}[Proof of \prettyref{lem:joint}]
(i) By \prettyref{lem:wg-ts-converge}\ref{enu:Q_wg}, we know that
$Q_{i,T}^{{\rm WG}}\to_{p}V_{xx}$ as $T\to\infty$ for each $i$.
Since $Q_{i,T}^{\mathrm{WG}}$ is nonnegative and integrable, if we
show $\mathbb{E}\bigl(Q_{i,T}^{\mathrm{WG}}\bigr)\to V_{xx}$ as $T\to\infty$,
then $Q_{i,T}^{\mathrm{WG}}$ is u.i.\ in $T$ by \prettyref{lem:u.i.}.
Using the DGP formula $x_{i,t}=\sum_{j=1}^{t}\rho^{*t-j}v_{i,j}+\rho^{*t}x_{i,0}$
we have
\begin{align}
\mathbb{E}\bigl(Q_{i,T}^{\mathrm{WG}}\bigr) & =\frac{1}{T^{1+\gamma}}\sum_{t=1}^{T}\mathbb{E}\left[\Biggl(\sum_{j=1}^{t}\rho^{*t-j}v_{i,j}+\rho^{*t}x_{i,0}\Biggr)^{2}\right]\nonumber \\
 & =\frac{1}{T^{1+\gamma}}\Biggl[\sum_{t=1}^{T}\sum_{j=1}^{t}\rho^{*2(t-j)}\mathbb{E}(v_{i,j}^{2})+\sum_{t=1}^{T}\rho^{*2t}\mathbb{E}(x_{i,0}^{2})\Biggr]\nonumber \\
 & =\frac{1}{T^{1+\gamma}}\biggl[\omega_{vv}^{*}\biggl(\frac{T}{1-\rho^{*2}}-\frac{\rho^{*2}(1-\rho^{*2T})}{(1-\rho^{*2})^{2}}\biggr)+\dfrac{\rho^{*2}(1-\rho^{*2T})}{1-\rho^{*2}}\mathbb{E}(x_{i,0}^{2})\biggr]\nonumber \\
 & \to V_{xx}=\begin{cases}
\omega_{vv}^{*}/(1-\rho^{*2}), & \text{if }\gamma=0,\\
\omega_{vv}^{*}/(-2c^{*}), & \text{if }\gamma\in(0,1),
\end{cases}\qquad\text{as }T\to\infty,\label{eq:Ex2 lim}
\end{align}
where the last convergence is due to, as $T\to\infty$,
\begin{align*}
\frac{\omega_{vv}^{*}}{T^{1+\gamma}}\cdot\frac{T}{1-\rho^{*2}} & =\dfrac{\omega_{vv}^{*}}{T^{1+\gamma}}\cdot\frac{T}{\frac{-c^{*}}{T^{\gamma}}(2+\frac{c^{*}}{T^{\gamma}})}\to V_{xx}=\begin{cases}
\omega_{vv}^{*}/(1-\rho^{*2}), & \text{if }\gamma=0,\\
\omega_{vv}^{*}/(-2c^{*}), & \text{if }\gamma\in(0,1),
\end{cases}\\
\frac{\omega_{vv}^{*}}{T^{1+\gamma}}\frac{\rho^{*2}(1-\rho^{*2T})}{(1-\rho^{*2})^{2}} & =O\left(\dfrac{T^{2\gamma}}{T^{1+\gamma}}\right)\to0,\qquad\text{and}\\
\frac{1}{T^{1+\gamma}}\dfrac{\rho^{*2}(1-\rho^{*2T})}{1-\rho^{*2}}\mathbb{E}(x_{i,0}^{2}) & =O\left(\dfrac{T^{\gamma}}{T^{1+\gamma}}\right)\cdot O\left(T^{\gamma}\right)\to0.
\end{align*}
Hence, \eqref{eq:Ex2 lim} is proved and $Q_{i,T}^{\mathrm{WG}}$
is u.i.\ in $T$. By Corollary~1 of \citet{phillips1999linear},
we can establish the joint convergence
\[
\frac{1}{n}\sum_{i=1}^{n}Q_{i,T}^{{\rm WG}}\to_{p}\mathop{\mathrm{plim}}_{T\to\infty}Q_{i,T}^{{\rm WG}}=V_{xx}\qquad\text{as }(n,T)\to\infty.
\]

\medskip{}
(ii) By \prettyref{lem:epct}\ref{enu:E_sum_x_2}, we have as $T\to\infty$,
\[
\mathbb{E}(R_{i,T}^{{\rm WG}})=\frac{1}{T^{2+\gamma}}\cdot O(T^{1+2\gamma})=O\biggl(\frac{1}{T^{1-\gamma}}\biggr)\to0.
\]
In addition, $R_{i,T}^{{\rm WG}}$ is clearly nonnegative and integrable.
Thus, $R_{i,T}^{{\rm WG}}\to_{p}0$ as $T\to\infty$ by Markov's inequality
and it is u.i.\ in $T$ by \prettyref{lem:u.i.}. By Corollary~1
of \citet{phillips1999linear}, we have the joint convergence $\frac{1}{n}\sum_{i=1}^{n}R_{i,T}^{{\rm WG}}\to_{p}\mathop{\mathrm{plim}}_{T\to\infty}R_{i,T}^{{\rm WG}}=0$
as $(n,T)\to\infty$.

\medskip{}
(iii) It is clear that $Z_{i,T}^{{\rm WG}}$ has mean zero and is
square integrable. If we further show that $(Z_{i,T}^{{\rm WG}})^{2}$
is u.i.\ in $T$, then by Theorem~3 of \citet{phillips1999linear},
we would have the joint convergence
\[
\frac{1}{\sqrt{n}}\sum_{i=1}^{n}Z_{i,T}^{{\rm WG}}\to_{d}\mathcal{N}\Bigl(0,\lim_{T\to\infty}\mathbb{E}\bigl[(Z_{i,T}^{{\rm WG}})^{2}\bigr]\Bigr)\qquad\text{as }(n,T)\to\infty.
\]
We now show that $(Z_{i,T}^{{\rm WG}})^{2}$ is u.i.\ in $T$ and
$\lim_{T\to\infty}\mathbb{E}\bigl[(Z_{i,T}^{{\rm WG}})^{2}\bigr]=\omega_{ee}^{*}V_{xx}$.

By \prettyref{lem:wg-ts-converge}\ref{enu:Z_wg},
\[
Z_{i,T}^{{\rm WG}}\to_{d}Z_{\infty}\sim\mathcal{N}(0,\omega_{ee}^{*}V_{xx})\qquad\text{as }T\to\infty.
\]
Then by the continuous mapping theorem, $(Z_{i,T}^{{\rm WG}})^{2}\to_{d}Z_{\infty}^{2}$
as $T\to\infty,$ where $\mathbb{E}(Z_{\infty}^{2})=\omega_{ee}^{*}V_{xx}$.
By \prettyref{lem:u.i.}, to show that $(Z_{i,T}^{{\rm WG}})^{2}$
is u.i.\ in $T$, it suffices to show
\begin{equation}
\lim_{T\to\infty}\mathbb{E}\left[(Z_{i,T}^{{\rm WG}})^{2}\right]=\mathbb{E}(Z_{\infty}^{2})=\omega_{ee}^{*}V_{xx}.\label{eq:limEZWG2I0MI}
\end{equation}
Noting that $\{x_{i,t}e_{i,t+1}\}$ is an m.d.s., we have
\begin{align}
\mathbb{E}\left[(Z_{i,T}^{{\rm WG}})^{2}\right] & =\frac{1}{T^{1+\gamma}}\sum_{t=1}^{T}\mathbb{E}(x_{i,t}^{2}e_{i,t+1}^{2})=\frac{\omega_{ee}^{*}}{T^{1+\gamma}}\sum_{t=1}^{T}\mathbb{E}(x_{i,t}^{2})=\omega_{ee}^{*}\mathbb{E}\bigl(Q_{i,T}^{\mathrm{WG}}\bigr).\label{eq: EZ2EX2}
\end{align}
By \eqref{eq:Ex2 lim}, $\mathbb{E}\bigl(Q_{i,T}^{\mathrm{WG}}\bigr)\to V_{xx}$
as $T\to\infty$, and thus \eqref{eq:limEZWG2I0MI} holds and $(Z_{i,T}^{{\rm WG}})^{2}$
is u.i.\ in $T$, giving rise to the desired joint convergence.

\medskip{}
(iv) Note that
\[
\frac{1}{\sqrt{n}}\sum_{i=1}^{n}L_{i,T}^{{\rm WG}}=\frac{1}{\sqrt{n}}\sum_{i=1}^{n}Z_{i,T}^{{\rm WG}}-\frac{1}{\sqrt{n}}\sum_{i=1}^{n}\left[H_{i,T}^{{\rm WG}}-\mathbb{E}\left(H_{i,T}^{{\rm WG}}\right)\right].
\]
By (iii), it suffices to show
\begin{equation}
\frac{1}{\sqrt{n}}\sum_{i=1}^{n}\left[H_{i,T}^{{\rm WG}}-\mathbb{E}\left(H_{i,T}^{{\rm WG}}\right)\right]\to_{p}0\qquad\text{as }(n,T)\to\infty.\label{eq:H limit 0}
\end{equation}
By \prettyref{lem:epct}\ref{enu:E_sum_x_sum_e_2}, for any $\gamma\in[0,1)$,
\[
\mathbb{E}\left[(H_{i,T}^{{\rm WG}})^{2}\right]=\frac{1}{T^{3+\gamma}}\cdot O(T^{2(1+\gamma)})=O\left(\frac{1}{T^{1-\gamma}}\right)\to0\qquad\text{as }T\to\infty.
\]
Then by the i.i.d.\ condition across $i$, as $(n,T)\to\infty$:
\[
\mathbb{E}\left[\left(\frac{1}{\sqrt{n}}\sum_{i=1}^{n}\left[H_{i,T}^{{\rm WG}}-\mathbb{E}\left(H_{i,T}^{{\rm WG}}\right)\right]\right)^{2}\right]=\mathbb{E}\left(\left[H_{i,T}^{{\rm WG}}-\mathbb{E}\left(H_{i,T}^{{\rm WG}}\right)\right]^{2}\right)\leq\mathbb{E}\left[(H_{i,T}^{{\rm WG}})^{2}\right]\to0.
\]
By Markov's inequality we have \eqref{eq:H limit 0}. The desired
joint convergence then follows.
\end{proof}
\begin{proof}[Proof of Lemma~\ref{lem:lurjoint}]
(i) Recall $Q_{i,T}^{{\rm WG}}=T^{-2}\sum_{t=1}^{T}x_{i,t}^{2}$
when $\gamma=1$. By \prettyref{lem:epct}\ref{enu:E_sum_x2_2},
\[
\mathbb{E}\left[(Q_{i,T}^{{\rm WG}})^{2}\right]=T^{-4}\cdot O(T^{4})=O(1)\qquad\text{as }T\to\infty.
\]
Thus by \prettyref{lem:u.i.}, $Q_{i,T}$ is u.i.\ in $T$. By Lemma~1(c)
of \citet{phillips1987towards}, we obtain
\begin{equation}
Q_{i,T}^{{\rm WG}}\to_{d}\int_{0}^{1}J_{2,c^{*}}(r)^{2}\:dr\qquad\text{as }T\to\infty.\label{eq:QWG lim dis}
\end{equation}
Then, by Corollary~1 of \citet{phillips1999linear}, we have
\[
\frac{1}{n}\sum_{i=1}^{n}Q_{i,T}^{{\rm WG}}\to_{p}\mathbb{E}\left(\int_{0}^{1}J_{2,c^{*}}(r)^{2}\:dr\right)=\Omega_{c^{*}}\qquad\text{as }(n,T)\to\infty.
\]

\medskip{}
(ii) Recall $R_{i,T}^{{\rm WG}}=(T^{-3/2}\sum_{t=1}^{T}x_{i,t})^{2}$
when $\gamma=1$. By \prettyref{lem:epct}(ii),
\[
\mathbb{E}\left[(R_{i,T}^{{\rm WG}})^{2}\right]=T^{-6}\cdot O(T^{6})=O(1)\qquad\text{as }T\to\infty.
\]
Thus by \prettyref{lem:u.i.}(i), $R_{i,T}^{{\rm WG}}$ is u.i.\ in
$T$. By Lemma~1(b) of \citet{phillips1987towards} and the continuous
mapping theorem, we obtain $R_{i,T}^{{\rm WG}}\to_{d}\bigl(\int_{0}^{1}J_{2,c^{*}}(r)\,dr\bigr)^{2}$
as $T\to\infty.$ Thus, by Corollary~1 of \citet{phillips1999linear},
we have
\[
\frac{1}{n}\sum_{i=1}^{n}R_{i,T}^{{\rm WG}}\to_{p}\mathbb{E}\left[\biggl(\int_{0}^{1}J_{2,c^{*}}(r)\,dr\biggr)^{2}\right]=\Sigma_{c^{*}}\qquad\text{as }(n,T)\to\infty.
\]

\medskip{}
(iii) By Lemma~1(a)(b) of \citet{phillips1987towards}, we obtain
\[
\frac{1}{\sqrt{T}}\tilde{x}_{i,t=[Tr]}=\frac{1}{\sqrt{T}}x_{i,t=[Tr]}-\frac{1}{T^{3/2}}\sum_{s=1}^{T}x_{i,s}\to_{d}J_{2,c^{*}}(r)-\int_{0}^{1}J_{2,c^{*}}(\tau)\,d\tau
\]
as $T\to\infty,$ where $[Tr]$ means integer part of $Tr$. Using
the same argument from that lemma, we can show, as $T\to\infty$,
\begin{align}
Z_{i,T}^{{\rm WG}}-H_{i,T}^{\mathrm{WG}} & =\frac{1}{T}\sum_{t=1}^{T}\tilde{x}_{i,t}e_{i,t+1}\to_{d}\int_{0}^{1}\biggl(J_{2,c^{*}}(r)-\int_{0}^{1}J_{2,c^{*}}(\tau)\,d\tau\biggr)\,dB_{1}(r)\qquad\text{and}\nonumber \\
H_{i,T}^{\mathrm{WG}} & =\frac{1}{\sqrt{T}}\sum_{t=1}^{T}e_{i,t+1}\frac{1}{T^{3/2}}\sum_{t=1}^{T}x_{i,t}\to_{d}B_{1}(1)\int_{0}^{1}J_{2,c^{*}}(\tau)\,d\tau.\label{eq:H lim dist}
\end{align}
If we show that $(Z_{i,T}^{{\rm WG}})^{2}$ and $(H_{i,T}^{{\rm WG}})^{2}$
are u.i.\ in $T$, then by \prettyref{lem:u.i.},\footnote{By \prettyref{lem:epct}\ref{enu:E_sum_x_sum_e_2} we can also show
$\mathbb{E}\bigl[(H_{i,T}^{{\rm WG}})^{2}\bigr]=O(1)$ when $\gamma=1$.} as $T\to\infty$,
\begin{align*}
\mathbb{E}\bigl(H_{i,T}^{{\rm WG}}\bigr) & \to\mathbb{E}\left(B_{1}(1)\int_{0}^{1}J_{2,c^{*}}(\tau)\,d\tau\right)\qquad\text{and}\\
\mathbb{E}\bigl[(H_{i,T}^{{\rm WG}})^{2}\bigr] & \to\mathbb{E}\left[\left(B_{1}(1)\int_{0}^{1}J_{2,c^{*}}(\tau)\,d\tau\right)^{2}\right]<\infty
\end{align*}
which implies $\mathbb{E}\bigl[(H_{i,T}^{{\rm WG}})^{2}\bigr]=O(1)$.
Therefore, $L_{i,T}^{{\rm WG}}=Z_{i,T}^{{\rm WG}}-\bigl[H_{i,T}^{{\rm WG}}-\mathbb{E}\bigl(H_{i,T}^{{\rm WG}}\bigr)\bigr]$
is square u.i.\ in $T$ and, as $T\to\infty$,
\begin{equation}
L_{i,T}^{{\rm WG}}\to_{d}\int_{0}^{1}\biggl(J_{2,c^{*}}(r)-\int_{0}^{1}J_{2,c^{*}}(\tau)\,d\tau\biggr)\,dB_{1}(r)+\mathbb{E}\left(B_{1}(1)\int_{0}^{1}J_{2,c^{*}}(\tau)\,d\tau\right).\label{eq:limit dist L}
\end{equation}
Then, by Theorem~3 of \citet{phillips1999linear}, since $(L_{i,T}^{{\rm WG}})^{2}$
is u.i.\ in $T$,
\begin{align*}
\frac{1}{n}\sum_{i=1}^{n}L_{i,T}^{{\rm WG}} & \to_{d}\mathcal{N}(0,\Sigma_{\tilde{x}e})\qquad\text{as }(n,T)\to\infty.
\end{align*}
Therefore, it suffices to show the uniform integrability of $(Z_{i,T}^{{\rm WG}})^{2}$
and $(H_{i,T}^{{\rm WG}})^{2}$.

\textbf{Step I. }Showing $(Z_{i,T}^{{\rm WG}})^{2}$ is u.i.\ in
$T$. Recall $Z_{i,T}^{{\rm WG}}=T^{-1}\sum_{t=1}^{T}x_{i,t}e_{i,t+1}$
when $\gamma=1$. By Lemma~1(d) of \citet{phillips1987towards} and
the continuous mapping theorem,
\[
(Z_{i,T}^{{\rm WG}})^{2}\to_{d}\left(\int_{0}^{1}J_{2,c^{*}}(r)\,dB_{1}(r)\right)^{2}\qquad\text{as }T\to\infty.
\]
Now $\rho^{*}=1+c^{*}/T$. When $c^{*}\neq0$, similar to \eqref{eq:Ex2 lim},
we have
\begin{align*}
\mathbb{E}\bigl(Q_{i,T}^{\mathrm{WG}}\bigr) & =\frac{1}{T^{2}}\biggl[\omega_{vv}^{*}\biggl(\frac{T}{1-\rho^{*2}}-\frac{\rho^{*2}(1-\rho^{*2T})}{(1-\rho^{*2})^{2}}\biggr)+\dfrac{\rho^{*2}(1-\rho^{*2T})}{1-\rho^{*2}}\mathbb{E}(x_{i,0}^{2})\biggr]\\
 & =\frac{1}{T^{2}}\biggl[\omega_{vv}^{*}\biggl(\frac{T}{\frac{-c^{*}}{T}(2+\frac{c^{*}}{T})}-\frac{\rho^{*2}(1-\rho^{*2T})}{\frac{c^{*2}}{T^{2}}(2+\frac{c^{*}}{T})^{2}}\biggr)+O(T)\cdot o(T)\biggr]\\
 & \to\omega_{vv}^{*}\frac{e^{2c^{*}}-2c^{*}-1}{4c^{*2}}\qquad\text{as }T\to\infty,
\end{align*}
where it should be noted that $\mathbb{E}(x_{i,0}^{2})=o(T)$ by \prettyref{assump:initval}
and $\rho^{*T}\to\exp(c^{*})$. We still have \eqref{eq: EZ2EX2}
so that as $T\to\infty$:
\begin{align*}
\mathbb{E}[(Z_{i,T}^{{\rm WG}})^{2}] & =\omega_{ee}^{*}\mathbb{E}\bigl(Q_{i,T}^{\mathrm{WG}}\bigr)\to\omega_{ee}^{*}\omega_{vv}^{*}\frac{e^{2c^{*}}-2c^{*}-1}{4c^{*2}}=\mathbb{E}\left[\left(\int_{0}^{1}J_{2,c^{*}}(r)dB_{1}(r)\right)^{2}\right]
\end{align*}
When $c^{*}=0$, by the second row of \eqref{eq:Ex2 lim} and \eqref{eq: EZ2EX2},
as $T\to\infty$,
\begin{align*}
\mathbb{E}[(Z_{i,T}^{{\rm WG}})^{2}] & =\omega_{ee}^{*}\mathbb{E}\bigl(Q_{i,T}^{\mathrm{WG}}\bigr)=\omega_{ee}^{*}\omega_{vv}^{*}\left[\frac{1}{T^{2}}\sum_{t=1}^{T}t+\frac{1}{T^{2}}\cdot o(T^{2})\right]\\
 & \to\frac{\omega_{ee}^{*}\omega_{vv}^{*}}{2}=\mathbb{E}\left[\left(\int_{0}^{1}B_{2}(r)\,dB_{1}(r)\right)^{2}\right].
\end{align*}
Thus by \prettyref{lem:u.i.}, $(Z_{i,T}^{{\rm WG}})^{2}$ is u.i.\ in
$T$.

\textbf{Step II}. Showing  $(H_{i,T}^{{\rm WG}})^{2}$ is u.i.\ in
$T$. Recall $H_{i,T}^{{\rm WG}}=T^{-2}\sum_{t=1}^{T}x_{i,t}\sum_{t=1}^{T}e_{i,t+1}$
when $\gamma=1$. By \eqref{eq:H lim dist} and the continuous mapping
theorem,
\[
(H_{i,T}^{{\rm WG}})^{2}\to_{d}\left(B_{1}(1)\int_{0}^{1}J_{2,c^{*}}(\tau)\,d\tau\right)^{2}\qquad\text{as }T\to\infty.
\]
Here we only focus on the case where $c^{*}<0$ so that $\rho^{*}=1+c^{*}/T<1$;
the case where $\rho^{*}=1$ is relatively simpler. Using the DGP
formula $x_{i,t}=\sum_{j=1}^{t}\rho^{*t-j}v_{i,j}+\rho^{*t}x_{i,0}$
we have
\[
\sum_{t=1}^{T}x_{i,t}=\sum_{j=1}^{T}\Biggl(\sum_{t=j}^{T}\rho^{*t-j}\Biggr)v_{i,j}+\text{\ensuremath{\sum_{t=1}^{T}}\ensuremath{\ensuremath{\rho^{*t}x_{i,0}}.}}
\]
For convenience, let $v_{i,0}=x_{i,0}$ and define
\[
r_{T,j}:=\begin{cases}
\sum_{t=j}^{T}\rho^{*t-j}, & \text{for }j>0,\\
\sum_{t=1}^{T}\ensuremath{\rho^{*t}}, & \text{for }j=0.
\end{cases}
\]
Then we can write $\sum_{t=1}^{T}x_{i,t}=\sum_{j=0}^{T}r_{T,j}v_{i,j}$.
We decompose $(H_{i,T}^{{\rm WG}})^{2}$ as
\begin{align*}
(H_{i,T}^{{\rm WG}})^{2} & =\frac{1}{T^{4}}\left(\sum_{t=1}^{T}x_{i,t}\right)^{2}\cdot\Biggl(\sum_{t=1}^{T}e_{i,t+1}^{2}+\sum_{t=1}^{T}\sum_{s\neq t}e_{i,s+1}e_{i,t+1}\Biggr)\\
 & =\frac{1}{T^{4}}\Biggl(\sum_{j=0}^{T}r_{T,j}v_{i,j}\Biggr)^{2}\cdot\left(I_{i,T}+\mathit{II_{i,T}}\right).
\end{align*}
Note that $r_{T,j}=O(T)$ uniformly for all $j\leq T$. Under \prettyref{assump:iid},
we deduce
\begin{align}
 & \mathbb{E}\left[\frac{1}{T^{4}}\Biggl(\sum_{j=0}^{T}r_{T,j}v_{i,j}\Biggr)^{2}\cdot I_{i,T}\right]=\frac{1}{T^{4}}\sum_{t=1}^{T}\mathbb{E}\left[\Biggl(\sum_{j=0}^{T}r_{T,j}v_{i,j}\Biggr)^{2}\cdot e_{i,t+1}^{2}\right]\nonumber \\
={} & \frac{1}{T^{4}}\sum_{t=1}^{T}\sum_{\substack{1\leq j\leq T-1\\
j\neq t+1
}
}r_{T,j}^{2}\mathbb{E}\left(v_{i,j}^{2}\cdot e_{i,t+1}^{2}\right)+\frac{1}{T^{4}}r_{T,0}^{2}\sum_{t=1}^{T}\mathbb{E}\left(v_{i,0}^{2}\cdot e_{i,t+1}^{2}\right)+\frac{1}{T^{4}}\sum_{t=1}^{T}r_{T,t+1}^{2}\mathbb{E}\left(v_{i,t+1}^{2}\cdot e_{i,t+1}^{2}\right)\nonumber \\
={} & \frac{\omega_{ee}^{*}\omega_{vv}^{*}}{T^{4}}\Biggl(\sum_{t=1}^{T}\sum_{j=1}^{T}r_{T,j}^{2}-\sum_{j=2}^{T}r_{T,j}^{2}\Biggr)+O(T^{-1})\cdot\mathbb{E}(x_{i,0}^{2})\omega_{ee}^{*}+O(T^{-1})\cdot O(1)\nonumber \\
={} & \frac{\omega_{ee}^{*}\omega_{vv}^{*}}{T^{3}}\sum_{j=1}^{T}r_{T,j}^{2}+o(1)=\dfrac{\omega_{ee}^{*}\sigma_{vv}^{*}}{[T(1-\rho^{*})]^{2}}\left[1-\dfrac{2(\rho^{*}-\rho^{*T+1})}{T(1-\rho^{*})}+\dfrac{\rho^{*2}-\rho^{*2(T+1)})}{T(1-\rho^{*2})}\right]+o(1)\nonumber \\
\to{} & \omega_{ee}^{*}\omega_{vv}^{*}\frac{2c^{*}+(1-e^{c^{*}})(3-e^{c^{*}})}{2c^{*3}}\qquad\text{as }T\to\infty,\label{eq:EI3}
\end{align}
and
\begin{align}
 & \mathbb{E}\left[\frac{1}{T^{4}}\Biggl(\sum_{j=0}^{T}r_{T,j}v_{i,j}\Biggr)^{2}\cdot\mathit{II_{i,T}}\right]=\frac{1}{T^{4}}\sum_{t=1}^{T}\sum_{s\neq t}\mathbb{E}\left[\Biggl(\sum_{j=0}^{T}r_{T,j}v_{i,j}\Biggr)^{2}\cdot e_{i,s+1}e_{i,t+1}\right]\nonumber \\
 & =\frac{2}{T^{4}}\sum_{t=1}^{T-1}\sum_{s\neq t}r_{T,s+1}r_{T,t+1}\mathbb{E}(v_{i,s+1}e_{i,s+1}\cdot v_{i,t+1}e_{i,t+1})\nonumber \\
 & =\frac{2}{T^{4}}\sum_{t=1}^{T-1}\sum_{s\neq t}r_{T,s+1}r_{T,t+1}\omega_{ev}^{*2}=\frac{2\omega_{ev}^{*2}}{T^{4}}\left(\sum_{t=1}^{T-1}r_{T,t+1}\right)^{2}-\frac{2\omega_{ev}^{*2}}{T^{4}}\sum_{t=1}^{T-1}r_{T,t+1}^{2}\nonumber \\
 & =2\omega_{ev}^{*2}\left[\dfrac{T-1}{T^{2}(1-\rho^{*})}-\dfrac{\rho^{*}-\rho^{*T}}{T^{2}(1-\rho^{*})^{2}}\right]^{2}+o(1)\nonumber \\
 & \to2\omega_{ev}^{*2}\left(\frac{1}{-c^{*}}-\frac{1-e^{c^{*}}}{c^{*2}}\right)^{2}=2\omega_{ev}^{*2}\frac{(e^{c^{*}}-c^{*}-1)^{2}}{c^{*4}}.\label{eq:EII3}
\end{align}
By \eqref{eq:EI3} and \eqref{eq:EII3}, as $T\to\infty$,
\begin{align*}
\mathbb{E}\bigl[(H_{i,T}^{{\rm WG}})^{2}\bigr] & \to\omega_{ee}^{*}\omega_{vv}^{*}\frac{2c^{*}+(1-e^{c^{*}})(3-e^{c^{*}})}{2c^{*3}}+2\omega_{ev}^{*2}\frac{(e^{c^{*}}-c^{*}-1)^{2}}{c^{*4}}\\
 & =\mathbb{E}\left[\left(B_{1}(1)\int_{0}^{1}J_{2,c^{*}}(\tau)\,d\tau\right)^{2}\right],
\end{align*}
where the last equality is by \eqref{eq:E of squared B_int_J}. By
\prettyref{lem:u.i.}, $(H_{i,T}^{{\rm WG}})^{2}$ is u.i.\ in $T$.
This completes the proof of \prettyref{lem:lurjoint}.
\end{proof}
\begin{proof}[Proof of \prettyref{lem:esterror}]
By \eqref{eq:sigmaWG lim} we have $1/\varsigma^{{\rm WG}}=O_{p}(\sqrt{nT^{1+\gamma}})$.
Then
\begin{align*}
r_{n,T}^{\mathrm{WG}}(\hat{\rho}) & =\bigl[b_{n,T}^{\mathrm{WG}}(\hat{\rho})-b_{n,T}^{\mathrm{WG}}(\rho^{*})\bigr]O_{p}(\sqrt{nT^{1+\gamma}})=\Biggl(\sum_{t=2}^{T_{1}}\sum_{s=2}^{t}\hat{\rho}^{t-s}-\sum_{t=2}^{T_{1}}\sum_{s=2}^{t}\rho^{*t-s}\Biggr)O_{p}\left(\sqrt{\frac{n}{T^{3+\gamma}}}\right)\\
 & =\sum_{t=2}^{T_{1}}\sum_{s=2}^{t}(\hat{\rho}^{t-s}-\rho^{*t-s})O_{p}\left(\sqrt{\frac{n}{T^{3+\gamma}}}\right).
\end{align*}
It remains to bound $\sum_{t=2}^{T_{1}}\sum_{s=2}^{t}(\hat{\rho}^{t-s}-\rho^{*t-s})$.
Define
\[
f_{T}(\rho):=\sum_{t=2}^{T_{1}}\sum_{s=2}^{t}\rho^{t-s}=\frac{T_{1}}{1-\rho}-\frac{1}{(1-\rho)^{2}}+\frac{\rho^{T_{1}}}{(1-\rho)^{2}}.
\]
By the differential mean value theorem, there exists some $\check{\rho}$
between $\hat{\rho}$ and $\rho^{*}$ such that
\[
f_{T}(\hat{\rho})-f_{T}(\rho^{*})=f_{T}^{\prime}(\check{\rho})(\hat{\rho}-\rho^{*}).
\]
When $\gamma<1$ and $c^{*}<0$, by $\hat{\rho}-\rho^{*}=O_{p}\bigl(T^{-\frac{1+\gamma}{2}}\bigr)$
we have that with high probability,
\[
0<1+\frac{2c^{*}}{T^{\gamma}}\leq\hat{\rho}=1+\frac{c^{*}}{T^{\gamma}}+O_{p}\biggl(\frac{1}{T^{(1+\gamma)/2}}\biggr)\leq1+\frac{c^{*}}{2T^{\gamma}}
\]
for $T$ large enough. This together with the fact that $\check{\rho}$
is between $\rho^{*}$ and $\hat{\rho}$ implies
\[
1+\frac{2c^{*}}{T^{\gamma}}\leq\check{\rho}\leq1+\frac{c^{*}}{2T^{\gamma}}
\]
for $T$ large enough. Accordingly, we have
\[
f_{T}^{\prime}(\check{\rho})=\frac{T_{1}}{(1-\check{\rho})^{2}}-\frac{2}{(1-\check{\rho})^{3}}+\frac{T_{1}\check{\rho}^{T_{1}-1}}{(1-\check{\rho})^{2}}+\frac{2\check{\rho}^{T_{1}}}{(1-\check{\rho})^{3}}=O_{p}\left(T^{1+2\gamma}\right).
\]
Furthermore, when $\gamma=1$, $\hat{\rho}=1+O_{p}(T^{-1})$ and thus
$\check{\rho}=1+O_{p}(T^{-1})$. Then
\begin{align*}
f_{T}^{\prime}(\check{\rho}) & =\sum_{t=3}^{T_{1}}\sum_{s=3}^{t}(t-s)\check{\rho}^{t-s-1}\leq\sum_{t=3}^{T_{1}}\sum_{s=3}^{t}(t-s)\bigl[1+O_{p}(T^{-1})\bigr]^{T}\\
 & =O(T^{3})\cdot O_{p}(1)=O_{p}(T^{3}).
\end{align*}
Thus, for any $\gamma\in[0,1]$, $f'_{T}(\check{\rho})=O_{p}(T^{1+2\gamma})$
and it follows that
\[
r_{n,T}^{\mathrm{WG}}(\hat{\rho})=O_{p}\left(T^{1+2\gamma}|\hat{\rho}-\rho^{*}|\cdot\sqrt{\frac{n}{T^{3+\gamma}}}\right)=O_{p}\left(\sqrt{\frac{n}{T^{1-3\gamma}}}|\hat{\rho}-\rho^{*}|\right).
\]
This completes the proof.
\end{proof}

\section{Proofs of Technical Lemmas for IVX\label{sec:Proofs-for-IVX-1}}

We note that the lower bounds imposed on $\theta$ and $\gamma$ in
\citet{phillips2009econometric}, which are required to handle the
effect of the long run covariances, are not necessary for m.d.s.;
see their Proposition~A2. Also, the asymptotics in \citet{phillips2009econometric}
for LUR accommodates the locally explosive regressor with $\rho^{*}=1+c^{*}/T$
where $c^{*}>0$; see also the paragraph right before Theorem~2.1
of \citet{phillips2016robust}.
\begin{proof}[Proof of \prettyref{lem:ivxepct}]
For simplicity, we assume $\rho^{*}\geq0$ and $\rho_{z}\neq\rho^{*}$
without loss of generality. Note that $\zeta_{i,t}$ and $x_{i,t}$
are AR(1) processes with coefficients $\rho_{z}=1+c_{z}/T^{\theta}$
and $\rho^{*}=1+c^{*}/T^{\gamma}$. Thus, $(1-\rho_{z}\rho^{*})^{-1}=O(T^{\theta\wedge\gamma})$
and Part~\ref{enu:zeta_x} follows from \prettyref{lem:generic_two_ARs}\ref{enu:Esum2}--\ref{enu:Esumsum2}.

\textbf{Item \ref{enu:phi_x}.} Note that
\begin{align}
\psi_{i,t} & =\sum_{j=0}^{t-1}\rho_{z}^{t-1-j}x_{i,j}=\sum_{j=0}^{t-1}\rho_{z}^{t-1-j}\left[\sum_{\ell=-\infty}^{j}\pi_{i,T}(j,\ell)\varepsilon_{i,\ell}+\rho^{*j}x_{i,0}\right]\nonumber \\
 & =\sum_{\ell=-\infty}^{t-1}\Biggl(\sum_{j=\ell\vee0}^{t-1}\rho_{z}^{t-1-j}\pi_{i,T}(j,\ell)\Biggr)\varepsilon_{i,\ell}+\sum_{j=0}^{t-1}\rho_{z}^{t-1-j}\rho^{*j}x_{i,0}\nonumber \\
 & =:\sum_{\ell=-\infty}^{t-1}P_{i}(t,\ell)\varepsilon_{i,\ell}+\frac{\rho_{z}^{t}-\rho^{*t}}{\rho_{z}-\rho^{*}}x_{i,0},\label{eq:psi eq Pv}
\end{align}
where
\begin{equation}
P_{i}(t,\ell):=\sum_{j=\ell\vee0}^{t-1}\rho_{z}^{t-1-j}\pi_{i,T}(j,\ell).\label{eq:def_P}
\end{equation}
By \prettyref{lem:ma_represent} we can show
\[
|P_{i}(t,\ell)|\leq\begin{cases}
\dfrac{\rho_{z}^{t}-(\rho^{*}\vee q_{\nu})^{t}}{\rho_{z}-(\rho^{*}\vee q_{\nu})}q_{\nu}^{1-\ell} & \ell\leq0,\\
\dfrac{\rho_{z}^{t-\ell}-(\rho^{*}\vee q_{\nu})^{t-\ell}}{\rho_{z}-(\rho^{*}\vee q_{\nu})} & \ell\geq1,
\end{cases}
\]
where $q_{\nu}:=\exp(-C_{g})$. We thus have
\begin{align*}
\psi_{i,t}x_{i,t} & =\sum_{k=-\infty}^{t-1}P_{i}(t,k)\varepsilon_{i,k}\sum_{\ell=-\infty}^{t}\pi_{i}(t,\ell)\varepsilon_{i,\ell}+\frac{\rho_{z}^{t}-\rho^{*t}}{\rho_{z}-\rho^{*}}x_{i,0}\sum_{\ell=-\infty}^{t}\pi_{i}(t,\ell)\varepsilon_{i,\ell}\\
 & \qquad+\rho^{*t}x_{i,0}\sum_{k=-\infty}^{t-1}P_{i}(t,k)\varepsilon_{i,k}+\frac{(\rho_{z}\rho^{*})^{t}-\rho^{*2t}}{\rho_{z}-\rho^{*}}x_{i,0}^{2}\\
 & =:\Psi_{1,i,t}+\Psi_{2,i,t}+\Psi_{3,i,t}+\Psi_{4,i,t}.
\end{align*}
Let $\Theta_{m,i,T}:=\sum_{t=1}^{T}\Phi_{m,t}$ for $m=1,2,3,4$.
Then,
\begin{align*}
\mathbb{E}\left[\left(\sum_{t=1}^{T}\psi_{i,t}x_{i,t}\right)^{2}\right] & =\mathbb{E}\left[\left(\Theta_{1,i,T}+\Theta_{2,i,T}+\Theta_{3,i,T}+\Theta_{4,i,T}\right)^{2}\right]\\
 & \leq4\mathbb{E}\left(\Theta_{1,i,T}^{2}+\Theta_{2,i,T}^{2}+\Theta_{3,i,T}^{2}+\Theta_{4,i,T}^{2}\right).
\end{align*}
Using the same argument for proving \eqref{eq:Xi1_bound}, we can
show $\mathbb{E}\bigl(\Theta_{1,i,T}^{2}\bigr)=O\left(T^{2[1+\gamma+(\theta\wedge\gamma)]}\right)$.
For $\mathbb{E}\bigl(\Theta_{2,i,T}^{2}\bigr)$ and $\mathbb{E}\bigl(\Theta_{3,i,T}^{2}\bigr)$,
similar to the proof for \eqref{eq:Xi2_bound}, we can also show they
are of order $O\left(T^{2[1+\gamma+(\theta\wedge\gamma)]}\right)$.
For $\mathbb{E}\bigl(\Theta_{4,i,T}^{2}\bigr)$, it is easy to see
\[
\mathbb{E}\left(\Theta_{4,i,T}^{2}\right)=\mathbb{E}\left(x_{i,0}^{4}\right)\left(\sum_{t=1}^{T}\frac{(\rho_{z}\rho^{*})^{t}-\rho^{*2t}}{\rho_{z}-\rho^{*}}\right)^{2}=O\left(T^{2\gamma}\right)\cdot O\left(T^{2(\theta+\gamma)}\right)=O\left(T^{2(\theta+2\gamma)}\right).
\]
These results lead to the desired order

\begin{equation}
\mathbb{E}\left[\left(\sum_{t=1}^{T}\psi_{i,t}x_{i,t}\right)^{2}\right]=O\bigl(T^{2[1+\gamma+(\theta\wedge\gamma)]}\bigr).\label{eq:EpsiX2}
\end{equation}

To bound\textbf{ $\mathbb{E}\bigl[\bigl(\sum_{t=1}^{T}\psi_{i,t}\sum_{t=1}^{T}x_{i,t}\bigr)^{2}\bigr]$},\textbf{
}we first bound $\mathbb{E}\bigl[(\sum_{t=1}^{T}\psi_{i,t})^{4}\bigr]$.
Note that \eqref{eq:psi eq Pv} yields
\begin{align}
\sum_{t=1}^{T}\psi_{i,t} & =\sum_{t=1}^{T}\sum_{\ell=-\infty}^{t-1}P_{i}(t,\ell)\varepsilon_{i,\ell}+\sum_{t=1}^{T}\frac{\rho_{z}^{t}-\rho^{*t}}{\rho_{z}-\rho^{*}}x_{i,0}\nonumber \\
 & =\sum_{\ell=-\infty}^{T-1}\sum_{t=(\ell+1)\vee1}^{T}P_{i}(t,\ell)\varepsilon_{i,\ell}+O(T^{\theta+\gamma})x_{i,0}=:\sum_{\ell=-\infty}^{T-1}S_{i,T}(\ell)\varepsilon_{i,\ell}+O(T^{\theta+\gamma})x_{i,0},\label{eq:sum_psi}
\end{align}
where, using the definition of $P_{i}(t,\ell)$ given by \eqref{eq:def_P},
\[
S_{i,T}(\ell):=\sum_{t=(\ell+1)\vee1}^{T}P_{i}(t,\ell)=\begin{cases}
O(T^{\theta+\gamma})q_{\nu}^{1-\ell} & \ell\leq0,\\
O(T^{\theta+\gamma}) & \ell>1.
\end{cases}
\]
The first term in \eqref{eq:sum_psi} has order (similar to the derivation
of \eqref{eq:sum_sum_pi_4}, using the cumulant condition)
\begin{align}
 & \quad\mathbb{E}\left[\left(\sum_{\ell=-\infty}^{T-1}S_{i,T}(\ell)\varepsilon_{i,\ell}\right)^{4}\right]\nonumber \\
 & =\sum_{\ell=-\infty}^{T-1}\sum_{k\neq\ell}[S_{i,T}(k)]^{3}S_{i,T}(\ell)\mathbb{E}(\varepsilon_{i,k}^{3}\varepsilon_{i,\ell})+\sum_{\ell=-\infty}^{T-1}\sum_{k\neq\ell}[S_{i,T}(k)]^{2}[S_{i,T}(\ell)]^{2}\mathbb{E}(\varepsilon_{i,k}^{2}\varepsilon_{i,\ell}^{2})\nonumber \\
 & \qquad+\sum_{k=-\infty}^{T-1}\sum_{\substack{\ell,j=\infty\\
\ell\neq j
}
}^{k-1}[S_{i,T}(k)]^{2}S_{i,T}(\ell)S_{i,T}(j)\mathbb{E}(\varepsilon_{i,k}^{2}\varepsilon_{i,\ell}\varepsilon_{i,j})\nonumber \\
 & =O\bigl(T^{1+4\theta+4\gamma}\bigr)+O\bigl(T^{2+4\theta+4\gamma}\bigr)+O\bigl(T^{4\theta+4\gamma}\bigr)=O\bigl(T^{2+4\theta+4\gamma}\bigr).\label{eq:E psi 4 term 1}
\end{align}
The second term in \eqref{eq:sum_psi} is bounded by
\begin{equation}
\mathbb{E}\left[\left(O(T^{\theta+\gamma})x_{i,0}\right)^{4}\right]=O(T^{4\theta+4\gamma})\mathbb{E}(x_{i,0}^{4})=O(T^{4\theta+6\gamma}).\label{eq:Epsi 4 term 2}
\end{equation}
It follows from \eqref{eq:E psi 4 term 1} and \eqref{eq:Epsi 4 term 2}
that
\begin{equation}
\mathbb{E}\left[\left(\sum_{t=1}^{T}\psi_{i,t}\right)^{4}\right]=O\bigl(T^{2+4\theta+4\gamma}\bigr).\label{eq:Esum_psi4}
\end{equation}
By \prettyref{lem:generic_AR}\ref{enu:E_sum_eta_4} we have
\begin{equation}
\mathbb{E}\left[\left(\sum_{t=1}^{T}x_{i,t}\right)^{4}\right]=O(T^{2+4\gamma}).\label{eq:E sum x4}
\end{equation}
Then, by the Cauchy-Schwarz inequality, \eqref{eq:Esum_psi4} and
\eqref{eq:E sum x4} together yield
\begin{equation}
\mathbb{E}\left[\left(\sum_{t=1}^{T}\psi_{i,t}\sum_{t=1}^{T}x_{i,t}\right)^{2}\right]\leq\left\{ \mathbb{E}\left[\left(\sum_{t=1}^{T}\psi_{i,t}\right)^{4}\right]\cdot\mathbb{E}\left[\left(\sum_{t=1}^{T}x_{i,t}\right)^{4}\right]\right\} ^{1/2}=O\bigl(T^{2(1+\theta+2\gamma)}\bigr).\label{eq:EpsiEX2}
\end{equation}
By \eqref{eq:EpsiX2} and \eqref{eq:EpsiEX2} we complete the proof
of Item \ref{enu:phi_x}.

\textbf{Item \ref{enu:zx} }follows from \ref{enu:zeta_x}, \ref{enu:phi_x},
and the decomposition \eqref{eq:zeta decom 1}.

\textbf{Item \ref{enu:zeta_e} }follows from \prettyref{lem:generic_two_ARs}\ref{enu:Esumsum2}
and the fact that\textbf{ $z_{i,t}$ }and $e_{i,t}$ are AR(1) with
autoregressive coefficients $\rho_{z}$ and $0$, respectively.

\textbf{Item \ref{enu:phi_e}.} By \prettyref{lem:generic_AR}\ref{enu:E_sum_eta_4},
we have $\mathbb{E}\bigl[\bigl(\sum_{t=1}^{T}e_{i,t+1}\bigr)^{4}\bigr]=O(T^{2})$.
This result together with \eqref{eq:E psi 4 term 1} yields, by the
Cauchy-Schwarz inequality,
\[
\mathbb{E}\left[\left(\sum_{t=1}^{T}\psi_{i,t}\sum_{t=1}^{T}e_{i,t+1}\right)^{2}\right]\leq\left\{ \mathbb{E}\left[\left(\sum_{t=1}^{T}\psi_{i,t}\right)^{4}\right]\cdot\mathbb{E}\left[\left(\sum_{t=1}^{T}e_{i,t+1}\right)^{4}\right]\right\} ^{1/2}=O\bigl(T^{2(1+\theta+\gamma)}\bigr).
\]

\textbf{Item \ref{enu:ze}.} If $\theta\leq\gamma$, by \ref{enu:zeta_e},
\ref{enu:phi_e}, and using the decomposition \eqref{eq:zeta decom 1}
we have
\begin{align}
 & \mathbb{E}\left[\left(\sum_{t=1}^{T}z_{i,t}\sum_{t=1}^{T}e_{i,t+1}\right)^{2}\right]\lesssim\mathbb{E}\left[\left(\sum_{t=1}^{T}\zeta_{i,t}\sum_{t=1}^{T}e_{i,t+1}\right)^{2}\right]+(1-\rho^{*})^{2}\mathbb{E}\left[\left(\sum_{t=1}^{T}\psi_{i,t}\sum_{t=1}^{T}e_{i,t+1}\right)^{2}\right]\nonumber \\
 & \qquad=O\bigl(T^{2(1+\theta)}\bigr)+O(T^{-2\gamma})\cdot O\bigl(T^{2(1+\theta+\gamma)}\bigr)=O\bigl(T^{2(1+\theta)}\bigr),\label{eq:E sum zeta sum e2 case 1}
\end{align}
where the inequality uses $(a+b)^{2}\leq2(a^{2}+b^{2})$.

If $\gamma<\theta$, we use the other decomposition \eqref{eq:zeta decom 2}
to get
\begin{align}
 & \mathbb{E}\left[\left(\sum_{t=1}^{T}z_{i,t}\sum_{t=1}^{T}e_{i,t+1}\right)^{2}\right]\nonumber \\
\lesssim{} & \mathbb{E}\left[\left(\sum_{t=1}^{T}x_{i,t}\sum_{t=1}^{T}e_{i,t+1}\right)^{2}\right]+(1-\rho_{z})^{2}\mathbb{E}\left[\left(\sum_{t=1}^{T}\psi_{i,t}\sum_{t=1}^{T}e_{i,t+1}\right)^{2}\right]\nonumber \\
 & +\left(\sum_{t=1}^{T}\rho_{z}^{t}\right)^{2}\mathbb{E}\left[\left(x_{i,0}\sum_{t=1}^{T}e_{i,t+1}\right)^{2}\right]\nonumber \\
={} & O\bigl(T^{2(1+\gamma)}\bigr)+O(T^{-2\theta})\cdot O\bigl(T^{2(1+\theta+\gamma)}\bigr)+O(T^{2\theta})\cdot O(T^{1+\gamma})\nonumber \\
={} & O\bigl(T^{2(1+\gamma)}\vee T^{1+2\theta+\gamma}\bigr)=O(T^{2+\theta+\gamma}).\label{eq:E sum zeta sum e 2 case 2}
\end{align}
We combine \eqref{eq:E sum zeta sum e2 case 1} and \eqref{eq:E sum zeta sum e 2 case 2}
to get $\mathbb{E}\bigl[\bigl(\sum_{t=1}^{T}z_{i,t}\sum_{t=1}^{T}e_{i,t+1}\bigr)^{2}\bigr]=O\bigl(T^{2+\theta+(\theta\wedge\gamma)}\bigr)$.

\textbf{Item \ref{enu:z2}.} Case I: $\theta\leq\gamma$. By \prettyref{lem:generic_two_ARs}\ref{enu:Esum2},
we have
\begin{equation}
\mathbb{E}\left[\left(\sum_{t=1}^{T}\zeta_{i,t}^{2}\right)^{2}\right]=O\bigl(T^{2(1+\theta)}\bigr).\label{eq:E sum z2 2}
\end{equation}
By \eqref{eq:psi eq Pv}, we have
\begin{equation}
\sum_{t=1}^{T}\psi_{i,t}^{2}=\sum_{t=1}^{T}\left(\sum_{\ell=-\infty}^{t-1}P_{i}(t,\ell)\varepsilon_{i,\ell}+\frac{\rho_{z}^{t}-\rho^{*t}}{\rho_{z}-\rho^{*}}x_{i,0}\right)^{2}=\sum_{t=1}^{T}\left(\Pi_{1,i,t}+\Pi_{2,i,t}\right)^{2}\lesssim\sum_{t=1}^{T}\left(\Pi_{1,i,t}^{2}+\Pi_{2,i,t}^{2}\right).\label{eq:sum psi2}
\end{equation}
Note that
\[
\mathbb{E}\left[\left(\sum_{t=1}^{T}\Pi_{1,i,t}^{2}\right)^{2}\right]=\sum_{t=1}^{T}\mathbb{E}\left(\Pi_{1,i,t}^{4}\right)+2\sum_{t=1}^{T}\sum_{s=1}^{t-1}\mathbb{E}\left(\Pi_{1,i,s}^{2}\Pi_{1,i,t}^{2}\right).
\]
Similar to the derivation of \eqref{eq:sum_sum_pi_4}, we can deduce
\[
\sum_{t=1}^{T}\mathbb{E}\left(\Pi_{1,i,t}^{4}\right)\leq T\cdot\sup_{t\leq T}\mathbb{E}\left(\Pi_{1,i,t}^{4}\right)=O\bigl(T^{4(\theta\wedge\gamma)+2(\theta\vee\gamma)]}\bigr).
\]
For the second term, arguing as in the proof of \prettyref{lem:generic_two_ARs}\ref{enu:Esum2},
we can show
\begin{align*}
\sum_{t=1}^{T}\sum_{s=1}^{t-1}\mathbb{E}\left(\Pi_{1,i,s}^{2}\Pi_{1,i,t}^{2}\right) & =\sum_{t=1}^{T}\sum_{s=1}^{t-1}\sum_{j,k=-\infty}^{s-1}\sum_{\ell,m=-\infty}^{t-1}P_{i}(s,j)P_{i}(s,k)P_{i}(t,\ell)P_{i}(t,m)\varepsilon_{i,j}\varepsilon_{i,k}\varepsilon_{i,\ell}\varepsilon_{i,m}\\
 & =O\bigl(T^{2[1+2(\theta\wedge\gamma)+(\theta\vee\gamma)]}\bigr).
\end{align*}
It follows that
\begin{equation}
\mathbb{E}\left[\left(\sum_{t=1}^{T}\Pi_{1,i,t}^{2}\right)^{2}\right]=O\bigl(T^{2[1+2(\theta\wedge\gamma)+(\theta\vee\gamma)]}\bigr).\label{eq:E psi I}
\end{equation}
In addition,
\begin{align}
\mathbb{E}\left[\left(\sum_{t=1}^{T}\left(\Pi_{2,i,t}\right)^{2}\right)^{2}\right] & =\mathbb{E}(x_{i,0}^{4})\left[\sum_{t=1}^{T}\left(\frac{\rho_{z}^{t}-\rho^{*t}}{\rho_{z}-\rho^{*}}\right)^{2}\right]^{2}\nonumber \\
 & =O(T^{2\gamma})\cdot O(T^{4(\theta\wedge\gamma)+(\theta\vee\gamma)})=O\bigl(T^{2[\gamma+2(\theta\wedge\gamma)+(\theta\vee\gamma)]}\bigr).\label{eq:E psi II}
\end{align}
By \eqref{eq:sum psi2}, \eqref{eq:E psi I} and \eqref{eq:E psi II},
\begin{equation}
\mathbb{E}\left[\left(\sum_{t=1}^{T}\psi_{i,t}^{2}\right)^{2}\right]=O\bigl(T^{2[1+2(\theta\wedge\gamma)+(\theta\vee\gamma)]}\bigr).\label{eq:E sum psi2 2}
\end{equation}
Thus, \eqref{eq:E sum z2 2}, \eqref{eq:E sum psi2 2}, and the decomposition
\eqref{eq:zeta decom 1} together yield
\begin{align}
\mathbb{E}\left[\left(\sum_{t=1}^{T}z_{i,t}^{2}\right)^{2}\right] & \leq\mathbb{E}\left[\left(2\sum_{t=1}^{T}\left[\zeta_{i,t}^{2}+(1-\rho^{*})^{2}\psi_{i,t}^{2}\right]\right)^{2}\right]\nonumber \\
 & \lesssim\mathbb{E}\left[\left(\sum_{t=1}^{T}\zeta_{i,t}^{2}\right)^{2}\right]+(1-\rho^{*})^{4}\mathbb{E}\left[\left(\sum_{t=1}^{T}\psi_{i,t}^{2}\right)^{2}\right]\nonumber \\
 & =O\bigl(T^{2(1+\theta)}\bigr)+O\bigl(T^{2(1+2\theta-\gamma)}\bigr)=O\bigl(T^{2(1+\theta)}\bigr).\label{eq:E sum zeta2 2 case1}
\end{align}

Case II: $\gamma<\theta$. By \prettyref{lem:generic_two_ARs}\ref{enu:Esum2},
we have
\begin{equation}
\mathbb{E}\left[\left(\sum_{t=1}^{T}x_{i,t}^{2}\right)^{2}\right]=O\bigl(T^{2(1+\gamma)}\bigr).\label{eq:E sum x2 2}
\end{equation}
We use the other decomposition \eqref{eq:zeta decom 2} to deduce
\begin{align}
\mathbb{E}\left[\left(\sum_{t=1}^{T}z_{i,t}^{2}\right)^{2}\right] & \lesssim\mathbb{E}\left[\left(\sum_{t=1}^{T}\left[x_{i,t}^{2}+\rho_{z}^{2t}x_{i,0}^{2}+(1-\rho_{z})^{2}\psi_{i,t}^{2}\right]\right)^{2}\right]\nonumber \\
 & \lesssim\mathbb{E}\left[\left(\sum_{t=1}^{T}x_{i,t}^{2}\right)^{2}\right]+\left(\sum_{t=1}^{T}\rho_{z}^{2t}\right)^{2}\mathbb{E}(x_{i,0}^{4})+(1-\rho_{z})^{4}\mathbb{E}\left[\left(\sum_{t=1}^{T}\psi_{i,t}^{2}\right)^{2}\right]\nonumber \\
 & =O\bigl(T^{2(1+\gamma)}\bigr)+O\bigl(T^{2(\theta+\gamma)}\bigr)+O\bigl(T^{2(1+2\gamma-\theta)}\bigr)=O\bigl(T^{2(1+\gamma)}\bigr),\label{eq:E sum zeta2 2 case 2}
\end{align}
where the third line is due to \eqref{eq:E sum psi2 2}, \eqref{eq:E sum x2 2}
and \prettyref{assump:initval}. By \eqref{eq:E sum zeta2 2 case1}
and \eqref{eq:E sum zeta2 2 case 2}, we conclude that $\mathbb{E}\bigl[\bigl(\sum_{t=1}^{T}z_{i,t}^{2}\bigr)^{2}\bigr]=O\bigl(T^{2[1+(\theta\wedge\gamma)]}\bigr)$.

For\textbf{ Item \ref{enu:z4}}, by \prettyref{lem:generic_AR}\ref{enu:E_sum_eta_4}
we have
\begin{equation}
\mathbb{E}\left[\left(\sum_{t=1}^{T}\zeta_{i,t}\right)^{4}\right]=O(T^{2+4\theta}).\label{eq:E sum z4}
\end{equation}
If $\theta\leq\gamma$, using the decomposition \eqref{eq:zeta decom 1},
\eqref{eq:E sum z4} and \eqref{eq:Esum_psi4} together we have
\begin{align*}
\mathbb{E}\left[\left(\sum_{t=1}^{T}z_{i,t}\right)^{4}\right] & \lesssim\mathbb{E}\left[\left(\sum_{t=1}^{T}\zeta_{i,t}\right)^{4}\right]+(1-\rho^{*})^{4}\mathbb{E}\left[\left(\sum_{t=1}^{T}\psi_{i,t}\right)^{4}\right]\\
 & =O(T^{2+4\theta})+O(T^{2+4\theta})=O(T^{2+4\theta}).
\end{align*}
If $\gamma<\theta$, using the other decomposition \eqref{eq:zeta decom 2},
then by \eqref{eq:E sum x4}, \eqref{eq:Esum_psi4}, and \prettyref{assump:initval}
we have
\begin{align*}
\mathbb{E}\left[\left(\sum_{t=1}^{T}z_{i,t}\right)^{4}\right] & \lesssim\mathbb{E}\left[\left(\sum_{t=1}^{T}x_{i,t}\right)^{4}\right]+\left(\sum_{t=1}^{T}\rho_{z}^{t}\right)^{4}\mathbb{E}(x_{i,0}^{4})+(1-\rho_{z})^{4}\mathbb{E}\left[\left(\sum_{t=1}^{T}\psi_{i,t}\right)^{4}\right]\\
 & =O(T^{2+4\gamma})+O(T^{4\theta+2\gamma})+O(T^{2+4\theta})=O\bigl(T^{2(1+\theta+\gamma)}\bigr).
\end{align*}
The two cases can be jointly expressed as $\mathbb{E}\bigl[\bigl(\sum_{t=1}^{T}z_{i,t}\bigr)^{4}\bigr]=O\left(T^{2[1+\theta+(\theta\wedge\gamma)]}\right)$.

\textbf{Equation~\eqref{eq:E sum zeta sum e}.} Let $v_{i,0}:=x_{i,0}$.
Note that by construction, $z_{i,t}$ can be written as
\begin{align*}
z_{i,t} & =\sum_{j=1}^{t}\rho_{z}^{t-j}(x_{i,j}-x_{i,j-1})=\sum_{j=1}^{t}\rho_{z}^{t-j}[(\rho^{*}-1)x_{i,j-1}+v_{i,j}]\\
 & =\sum_{j=1}^{t}\rho_{z}^{t-j}\left[(\rho^{*}-1)\left(\sum_{k=0}^{j-1}\rho^{*j-1-k}v_{i,k}\right)+v_{i,j}\right]\\
 & =\sum_{k=0}^{t-1}\Biggl[\sum_{j=k+1}^{t}(\rho^{*}-1)\rho_{z}^{t-j}\rho^{*j-1-k}\Biggr]v_{i,k}+\sum_{j=1}^{t}\rho_{z}^{t-j}v_{i,j}.
\end{align*}
It follows that
\begin{align*}
\sum_{t=1}^{T}z_{i,t} & =\sum_{t=1}^{T}\left\{ \sum_{k=0}^{t-1}\Biggl[\sum_{j=k+1}^{t}(\rho^{*}-1)\rho_{z}^{t-j}\rho^{*j-1-k}\Biggr]v_{i,k}+\sum_{j=1}^{t}\rho_{z}^{t-j}v_{i,j}\right\} \\
 & =\sum_{k=0}^{T-1}\Biggl[\sum_{t=k+1}^{T}\sum_{j=k+1}^{t}(\rho^{*}-1)\rho_{z}^{t-j}\rho^{*j-1-k}\Biggr]v_{i,k}+\sum_{j=1}^{T}\Biggl(\sum_{t=j}^{T}\rho_{z}^{t-j}\Biggr)v_{i,j}\\
 & =\sum_{k=0}^{T-1}\Biggl[(\rho^{*}-1)\sum_{t=1}^{T-k}\frac{\rho_{z}^{t}-\rho^{*t}}{\rho_{z}-\rho^{*}}\Biggr]v_{i,k}+\sum_{j=1}^{T}\Biggl(1+\sum_{t=1}^{T-j}\rho_{z}^{t}\Biggr)v_{i,j}\\
 & =(\rho^{*}-1)\sum_{t=1}^{T}\frac{\rho_{z}^{t}-\rho^{*t}}{\rho_{z}-\rho^{*}}v_{i,0}+\sum_{k=1}^{T}\left(\frac{\rho_{z}^{T-k+1}-\rho^{*T-k+1}}{\rho_{z}-\rho^{*}}\right)v_{i,k}.
\end{align*}
Thus,
\begin{align*}
\psi_{0}(\rho^{*},\rho_{z}) & =(\rho^{*}-1)\sum_{t=1}^{T}\frac{\rho_{z}^{t}-\rho^{*t}}{\rho_{z}-\rho^{*}},\\
\psi_{k}(\rho^{*},\rho_{z}) & =\frac{\rho_{z}^{T-k+1}-\rho^{*T-k+1}}{\rho_{z}-\rho^{*}}\qquad k=1,\dots,T.
\end{align*}
Under stationarity, $\mathbb{E}(v_{i,k}e_{i,s+1})$ only depends on
the time gap $|k-s-1|$. Therefore
\[
\mathbb{E}(v_{i,k}e_{i,s+1})=\begin{cases}
0, & s+1>k,\\
\omega_{ev}(k-s-1), & s+1\leq k,
\end{cases}
\]
where $\omega_{ev}(h)\equiv\omega_{ev,h}^{*}$ denotes the ``covariance
function.'' We then have
\begin{align*}
\mathbb{E}\left(\sum_{t=1}^{T}z_{i,t}\sum_{s=1}^{T}e_{i,s+1}\right) & =\sum_{k=2}^{T}\psi_{k}(\rho^{*},\rho_{z})\mathbb{E}\left(v_{i,k}\sum_{s=1}^{k-1}e_{i,s+1}\right)\\
 & =\sum_{k=2}^{T}\psi_{k}(\rho^{*},\rho_{z})\sum_{s=1}^{k-1}\omega_{ev}(k-s-1)\\
 & =\sum_{k=2}^{T}\psi_{k}(\rho^{*},\rho_{z})\sum_{h=0}^{k-2}\omega_{ev}(h)=\sum_{h=0}^{T-2}\left[\sum_{k=h+2}^{T}\psi_{k}(\rho^{*},\rho_{z})\right]\omega_{ev,h}^{*}\\
 & =:\sum_{h=0}^{T-2}\Psi_{h,T}(\rho^{*},\rho_{z})\omega_{ev,h}^{*},
\end{align*}
where
\begin{equation}
\Psi_{h,T}(\rho^{*},\rho_{z}):=\sum_{k=h+2}^{T}\psi_{k}(\rho^{*},\rho_{z})=\frac{1}{\rho_{z}-\rho^{*}}\left(\frac{\rho_{z}-\rho_{z}^{T-h}}{1-\rho_{z}}-\frac{\rho^{*}-\rho^{*T-h}}{1-\rho^{*}}\right).\label{eq:Psi-h}
\end{equation}
In addition, noting that $v_{i,t}=\sum_{s=0}^{\infty}g_{s}\varepsilon_{i,t-s}$,
we have $|\omega_{ev,h}|\lesssim|g_{h}|\lesssim q_{\nu}^{h}$ where
$q_{\nu}:=\exp(-C_{g})$. Hence,
\[
\left|\sum_{h=0}^{T-2}\Psi_{h,T}(\rho^{*},\rho_{z})\omega_{ev,h}\right|\leq\sum_{h=0}^{T-2}\frac{1}{|\rho_{z}-\rho^{*}|}\left(\frac{|\rho_{z}|+|\rho_{z}^{T-h}|}{|1-\rho_{z}|}+\frac{|\rho^{*}|+|\rho^{*T-h}|}{|1-\rho^{*}|}\right)q_{\nu}^{h}=O(T^{\theta+\gamma}).
\]
This establishes \eqref{eq:E sum zeta sum e}. The proof of \prettyref{lem:ivxepct}
is thus complete.
\end{proof}
\begin{rem}
The following proof invokes Lemma~3.1 of \citet{phillips2009econometric},
which relies on a prerequisite $\theta>1/2$ when the innovation in
$x_{i,t}$ is a weakly dependent linear process. However, this requirement
turns out to be unnecessary if the coefficients of the linear process
are exponentially decaying as specified in \prettyref{assump:innov}\ref{enu:LP}.
The argument is as follows. We can show that, similar to the proof
of \prettyref{lem:ivxepct}\ref{enu:phi_x}, $\mathbb{E}\bigl[\bigl(\sum_{t=1}^{T}\psi_{i,t}v_{i,t}\bigr)^{2}\bigr]=O\bigl(T^{2[1+(\theta\wedge\gamma)]}\bigr)$,
which implies $\sum_{t=1}^{T}\psi_{i,t}v_{i,t}=O_{p}\left(T^{1+(\theta\wedge\gamma)}\right)$.
This, however, is different from Equation~(41) of \citet{phillips2009econometric},
which is $\sum_{t=1}^{T}\psi_{i,t}v_{i,t}=o_{p}\bigl(T^{\frac{1}{2}[1+(\theta\wedge\gamma)]+(\theta\vee\gamma)}\bigr)$
provided $\theta>1/2$, but $\sum_{t=1}^{T}\psi_{i,t}v_{i,t}=O_{p}\left(T^{1+(\theta\wedge\gamma)}\right)$
would be sufficient for the proof of their Lemma~3.1 to still hold
without imposing $\theta>1/2$.
\end{rem}
\begin{proof}[Proof of \prettyref{lem:ivx_marginal}\ref{enu:Q_dto}]
For notational simplicity, we write $J_{v,c^{*}}(r)$ in short as
$J_{v,c^{*}}$ and $B_{v}(r)$ as $B_{v}$ in the proofs.

\textbf{CASE I:} $0<\theta<\gamma\leq1$. According to Equation~(20)
of \citet{phillips2009econometric}, we have, as $T\to\infty$,
\[
\frac{1}{T^{1+\theta}}\sum_{t=1}^{T}\zeta_{i,t}x_{i,t}\to_{d}\begin{cases}
\omega_{vv}^{*}/(-c_{z}), & \text{if }\gamma\in(0,1),\\
\bigl(\int_{0}^{1}J_{v,c^{*}}\,dB_{v}+\omega_{vv}^{*}\bigr)/(-c_{z}), & \text{if }\gamma=1.
\end{cases}
\]
The proof of Lemma~3.1(ii), specifically the second equation on page
28 of \citet{phillips2009econometric} gives us, as $T\to\infty$,
\begin{align*}
\frac{c^{*}}{T^{1+\theta+\gamma}}\sum_{t=1}^{T}\psi_{i,t}x_{i,t} & =-\frac{c^{*}}{c_{z}}\frac{1}{T^{1+\gamma}}\sum_{t=1}^{T}x_{i,t}^{2}+o_{p}(1)\to_{d}\begin{cases}
\omega_{vv}^{*}/(2c_{z}), & \text{if }\gamma\in(0,1),\\
-\frac{c^{*}}{c_{z}}\int_{0}^{1}J_{v,c^{*}}^{2}, & \text{if }\gamma=1.
\end{cases}
\end{align*}
where the convergence applies \prettyref{lem:wg-ts-converge} and
\eqref{eq:QWG lim dis}. Then, by \eqref{eq:zeta decom 1},
\[
\frac{1}{T^{1+\theta}}\sum_{t=1}^{T}z_{i,t}x_{i,t}\to_{d}\begin{cases}
\omega_{vv}^{*}/(-2c_{z}), & \text{if }\gamma\in(0,1),\\
\bigl(\int_{0}^{1}J_{v,c^{*}}\,dB_{v}+\omega_{vv}^{*}+c^{*}\int_{0}^{1}J_{v,c^{*}}^{2}\bigr)/(-c_{z}), & \text{if }\gamma=1.
\end{cases}
\]
 as $T\to\infty$. Since $T^{\theta}[1-(\rho^{*}\rho_{z})^{2}]\to-2c_{z}$,
we have
\begin{align*}
Q_{i,T} & =T^{\theta}[1-(\rho^{*}\rho_{z})^{2}]\cdot\frac{1}{T^{1+\theta}}\sum_{t=1}^{T}z_{i,t}x_{i,t}\to_{d}\begin{cases}
\omega_{vv}^{*}, & \text{if }\gamma\in(0,1),\\
2\bigl(\int_{0}^{1}J_{v,c^{*}}\,dB_{v}+\omega_{vv}^{*}+c^{*}\int_{0}^{1}J_{v,c^{*}}^{2}\bigr), & \text{if }\gamma=1.
\end{cases}
\end{align*}

\textbf{CASE II:} $0<\theta=\gamma<1$. By Lemma~3.6(i) of \citet{phillips2009econometric},
we have
\[
\frac{1}{T^{1+\theta}}\sum_{t=1}^{T}z_{i,t}x_{i,t}\to_{p}\dfrac{\omega_{vv}^{*}}{-2(c^{*}+c_{z})},
\]
which implies $Q_{i,T}\to_{p}\omega_{vv}^{*}$ in view of $T^{\theta}[1-(\rho^{*}\rho_{z})^{2}]\to-2(c_{z}+c^{*})$.

\textbf{CASE III:} $0\leq\gamma<\theta<1$. By Lemma~3.5(ii) of \citet{phillips2009econometric},\footnote{\label{fn:stationary-1}Although $\gamma=0$ is not considered in
that lemma, it can be easily verified that the result and its proof
applies to the stationary case. If $x_{i,t}$ is stationary, we simply
invoke the ergodic theorem to get convergence in probability. This
argument about $\gamma=0$ applies to other places that cite the lemmas
in \citet{phillips2009econometric}.}
\[
\frac{1}{T^{1+\gamma}}\sum_{t=1}^{T}z_{i,t}x_{i,t}\to_{p}\begin{cases}
\omega_{vv}^{*}/(1-\rho^{*2}), & \text{if stationary (\ensuremath{\gamma=0})},\\
\omega_{vv}^{*}/(-2c^{*}), & \text{if MI (\ensuremath{0<\gamma<\theta})}.
\end{cases}
\]
as $T\to\infty$. Also,
\[
T^{\gamma}[1-(\rho^{*}\rho_{z})^{2}]\to\begin{cases}
1-\rho^{*2} & \text{if stationary (\ensuremath{\gamma=0})},\\
-2c^{*}, & \text{if MI (\ensuremath{0<\gamma<\theta})}.
\end{cases}
\]
 as $T\to\infty$. It follows that
\begin{align*}
Q_{i,T} & =T^{\gamma}[1-(\rho^{*}\rho_{z})^{2}]\cdot\frac{1}{T^{1+\gamma}}\sum_{t=1}^{T}z_{i,t}x_{i,t}\to_{d}\omega_{vv}^{*}.
\end{align*}
This completes proof of \prettyref{lem:ivx_marginal}\ref{enu:Q_dto}.
\end{proof}
\begin{proof}[Proof of \prettyref{lem:ivx_marginal}\ref{enu:R_dto}]
\textbf{ CASE I:} $0<\theta<\gamma\leq1$. When $\gamma=1$, we apply
the standard limit theory in \citet{phillips1987towards} to get,
as $T\to\infty$,
\begin{align*}
\frac{1}{T^{\frac{1}{2}+\theta}}\sum_{t=1}^{T}\zeta_{i,t}\frac{1}{T^{3/2}}\sum_{t=1}^{T}x_{i,t} & \to_{d}-\frac{1}{c_{z}}B_{v}(1)\int_{0}^{1}J_{v,c^{*}},\\
\frac{1}{T^{\frac{3}{2}+\theta}}\sum_{t=1}^{T}\psi_{i,t}\frac{1}{T^{3/2}}\sum_{t=1}^{T}x_{i,t} & \to_{d}-\frac{1}{c_{z}}\left(\int_{0}^{1}J_{v,c^{*}}\right)^{2}.
\end{align*}
When $\gamma<1$, we know from \prettyref{lem:wg-ts-converge} that
$\sum_{t=1}^{T}x_{i,t}=O_{p}\bigl(T^{\frac{1}{2}+\gamma}\bigr)$ and
$\sum_{t=1}^{T}\zeta_{i,t}=O_{p}\bigl(T^{\frac{1}{2}+\theta}\bigr)$.
In addition, by \eqref{eq:Esum_psi4}, we have $\sum_{t=1}^{T}\psi_{i,t}=O_{p}\bigl(T^{\frac{1}{2}+\theta+\gamma}\bigr)$.
It follows that
\begin{align*}
\frac{1}{T^{\frac{1}{2}+\theta}}\sum_{t=1}^{T}\zeta_{i,t}\frac{1}{T^{\frac{1}{2}+\gamma}}\sum_{t=1}^{T}x_{i,t} & \to_{p}0\quad\text{and}\quad\frac{1}{T^{\frac{3}{2}+\theta}}\sum_{t=1}^{T}\psi_{i,t}\frac{1}{T^{\frac{1}{2}+\gamma}}\sum_{t=1}^{T}x_{i,t}\to_{p}0.
\end{align*}
Since $T^{\theta}[1-(\rho^{*}\rho_{z})^{2}]\to-2c_{z}$, we have
\begin{align*}
R_{i,T} & =T^{\theta}[1-(\rho^{*}\rho_{z})^{2}]\cdot\frac{1}{T^{2+\theta}}\sum_{t=1}^{T}z_{i,t}\sum_{t=1}^{T}x_{i,t}\\
 & \to_{d}\begin{cases}
0, & \text{if }\gamma\in(0,1),\\
2\left[B_{v}(1)\int_{0}^{1}J_{v,c^{*}}+\bigl(\int_{0}^{1}J_{v,c^{*}}\bigr)^{2}\right], & \text{if }\gamma=1.
\end{cases}
\end{align*}

\textbf{CASE II:} $0\leq\gamma<\theta<1$ or $\gamma=\theta\in(0,1)$.
By \prettyref{lem:ivxepct}\ref{enu:zx} and the fact that $T^{\gamma}[1-(\rho^{*}\rho_{z})^{2}]=O(1)$
we have
\[
\mathbb{E}\left(R_{i,T}^{2}\right)=\left[\frac{1-(\rho^{*}\rho_{z})^{2}}{T^{2}}\right]^{2}O\bigl(T^{2(1+\theta+\gamma)}\bigr)=O\left(\frac{1}{T^{2(1-\theta)}}\right)\to0,
\]
which implies $R_{i,T}\to_{p}0.$ This completes the proof of \prettyref{lem:ivx_marginal}\ref{enu:R_dto}.
\end{proof}
\begin{proof}[Proof of \prettyref{lem:ivx_marginal}\ref{enu:S_pto}]
\textbf{CASE I:} $0<\theta<\gamma\leq1$. According to Lemma 3.1(iii)
and Equation~(14) of \citet{phillips2009econometric}, we have $T^{-(1+\theta)}\sum_{t=1}^{T}z_{i,t}^{2}\to_{p}-\omega_{vv}^{*}/(2c_{z})$
as $T\to\infty$. Since $T^{\theta}[1-(\rho^{*}\rho_{z})^{2}]\to-2c_{z}$,
we thus have
\begin{align*}
S_{i,T} & =T^{\theta}[1-(\rho^{*}\rho_{z})^{2}]\cdot\frac{1}{T^{1+\theta}}\sum_{t=1}^{T}z_{i,t}^{2}\to_{p}\omega_{vv}^{*}.
\end{align*}

\textbf{CASE II:} $0<\theta=\gamma<1$. By Lemma~3.6(ii) of \citet{phillips2009econometric},
we have
\[
\frac{1}{T^{1+\theta}}\sum_{t=1}^{T}z_{i,t}^{2}\to_{p}\dfrac{\omega_{vv}^{*}}{-2(c^{*}+c_{z})},
\]
which implies $S_{i,T}\to_{p}\omega_{vv}^{*}$ in view of $T^{\theta}[1-(\rho^{*}\rho_{z})^{2}]\to-2(c_{z}+c^{*})$.

\textbf{CASE III:} $0\leq\gamma<\theta<1$. By Lemma~3.5(ii) of \citet{phillips2009econometric},
we have
\[
S_{i,T}:=\frac{1}{T^{1+\gamma}}\sum_{t=1}^{T}z_{i,t}^{2}\to_{p}\begin{cases}
\omega_{vv}^{*}/(1-\rho^{*2}), & \text{if stationary (\ensuremath{\gamma=0})},\\
\omega_{vv}^{*}/(-2c^{*}), & \text{if MI (\ensuremath{0<\gamma<\theta})}.
\end{cases}
\]
Also note that
\[
T^{\gamma}[1-(\rho^{*}\rho_{z})^{2}]\to\begin{cases}
1-\rho^{*2}, & \text{if stationary (\ensuremath{\gamma=0})},\\
-2c^{*}, & \text{if MI (\ensuremath{0<\gamma<\theta})}.
\end{cases}
\]
It follows that
\begin{align*}
S_{i,T} & =T^{\gamma}[1-(\rho^{*}\rho_{z})^{2}]\cdot\frac{1}{T^{1+\gamma}}\sum_{t=1}^{T}z_{i,t}x_{i,t}\to_{p}\omega_{vv}^{*}.
\end{align*}
The proof shows that $S_{i,T}$ has the same probability limit as
$Q_{i,T}$ whenever $\gamma<1$.
\end{proof}
\begin{proof}[Proof of \prettyref{lem:ivx_marginal}\ref{enu:Z_dto}]
\textbf{CASE I:} $0<\theta<\gamma\leq1$. According to Lemma~3.1(i)
of \citet{phillips2009econometric}, as $T\to\infty$,
\begin{align*}
\frac{1}{T^{\frac{1}{2}(1+\theta)}}\sum_{t=1}^{T}z_{i,t}e_{i,t+1} & =\frac{1}{T^{\frac{1}{2}(1+\theta)}}\sum_{t=1}^{T}\zeta_{i,t}e_{i,t+1}+o_{p}(1)\to_{d}\mathcal{N}\left(0,\frac{\omega_{ee}^{*}\omega_{vv}^{*}}{-2c_{z}}\right),
\end{align*}
where the weak limit applies the fact that $\zeta_{i,t}$ is a mildly
integrated process with $\rho_{z}=1+c_{z}/T^{\theta}$, and Lemma~B4
of \citet{Kostakis2015}. Since $T^{\theta}[1-(\rho^{*}\rho_{z})^{2}]\to-2c_{z}$,
we have
\[
Z_{i,T}=\sqrt{T^{\theta}[1-(\rho^{*}\rho_{z})^{2}]}\cdot\frac{1}{T^{\frac{1}{2}(1+\theta)}}\sum_{t=1}^{T}z_{i,t}e_{i,t+1}\to_{d}\sqrt{-2c_{z}}\cdot\mathcal{N}\left(0,\frac{\omega_{ee}^{*}\omega_{vv}^{*}}{-2c_{z}}\right)=\mathcal{N}(0,\omega_{ee}^{*}\omega_{vv}^{*}).
\]

\textbf{CASE II:} $0<\theta=\gamma<1$. By Lemma~3.3 of \citet{magdalinosLeastSquaresIvx2022}
and Lemma~3.6(ii) of \citet{phillips2009econometric}, we have
\begin{align*}
\frac{1}{T^{\frac{1}{2}(1+\gamma)}}\sum_{t=1}^{T}z_{i,t}e_{i,t+1} & \to_{d}\mathcal{N}\left(0,\left[\mathop{\mathrm{plim}}_{T\to\infty}\frac{1}{T^{1+\gamma}}\sum_{t=1}^{T}z_{i,t}z_{i,t}\right]\omega_{ee}^{*}\right)=\mathcal{N}\left(0,\frac{\omega_{ee}^{*}\omega_{vv}^{*}}{-2(c^{*}+c_{z}\text{)}}\right).
\end{align*}
Since $T^{\gamma}[1-(\rho^{*}\rho_{z})^{2}]\to-2(c^{*}+c_{z}),$ we
have
\begin{align*}
Z_{i,T} & =\sqrt{T^{\gamma}[1-(\rho^{*}\rho_{z})^{2}]}\cdot\frac{1}{T^{\frac{1}{2}(1+\gamma)}}\sum_{t=1}^{T}z_{i,t}e_{i,t+1}\\
 & \to_{d}\sqrt{-2(c^{*}+c_{z})}\cdot\mathcal{N}\left(0,\frac{\omega_{ee}^{*}\omega_{vv}^{*}}{-2(c^{*}+c_{z})}\right)=\mathcal{N}(0,\omega_{ee}^{*}\omega_{vv}^{*}).
\end{align*}

\textbf{CASE III:} $0\leq\gamma<\theta<1$. When $\gamma=0$, following
the CLT in the Online Appendix of \citet[Eqs. (31)--(32)]{Kostakis2015}
for the stationary case, we have
\[
\frac{1}{\sqrt{T}}\sum_{t=1}^{T}z_{i,t}e_{i,t+1}\to_{d}\mathcal{N}\left(0,{\rm S}_{0,xe}\right).
\]
When $\gamma>0,$ by Lemma~3.5(i) of \citet{phillips2009econometric},
we have
\begin{align*}
\frac{1}{T^{\frac{1}{2}(1+\gamma)}}\sum_{t=1}^{T}z_{i,t}e_{i,t+1} & =\frac{1}{T^{\frac{1}{2}(1+\gamma)}}\sum_{t=1}^{T}x_{i,t}e_{i,t+1}+o_{p}(1)\to_{d}\mathcal{N}\left(0,\frac{\omega_{ee}^{*}\omega_{vv}^{*}}{-2c^{*}}\right),
\end{align*}
where the weak limit again applies Lemma~B4 of \citet{Kostakis2015}.
Since
\[
T^{\gamma}[1-(\rho^{*}\rho_{z})^{2}]\to\begin{cases}
1-\rho^{*2} & \text{if stationary (\ensuremath{\gamma=0})},\\
-2c^{*}, & \text{if MI (\ensuremath{0<\gamma<\theta})},
\end{cases}
\]
we have
\[
Z_{i,T}=\sqrt{T^{\gamma}[1-(\rho^{*}\rho_{z})^{2}]}\cdot\frac{1}{T^{\frac{1}{2}(1+\gamma)}}\sum_{t=1}^{T}z_{i,t}e_{i,t+1}\to_{d}\mathcal{N}(0,{\rm S}_{xe}).
\]
We complete the proof of \prettyref{lem:ivx_marginal}\ref{enu:Z_dto}.
\end{proof}
\begin{proof}[Proof of \prettyref{lem:ivxjoint}]
\ref{enu:Q_pto} By \prettyref{lem:ivx_marginal}\ref{enu:Q_dto}
and Corollary~1 of \citet{phillips1999linear}, it suffices to show
that $Q_{i,T}$ is u.i.~in $T$. By \prettyref{lem:ivxepct}\ref{enu:zx}
we deduce that
\[
\mathbb{E}\left(Q_{i,T}^{2}\right)=\left(\frac{1-(\rho^{*}\rho_{z})^{2}}{T}\right)^{2}\mathbb{E}\left[\left(\sum_{t=1}^{T}z_{i,t}x_{i,t}\right)^{2}\right]=O\left(\dfrac{1}{T^{2+2(\theta\wedge\gamma)}}\right)\cdot O\bigl(T^{2[1+(\theta\wedge\gamma)]}\bigr)=O(1).
\]
Then $Q_{i,T}$ is u.i.~in $T$ by \prettyref{lem:u.i.}.

\ref{enu:R_pto} By \prettyref{lem:ivx_marginal}\ref{enu:R_dto}
and Corollary~1 of \citet{phillips1999linear}, it suffices to show
that $R_{i,T}$ is u.i.~in $T$. By \prettyref{lem:ivxepct}\ref{enu:zx}
we deduce that
\[
\mathbb{E}\left(R_{i,T}^{2}\right)=\left(\frac{1-(\rho^{*}\rho_{z})^{2}}{T^{2}}\right)^{2}\mathbb{E}\left[\left(\sum_{t=1}^{T}z_{i,t}\sum_{t=1}^{T}x_{i,t}\right)^{2}\right]=O\left(\dfrac{1}{T^{4+2(\theta\wedge\gamma)}}\right)\cdot O\bigl(T^{2(1+\theta+\gamma)}\bigr)=O(1).
\]
Then $R_{i,T}$ is u.i.~in $T$ by \prettyref{lem:u.i.}.

\ref{enu:S_pto_joint} By \prettyref{lem:ivxepct}\ref{enu:zeta_e}
we have
\[
\mathbb{E}\left(S_{i,T}^{2}\right)=\left(\frac{1-(\rho^{*}\rho_{z})^{2}}{T}\right)^{2}\mathbb{E}\left[\left(\sum_{t=1}^{T}z_{i,t}^{2}\right)^{2}\right]=O\left(\dfrac{1}{T^{2+2(\theta\wedge\gamma)}}\right)\cdot O\bigl(T^{2[1+(\theta\wedge\gamma)]}\bigr)=O(1).
\]
This indicates that $S_{i,T}$ is u.i.~in $T$ by \prettyref{lem:u.i.}.
Then by Corollary~1 of \citet{phillips1999linear}, we conclude $n^{-1}\sum_{i=1}^{n}S_{i,T}\to_{p}\omega_{vv}^{*}$
as $(n,T)\to\infty$.

\ref{enu:L_dto} By definition, $L_{i,T}=Z_{i,T}-[H_{i,T}-\mathbb{E}(H_{i,T})]$.
The proof consists of two steps: Step I: showing $n^{-1/2}\sum_{i=1}^{n}Z_{i,T}\to_{d}\mathcal{N}(0,{\rm S}_{xe})$
as $(n,T)\to\infty$; Step II: showing $n^{-1/2}\sum_{i=1}^{n}[H_{i,T}-\mathbb{E}(H_{i,T})]\to_{p}0$
as $(n,T)\to\infty$.

\textbf{Step I. }We show the asymptotic normality of $n^{-1/2}\sum_{i=1}^{n}Z_{i,T}$.
First, we show that $Z_{i,T}^{2}$ is u.i.~in $T$. Note that in
\ref{enu:S_pto_joint} we have shown the u.i.\ of $S_{i,T}$ and
\prettyref{lem:ivx_marginal}\ref{enu:S_pto} gives us $S_{i,T}\to_{p}\omega_{vv}^{*}$.
Hence, by \prettyref{lem:u.i.} we have $\mathbb{E}(S_{i,T})\to\omega_{vv}^{*}$
as $T\to\infty$. It is easy to see that $\{z_{i,t}e_{i,t+1}\}$ is
an m.d.s., so we have
\[
\mathbb{E}\left(Z_{i,T}^{2}\right)=\frac{1-(\rho^{*}\rho_{z})^{2}}{T}\sum_{t=1}^{T}\mathbb{E}\left(z_{i,t}^{2}e_{i,t+1}^{2}\right).
\]
We now show that
\begin{equation}
\lim_{T\to\infty}\mathbb{E}\left(Z_{i,T}^{2}\right)={\rm S}_{xe}.\label{eq:E_Z2_converge}
\end{equation}
When $\gamma=0$, the asymptotics in the Online Appendix of \citet[Lemmas B2(iv) and B4(iii)]{Kostakis2015}
yields that
\[
\frac{1-(\rho^{*}\rho_{z})^{2}}{T}\sum_{t=1}^{T}\mathbb{E}\left(z_{i,t}^{2}e_{i,t+1}^{2}\right)\to(1-\rho^{*2}){\rm S}_{0,xe}={\rm S}_{xe}.
\]
We now focus on $\gamma>0$. First we express $z_{i,t}$ as a linear
process:
\begin{align*}
z_{i,t} & =\sum_{k=1}^{t}\left(\frac{\rho_{z}-1}{\rho_{z}-\rho^{*}}\rho_{z}^{t-k}-\frac{\rho^{*}-1}{\rho_{z}-\rho^{*}}\rho^{*t-k}\right)v_{i,k}+(\rho^{*}-1)\frac{\rho_{z}^{t}-\rho^{*t}}{\rho_{z}-\rho^{*}}x_{i,0}\\
 & =\sum_{\ell=-\infty}^{t}\left[\sum_{k=\ell\vee1}^{t}B_{T}(t,k)g_{k-\ell}\right]\varepsilon_{i,\ell}+(\rho^{*}-1)\frac{\rho_{z}^{t}-\rho^{*t}}{\rho_{z}-\rho^{*}}x_{i,0}\\
 & =\sum_{\ell=-\infty}^{t}M_{T}(t,\ell)\varepsilon_{i,\ell}+(\rho^{*}-1)\frac{\rho_{z}^{t}-\rho^{*t}}{\rho_{z}-\rho^{*}}x_{i,0}=:\Upsilon_{1,i,t}+\Upsilon_{2,i,t},
\end{align*}
where we denote
\begin{align*}
B_{T}(t,k) & :=\frac{\rho_{z}-1}{\rho_{z}-\rho^{*}}\rho_{z}^{t-k}-\frac{\rho^{*}-1}{\rho_{z}-\rho^{*}}\rho^{*t-k},\\
M_{T}(t,\ell) & :=\sum_{k=\ell\vee1}^{t}B(t,k)g_{k-\ell}.
\end{align*}
We make explicit the dependence of $B_{T}(t,k)$ and $M_{T}(t,\ell)$
on $T$ since $\rho_{z}$ and $\rho^{*}$ are dependent on $T$. First,
we can deduce
\begin{eqnarray}
 &  & \frac{1-(\rho^{*}\rho_{z})^{2}}{T}\sum_{t=1}^{T}\mathbb{E}\left(\Upsilon_{2,i,t}^{2}e_{i,t+1}^{2}\right)\nonumber \\
 & \leq & \frac{1-(\rho^{*}\rho_{z})^{2}}{T}(\rho^{*}-1)^{2}\sum_{t=1}^{T}\left(\frac{\rho_{z}^{t}-\rho^{*t}}{\rho_{z}-\rho^{*}}\right)^{2}\sqrt{\mathbb{E}\bigl(x_{i,0}^{4}\bigr)\mathbb{E}\bigr(e_{i,1}^{4}\bigr)}\nonumber \\
 & = & O\bigl(T^{-1-(\theta\wedge\gamma)}\bigr)\cdot O\left(T^{-2\gamma}\right)\cdot O\bigl(T^{2(\theta\wedge\gamma)+(\theta\vee\gamma)}\bigr)\cdot O\left(T^{\gamma}\right)=O\bigr(T^{\theta-1}\bigr).\label{eq:upsilon2}
\end{eqnarray}
Moreover, we have
\begin{align*}
 & \frac{1-(\rho^{*}\rho_{z})^{2}}{T}\sum_{t=1}^{T}\mathbb{E}\left(\Upsilon_{1,i,t}^{2}\right)\\
={} & \frac{1-(\rho^{*}\rho_{z})^{2}}{T}\sum_{t=1}^{T}\sum_{k=1}^{t}[B_{T}(t,k)]^{2}\Gamma_{vv}(0)+2\frac{1-(\rho^{*}\rho_{z})^{2}}{T}\sum_{t=1}^{T}\sum_{h=1}^{t-1}\sum_{\ell=1}^{t-h}B_{T}(t,\ell+h)B_{T}(t,\ell)\Gamma_{vv}(h).
\end{align*}
We now show this expression tends to $\sum_{\ell=-\infty}^{\infty}\Gamma_{vv}(\ell)$
by showing that the first term converges to $\Gamma_{vv}(0)$ as $T\to\infty$
and the second term converges to $2\sum_{\ell=1}^{\infty}\Gamma_{vv}(\ell)$.
The convergence of the first term holds since, with careful calculation,
we can show that
\[
\lim_{T\to\infty}\frac{1-(\rho^{*}\rho_{z})^{2}}{T}\sum_{t=1}^{T}\sum_{k=1}^{t}[B_{T}(t,k)]^{2}=1.
\]
For the second term, we have
\begin{align*}
 & \frac{1-(\rho^{*}\rho_{z})^{2}}{T}\sum_{t=1}^{T}\sum_{h=1}^{t-1}\sum_{\ell=1}^{t-h}B_{T}(t,\ell+h)B_{T}(t,\ell)\Gamma_{vv}(h)\\
 & \quad=\frac{1}{T}\sum_{t=1}^{T}\sum_{h=1}^{t-1}\left(\sum_{\ell=1}^{t-h}\left[1-(\rho^{*}\rho_{z})^{2}\right]B_{T}(t,\ell+h)B_{T}(t,\ell)\right)\Gamma_{vv}(h)\\
 & \quad=\frac{1}{T}\sum_{t=1}^{T}\sum_{h=1}^{t-1}D_{T}(t,h)\Gamma_{vv}(h)=\sum_{h=1}^{T-1}\left[\frac{1}{T}\sum_{t=h+1}^{T}D_{T}(t,h)\right]\Gamma_{vv}(h),
\end{align*}
where
\[
D_{T}(t,h):=\sum_{\ell=1}^{t-h}\left[1-(\rho^{*}\rho_{z})^{2}\right]B_{T}(t,\ell+h)B_{T}(t,\ell).
\]
It is easy to verify that
\[
\lim_{T\to\infty}\frac{1}{T}\sum_{t=h+1}^{T}D_{T}(t,h)=1\quad\text{and}\quad\sup_{T,h}\left|\frac{1}{T}\sum_{t=h+1}^{T}D_{T}(t,h)\right|<\infty,
\]
Therefore, by \prettyref{lem:series} we have
\[
\sum_{h=1}^{T-1}\left[\frac{1}{T}\sum_{t=h+1}^{T}D_{T}(t,h)\right]\Gamma_{vv}(h)\to\sum_{h=1}^{\infty}\Gamma_{vv}(h).
\]
Hence, it follows that
\begin{equation}
\frac{1-(\rho^{*}\rho_{z})^{2}}{T}\sum_{t=1}^{T}\mathbb{E}\left(\Upsilon_{1,i,t}^{2}\right)\to\omega_{vv}^{*}=\sum_{\ell=-\infty}^{\infty}\Gamma_{vv}(\ell).\label{eq:upsilon1}
\end{equation}

By \eqref{eq:upsilon2} and \eqref{eq:upsilon1}, to show \eqref{eq:E_Z2_converge},
it suffices to show that, by \prettyref{assump:innov}\ref{enu:hetero_e},
\begin{equation}
\frac{1-(\rho^{*}\rho_{z})^{2}}{T}\sum_{t=1}^{T}\mathbb{E}\left[\Upsilon_{1,i,t}^{2}(e_{i,t+1}^{2}-\omega_{ee}^{*})\right]=\frac{1-(\rho^{*}\rho_{z})^{2}}{T}\sum_{t=1}^{T}\mathbb{E}\left[\Upsilon_{1,i,t}^{2}(h_{i,t+1}-\omega_{ee}^{*})\right]\to0.\label{eq:Upsilon1_e_2}
\end{equation}
We have
\begin{align*}
\sum_{t=1}^{T}\mathbb{E}\left[\Upsilon_{1,i,t}^{2}(h_{i,t+1}-\omega_{ee}^{*})\right] & =\sum_{t=1}^{T}\sum_{\ell=-\infty}^{t}[M_{T}(t,\ell)]^{2}\mathbb{E}\left[\varepsilon_{i,\ell}^{2}(h_{i,t+1}-\omega_{ee}^{*})\right]\\
 & \qquad+2\sum_{t=1}^{T}\sum_{\ell=-\infty}^{t}\sum_{k=-\infty}^{\ell-1}M_{T}(t,\ell)M_{T}(t,k)\mathbb{E}\left[\varepsilon_{i,\ell}\varepsilon_{i,k}(h_{i,t+1}-\omega_{ee}^{*})\right]\\
 & =:A_{1,i,T}+2A_{2,i,T}.
\end{align*}
First note that
\begin{align*}
A_{1,i,T} & =\sum_{t=1}^{T}\sum_{\ell=-\infty}^{t}[M_{T}(t,\ell)]^{2}\mathbb{E}\left[\left(\varepsilon_{i,\ell}^{2}-\sigma_{\varepsilon\varepsilon}^{*}+\sigma_{\varepsilon\varepsilon}^{*}\right)(h_{i,t+1}-\omega_{ee}^{*})\right]\\
 & =\sum_{t=1}^{T}\sum_{\ell=-\infty}^{t}[M_{T}(t,\ell)]^{2}\mathbb{E}\left[\left(\varepsilon_{i,\ell}^{2}-\sigma_{\varepsilon\varepsilon}^{*}\right)h_{i,t+1}\right]\\
 & =\sum_{\ell=-\infty}^{T}\sum_{t=\ell\vee1}^{T}[M_{T}(t,\ell)]^{2}\mathbb{E}\left[\left(\varepsilon_{i,\ell}^{2}-\sigma_{\varepsilon\varepsilon}^{*}\right)h_{i,t+1}\right].
\end{align*}
Denote $\Sigma_{i,j,t}:=\left(\varepsilon_{i,j}^{2}-\sigma_{\varepsilon\varepsilon}^{*}\right)h_{i,t+1}$.
Using \prettyref{assump:innov}\ref{enu:hetero_e}, we have
\begin{align}
\Sigma_{i,j,t} & =\sum_{m=1}^{p\vee r}c_{m}\Sigma_{i,j,t-m}+\nu_{i,j,t},\label{eq:Sigma_AR}\\
\nu_{i,j,t} & =\left(\varepsilon_{i,j}^{2}-\sigma_{\varepsilon\varepsilon}^{*}\right)\left[\phi+\sum_{k=1}^{r}a_{k}\left(e_{i,t+1-k}^{2}-h_{i,t+1-k}\right)\right],\nonumber \\
c_{m} & =\begin{cases}
a_{m}+b_{m} & i\leq p\wedge r,\\
a_{m} & p<i\leq r,\\
b_{m} & r<i\leq p.
\end{cases}\nonumber
\end{align}
Note that
\begin{align*}
\mathbb{E}(\nu_{i,j,t}) & =\mathbb{E}\left[\left(\varepsilon_{i,j}^{2}-\sigma_{\varepsilon\varepsilon}^{*}\right)\phi\right]+\sum_{k=1}^{r}a_{k}\mathbb{E}\left[\left(\varepsilon_{i,j}^{2}-\sigma_{\varepsilon\varepsilon}^{*}\right)\left(e_{i,t+1-k}^{2}-h_{i,t+1-k}\right)\right],
\end{align*}
where the first term is obviously zero, and the second term is also
zero for $j\leq t-r$ because
\begin{align*}
\mathbb{E}\left[\left(\varepsilon_{i,j}^{2}-\sigma_{\varepsilon\varepsilon}^{*}\right)\left(e_{i,t+1-k}^{2}-h_{i,t+1-k}\right)\right] & =\mathbb{E}\left\{ \mathbb{E}_{t-k}\left[\left(\varepsilon_{i,j}^{2}-\sigma_{\varepsilon\varepsilon}^{*}\right)\left(e_{i,t+1-k}^{2}-h_{i,t+1-k}\right)\right]\right\} \\
 & =\mathbb{E}\left[\left(h_{i,t+1-k}-h_{i,t+1-k}\right)\left(\varepsilon_{i,j}^{2}-\sigma_{\varepsilon\varepsilon}^{*}\right)\right]=0.
\end{align*}
Hence, for each $j$, for $t\geq j+r$, we have $\mathbb{E}(\nu_{i,j,t})=0$
which indicates $\mathbb{E}(\Sigma_{i,j,t})=0$ in view of \eqref{eq:Sigma_AR}.
For $t<j+r$, we simply use
\[
|\mathbb{E}(\Sigma_{i,j,t})|=\left|\mathbb{E}\left(\varepsilon_{i,\ell}^{2}h_{i,t+\ell}\right)-\sigma_{\varepsilon\varepsilon}^{*}\omega_{ee}^{*}\right|\leq\sqrt{\mathbb{E}\bigl(\varepsilon_{i,1}^{4}\bigr)\mathbb{E}\bigl(e_{i,1}^{4}\bigr)}+\sigma_{\varepsilon\varepsilon}^{*}\omega_{ee}^{*}=O(1).
\]
Thus,
\[
\frac{1-(\rho^{*}\rho_{z})^{2}}{T}|A_{1,i,T}|\lesssim\frac{1-(\rho^{*}\rho_{z})^{2}}{T}\sum_{\ell=-\infty}^{T}\sum_{t=\ell\vee1}^{\ell+r}[M_{T}(t,\ell)]^{2}=O\left(1-(\rho^{*}\rho_{z})^{2}\right)\to0
\]
when $\gamma>0$. We can use the same argument to show that
\[
\lim_{T\to\infty}\frac{1-(\rho^{*}\rho_{z})^{2}}{T}|A_{2,i,T}|=0.
\]
Consequently, \eqref{eq:Upsilon1_e_2} holds. It then follows that
\begin{align*}
\mathbb{E}\left(Z_{i,T}^{2}\right) & =\frac{1-(\rho^{*}\rho_{z})^{2}}{T}\sum_{t=1}^{T}\mathbb{E}\left(z_{i,t}^{2}e_{i,t+1}^{2}\right)\\
 & =\frac{1-(\rho^{*}\rho_{z})^{2}}{T}\sum_{t=1}^{T}\mathbb{E}\left(\Upsilon_{1,i,t}^{2}e_{i,t+1}^{2}\right)+\frac{1-(\rho^{*}\rho_{z})^{2}}{T}\sum_{t=1}^{T}\mathbb{E}\left(\Upsilon_{2,i,t}^{2}e_{i,t+1}^{2}\right)\\
 & \qquad+2\frac{1-(\rho^{*}\rho_{z})^{2}}{T}\sum_{t=1}^{T}\mathbb{E}\left(\Upsilon_{1,i,t}\Upsilon_{2,i,t}e_{i,t+1}^{2}\right)\\
 & =\frac{1-(\rho^{*}\rho_{z})^{2}}{T}\sum_{t=1}^{T}\mathbb{E}\left(\Upsilon_{1,i,t}^{2}(e_{i,t+1}^{2}-\omega_{ee}^{*})\right)+\frac{1-(\rho^{*}\rho_{z})^{2}}{T}\sum_{t=1}^{T}\mathbb{E}\left(\Upsilon_{1,i,t}^{2}\right)\omega_{ee}^{*}+o(1)\\
 & =o(1)+\omega_{ee}^{*}\omega_{vv}^{*}+o(1)\to\omega_{ee}^{*}\omega_{vv}^{*}.
\end{align*}

By \prettyref{lem:ivx_marginal}\ref{enu:Z_dto} and the continuous
mapping theorem,
\[
Z_{i,T}^{2}\to_{d}Z_{\infty}^{2}\qquad\text{as }T\to\infty,
\]
where $\mathbb{E}(Z_{\infty}^{2})={\rm S}_{xe}$. Thus, $Z_{i,T}^{2}$
is u.i.~in $T$ by \prettyref{lem:u.i.}. By Theorem~3 of \citet{phillips1999linear},
we conclude that $n^{-1/2}\sum_{i=1}^{n}Z_{i,T}\to_{d}\mathcal{N}(0,{\rm S}_{xe})$
as $(n,T)\to\infty$.

\textbf{Step II. }By the i.i.d.\ condition across $i$ and \prettyref{lem:ivxepct}\ref{enu:ze},
as $(n,T)\to\infty$,
\begin{align}
 & \mathbb{E}\left[\left(\frac{1}{\sqrt{n}}\sum_{i=1}^{n}[H_{i,T}-\mathbb{E}(H_{i,T})]\right)^{2}\right]=\mathbb{E}\left([H_{i,T}-\mathbb{E}(H_{i,T})]^{2}\right)\leq\mathbb{E}(H_{i,T}^{2})\nonumber \\
 & \qquad=\frac{1-(\rho^{*}\rho_{z})^{2}}{T^{3}}\mathbb{E}\left[\left(\sum_{t=1}^{T}z_{i,t}\sum_{t=1}^{T}e_{i,t+1}\right)^{2}\right]\nonumber \\
 & \qquad=O\left(\dfrac{1}{T^{3+(\theta\wedge\gamma)}}\right)\cdot O\bigl(T^{2+\theta+(\theta\wedge\gamma)}\bigr)=O\left(\frac{1}{T^{1-\theta}}\right)\to0,\label{eq: var H o(1)}
\end{align}
which, by Markov's inequality, implies $n^{-1/2}\sum_{i=1}^{n}[H_{i,T}-\mathbb{E}(H_{i,T})]\to_{p}0$
as $(n,T)\to\infty$. It then follows that $n^{-1/2}\sum_{i=1}^{n}L_{i,T}\to_{d}\mathcal{N}(0,{\rm S}_{xe})$.
\end{proof}
\begin{proof}[Proof of \prettyref{lem:Omega_hat}]
We only show \eqref{eq:hat omg 12 rate}. The other two formulae
can be deduced in the same manner. By elementary calculations, we
can decompose the estimation error into
\[
\hat{\omega}_{ev,h}(\hat{\beta},\hat{\rho})-\omega_{h}^{*}=W_{1,h,n,T}+W_{2,h,n,T}+W_{3,h,n,T}+W_{4,h,n,T},
\]
where
\begin{align*}
W_{1,h,n,T} & =\dfrac{1}{n(T-h)}\sum_{i=1}^{n}\sum_{t=1}^{T-h}\tilde{x}_{i,t-1}\tilde{x}_{i,t+h-1}(\hat{\beta}-\beta^{*})(\hat{\rho}-\rho^{*}),\\
W_{2,h,n,T} & =\dfrac{1}{n(T-h)}\sum_{i=1}^{n}\sum_{t=1}^{T-h}\tilde{x}_{i,t+h-1}e_{i,t}(\rho^{*}-\hat{\rho}),\\
W_{3,h,n,T} & =\dfrac{1}{n(T-h)}\sum_{i=1}^{n}\sum_{t=1}^{T-h}\tilde{x}_{i,t-1}v_{i,t+h}(\beta^{*}-\hat{\beta}),\\
W_{4,h,n,T} & =\dfrac{1}{n(T-h)}\sum_{i=1}^{n}\sum_{t=1}^{T-h}(e_{i,t}-\bar{e}_{i})v_{i,t+h}-\omega_{ev,h}^{*}.
\end{align*}
By the same argument as in the proof of \prettyref{lem:generic_two_ARs}\ref{enu:Esum2}
we can get for $G=o(T)$,
\[
\sup_{h\leq G}\mathbb{E}\left[\left(\sum_{t=1}^{T}\tilde{x}_{i,t}\tilde{x}_{i,t+h}\right)^{2}\right]=O(T^{2+2\gamma}),
\]
which implies
\[
\sup_{h\leq G}\mathbb{E}(W_{1,h,n,T}^{2})=O\big((\hat{\rho}-\rho^{*})^{2}(\hat{\beta}-\beta^{*})^{2}T^{2\gamma}\bigr).
\]
Similarly we can deduce
\[
\sup_{h\leq G}\mathbb{E}(W_{2,h,n,T}^{2})=O\big((\hat{\rho}-\rho^{*})^{2}\bigr)\quad\text{and}\quad\sup_{h\leq G}\mathbb{E}(W_{3,h,n,T}^{2})=O\big((\hat{\beta}-\beta^{*})^{2}\bigr).
\]
Lastly, note that by the i.i.d.\ condition across $i$ and $\mathbb{E}(e_{i,t+1}^{2}v_{i,t+1}^{2})<\infty$
uniformly for all $t$ (which is implied by \prettyref{assump:innov}\ref{enu:w_cumu}),
\begin{eqnarray*}
 &  & \sup_{h\leq G}\mathbb{E}\left[\left(\dfrac{1}{\sqrt{n(T-h)}}\sum_{i=1}^{n}\sum_{t=1}^{T-h}(e_{i,t}v_{i,t+h}-\omega_{ev,h}^{*})\right)^{2}\right]\\
 & = & \sup_{h\leq G}\frac{1}{T-h}\sum_{t=1}^{T-h}\mathbb{E}\left[(e_{i,t}v_{i,t+h}-\omega_{ev,h}^{*})^{2}\right]=O(1),
\end{eqnarray*}
Moreover, we can deduce by \prettyref{lem:generic_two_ARs}\ref{enu:Esumsum2}
that
\begin{eqnarray*}
 &  & \sup_{h\leq G}\mathbb{E}\left[\left(\dfrac{1}{\sqrt{n}(T-h)}\sum_{i=1}^{n}\left[\sum_{t=1}^{T-h}e_{i,t}\sum_{t=1}^{T-h}v_{i,t+h}\right]\right)^{2}\right]\\
 & = & \sup_{h\leq G}\mathbb{E}\left[\left(\dfrac{1}{T-h}\sum_{t=1}^{T-h}e_{i,t}\sum_{t=1}^{T-h}v_{i,t+h}\right)^{2}\right]=O(1).
\end{eqnarray*}
It thus follows that
\[
\sup_{h\leq G}\mathbb{E}(W_{4,h,n,T}^{2})=O\left(\frac{1}{nT}\right).
\]
These combined yield
\begin{align*}
\mathbb{E}\left|\sum_{h=1}^{G}|\hat{\omega}_{ev,h}(\hat{\beta},\hat{\rho})-\omega_{ev,h}^{*}|\right| & \leq\sum_{h=1}^{G}\mathbb{E}\left(|\hat{\omega}_{ev,h}(\hat{\beta},\hat{\rho})-\omega_{ev,h}^{*}|\right)\\
 & \lesssim\sum_{h=1}^{G}\mathbb{E}\left(|W_{1,h,n,T}|+|W_{2,h,n,T}|+|W_{3,h,n,T}|+|W_{4,h,n,T}|\right)\\
 & =O\left(G\left[\frac{1}{\sqrt{nT}}+|\hat{\rho}-\rho^{*}|+|\hat{\beta}-\beta^{*}|+T^{\gamma}|\hat{\rho}-\rho^{*}||\hat{\beta}-\beta^{*}|\right]\right)
\end{align*}
as well as
\begin{align*}
\mathbb{E}\left|\sum_{h=1}^{G}[\hat{\omega}_{ev,h}(\hat{\beta},\hat{\rho})-\omega_{ev,h}^{*}]^{2}\right| & \leq\sum_{h=1}^{G}\mathbb{E}\left([\hat{\omega}_{ev,h}(\hat{\beta},\hat{\rho})-\omega_{ev,h}^{*}]^{2}\right)\\
 & \lesssim\sum_{h=1}^{G}\mathbb{E}\left(W_{1,h,n,T}^{2}+W_{2,h,n,T}^{2}+W_{3,h,n,T}^{2}+W_{4,h,n,T}^{2}\right)\\
 & =O\left(G^{2}\left[\frac{1}{nT}+(\hat{\rho}-\rho^{*})^{2}+(\hat{\beta}-\beta^{*})^{2}+T^{2\gamma}(\hat{\rho}-\rho^{*})^{2}(\hat{\beta}-\beta^{*})^{2}\right]\right).
\end{align*}
This implies
\[
\sum_{h=1}^{G}|\hat{\omega}_{ev,h}(\hat{\beta},\hat{\rho})-\omega_{ev,h}^{*}|=O_{p}\left(G\left[\frac{1}{\sqrt{nT}}+|\hat{\rho}-\rho^{*}|+|\hat{\beta}-\beta^{*}|+T^{\gamma}|\hat{\rho}-\rho^{*}||\hat{\beta}-\beta^{*}|\right]\right)
\]
and
\[
\sum_{h=1}^{G}[\hat{\omega}_{ev,h}(\hat{\beta},\hat{\rho})-\omega_{ev,h}^{*}]^{2}=O_{p}\left(G^{2}\left[\frac{1}{nT}+(\hat{\rho}-\rho^{*})^{2}+(\hat{\beta}-\beta^{*})^{2}+T^{2\gamma}(\hat{\rho}-\rho^{*})^{2}(\hat{\beta}-\beta^{*})^{2}\right]\right).
\]
This completes the proof.
\end{proof}
\begin{proof}[Proof of \prettyref{lem:ivxerror}]
Without loss of generality, we assume $\rho^{*}>0$. By \prettyref{lem:ivxjoint}\ref{enu:Q_pto}\ref{enu:R_pto},
as $(n,T)\to\infty$,
\[
\dfrac{1-(\rho_{z}\rho^{*})^{2}}{nT}\sum_{i=1}^{n}\sum_{t=1}^{T}\tilde{z}_{i,t}x_{i,t}\to_{p}\mathbb{E}(Q_{zx}-R_{zx})\neq0.
\]
Hence, in view of $1-(\rho_{z}\rho^{*})^{2}=T^{-(\theta\wedge\gamma)}$,
we have
\begin{equation}
\left(\frac{1}{n}\sum_{i=1}^{n}\sum_{t=1}^{T}\tilde{z}_{i,t}x_{i,t}\right)^{-1}=O_{p}\left(\dfrac{1}{T^{1+(\theta\wedge\gamma)}}\right).\label{eq:1 over zeta x}
\end{equation}
It follows that
\begin{align*}
 & b_{n,T}^{\mathrm{IVX}}(\hat{\rho})-b_{n,T}^{\mathrm{IVX}}(\rho^{*})\\
={} & \left(\sum_{h=0}^{G}\Psi_{h,T}(\hat{\rho},\rho_{z})\hat{\omega}_{ev,h}-\sum_{h=0}^{T-2}\Psi_{h,T}(\rho^{*},\rho_{z})\omega_{ev,h}^{*}\right)\cdot O_{p}\left(\dfrac{1}{T^{2+(\theta\wedge\gamma)}}\right)\\
={} & \left(\sum_{h=0}^{G}\Psi_{h,T}(\hat{\rho},\rho_{z})(\hat{\omega}_{ev,h}-\omega_{ev,h}^{*})+\sum_{h=0}^{G}[\Psi_{h,T}(\hat{\rho},\rho_{z})-\Psi_{h,T}(\rho^{*},\rho_{z})]\omega_{ev,h}^{*}\right.\\
 & -\left.\sum_{h=G+1}^{T-2}\Psi_{h,T}(\rho^{*},\rho_{z})\omega_{ev,h}^{*}\right)\cdot O_{p}\left(\dfrac{1}{T^{2+(\theta\wedge\gamma)}}\right)
\end{align*}
Since $\hat{\rho}-\rho^{*}=O_{p}(T^{-\eta})$,
\[
\hat{\rho}-1=\hat{\rho}-\rho^{*}+\rho^{*}-1=O_{p}\left(\frac{1}{T^{\eta}}+\frac{1}{T^{\gamma}}\right)=O_{p}(T^{-(\eta\wedge\gamma)}).
\]
Thus by \eqref{eq:Psi-h} it holds that
\begin{equation}
\sum_{h=0}^{G}[\Psi_{h,T}(\hat{\rho},\rho_{z})]^{2}=O_{p}\bigl(GT^{2(\theta\wedge\gamma)+2[(\eta\wedge\gamma)\vee\theta]}\bigr)=O_{p}\bigl(GT^{2(\theta\wedge\gamma)+2(\gamma\vee\theta)}\bigr)=O_{p}\bigl(GT^{2(\theta+\gamma)}\bigr).\label{eq:sum_psi_2}
\end{equation}
By \prettyref{cor:omega_mixed}, we have
\begin{equation}
\sum_{h=0}^{G}(\hat{\omega}_{ev,h}-\omega_{ev,h}^{*})^{2}=O_{p}\left(\dfrac{G}{nT}+\frac{G}{T^{2}}\right).\label{eq:sum_omega_divx_2}
\end{equation}
Then by the Cauchy-Schwarz inequality, \eqref{eq:sum_psi_2} and \eqref{eq:sum_omega_divx_2}
lead to
\begin{align}
\sum_{h=0}^{G}\Psi_{h,T}(\hat{\rho},\rho_{z})(\hat{\omega}_{ev,h}-\omega_{ev,h}^{*}) & \leq\left(\sum_{h=0}^{G}[\Psi_{h,T}(\hat{\rho},\rho_{z})]^{2}\right)^{1/2}\left(\sum_{h=0}^{G}(\hat{\omega}_{ev,h}-\omega_{ev,h}^{*})^{2}\right)^{1/2}\nonumber \\
 & =O_{p}\left(\dfrac{GT^{\theta+\gamma}}{\sqrt{nT}}+\frac{G}{T^{1-\theta-\gamma}}\right).\label{eq:sum_psi_omega_diff}
\end{align}

By the differential mean value theorem, there exists a $\check{\rho}$
between $\hat{\rho}$ and $\rho^{*}$ such that $\Psi_{h,T}(\hat{\rho},\rho_{z})-\Psi_{h,T}(\rho^{*},\rho_{z})=\frac{d}{d\rho}\Psi_{h,T}(\check{\rho},\rho_{z})(\hat{\rho}-\rho^{*}).$
To bound $\frac{d}{d\rho}\Psi_{h,T}(\check{\rho},\rho_{z})$, we first
prove a useful asymptotic result
\begin{equation}
\sup_{h\leq T}{}\left|\dfrac{1-(T-h)\check{\rho}^{T-h-1}}{1-\check{\rho}}\right|=O(T^{2\gamma}).\label{eq:useful asymp}
\end{equation}
Since $\hat{\rho}-1=O_{p}(T^{-(\eta\wedge\gamma)})$ and $\check{\rho}$
is between $\hat{\rho}$ and $\rho^{*}$, we have $1-O_{p}(T^{-(\eta\wedge\gamma)})\leq\check{\rho}\leq1+O_{p}(T^{-\gamma})$,
which gives rise to $1-\check{\rho}=O_{p}(T^{-\gamma})$.

When $c^{*}<0$ and $\gamma\in[0,1)$, for any constant $\alpha$,
as $T\to\infty$ we have $T^{\alpha}\check{\rho}^{T}\leq T^{\alpha}\bigl(1+\frac{c^{*}}{T^{\gamma}}\bigr)^{T}\to0$
with probability approaching one, where the convergence is due to
\[
\lim_{T\to\infty}T^{\alpha}\left(1+\frac{c^{*}}{T^{\gamma}}\right)^{T}=\exp\left(\lim_{x\to0^{+}}\frac{\log(1+c^{*}x^{\gamma})-\alpha x\log x}{x}\right)=\exp(-\infty)=0.
\]
This implies that $T^{\alpha}\check{\rho}^{T}=o_{p}(1)$ for any constant
$\alpha$. Therefore, $\sup_{h\leq T}{}\bigl|\frac{1-(T-h)\check{\rho}^{T-h-1}}{1-\check{\rho}}\bigr|=O_{p}(T^{\gamma}).$
When $\gamma=1$, $\check{\rho}^{T}\leq\bigl(1+\frac{|c^{*}|}{T^{\gamma}}\bigr)^{T}\to\exp(|c^{*}|)$
and thus $\sup_{h\leq T}\bigl|\frac{1-(T-h)\check{\rho}^{T-h-1}}{1-\check{\rho}}\bigr|=O_{p}(T^{2\gamma}).$
Then \eqref{eq:useful asymp} is verified.

For $\rho\neq1$, $\Psi_{h,T}(\rho,\rho_{z})$ can be written as
\[
\Psi_{h,T}(\rho,\rho_{z})=\frac{1}{\rho_{z}-\rho}\left(\frac{\rho_{z}-\rho_{z}^{T-h}}{1-\rho_{z}}-\frac{\rho-\rho^{T-h}}{1-\rho}\right).
\]
Then, we have
\begin{eqnarray*}
 &  & \sup_{h\leq T}{}\left|\frac{d}{d\rho}\Psi_{h,T}(\check{\rho},\rho_{z})\right|\\
 & = & \sup_{h\leq T}{}\left|\dfrac{1}{(\rho_{z}-\check{\rho})^{2}}\left(\dfrac{\rho_{z}-\rho_{z}^{T-h}}{1-\rho_{z}}-\dfrac{\check{\rho}-\check{\rho}^{T-h}}{1-\check{\rho}}\right)-\dfrac{1}{\rho_{z}-\check{\rho}}\left(\dfrac{\check{\rho}-\check{\rho}^{T-h}}{(1-\check{\rho})^{2}}+\dfrac{1-(T-h)\check{\rho}^{T-h-1}}{1-\check{\rho}}\right)\right|\\
 & = & O_{p}\bigl(T^{2(\theta\wedge\gamma)}\cdot(T^{\theta}+T^{\gamma})\bigr)+O_{p}\bigl(T^{\theta\wedge\gamma}\cdot(T^{2\gamma}+T^{2\gamma})\bigr)\\
 & = & O_{p}\bigl(T^{2(\theta\wedge\gamma)+(\theta\vee\gamma)}\bigr)+O_{p}\bigl(T^{(\theta\wedge\gamma)+2\gamma}\bigr)=O_{p}(T^{\theta+2\gamma}),
\end{eqnarray*}
where the second line applies \eqref{eq:useful asymp}. It follows
that
\begin{align}
\left|\sum_{h=0}^{G}[\Psi_{h,T}(\hat{\rho},\rho_{z})-\Psi_{h,T}(\rho^{*},\rho_{z})]\omega_{ev,h}^{*}\right| & \leq\left(\sup_{h\leq T}{}\left|\Psi_{h,T}(\hat{\rho},\rho_{z})-\Psi_{h,T}(\rho^{*},\rho_{z})\right|\right)\sum_{h=0}^{G}|\omega_{ev,h}^{*}|\nonumber \\
 & \leq\sup_{h\leq T}{}\left|\frac{d}{d\rho}\Psi_{h,T}(\check{\rho},\rho_{z})\right||\hat{\rho}-\rho^{*}|\cdot O(1)\nonumber \\
 & =O_{p}(T^{\theta+2\gamma}|\hat{\rho}-\rho^{*}|).\label{eq:sum_psi_diff}
\end{align}
Lastly, noting that $v_{i,t}=\sum_{s=0}^{\infty}g_{s}\varepsilon_{i,t-s}$,
we have $|\omega_{ev,h}|\lesssim|g_{h}|\lesssim q_{\nu}^{h}$ where
$q_{\nu}=\exp(-C_{g})$ and therefore
\begin{align}
\left|\sum_{h=G+1}^{T-2}\Psi_{h,T}(\rho^{*},\rho_{z})\omega_{ev,h}^{*}\right| & \lesssim\sum_{h=G+1}^{T-2}\frac{1}{|\rho_{z}-\rho^{*}|}\left(\frac{|\rho_{z}|+|\rho_{z}^{T-h}|}{|1-\rho_{z}|}+\frac{|\rho^{*}|+|\rho^{*T-h}|}{|1-\rho^{*}|}\right)q_{\nu}^{h}\nonumber \\
 & =O\bigl(T^{\theta+\gamma}\cdot q_{\nu}^{G}\bigr).\label{eq:sum_psi_omega}
\end{align}

Combining \eqref{eq:sum_psi_omega_diff}, \eqref{eq:sum_psi_diff}
and \eqref{eq:sum_psi_omega} we conclude that
\begin{align*}
 & b_{n,T}^{\mathrm{IVX}}(\hat{\rho})-b_{n,T}^{\mathrm{IVX}}(\rho^{*})\\
={} & O_{p}\left(\dfrac{1}{T^{2+(\theta\wedge\gamma)}}\right)\cdot O_{p}\left(\dfrac{GT^{\theta+\gamma}}{\sqrt{nT}}+\frac{G}{T^{1-\theta-\gamma}}+T^{\theta+2\gamma}|\hat{\rho}-\rho^{*}|+T^{\theta+\gamma}\cdot q^{G}\right)\\
={} & O_{p}\left(\dfrac{G}{\sqrt{nT^{5-2(\theta\vee\gamma)}}}+\frac{G}{T^{3-(\theta\vee\gamma)}}+\frac{|\hat{\rho}-\rho^{*}|}{T^{2-(\theta\vee\gamma)-\gamma}}+\frac{q^{G}}{T^{2-(\theta\vee\gamma)}}\right).
\end{align*}
We complete the proof.
\end{proof}

\section{Proof of Moments of Stochastic Integrals\label{sec:Proof-of-Moments}}
\begin{proof}[Proof of \prettyref{lem:expectations}]
For (\ref{eq:E int J2}), we have
\[
\mathbb{E}\left[\int_{0}^{1}J_{2,c^{*}}(r)^{2}\,dr\right]=\int_{0}^{1}\mathbb{E}[J_{2,c^{*}}(r)^{2}]\,dr=\omega_{vv}^{*}\int_{0}^{1}\frac{1}{2c^{*}}(e^{2rc^{*}}-1)\,dr=\omega_{vv}^{*}\frac{e^{2c^{*}}-2c^{*}-1}{4c^{*2}},
\]
where the first equality is due to Fubini's theorem and the second
uses \eqref{eq:E JJ}.

For (\ref{eq:E of squared int J}), we have
\begin{eqnarray*}
 &  & \mathbb{E}\left[\biggl(\int_{0}^{1}J_{2,c^{*}}(r)\,dr\biggr)^{2}\right]\\
 & = & \int_{r=0}^{1}\int_{s=0}^{1}\mathbb{E}[J_{2,c^{*}}(r)J_{2,c^{*}}(s)]\,ds\,dr\\
 & = & \int_{r=0}^{1}\int_{s=0}^{r}\mathbb{E}[J_{2,c^{*}}(r)J_{2,c^{*}}(s)]\,ds\,dr+\int_{r=0}^{1}\int_{s=r}^{1}\mathbb{E}[J_{2,c^{*}}(r)J_{2,c^{*}}(s)]\,ds\,dr\\
 & = & 2\int_{r=0}^{1}\int_{s=0}^{r}\mathbb{E}[J_{2,c^{*}}(r)J_{2,c^{*}}(s)]\,ds\,dr=2\int_{r=0}^{1}\int_{s=0}^{r}\frac{\omega_{vv}^{*}}{2c^{*}}[e^{c^{*}(r+s)}-e^{c^{*}(r-s)}]\,ds\,dr\\
 & = & \frac{\omega_{vv}^{*}}{c^{*2}}\int_{r=0}^{1}[e^{2c^{*}r}-2e^{c^{*}r}+1]\,dr=\omega_{vv}^{*}\frac{2c^{*}+(1-e^{c^{*}})(3-e^{c^{*}})}{2c^{*3}},
\end{eqnarray*}
where the first line is due to Fubini's theorem, and the fourth equality
applies the fact that for $s\leq r$,
\begin{eqnarray}
 &  & \mathbb{E}[J_{2,c^{*}}(r)J_{2,c^{*}}(s)]\nonumber \\
 & = & \mathbb{E}\left[\int_{u=0}^{r}e^{(r-u)c^{*}}\,dB_{2}(u)\int_{\tau=0}^{s}e^{(s-\tau)c^{*}}\,dB_{2}(\tau)\right]\nonumber \\
 & = & e^{c^{*}(r+s)}\mathbb{E}\left[\int_{u=0}^{r}e^{-c^{*}u}\,dB_{2}(u)\int_{\tau=0}^{s}e^{-c^{*}\tau}\,dB_{2}(\tau)\right]\nonumber \\
 & = & e^{c^{*}(r+s)}\mathbb{E}\left[\omega_{vv}^{*}\int_{u=0}^{s}e^{-2c^{*}u}\,du\right]=\frac{\omega_{vv}^{*}}{2c^{*}}[e^{c^{*}(r+s)}-e^{c^{*}(r-s)}],\label{eq:E JJ}
\end{eqnarray}
where the third line uses Itô's isometry (cf.\ \citet{karatzas2014brownian}
Equation (2.14) in Section~3.2).

For (\ref{eq:E of squared B_int_J}), we first derive some useful
quantities. Let $K_{2,c^{*}}(r):=\int_{0}^{r}e^{-c^{*}\tau}\,dB_{2}(\tau)$.
By Itô's isometry, we have
\begin{align}
\mathbb{E}\left[K_{2,c^{*}}(r)^{2}\right] & =\omega_{vv}^{*}\int_{0}^{r}e^{-2c^{*}\tau}\,d\tau=\frac{\omega_{vv}^{*}}{2c^{*}}(1-e^{-2c^{*}r}),\label{eq:E K2}\\
\mathbb{E}\bigl[B_{1}(r)K_{2,c^{*}}(r)\bigr] & =\omega_{ev}^{*}\int_{0}^{r}e^{-c^{*}\tau}\,d\tau=\frac{\omega_{ev}^{*}}{c^{*}}(1-e^{-c^{*}r}).\label{eq:E BK}
\end{align}
Let $W(r)$ and $W^{\perp}(r)$ be two independent standard Brownian
motions defined on the same probability space as $\bm{B}(r)$. Specifically,
they are constructed by
\begin{align*}
W(r) & =\frac{1}{\sqrt{\omega_{ee}^{*}}}B_{1}(r),\ \ \text{and }\ \ W^{\perp}(r)=\frac{1}{\sqrt{1-\varrho^{*2}}}\left[\frac{1}{\sqrt{\omega_{vv}^{*}}}B_{2}(r)-\frac{\varrho^{*}}{\sqrt{\omega_{ee}^{*}}}B_{1}(r)\right],
\end{align*}
where $\varrho^{*}:=\omega_{ev}^{*}/\sqrt{\omega_{ee}^{*}\omega_{vv}^{*}}$
is the correlation between $B_{1}(r)$ and $B_{2}(r)$. In matrix
form, we have
\[
\begin{bmatrix}B_{1}(r)\\
B_{2}(r)
\end{bmatrix}=\begin{bmatrix}\sqrt{\omega_{ee}^{*}} & 0\\
\sqrt{\omega_{vv}^{*}}\varrho^{*} & \sqrt{\omega_{vv}^{*}(1-\varrho^{*2})}
\end{bmatrix}\begin{bmatrix}W(r)\\
W^{\perp}(r)
\end{bmatrix}.
\]
Then, $[B_{1}(r),K_{2,c^{*}}(r)]'$ satisfy the following stochastic
differential equation:
\[
\begin{bmatrix}dB_{1}(r)\\
dK_{2,c^{*}}(r)
\end{bmatrix}=\begin{bmatrix}1 & 0\\
0 & e^{-c^{*}r}
\end{bmatrix}\begin{bmatrix}dB_{1}(r)\\
dB_{2}(r)
\end{bmatrix}=\begin{bmatrix}\sqrt{\omega_{ee}^{*}} & 0\\
\sqrt{\omega_{vv}^{*}}\varrho^{*}e^{-c^{*}r} & \sqrt{\omega_{ee}^{*}(1-\varrho^{*2})}e^{-c^{*}r}
\end{bmatrix}\begin{bmatrix}dW(r)\\
dW^{\perp}(r)
\end{bmatrix}.
\]
Invoking Itô's lemma (cf.\ \citet{karatzas2014brownian} Theorem~3.6
in Section~3.3), the stochastic differential of $B_{1}(r)^{2}K_{2,c^{*}}(r)$
is given by
\begin{align*}
d[B_{1}(r)^{2}K_{2,c^{*}}(r)] & =\left[\omega_{ee}^{*}K_{2,c^{*}}(r)+2\sqrt{\omega_{ee}^{*}\omega_{vv}^{*}}\varrho^{*}e^{-c^{*}r}B_{1}(r)\right]\,dr\\
 & \quad\,+\left[2\sqrt{\omega_{ee}^{*}}B_{1}(r)K_{2,c^{*}}(r)+\sqrt{\omega_{vv}^{*}}\varrho^{*}e^{-c^{*}r}B_{1}(r)^{2}\right]\,dW(r)\\
 & \quad\,+\sqrt{\omega_{vv}^{*}(1-\varrho^{*2})}e^{-c^{*}r}B_{1}(r)^{2}\,dW^{\perp}(r),
\end{align*}
or in the integral form:
\begin{align*}
B_{1}(r)^{2}K_{2,c^{*}}(r) & =\int_{0}^{r}\left[\omega_{11}^{*}K_{2,c^{*}}(\tau)+2\sqrt{\omega_{ee}^{*}\omega_{vv}^{*}}\varrho^{*}e^{-c^{*}\tau}B_{1}(\tau)\right]\,d\tau\\
 & \quad\,+\int_{0}^{r}\left[2\sqrt{\omega_{ee}^{*}}B_{1}(\tau)K_{2,c^{*}}(\tau)+\sqrt{\omega_{vv}^{*}}\varrho^{*}e^{-c^{*}\tau}B_{1}(\tau)^{2}\right]\,dW(\tau)\\
 & \quad\,+\int_{0}^{r}\sqrt{\omega_{vv}^{*}(1-\varrho^{*2})}e^{-c^{*}\tau}B_{1}(\tau)^{2}\,dW^{\perp}(\tau).
\end{align*}
Its conditional expectation with respect to $\mathcal{F}_{s}$ for
$s\leq r$ equals
\begin{align*}
\mathbb{E}[B_{1}(r)^{2}K_{2,c^{*}}(r)|\mathcal{F}_{s}] & =\int_{0}^{s}\left[\omega_{ee}^{*}K_{2,c^{*}}(\tau)+2\omega_{ev}^{*}e^{-c^{*}\tau}B_{1}(\tau)\right]\,d\tau\\
 & \quad\,+\int_{s}^{r}\left[\omega_{ee}^{*}K_{2,c^{*}}(s)+2\omega_{ev}^{*}e^{-c^{*}\tau}B_{1}(s)\right]\,d\tau\\
 & \quad\,+\int_{0}^{s}\left[2\sqrt{\omega_{ee}^{*}}B_{1}(\tau)K_{2,c^{*}}(\tau)+\sqrt{\omega_{vv}^{*}}\varrho^{*}e^{-c^{*}\tau}B_{1}(\tau)^{2}\right]\,dW(\tau)\\
 & \quad\,+\int_{0}^{s}\sqrt{\omega_{vv}^{*}(1-\varrho^{*2})}e^{-c^{*}\tau}B_{1}(\tau)^{2}\,dW^{\perp}(\tau).
\end{align*}
Then we have, for $s\leq r$,
\begin{align}
 & \mathbb{E}[B_{1}(r)^{2}K_{2,c^{*}}(r)K_{2,c^{*}}(s)]=\mathbb{E}\left[\mathbb{E}[B_{1}(r)^{2}K_{2,c^{*}}(r)|\mathcal{F}_{s}]K_{2,c^{*}}(s)\right]\nonumber \\
={} & \int_{0}^{s}\left(\omega_{ee}^{*}\mathbb{E}[K_{2,c^{*}}(\tau)^{2}]+2\omega_{ev}^{*}e^{-c^{*}\tau}\mathbb{E}\bigl[B_{1}(\tau)K_{2,c^{*}}(\tau)\bigr]\right)\,d\tau\nonumber \\
 & +\int_{s}^{r}\left(\omega_{ee}^{*}\mathbb{E}[K_{2,c^{*}}(s)^{2}]+2\omega_{ev}^{*}e^{-c^{*}\tau}\mathbb{E}\bigl[B_{1}(s)K_{2,c^{*}}(s)\bigr]\right)\,d\tau\nonumber \\
 & +\int_{0}^{s}\sqrt{\omega_{vv}^{*}}\varrho^{*}e^{-c^{*}\tau}\left(2\sqrt{\omega_{ee}^{*}}\mathbb{E}\bigl[B_{1}(\tau)K_{2,c^{*}}(\tau)\bigr]+\sqrt{\omega_{vv}^{*}}\varrho^{*}e^{-c^{*}\tau}\mathbb{E}\left[B_{1}(\tau)^{2}\right]\right)\,d\tau\nonumber \\
 & +\int_{0}^{s}\omega_{vv}^{*}(1-\varrho^{*2})e^{-2c^{*}\tau}\mathbb{E}\left[B_{1}(\tau)^{2}\right]\,d\tau\nonumber \\
={} & \frac{\omega_{ee}^{*}\omega_{vv}^{*}}{2c^{*}}\int_{0}^{s}(1-e^{-2c^{*}\tau})\,d\tau+\frac{2\omega_{ev}^{*2}}{c^{*}}\int_{0}^{s}e^{-c^{*}\tau}(1-e^{-c^{*}\tau})\,d\tau\nonumber \\
 & +\frac{\omega_{ee}^{*}\omega_{vv}^{*}}{2c^{*}}(r-s)(1-e^{-2c^{*}s})+\frac{2\omega_{ev}^{*2}}{c^{*}}(1-e^{-c^{*}s})\int_{s}^{r}e^{-c^{*}\tau}\,d\tau\nonumber \\
 & +\int_{0}^{s}e^{-c^{*}\tau}\left(\frac{2\omega_{ev}^{*2}}{c^{*}}(1-e^{-c^{*}\tau})+\omega_{ev}^{*2}e^{-c^{*}\tau}\tau\right)\,d\tau\nonumber \\
 & +\int_{0}^{s}(\omega_{ee}^{*}\omega_{vv}^{*}-\omega_{ev}^{*2})e^{-2c^{*}\tau}\tau\,d\tau\nonumber \\
={} & \frac{\omega_{ee}^{*}\omega_{vv}^{*}}{2c^{*}}\left[s-\frac{1}{2c^{*}}(1-e^{-2c^{*}s})\right]+\biggl[\frac{\omega_{ev}^{*}}{c^{*}}(1-e^{-c^{*}s})\biggr]^{2}\nonumber \\
 & +\frac{\omega_{ee}^{*}\omega_{vv}^{*}}{2c^{*}}(r-s)(1-e^{-2c^{*}s})+\frac{2\omega_{ev}^{*2}}{c^{*2}}(1-e^{-c^{*}s})(e^{-c^{*}s}-e^{-c^{*}r})\nonumber \\
 & +\biggl[\frac{\omega_{ev}^{*}}{c^{*}}(1-e^{-c^{*}s})\biggr]^{2}+\frac{\omega_{ee}^{*}\omega_{vv}^{*}}{4c^{*2}}\bigl[1-(2c^{*}s+1)e^{-2c^{*}s}\bigr]\nonumber \\
={} & \frac{\omega_{ee}^{*}\omega_{vv}^{*}}{2c^{*}}r(1-e^{-2c^{*}s})+\frac{2\omega_{ev}^{*2}}{c^{*2}}(1-e^{-c^{*}s})(1-e^{-c^{*}r}),\label{eq:E B2KK}
\end{align}
where the third equality uses \eqref{eq:E K2} and \eqref{eq:E BK}.
Finally, we have
\begin{align*}
 & \mathbb{E}\left[\biggl(B_{1}(1)\int_{0}^{1}J_{2,c^{*}}(r)\,dr\biggr)^{2}\right]=\int_{r=0}^{1}\int_{s=0}^{1}\mathbb{E}[B_{1}(1)^{2}J_{2,c^{*}}(r)J_{2,c^{*}}(s)]\,ds\,dr\\
={} & 2\int_{r=0}^{1}\int_{s=0}^{r}\mathbb{E}[B_{1}(1)^{2}J_{2,c^{*}}(r)J_{2,c^{*}}(s)]\,ds\,dr\\
={} & 2\int_{r=0}^{1}\int_{s=0}^{r}e^{c^{*}(r+s)}\mathbb{E}[B_{1}(r)^{2}K_{2,c^{*}}(r)K_{2,c^{*}}(s)]\,ds\,dr\\
 & +2\omega_{ee}^{*}\int_{r=0}^{1}\int_{s=0}^{r}(1-r)\mathbb{E}[J_{2,c^{*}}(r)J_{2,c^{*}}(s)]\,ds\,dr\\
={} & \frac{\omega_{ee}^{*}\omega_{vv}^{*}}{c^{*}}\int_{r=0}^{1}\int_{s=0}^{r}[e^{c^{*}(r+s)}-e^{c^{*}(r-s)}]\,ds\,dr+\frac{4\omega_{ev}^{*2}}{c^{*2}}\int_{r=0}^{1}\int_{s=0}^{r}(e^{c^{*}s}-1)(e^{c^{*}r}-1)\,ds\,dr\\
={} & \omega_{ee}^{*}\omega_{vv}^{*}\frac{2c^{*}+(1-e^{c^{*}})(3-e^{c^{*}})}{2c^{*3}}+2\omega_{ev}^{*2}\frac{(e^{c^{*}}-c^{*}-1)^{2}}{c^{*4}},
\end{align*}
where the third equality uses the fact that $B_{1}(r)^{2}-\omega_{ee}^{*}r$
is a martingale and the fourth equality uses \eqref{eq:E JJ} and
\eqref{eq:E B2KK}.

If $c^{*}=0$, we can perform the same calculations as above. A more
convenient way is to invoke the dominated convergence theorem to allow
interchanging expectation and limit.
\end{proof}

\end{document}